\documentclass[12pt]{dalthesis}
\usepackage{amsmath,amssymb,latexsym,graphicx,fancyhdr}
\usepackage[mathscr]{euscript}
\usepackage{algorithm}
\usepackage[noend]{algorithmic}
\usepackage{tabulary}
\usepackage{colortbl}
\usepackage{calc}
\usepackage{url}

\def\favoritefont{\bfseries \sffamily}
\makeatletter
\long\def\@makecaption#1#2{
   \vskip 10pt
   \setbox\@tempboxa\hbox{{\footnotesize {\favoritefont #1.} {\sffamily #2}}}
   \ifdim \wd\@tempboxa >\hsize         
       {\footnotesize {\favoritefont #1.} {\sffamily #2}\par}
     \else                              
       \hbox to\hsize{\hfil\box\@tempboxa\hfil}
   \fi}

\def\@begintheorem#1#2{\trivlist 
   \item[\hskip \labelsep{\favoritefont #1\ #2}]\itshape}

\def\@opargbegintheorem#1#2#3{\trivlist
   \item[\hskip \labelsep{\favoritefont #1\ #2}{\bf\ (#3)}]\itshape}

\def\QED{\ensuremath{{\Box}}}
\def\markatright#1{\leavevmode\unskip\nobreak\quad\hspace*{\fill}{#1}}
\newenvironment{proof}
	{\begin{trivlist}\item[\hskip\labelsep{\favoritefont Proof:}]}
	{\markatright{\QED}\end{trivlist}}

\newtheorem{theorem}{Theorem}
\newtheorem{lemma}[theorem]{Lemma}
\newtheorem{corollary}[theorem]{Corollary}
\newtheorem{obs}{Observation}

\newtheorem{definition}{Definition}
\newcommand{\qed}{}
\newcommand{\IR}{\ensuremath{\mathbb{R}}} 
 
\newcommand{\IN}{\ensuremath{\mathbb{N}}}

\newcommand{\lee}{\leqslant}
\newcommand{\gee}{\geqslant}
\newcommand{\ceil}[1]{{\left\lceil{#1}\right\rceil}}
\newcommand{\floor}[1]{{\left\lfloor{#1}\right\rfloor}}

\newcommand{\set}[1]{{\{ #1 \}}}
\newcommand{\seq}[1]{{\left< #1 \right>}}

\newcommand{\eps}{\varepsilon}
\newcommand{\bslash}{\!\setminus\!}
\newcommand{\bigOmega}{{\rm\Omega}}
\newcommand{\etal}{{\em et~al.\/}}
\newcommand{\REM}[1]{}

\newcommand{\eq}{{\ \leftarrow\ }}

\newcommand{\CF}{{\mathscr F}}
\newcommand{\CR}{{\mathscr R}}
\newcommand{\CI}{{\mathscr I}}
\newcommand{\CB}{{\mathscr B}}
\newcommand{\CP}{{\mathscr P}}
\newcommand{\CT}{{\mathscr T}}
\newcommand{\CS}{{\mathscr S}}
\newcommand{\CC}{{\mathscr C}}

\newcommand{\CH}{{\mathscr H}}

\newcommand{\Frechet}{Fr\'echet }

\newcommand{\distF}{\delta_F}
\newcommand{\distFS}{\delta_{\bar{F}}} 
\newcommand{\distWeakF}{\delta_{\bar{N}}} 
\newcommand{\distClosedF}{\delta_{\bar{C}}} 
\newcommand{\distPartialF}{\delta_{\bar{P}}} 
\newcommand{\distDisF}{\delta_{dF}} 
\newcommand{\distC}{\delta_C} 
\newcommand{\distSetF}{\delta_{\CF}} 

\newcommand{\SC}{slope-constrained }

\newcommand{\cell}[1]{{\CC_{#1}}}
\newcommand{\BNM}{\CB_{n \times m}}
\newcommand{\BNNM}{\CB_{2n \times m}}
\newcommand{\CO}{{\mathscr O}}
\newcommand{\COB}{{\bar{\CO}}}

\newcommand{\Feps}{\CF_\eps}

\newcommand{\LF}{L^\CF}
\newcommand{\BF}{B^\CF}
\newcommand{\LR}{L^\CR}
\newcommand{\BR}{B^\CR}
\newcommand{\LT}{L^{T}}
\newcommand{\BT}{B^{T}}

\newcommand{\OP}{\CT}

\newcommand{\Xmin}[2]{{#1}_{\min}(#2)}
\newcommand{\Xmax}[2]{{#1}_{\max}(#2)}
\newcommand{\umin}[1]{\Xmin{\bar{v}}{#1}}
\newcommand{\umax}[1]{\Xmax{\bar{v}}{#1}}
\newcommand{\vmin}[1]{\Xmin{v}{#1}}
\newcommand{\vmax}[1]{\Xmax{v}{#1}}

\newcommand{\minS}[1]{\mbox{minSlope}_{#1}}
\newcommand{\maxS}[1]{\mbox{maxSlope}_{#1}}

\newcommand{\entry}[1]{\mbox{entry}(\cell{#1})}
\newcommand{\exit}[1]{\mbox{exit}(\cell{#1})}
\newcommand{\proj}[1]{\pi_{#1}}

\newcommand{\Left}{\mbox{left}}
\newcommand{\Right}{\mbox{right}}

\newcommand{\union}{\mbox{\sc U}}
\newcommand{\lei}{\prec}

\newcommand{\Ho}[2]{\CH_{\overline{#1},\overline{#2}}}
\newcommand{\But}[2]{\CB_{\overline{#1},\overline{#2}}}
\newcommand{\Ov}[1]{\overline{#1}}
\newcommand{\PP}{{F}}
\newcommand{\SP}{\pi} 
\newcommand{\Lin}{\overleftrightarrow{cd}} 
\newcommand{\Dir}{\overrightarrow} 

\newcommand{\For}{{\bf for }}
\newcommand{\Foreach}{{\bf for each }}

\newcommand{\Do}{{\bf do }}

\newcommand{\row}[1]{{\CR_{#1}}}
\newcommand{\RSet}[2]{{\mathscr RS}_{#1}^{#2}} 

\newcommand{\fs}{free-space }
\newcommand{\FD}{\mathscr {FD}}
\newcommand{\FS}{\mathscr {FS}}
\newcommand{\provided}{{\ | \ }}
\newcommand{\lr}{\mbox{\sc Leftmost-Reachable}}
\newcommand{\rl}{\mbox{\sc Rightmost-Take-Off}}
\newcommand{\topp}{\text{top}}
\newcommand{\up}{\text{up}}
\newcommand{\nil}{\mbox{null}}

\newcommand{\bigTheta}{{\rm\Theta}}

\newcommand{\CD}{{\mathscr D}}

\newcommand{\Next}{\mbox{next}}

\newcommand{\Prev}{\mbox{prev}}
\newcommand{\CL}{{\mathbb L}}

\newcommand{\F}{\CF}
\newcommand{\R}{\CR}
\newcommand{\RE}{\mbox{\sc R}}

\newcommand{\Last}{\mbox{last}}
\newcommand{\reach}{\leadsto}

\newcommand{\lp}{\ell}
\newcommand{\rp}{r}

\newcommand{\lex}{\preceq}

\newcommand{\pset}{S}
\newcommand{\ap}{\oplus}
\newcommand{\gre}{{g}}
\newcommand{\sma}{{s}}

\newcommand{\Seg}[1]{{\overline{#1}}}
\newcommand{\ri}{r}

\newcommand{\CQ}{{\mathscr Q}}
\newcommand{\sq}{{\CS\CQ}}
\newcommand{\cfev}{{l}} 

\begin{document}

\title{
Applied Similarity Problems Using \Frechet Distance 
        \\}
\author{Kaveh Shahbaz}
\submitdate{Feb 15, 2011} \copyrightyear{2011}

\frontmatter


\pagestyle{fancy}
\fancyhf{}
\renewcommand{\headrulewidth}{0pt}
\fancyhead[LE,RO]{\thepage}
\renewcommand{\chaptermark}[1]{%
\markboth{\chaptername \ \thechapter.\ #1}{}}
\lhead{\nouppercase{\leftmark}}

\chapter {Introduction and Motivation}
\pagenumbering{arabic}
\setcounter{page}{1}
\label{ch:Background}


The problem of curve matching appears in a variety of different domains, like 
shape matching, GIS applications~\cite{AltERW03a, Buchin10, Appx-MM}, 
pattern recognition~\cite{FDRevisited,JiangXZ08}, 
computer vision~\cite{AltBook2009}, speech recognition~\cite{FDSpeech},  
time series analysis~\cite{FDTime},  
and signature verification~\cite{FDHandwriting,SriraghavendraKB07}.  The main questions
associated with curve matching in a specific domain
are: What is a good measure of similarity between curves? How can we compute it (or some approximation
of it) efficiently? Other questions that are often
of interest include: given a database of curves and a candidate
curve, can we find a nearest neighbor to this curve
in the database? Can we cluster curves with respect to
a given measure of similarity?

Curve matching has been studied extensively by computational 
geometers. 
The Hausdorff distance and the \Frechet distance are the most well-known distance measures 
to assess the resemblance of two curves
(see \cite{SomeOtherMetrics} for some other metrics such as 
the bottleneck distance, the volume of symmetric difference).
The Hausdorff distance between two curves $P$ and $Q$ is the smallest $\delta$,
such that $P$ is completely contained in the $\delta$-neighborhood of $Q$, and vice
versa. 
Although the Hausdorff distance is arguably a  natural distance
measure between curves and/or compact sets, it is too ”static”,
in the sense that it neither considers 
direction nor any dynamics of the motion along the curves
(see Figure \ref{fig:Haus}). The \Frechet distance deals with
this problem. It takes the order between points along 
the curves into consideration, making it a better 
measure of similarity for curves than alternatives 
such as the Hausdorff distance. 

The \Frechet distance was first defined by Maurice \Frechet in 1906~\cite{FirstFD}. 
While known as a famous distance measure in the field of mathematics 
(more specifically, abstract spaces), 
it was Alt and Godau~\cite{AltG95} who first applied it in measuring the similarity of polygonal curves in early 1990s. 

An intuitive way to understand the \Frechet metric is as follows:
imagine a man is walking his
dog, he is walking on one curve, the dog on the other. Both are allowed to
control their speeds, but are not allowed to go backwards. Then, the \Frechet
distance of the curves is the minimal length of a leash that is necessary.

Alt and Godau~\cite{AltG95} proposed an $O(n^2 \log n)$ time 
algorithm to compute the \Frechet distance, 
where $n$ is the total complexity of the curves. 
Since that time,
\Frechet metric has received much attention as a measure
of curve similarity and many variants have been studied. These include 
minimizing the \Frechet distance under various classes of transformations~\cite{AltTranslation,Mosig2005}, 
extending it to graphs~\cite{AltERW03a,VehicleTracking}, piecewise smooth curves~\cite{smoothFD}, simple polygons~\cite{Buchin2006}, surfaces~\cite{Alt2009Surface}, and
to more general metric spaces~\cite{WenkC08a,Chambers10,Cook2009},
in curve simplification~\cite{Agarwal2002}, 
protein structure alignment~\cite{JiangXZ08,FDRevisited} and morphing~\cite{GuibasNoCross}.


%


\begin{figure}[h]
	\centering
	\includegraphics[width=0.7\columnwidth]{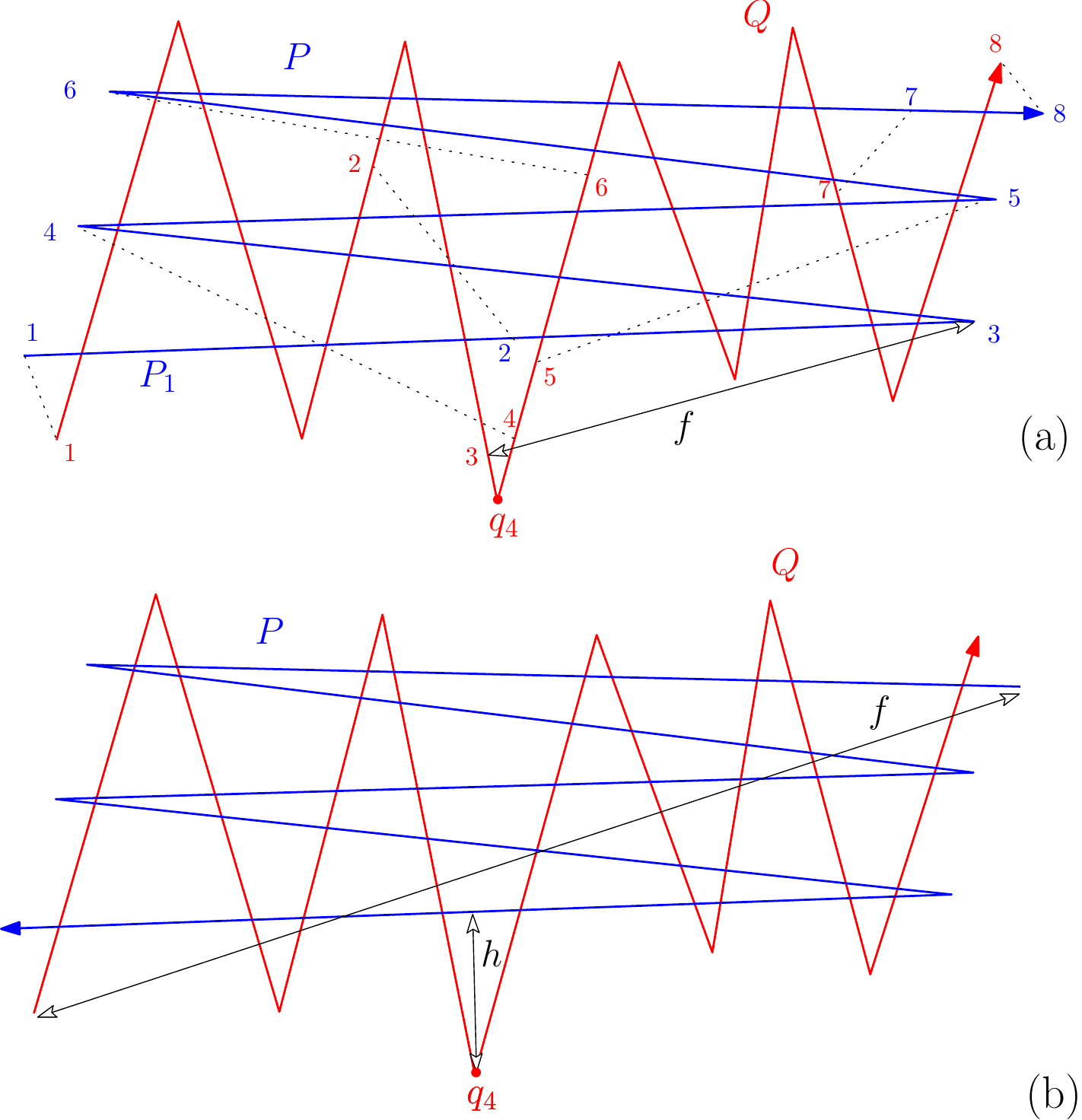}
	\caption{ Hausdorff vs. \Frechet distance. Shows two curves $P$ and $Q$ with small Hausdorff distance $h$ having a large \Frechet distance $f$.
	(a) The \Frechet distance is indicated by $f$. The Hausdorff distance is the distance from vertex $q_4$ to $P_1$. A sample walk is also shown with
a sequence of the locations of the moving objects.
(b) The direction of $P$ is reversed. The \Frechet distance is not same as before but the Hausdorff distance remains unchanged.
	Applet of Pelletier~\cite{FDApplet} is used to compute \Frechet distance.
	}
	\label{fig:Haus}
\end{figure}


\section{Contributions of the Thesis}

The main contributions of this thesis are summarized below: 
\begin{itemize}
\item We introduce a generalization of the well-known \Frechet 
distance between two polygonal curves which incorporates speed limits. 
We provide efficient 
algorithms for computing that metric~\cite{oursCCCG2009,oursSpeedJournal}.

\item We present an algorithm which computes 
the speed-constrained \Frechet distance when the
input curves are restricted to be inside a simple polygon \cite{oursCCCG2010}.

\item We introduce a new data structure called  
the \emph{free-space map} which can be used to solve several variants of 
\Frechet distance problems efficiently.
We improve algorithms for partial curve matching and
closed curve matching  using free-space map. 
We also obtain an improved algorithm for
the map matching algorithm of Alt \etal~\cite{AltERW03a}
for the case when the map is a directed acyclic graph.
We also solve \emph{minimum/maximum walk} problem efficiently using
our data structure~\cite{oursESA2011,oursPartialAlgorithmica}. 

\item We introduce the \emph{curve-pointset matching} problem and present 
an efficient algorithm to solve it~\cite{oursCCCG2011}. 

\item We provide NP-completeness proof of \emph{all-points curve-pointset matching} problem.

\end{itemize}

\section{Organization of the Thesis}

This thesis is organized as follows. In the next chapter, 
we describe the algorithm due to Alt and Godau \cite{AltG95} 
for computing the \Frechet distance. In addition, we summarize different variants of \Frechet distance problem 
which have been studied and
describe briefly the techniques for solving them.

Next, in Chapter \ref{ch:speedFD},
we introduce a new generalization of 
\Frechet distance and provide an efficient algorithm for computing it. 
The classical \Frechet distance between two polygonal curves corresponds to the 
maximum distance between two point objects that traverse the curves with arbitrary non-negative speeds.
Here, we consider a problem instance in which the speed of traversal 
along each segment of the curves is restricted to be within a specified range. 
We provide an efficient algorithm that  
decides in $O(n^2\log n)$ time whether the \Frechet distance with speed limits
between two polygonal curves is at most $\eps$,
where $n$ is the number of segments in the curves, and
$\eps \gee 0$ is an input parameter. 
We then use our solution to this decision problem 
to find the exact \Frechet distance with speed limits
in $O(n^3 \log n)$ time.

Given two polygonal curves inside a simple polygon,
in Chapter \ref{ch:speed-geodesic}, we study the problem of finding the 
\Frechet distance between the two curves
under the following two conditions
(i) the distance between two points on the curves is measured as the length of the shortest path
between them lying inside the simple polygon, and 
(ii) the traversal along each segment of the polygonal curves 
is restricted to be between a minimum and a maximum permissible speed
assigned to that segment.%
We provide an algorithm that decides in $O(n^2 (k + n))$ time 
whether the speed-constrained geodesic \Frechet distance between two polygonal curves inside a 
simple polygon is within a given value $\eps$, 
where $n$ is the number of segments in the curves, and $k$ is the complexity of the polygon. 

In Chapter \ref{ch:partial}, 
we propose a new data structure, \emph{free-space map}, 
that enables us to solve several variants of the \Frechet distance problem efficiently.
Our data structure encapsulates all the
information available in the 
free-space diagram of Alt and Godau~\cite{AltG95} to compute the
\Frechet distance. 
In addition, our data structure is capable of answering
more general type of queries than 
the free-space diagram. 
Given that the free-space map has the same size  and 
construction time 
($O(n^2)$, $n$ is the total complexity of the curves)
as the standard free-space diagram,
 it can be viewed as a powerful alternative.
 
Using our new data structure,
we present improved algorithms 
for several variants of
the \Frechet distance problem.
In particular, we improve the 
$O(n^2\log^2n)$ time algorithm 
for computing the partial \Frechet distance in \cite{AltG95}, 
by a $\log n$ factor. 
Also, we obtain improved algorithms 
for computing \Frechet distance between 
two closed curves, and the so-called \emph{minimum/maximum walk} problem. 
Our data structure leads to efficient 
algorithms for
the map matching algorithm of Alt \etal 
\cite{AltERW03a}
for the case when the map is a directed acyclic graph.

In Chapter \ref{ch:StayClose},
we examine the following variant
of the \Frechet distance problem, 
which we refer to as the \emph{Curve-Pointset Matching (CPM)} problem. 
Given a pointset $\pset$ of size $k$ and a polygonal 
curve $P$ of size $n$ in $\IR^d$,
we study the problem of finding a polygonal curve $Q$ whose vertices are from  $\pset$,	
and has a minimum \Frechet distance to $P$. 
In the decision version of that problem, 
given a distance $\eps \ge 0$, we 
present an $O(nk^2)$ time algorithm 
to decide if exists a curve $Q$ through some points of $\pset$
in $\eps$-\Frechet distance to curve $P$, 
where vertices of $Q$ are from $\pset$, 
and curve $Q$ need not contain all points of $\pset$ 
and may use a point of $S$ multiple times. 
Also, we show that the
curve of minimum \Frechet distance can be computed in
time $O(nk^2 \log(nk))$.
As a by-product of our result, 
we improve the map matching algorithm of Alt \etal\ \cite{AltERW03a}
by a $\log k$ factor for the case when the map is a complete graph.
Finally, in Chapter \ref{ch:NP-Complete}, we 
study the same problem as in the previous chapter, 
under the new condition that 
curve $Q$ must visit every point in the pointset $\pset$.
We refer to this problem 
as All-Points CPM problem and 
we show that it is NP-complete. 

\REM{
Finally, in Chapter \ref{ch:SpecialCase},
we study some special cases of  All-Points CPM problem, 
where  instead of a general curve, the input to the problem is 
a convex $xy$-monotone curve, 
convex $x$-monotone curve (or convex $y$-monotone curve) or
a convex polygon. 
}


\chapter{Related Work}
\label{ch:related}

\section{Classical \Frechet Distance Problem }
\label{sec:classicalFD}
The \Frechet distance is a metric to measure the similarity of polygonal curves. It was first defined by a French mathematician,
Maurice \Frechet \cite{FirstFD}.
The \Frechet distance between two curves is 
often referred to as a dog-leash distance because
it can be interpreted as the minimum-length leash required for a person to walk a dog,
if the person and the dog, each travels from its respective starting position to its ending position, without ever letting go off the leash or backtracking.
The length of the leash determines how similar the two curves are to each other:
a short leash means the curves are similar,
and a long leash means that the curves are different from each other.

Two problem instances naturally arise:  decision and optimization.
In the {\em decision problem}, one wants to decide whether two polygonal curves $P$  and $Q$
are within $\eps$ \Frechet distance from each other, i.e., if a leash of given length $\eps$ suffices.
In the {\em optimization problem}, one wants to determine the minimum such $\eps$.
In~\cite{AltG95}, Alt and Godau gave an $O(n^2)$ 
time algorithm for the decision problem,
where $n$ is the total number of segments in the curves.
They also solved the corresponding optimization problem in $O(n^2\log n)$ time.
Here, we first describe their decision algorithm:

\paragraph{Polygonal Curve (or Polyline).}
A {\em polygonal curve\/} in $\IR^d$ is  a continuous function 
$P:[0,n] \rightarrow \IR^d$ with $n \in \IN$, 
such that for each $i \in \set{0, \ldots, n-1}$,
the restriction of $P$ to the interval $[i, i+1]$ 
is affine (i.e., forms a line segment).
The integer $n$ is called the {\em length\/} of $P$.
Moreover, the sequence ${P(0), \ldots, P(n)}$ represents the set of {\em vertices\/} of $P$.
For each $i \in \set{1, \ldots, n}$, 
we denote the line segment $P(i-1)P(i)$ by $P_i$.


\paragraph{\Frechet Distance.}

A {\em monotone parametrization} of $[0,n]$ 
is a continuous non-decreasing function $\alpha: [0,1] \rightarrow [0,n]$
with $\alpha(0)=0$ and $\alpha(1)=n$.
Given two polygonal curves $P$ and $Q$ of lengths $n$ and $m$ respectively, 
the {\em \Frechet distance\/} between $P$ and $Q$ is defined as
\[
	\distF(P,Q) = \inf_{\alpha, \beta} \max_{t \in [0,1]} d( P(\alpha(t)), Q(\beta(t)) ),
\]
where $d$ is the Euclidean  distance, and $\alpha$ and $\beta$ range over all monotone parameterizations of 
$[0,n]$ and $[0,m]$, respectively.

\paragraph{Free-Space Diagram.}
To compute the \Frechet distance, a way of representing all possible 
person and dog walks is needed.
Let $\BNM = [0,n] \times [0,m]$ be an $n$ by $m$ rectangle  in the plane.
Each point $(s,t) \in \BNM$ uniquely represents a pair of points
$(P(s),Q(t))$ on the polygonal curves $P$ and $Q$.
We decompose $\BNM$ into
$n\cdot m$ unit grid cells $\cell{ij} = [i-1,i] \times [j-1,j]$
for $(i,j) \in \set{1, \ldots, n} \times \set{1, \ldots, m}$,
where each cell $\cell{ij}$ corresponds to
a segment $P_i$ on $P$ and a segment $Q_j$ on $Q$.
Given a parameter $\eps \gee 0$,
the {\em free space\/} $\Feps$ is defined as
\[
	\Feps = \set{(s,t) \in \BNM \ | \ d(P(s),Q(t)) \lee \eps }.
\]
We call any point $p \in \Feps$ a {\em feasible\/} point.
An example of the free-space diagram for two curves $P$ and $Q$ 
is illustrated in Figure~\ref{fig:diagram}.a.
The free-space diagram was first used in~\cite{AltG95}
to find the standard \Frechet distance in near quadratic time.
Consider any segment $P_i$ and $Q_j$ from polygonal curves $P$ and $Q$, 
respectively. Then, the free space inside cell $\cell{ij}$ is convex 
and can be determined in $O(1)$ time by computing the intersection of a unit square and an ellipse~\cite{AltG95}.

\begin{figure}[h]
	\centering
	\includegraphics[width=0.60\columnwidth]{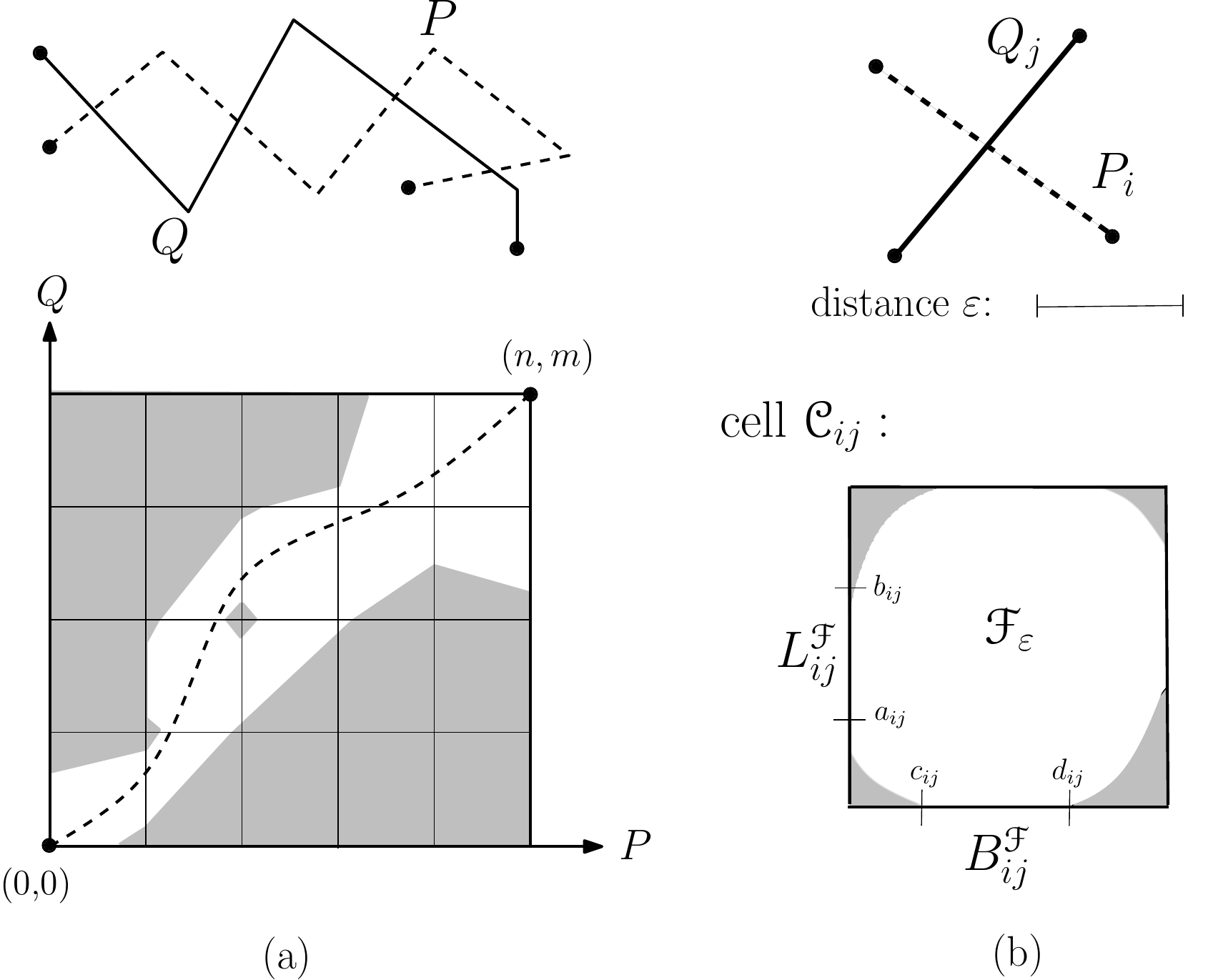}
	\caption{ (a) The free-space diagram for two polygonal curves $P$ and $Q$;
	(b) two segments $P_i$ and $Q_j$ and their corresponding free space.
	The diagram was generated using a Java applet developed by S. Pelletier~\cite{FDApplet}.}
	\label{fig:diagram}
\end{figure}


	Here, we show how the function of such an ellipse is computed:
	Let the coordinates of the endpoints of $P_i$ be: $(p_{1a}; p_{1b}); (p_{2a}; p_{2b})$
	and the coordinates of
	the endpoints of $Q_j$ be: $(q_{1a}; q_{1b}); (q_{2a}; q_{2b})$. 
	Let $P_i$ be defined by the function: $y = a_1x + b_1$
	and $Q_j$ be defined by the function: $y = a_2x + b_2$. Then: 
	\begin{displaymath}
	a_1 = \frac{ p_{2b} - p_{1b} }{ p_{2a} - p_{1a} }, b_1 = \frac{p_{2a}p_{1b} - p_{1a}p_{2b}}		{p_{2a} - p_{1a}}
	\end{displaymath}
	\begin{displaymath}
	a_2 = \frac{ q_{2b} - q_{1b} }{ q_{2a} - q_{1a} }, b_2 = \frac{q_{2a}q_{1b} - q_{1a}q_{2b}}{q_{2a} - q_{1a}}
	\end{displaymath}
	The points located in the 2-dimentional coordinate system of the parametrization of $P_i$
	and $Q_j$ satisfy: $x' \in [0,1], y' \in [0,1]$, thus the coordinate of points in $P_i$ are :

\abovedisplayshortskip=0pt
\belowdisplayshortskip=0pt
\abovedisplayskip=0pt
\belowdisplayskip=0pt

	\begin{displaymath}
	( (p_{2a} - p_{1a}) x' + p_{1a}, (p_{2b} - p_{1b})x' + p_{1b} ).
	\end{displaymath}
	The coordinates of the points in $Q$ are: 
	\begin{displaymath}
	( (q_{2a} - q_{1a}) y' + q_{1a}, (q_{2b} - q_{1b})y' + q_{1b} ).
	\end{displaymath}
	Every point in the free space inside $\cell{i,j}$, corresponds to exactly two points, one from $P_i$
	the other from $Q_j$ where their distance is less than or equal parameter $\epsilon$. Therefore, 
	\begin{displaymath}
	[ (p_{2a} - p_{1a})x' + p_{1a} - (q_{2a}-q_{1a}) ]^2 + [ (p_{2a} - p_{1b}) x' + p_{1b} - (q_{2b}-q_{1b})y'-q_{1b} ]^2 \le \epsilon^2
	\end{displaymath}
	by expanding the above inequality and considering the fact that :
	\begin{displaymath}
	(p_{2a}-p_{1a})^2 + (p_{2b} - p_{1b})^2 = |P_i|^2 ; (q_{2a} - q_{1a})^2 + (q_{2b}-q_{1b})^2 = |Q_i|^2.
	\end{displaymath}
	We derive the function of the ellipse as follows:
	\begin{eqnarray*}
	\lefteqn{|P_i|^2 x'^2 + |Q_j|^2y'^2 - 2 [(p_2a-p_1a)(q_2a-q_1a) + (p_2b-p_1b)(q_2b-q_1b)] x'y'} \\
	& & \mbox{}+2[(p_{1a}-q_{1a})(p_{2a}-p_{1a}) + (p_{1b}-q_{1b})(p_{2b}-p_{1b})]x'  \\
	& & \mbox{}-2[(p_{1a} - q_{1a})(q_{2a}-q_{1a}) + (p_{1b}-q_{1b})(q_{2b}-q_{1b})]y'  \\
	& & \mbox{}+(p_{1a}-q_{1a})^2 + (p_{1b} - a_{1b})^2 \le \epsilon ^2 
	\end{eqnarray*}	
	Since an ellipse is a convex  shape and the unit 
	square in the free-space diagram is convex too, 
	the intersection of two convex objects is convex and therefore, 
the free space inside each cell is convex. 
In addition, Alt and Godau~\cite{AltG95} observed that any $xy$-monotone 
path from $(0,0)$ to $(n,m)$ in the free space corresponds to traversals of $P$ and $Q$, 
where the traversing objects remain at a distance of at most $\eps$ from each other.

Based on the above observations, 
Alt and Godau~\cite{AltG95} provided an algorithm
to solve the decision problem (i.e., decide if $\distF(P,Q) 
\lee \eps$ for a given $\eps \gee 0$) in quadratic time as follows:

Let $L_{ij}$ (resp., $B_{ij}$) denote the left (bottom, resp.) line 
segment bounding $\cell{ij}$ (see Figure~\ref{fig:diagram}.b).
As a preprocessing step, the free space, $\Feps$, is computed by the algorithm.
Let $\LF_{ij} = L_{ij} \cap \Feps$ and $\BF_{ij} = B_{ij} \cap \Feps$ 
(see Figure~\ref{fig:diagram}.b).
Since $F_\eps$ is convex within $\cell{ij}$,
each of $\LF_{ij}$ and $\BF_{ij}$ is a line segment. 
The preprocessing step therefore involves computing line segments
$\LF_{ij}$ and $\BF_{ij}$ for all feasible pairs $(i,j)$, 
which can be done in $O(n^2)$ time.
A point $(s,t) \in \Feps$ is called {\em reachable\/} 
if there is a monotone path from $(0,0)$ to $(s,t)$ in $\Feps$.
Let $\LR_{ij}$ be the set of reachable points in $L_{ij}$,
and $\BR_{ij}$ be the set of reachable points in $B_{ij}$.
Observe that all non-empty sets $\LR_{ij}$ and $\BR_{ij}$ for each cell $\cell{ij}$ forms line 
segment~\cite{AltG95}. The algorithm processes the cells in the row-wise order, from $\cell{0,0}$ to $\cell{nm}$, 
and at each cell $\cell{ij}$, $\LR_{ij}$ and $\BR_{ij}$ are computed. 
Finally, at the last cell, if the top-right corner of $\BNM$ is reachable, 
``YES" is returned as the answer to the decision problem, otherwise ``NO" is returned.
Details are shown in Algorithm~\ref{alg:StandardFDec}. 
Given polygonal curves $P$ and $Q$ with total complexity $n$, 
Algorithm~\ref{alg:StandardFDec} decides in $O(n^2)$ time 
if $\distF(P,Q) \lee \eps$~\cite{AltG95}.


\begin{algorithm} [h]
\caption {\sc Standard \Frechet Decision Algorithm~\cite{AltG95} } \label{alg:StandardFDec}

\algsetup{indent=1.5em}
\begin{algorithmic}[1]
	\vspace{0.5em}
	\baselineskip=1\baselineskip

		\FOR { each cell $\cell{ij}$ } 
 
 		\STATE Compute $\LF_{ij}$  and $\BF_{ij}$
\ENDFOR

	\STATE Set $\LR_{0,0} = \BR_{0,0} = \set{(0,0)}$, \ 
		$\LR_{i,0} = \emptyset$ for $i \in \set{1, \ldots, n}$, \ 
		$\BR_{0,j} = \emptyset$ for $j \in \set{1, \ldots, m}$  
	\FOR {$i = 0$ to $n$} 
	 	\FOR {$j = 0$ to $m$}
			\STATE Compute  $\LR_{i+1,j}$ and $\BR_{i,j+1}$ from  $\LR_{i,j}$, $\BR_{i,j}$,  $\LF_{i+1,j}$ and $\BF_{i,j+1}$.
		\ENDFOR
	\ENDFOR
	\STATE\label{line:last} Return ``{\sc yes}" if $(n,m) \in \LR_{n+1,m}$, ``{\sc no}" otherwise. 

\end{algorithmic}
\end{algorithm}

The algorithm proposed by Alt and Godau for actually computing the \Frechet distance $\distF$ 
makes use of Algorithm \ref{alg:StandardFDec}, and the technique of parametric search of Megiddo~\cite{Megiddo83}, 
accompanied by a speedup technique due to Cole~\cite{Cole87}. 
The resulting algorithm has time complexity $O(n^2 \log n)$.

Let $\LF_{ij} = [a_{ij}, b_{ij}]$ and  $\BF_{ij} = [c_{ij}, d_{ij}]$ (see Figure~\ref{fig:diagram}.b). 
Notice that the free space, $\Feps$, is an increasing function of $\eps$, 
that is, for $\eps_1 \lee \eps_2$, we have $\CF_{\eps_1} \subseteq \CF_{\eps_2}$.
Therefore, to find the exact value of $\distF(P,Q)$,
we can start from $\eps = 0$, and continuously increase $\eps$ until
we reach the first point at which $\Feps$ contains a monotone path from $(0,0)$ to $(n,m)$.
This occurs at only one of the following ``critical values''~\cite{AltG95}:
\begin{itemize} \itemsep1pt
	\item[(A)] smallest $\eps$ for which $(0,0) \in \Feps$ or $(n,m) \in \Feps$. 
	These are the distances between starting point and endpoints of $P$ and $Q$.

	\item[(B)] smallest $\eps$ at which $\LF_{ij}$ or $\BF_{ij}$  becomes non-empty for some
 	pair $(i,j)$ (when a new passage opens between two adjacent cells in the diagram). 
	These are the distances between vertices of one curve and edges of the other (see Figure \ref{fig:TypeBC}a).
	
	\item[(C)] smallest $\eps$ at which $a_{ij} =b_{k\ell}$, or 
	$d_{ij} = c_{k\ell}$, for some $i,j,k$, and $\ell$, 
	(when a new horizontal or vertical passage opens within the diagram). 
	A critical distance of type (C) corresponds to the common 
	distance of two vertices of one curve 
	to the intersection point of their bisector 
	with an edge of the other curve~\cite{AltG95} (see Figure \ref{fig:TypeBC}b).
\end{itemize}

\begin{figure}[t]
	\centering
	\includegraphics[width=0.90\columnwidth]{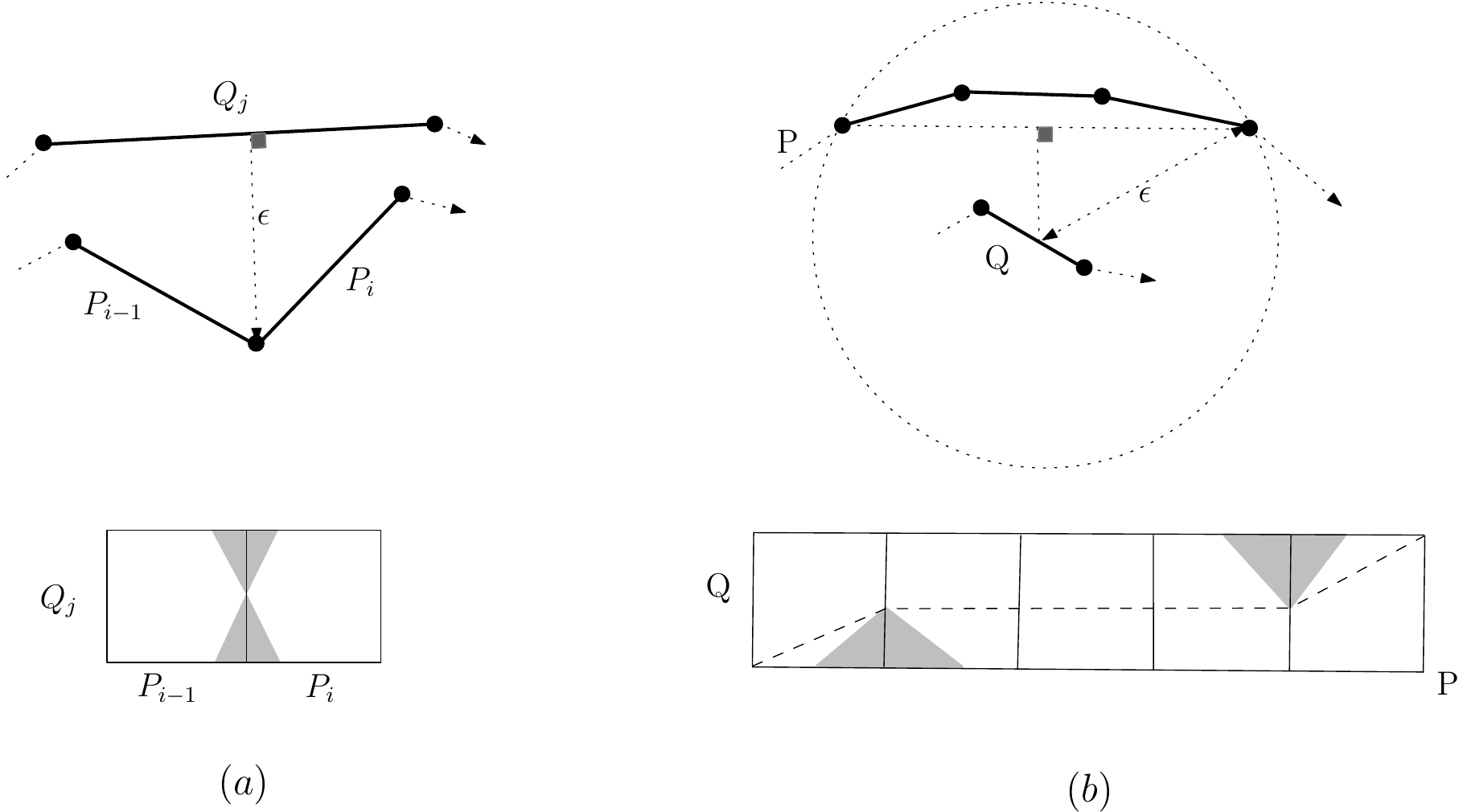}
	\caption{The geometric situations corresponding to Type (B) and Type (C)
 critical distances. 
(a) a new passage opens between two neighboring cells in the free-space diagram (b) a horizontal passage opens in the free-space diagram. }
	\label{fig:TypeBC}
\end{figure}

There are two critical values of type (A), 
$O(n^2)$ critical values of type (B), and $O(n^3)$ critical values of type (C),
each computable in $O(1)$ time.
Therefore, to find the exact value of $\distFS(P,Q)$,
one can compute all these $O(n^3)$ values, sort them, 
and do a binary search (equipped with Algorithm~\ref{alg:StandardFDec})
to find the smallest $\eps$ for which $\distF(P,Q) \lee \eps$,
in $O(n^3 \log n)$ total time.
However, as mentioned in~\cite{AltG95},
a parametric search method~\cite{Megiddo83,Cole87} can be applied
to the critical values of type (C) to get a faster algorithm.

The crucial observation made in~\cite{AltG95} is that
any comparison-based sorting algorithm that sorts
$a_{ij}, b_{ij}, c_{ij}$, and $d_{ij}$ (defined as functions of $\eps$)
has critical values that include those of type (C).
This is because the critical values of type (C) occur if 
$a_{ij} = b_{k\ell}$ or $d_{ij} = c_{k\ell}$, 
for some $i,j,k$, and $\ell$. Thus, Algorithm ~\ref{alg:CompFD}, 
uses parametic search to find the exact value of \Frechet distance.

\begin{algorithm} [t]
\caption {\sc Standard \Frechet Computation Algorithm~\cite{AltG95} } \label{alg:CompFD}
\algsetup{indent=1.5em}
\begin{algorithmic}[1]
	\vspace{0.5em}
	\baselineskip=0.9\baselineskip

\STATE Compute all critical values of types (A) and (B), and sort them.
	
	\STATE Binary search to find two consecutive values  $\eps_1$ and $\eps_2$ in  the sorted list
	such that $\distF \in [\eps_1,\eps_2]$.
	
	\STATE Let $S$ be the set of endpoints $a_{ij}, b_{ij}, c_{ij}$, $d_{ij}$ of intervals  
	$\LF_{ij}$ and $\BF_{ij}$ that are nonempty for $\eps \in [\eps_1,\eps_2]$.
	Use Cole's parametric search method~\cite{Cole87}  based on sorting the values in $S$ 
	to find the exact value of $\distF$.

\end{algorithmic}
\end{algorithm}

Steps 1 and 2 together take $O(n^2\log n)$ time. 
The parametric search in Step 3 takes $O(( k + T ) \log k)$ time,
where $k$ is the number of values to be sorted, 
and $T$ is the time needed by the decision algorithm.
In case of the standard \Frechet distance problem, $k = |S| = O(n^2)$, and $T = O(n^2)$.
We conclude that the exact \Frechet distance between two polygonal curves 
can be computed in $O(n^2 \log n)$ time~\cite{AltG95}.


\section{Variants of \Frechet Distance}

In this section, we summarize different variants 
of \Frechet distance metric
which have been studied in the literature.

\subsection{Weak \Frechet Distance}

One of the variants of the \Frechet metric studied in~\cite{AltG95} is the weak \Frechet
distance or non-monotone \Frechet distance. Coming back to the man-dog illustration 
of the \Frechet metric, in this instance,  
both the man and the dog are allowed to 
backtrack on their respective curves.

Let $\distWeakF(P,Q)$ denote the weak \Frechet distance between two polygonal curves $P$
and $Q$. In order to solve the decision and optimization problems, 
the same $m\times n$-diagram $\BNM$
can be used as in the previous section. 
Now the decision problem has a yes answer iff 
there exists a path from $(0,0)$ to $(m,n)$ in $\Feps$ which is not necessarily monotone~\cite{AltG95}. 
To solve the decision problem, an undirected labeled graph, $G=(V,E)$, is constructed 
on top of $\BNM$ as follows:

For each cell in the diagram, one node is added to $V$; two additional nodes $s$ and $t$
are added to the graph, where node $s$ represents point $(0,0)$
and node $t$ represents point $(m,n)$.
Two nodes are connected via an edge in the graph 
if their corresponding cells are adjacent in the diagram.
Furthermore, one edge connects node $s$ (resp., node $t$)  
to the node which corresponds to cell $\cell{11}$ 
(resp., cell $\cell{mn}$) as  depicted in Figure \ref{fig:weak}. 
The edge between two neighboring cells is labeled with a minimal $\epsilon$
for which there is a possible direct transition between the two cells within $\Feps$.
The edge $\{s,\cell{11} \}$ is labeled with the distance between starting points of the curves
and the edge  $\{s,\cell{11} \}$ is labeled with the distance between ending points of the curves. 
Let the weight of a path within $G$ be the largest weight of its edges. After constructing graph $G$, the decision problem has a positive answer 
iff there exists a path of weight $\epsilon$ between $s$ and $t$ within graph $G$.
This can be done by removing all edges of weight greater than $\epsilon$, 
and checking if $s$ and $t$ are in the same connected component, 
for example, by running BFS algorithm~\cite{AltG95}.

\begin{figure}[t]
	\centering
	\includegraphics[width=0.50\columnwidth]{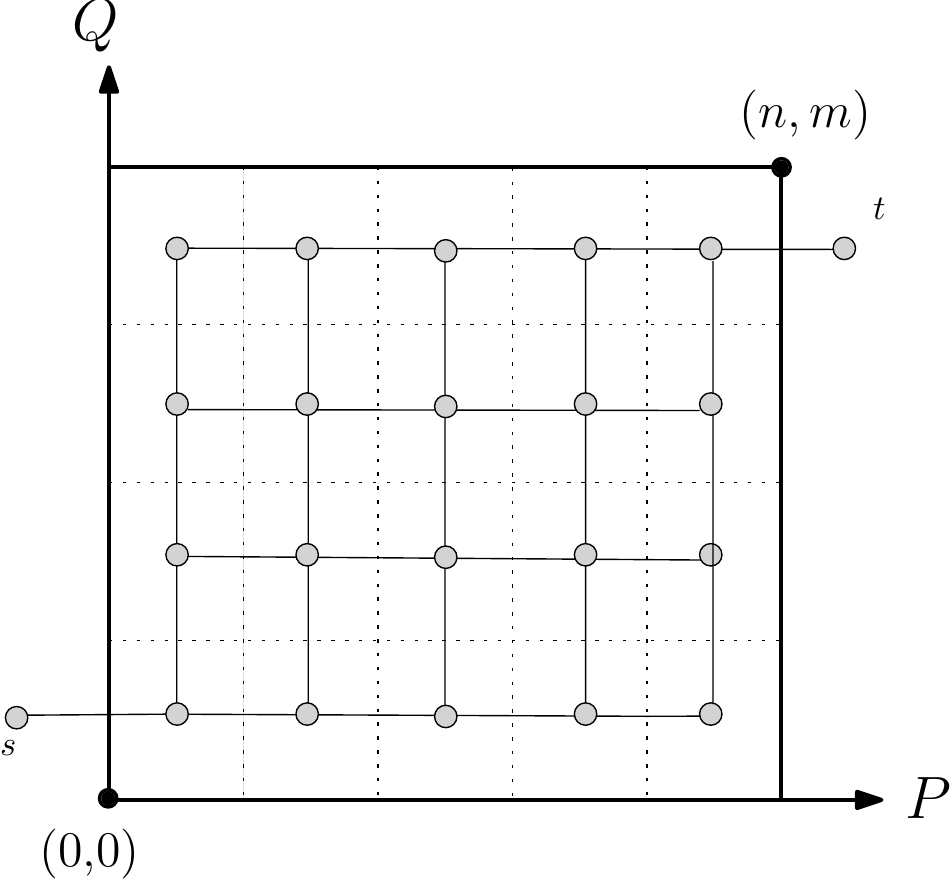}
	\caption{ The graph with grey nodes is built on top of $\BNM$ to compute the weak \Frechet distance}
	\label{fig:weak}
\end{figure}

The computation of the exact value of $\distWeakF(P,Q)$ consists of 
determining the minimum weight path within graph $G$ from $s$ to $t$.
This can be computed by using Prim's minimum spanning tree 
algorithm starting from $s$ and running it until the minimum spanning tree 
containing $s$ and $t$ is found. 
After finding 
the MST, one can run breadth first search algorithm to find 
a path from $s$ to $t$ in MST.  
We conclude that 
given two polygonal curves $P$ and $Q$ with total length $n$ and a distance $\epsilon$,
one can decide in $O(n^2)$ time if $\distWeakF(P,Q) \le \epsilon$ and 
the exact value of $\distWeakF(P,Q)$ can be found in $O(n^2 \log n)$~\cite{AltG95}.

\subsection{\Frechet Distance of a Set of Curves}
Dumitrescu \etal~\cite{SetofCurves} have extended the \Frechet distance 
notion between two curves to a set of curves and showed how to compute and approximate it .
The corresponding intuitive illustration is as follows. 
Suppose that points are moving, one on each of	given curves. 
The speed of each point may vary but no point is allowed to move backwards. 
Assume that all pairs of points are connected by strings of the same length. 
Then, the \Frechet distance of the set of curves is the minimum length of a connecting string that is necessary.

To compute the \Frechet distance of a set of $m$ 
curves $f_1,f_2,...,f_m$ (with complexity $n_1,n_2,...,n_m$, respectively), 
the approach of ~\cite{AltG95} can be adapted. 
First, a free-space diagram corresponding to each pair of curves is built. 
To answer the decision problem, one would
 check whether 
there exists a path from  $(0,...,0)$ to $(n_1,...,n_m)$ in free-space
diagram in $\IR^m$ which is monotone in all $m$ coordinates.
This  takes $O(n_1 ... n_m)$ time and  using parametric search, the resulting final algorithm
has time complexity $O(n_1...n_m\log(n_1..n_m))$.
In~\cite{SetofCurves}, a simple algorithm 
is proposed which computes 
the \Frechet distance of set of curves (i.e., $\distSetF$) approximately. 
Let $d_{ij} = \distF(f_i,f_j)$. Observe that 
$\distSetF \le \min_{1 \le i \le m} \max_{1\le j <k \le m} (d_{ij} + d_{ik})$~\cite{SetofCurves}.
Thus, one can compute all pairwise \Frechet distances and 
output  $\min_{1 \le i \le m} \max_{1\le j <k \le m} (d_{ij} + d_{ik})$ as
the \Frechet distance of a set of curves with the approximation ratio 2~\cite{SetofCurves}. The running time of this approach is $O(\sum_{1\le i <j \le m} n_i n_j \log (n_in_j))$ which is much better than that of the exact algorithm previously mentioned.

Alt \etal~\cite{AltTranslation} consider the problem of 
minimizing the \Frechet distance under translations:
Given two polygonal curves, search for a translation which, when
applied to the first curve, minimizes the \Frechet distance to the second one.
The decision algorithm decides whether there is a transformation
that, when applied to the first curve, results in a \Frechet 
distance less or equal than some given parameter $\epsilon$.
The runtime of the decision algorithm is $O( (mn)^3(m+n)^2)$. 
The parametric search adds only a logarithmic overhead,
since Cole's technique for parametric search based on sorting~\cite{Cole87} can be applied, so
the optimization problem can be solved in $O((mn)^3(m + n)^2 \log(m + n))$ time. 
In~\cite{JiangXZ08}, the authors present  algorithms for matching two polygonal chains 
in two dimensions to minimize their discrete \Frechet distance under 
translation and rotation.

%

\subsection{Average \Frechet and Summed \Frechet Distance }


Notice that the \Frechet metric is a max measure; 
it is defined as the maximum pointwise distance minimized over all parametrizations. 
This dependence on the maximum value can often lead to non-robust behavior, where small variations in
the input can distort the distance function by a large
amount. 
Consider for example the curves shown in Figure \ref{fig:FD-sum}. 
Assume one wants to match
the curve $f_2$ either to the curve $f_1$ or $f_3$. 
Intuitively, it seems that $f_3$ is the better
match. 
This is however not reflected by the \Frechet distance which is equal for 
both pairs of curves $(f_1, f_2)$ and $(f_2, f_3)$.
An average \Frechet distance 
was suggested in~\cite{VehicleTracking} which averages over certain 
distances instead of taking the maximum.
Efrat \etal~\cite{SumFD}  has combined {\em dynamic time warping} to compute an integral  version of the \Frechet distance, 
which can “smooth out” the impact of some outliers.
Dynamic time warping measure (DTW) is a measure which was
first proposed in the 60s as a measure of speech signal similarity.  
In the dog-man setting,
the DTW distance between two curves (defined as sequences
of points) is the sum of the leash lengths measured
at each (discrete) position (minimized over all
trajectories). 
%


\begin{figure}[t]
	\centering
	\includegraphics[width=0.65\columnwidth]{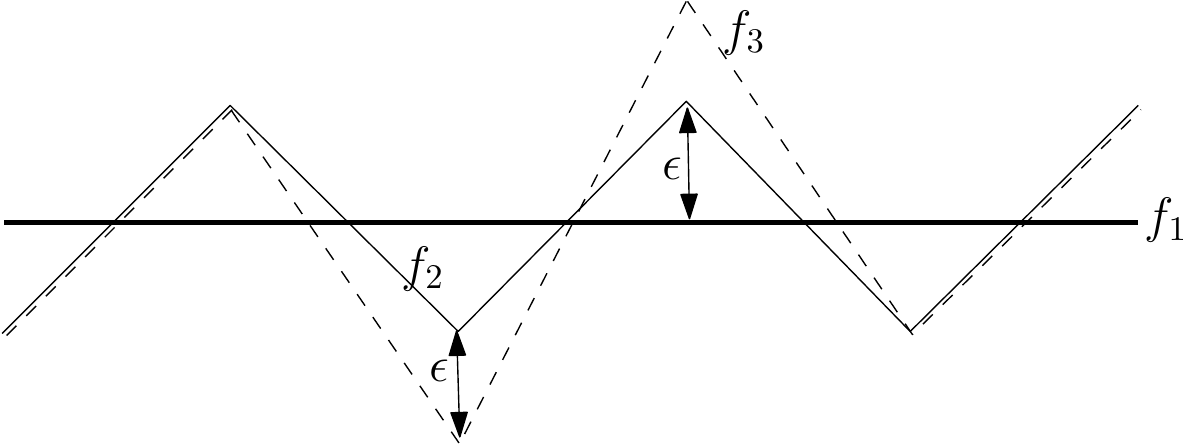}
	\caption{ $\distF(f_1,f_2) = \distF(f_2,f_3)$. But curve $f_2$ is more matched to curve $f_3$.	}
	\label{fig:FD-sum}
\end{figure}



\subsection{\Frechet Distance of Specific Families of Curves}

It has been an open problem to find a sub-quadratic algorithm 
for computing the \Frechet distance.
A lower bound of $\Omega(n \log n)$ is given by~\cite{LowerBound-FD}
for the problem of deciding whether the \Frechet 
distance between two curves is smaller 
or equal a given value. 
In~\cite{AltBook2009}, Alt conjectured that 
the decision problem may be 3SUM-hard~\cite{ClassOfN2}. 
In a very recent work~\cite{BuchinFasterFrechet},
Buchin~\etal~present an algorithm 
with total expected time
$O(n^2\sqrt{\log n}(\log \log n )^ {3/2})$
which is the first algorithm to achieve 
a running time of $o(n^2\log n)$ for computing \Frechet distance.
Furthermore, they 
show that there exists an algebraic decision tree 
for the decision problem of depth $O(n^{2-\gamma})$, 
for some $\gamma>0$. This 
provides some insights that might suggest
that the decision problem may not be 3SUM-hard.

The only subquadratic algorithms known are for
quite restricted classes of curves such as 
for closed convex curves and for $k$-bounded curves~\cite{Alt2003Comparison}, 
or discrete version of \Frechet distance \cite{DiscreteSUBQuadractic}. 
For a curve to be $k$-bounded means that for any two points on the curve
the portion of the curve in between them cannot be further away from either point than $\frac{k}{2}$
times the distance between the two points.  For closed convex curves the \Frechet distance
equals the Hausdorff distance and for $k$-bounded curves the \Frechet distance is at most
$(1 + k)$ times the Hausdorff distance, and hence the $O(n \log n)$ algorithm for the Hausdorff
distance applies.

The \Frechet distance of another family of curves, {\em $c$-packed} curves, 
is studied in a recent work by Driemel \etal~\cite{c-packedFD}.
A curve $P$ is called $c$-packed if the total length of $P$ inside any 
circle is bounded by $c$ times the radius of the circle. Intuitively, the constant $c$
measures how unrealistic the input curves are.

A $k$-bounded curve might have arbitrary length while maintaining a finite diameter, 
and as such may not be $c$-packed.
Unlike $k$-bounded curves, the \Frechet
distance between two $c$-packed curves might be arbitrarily 
larger than their Hausdorff distance. Indeed, $c$-packed
curves are considerably more general and a more natural
family of curves. For example, a $c$-packed curve might self cross and revisit the same location
several times, and the class of $c$-packed curves is closed under concatenation, 
none of which is true for $k$-bounded curves. Given two $c$-packed curves $P$ and $Q$
with total complexity $n$, a $(1+\epsilon)$-approximation of the \Frechet distance between them 
can be computed in $O(\frac{n}{\epsilon} + c n \log n)$ time~\cite{c-packedFD}.

In standard \Frechet metric, the objects are piecewise linear.
Rote~\cite{smoothFD} explores the \Frechet distance between more general curves where
each input curve is given as a sequence of smooth curve pieces that are “sufficiently well-behaved”, 
such as circular arcs, parabolic arcs, or some class of spline curves. 
He has shown that the combinatorial complexity, i.e., the number of steps, for solving the decision
 problem is not larger than for polygonal paths, $O(n^2)$ ($n$ is the total size of two given curves). 
Furthermore, under the assumption that the curves consist of algebraic pieces whose degree is bounded by a constant, 
the optimization problem can be solved in $O(n^2\log n)$ time, which matches the running time for the polygonal case.

\subsection{\Frechet Distance with no Leash Cross}
In the \Frechet metric, the leash is allowed to cross the two polylines.
A natural restriction to apply is to require that the leash not cross the polylines.
Efrat \etal~\cite{GuibasNoCross} has introduced two new metrics 
for measuring the distance between non-intersecting 
(not self-intersecting) polygonal curves:
Given two polylines with total complexity $n$, they 
present algorithms to compute the geodesic 
width of the two polylines in $O(n^2 \log^2 n)$ time 
using $O(n^2)$ space and the link width 
in $O(n^3 \log n)$ time using $O(n^2)$ working space
 where $n$ is the total number of edges of the
 polylines. Their computation of these metrics relies on two 
closely-related combinatorial 
strutures: the shortest-path diagram and the link diagram of a simple polygon.
The shortest-path (resp., link) diagram encodes the
Euclidean (resp., link) shortest path distance between 
all pairs of points on the boundary of the polygon. 
Later, Bespamyatnikh~\cite{BesNoCross} obtained a faster
algorithm for computing the geodesic width in $O(n^2)$ time, using $O(n)$ space.

\subsection{Directional-based \Frechet Distance}
Notice that small deviations in one curve can disproportionately 
influence the similarity of two curves.
Furthermore, translations and scalings can affect the result, 
and it is very difficult to make the \Frechet distance invariant under these
types of transformations.
To address these issues, ~\cite{Dir-FD} has proposed the 
{\em direction-based } \Frechet distance. 
Like the standard \Frechet distance, this measure
optimizes over all parametrizations for a pair of curves. Unlike the \Frechet
distance, it is based on differences between the directions of movement along the
curves, rather than on positional differences. 
Therefore, the directional-based \Frechet distance 
is invariant under translations and scalings.
It measures the similarity of polygonal curves by integrating over 
the angular differences between pairs of vectors.
The direction-based \Frechet distance of two polygonal curves with $m$ and $n$ vertices
can be computed in $O(mn)$ time, using $O(m + n)$ space~\cite{Dir-FD}.
Furthermore, the direction-based integral \Frechet distance is proposed
in \cite{Dir-FD} to ensure that small variations in one path do not disproportionately affect the similarity measure.

The measure most
closely related to the direction-based \Frechet 
distance is the turning angle distance~\cite{Arkin1990}. This distance measure
is essentially the same as the direction-based integral \Frechet distance, but
with the following important difference: the turning angle distance 
does not optimize over all possible one-to-one mappings between the two curves. 
Rather, the direction-based \Frechet distance optimizes over all possible one-to-one mappings between the two curves. 
The turning angle distance is 
easily  computed in $O(m + n)$ time for two
polygonal curves with $m$ and $n$ vertices~\cite{Arkin1990}. 
%

%
%

\subsection{\Frechet Distance of Closed Curves}
\label{sec:RelatedClosed}
Closed polygonal curves are curves with common starting and ending points. 
The man-dog illustration of \Frechet metric in this variant is as follows:  
the man and the dog 
are not only allowed to control their speeds, 
but also to choose optimal starting points on the closed curves to minimize 
the length of the leash.

Let $\distClosedF(P,Q)$ denote the \Frechet distance between two closed curves. 
Alt \etal~\cite{AltG95} proposed a polynomial time algorithm 
to solve the decision problem of $\distClosedF(P,Q)$ as follows.
First, a new diagram $\BNNM$ is constructed by 
concatenating two copies of $\BNM$ ~\cite{AltG95}. 
Then, a data structure is built on top of $\BNNM$ to check 
the following property in constant time:
 $\distClosedF(P,Q) \le \epsilon$ iff there exists
a $t \in [0,n]$  and a monotone curve from
$(t,0)$ to $(t+n,m)$ in the free space $\Feps$ of $\BNNM$~\cite{AltG95}.
Suppose diagram $D = \BNNM$ is given and 
$B,T, L,$ and $R$ are its bottom, top, left and right sides, respectively.
In the data structure, these sides are partitioned into some intervals 
where each interval is a connected subset of white points on the 
boundary of $D$. There are three types of intervals: 
%

\begin{itemize}\itemsep1pt

	\item $I$ is $n-$interval iff from no point on $I \subseteq B \cup L$, 
	a point on  $T \cup R$ can be reached by a monotone path 
	in $\Feps$ of $\BNNM$.
	
	\item $I$ is $r-$interval iff from any two points in $I \subseteq B \cup L$, the same set of points on $R \cup T$  can be reached.

	\item $I$ is $s-$interval iff from any point in $I \subseteq  L$ (resp., $I \subseteq  B$), 
	the horizontal (resp., the vertical) line segment connecting that point with $R$ (resp., $T$) lies 
	completely within $\Feps$.

\end{itemize} 

Two pointers $h$ and $\ell$ are attached
to each $r$-interval $I$:
pointer $h$ points to the highest point in $T \cup R$ that can be reached    
from $I$ and pointer $\ell$ points to
 the lowest point in  $T \cup R$ which is reachable from $I$. In addition, 
an $h$ pointer is assigned to each s-interval on $L$, and
an $\ell$ pointer is assigned to each  $s$-interval on $B$. Analogously, 
$T \cup R$ is partitioned into n-,s-, and r intervals depending on their reachability from 
$L \cup B$ and $h$ and $\ell$ pointers are attached to them. 

The data structure is constructed 
recursively by starting from diagram $D$
and splitting the diagram in half at its longer side into two diagrams 
$D_1$ and $D_2$. The recursion continues until 
 a $1 \times 1$- diagram, which is a cell, is reached.
For one cell, the partitioning and the pointers  can be found in $O(1)$ time. 

In order to merge the two solutions $D_1$ and $D_2$ 
into one for $D$, first the intervals on the
right side $R_1$ of $D_1$ are merged with the ones of the left side $L_2$ of $D_2$.
This causes a refinement of the partitions of $R_1$ and $L_2$. Each new 
interval inherits the type and pointers from the old interval of which it is 
subset. Then, the types and  pointers of the intervals on $L_1 \cup B_1$ and $T_2 \cup R_2$ are updated and the intervals and pointers on $T_1$ and $B_2$ remain unchanged.
Details of how intervals on $L_1 \cup B_1$ (or intervals on $T_2 \cup R_2$) 
are updated, can be found in~\cite{AltG95}. 
It is shown that the total time for merging is proportional to 
the number of intervals in the partitioning of $D_1$ and $D_2$; 
in the worst case, this number is 
 $O(nm)$. Thus, the runtime of the merging step is $O(nm)$
and consequently, the whole divide-and-conquer algorithm has $O(nm \log nm)$ running time.

Observe that given two points $u \in I \subseteq B$ and $v \in J \subseteq T$,
there exists a monotone path from $u$ to $v$ in $\Feps$ 	iff one of the 
following conditions (a) or (b) holds: 
(a) $I$ is an $r$-interval and $v$ lies between $h(I)$ and $\ell(I)$ 
(b)$I$ is type-s and $v$ lies to the right of $u$ and to the left of $\ell(I)$~\cite{AltG95}.
Having constructed the data structure on $\BNNM$, 
one can determine in $O(n)$ time by scanning intervals on the bottom and top 
side of $\BNNM$ simultaneously, if there exists $t \in [0,n]$  and a monotone curve from
$(t,0)$ to $(t+n,m)$ in $\Feps$ of $\BNNM$.
Given two closed curves $P$ and $Q$ with total length $n$, 
whether $\distClosedF(P,Q) \le \epsilon$
can be decided in $O(n^2 \log n)$ time. 
The exact value of 
$\distClosedF(P,Q)$ can be computed in $O(n^2 \log^2 n)$ time using parametric search.

	




\section{Partial Curve Matching}
\label{sec:RelatedPartial}
In this section, we discuss the problem
of measuring partial similarity between curves.

\subsection{Partial Curve Matching}
\label{sec:RelatedPartialMain}
Alt and Godau~\cite{AltG95} considered one natural partial similarity 
measure by computing the \Frechet distance between a single 
consecutive piece of subcurve of $P$ and another curve $Q$. 
Let $\distPartialF(P,Q) =$ inf $\{ \distF(R,Q) \ |$ where $R$ is a subcurve of $P \}$. 
The same technique for two closed curves 
 can be applied to solve the decision problem, i.e., to decide if $\distPartialF(P,Q) \le \epsilon$. Let  $P$
and $Q$ be two curves with length $n$ and $m$, respectively and a parameter $\epsilon \ge 0$ is given. Once we have constructed 
the data structure on top of $\BNM$, 
we only have to check the type of the intervals on the bottom side of 
$\BNM$. If all are of type $n$, then the answer is ``NO", otherwise the 
answer is ``YES". Therefore, the decision problem can be solved in  $O(n^2\log n)$ time and the exact value of $\distPartialF(P,Q)$ 
can be computed in $O(n^2\log^2 n)$ time using the 
parametric search~\cite{AltG95}.

The partial similarity measure introduced in~\cite{AltG95}  
only allows to have outliers in one of the input curve, and more importantly, 
it does not allow outliers
appearing in different (non-consecutive) locations along the input curve. 
In addition, the summed versions introduced in~\cite{SumFD} do 
not fully resolve the issue of partial similarity, especially when significant parts 
of the curves are dissimilar. 

Recently, Buchin \etal~\cite{ExactPartial} have
proposed a natural extension of the
\Frechet distance to measure the partial similarity between
curves. 
They introduce a continuous
partial curve similarity measure that allows general
types of outliers, and develop an exact algorithm to
compute it. The goal here is to maximize the total length of 
subcurves that are close to each other, where closeness is measured by 
the \Frechet distance. 

Specifically, given a distance threshold $\epsilon$ and two polygonal 
curves $P$ and $Q$, the partial \Frechet similarity between $P$ and $Q$ is the total length
of longest subcurves of $P$ and $Q$ that are matched with
\Frechet distance at most $\epsilon$. The \Frechet distance 
can be measured under any $L_p$ norm, and they consider the $L_1$ and
$L_\infty$ norms in~\cite{ExactPartial}.
The partial \Frechet similarity can be considered as
the length of the longest monotone path in a certain
polygonal domain with weighted regions, where the
weight is either 0 or 1. Hence computing that measure
bears similarity with the standard shortest path queries
in weighted regions. The algorithm in~\cite{ExactPartial} computes 
the partial \Frechet similarity measure in $O(mn(m + n) \log(mn))$ time, by constructing a "shortest-path map" type decomposition.

In another recent work \cite{BuchinLocallyFrechet}, 
Buchin \etal~introduce locally correct \Frechet matchings. They introduce the local correctness criterion for \Frechet matchings and prove
that there always exists at least one locally correct \Frechet matching between any
two polygonal curves. They provide an 
$O(n^3 \log n)$ algorithm to compute
such matching, where
$n$ is the total complexity of the two curves.

Although the \Frechet distance is considered to be a high quality metric to measure
the similarity between polygonal curves, it is very sensitive to the presence of outliers. 
In \cite{Jaywalking}, Driemel and Har-Peled
discuss a new notion of robust \Frechet distance, 
where they allow $k$ shortcuts between vertices of one
of the two curves, where $k$ 
is a constant given as an input parameter.
They provide a constant factor
approximation algorithm for finding the minimum 
\Frechet distance among all possible $k$-shortcuts.
However, their approach has this 
drawback that a shortcut is selected without considering the length of the ignored part. Therefore, such shortcuts may remove a significant portion of a curve.
Recently, in another work~\cite{SackAminPaper}, 
authors propose
an alternative \Frechet distance measure to tolerate outliers,
considering the length of portion of the curves that must 
be discarded. Roughly, their goal is to minimize the length of 
subcurves of two polygonal curves that need to be ignored to achieve a given  \Frechet distance.

\subsection{Map Matching}
\label{sec:RelatedMapMatching}

In GIS applications, the method of sampling the  movements of vehicles using GPS is affected 
by errors and consequently produces inaccurate trajectory data. To become useful, 
the data has to be related to the underlying road network by using map matching 
algorithms. A quality map matching algorithm utilizing the \Frechet distance is 
introduced in~\cite{AltERW03a}.

Given a planar graph $G$ as a road network and a polygonal curve $P$ as a 
trajectory of a vehicle, the objective is to find a path $\pi$  
in graph $G$ with minimum \Frechet distance to curve $P$. To find such a path, 
Alt \etal ~\cite{AltERW03a} generalized  the definition of free space 
between two curves to the free space between a graph and a curve as follows.

The free space of graph  $G = (V,E)$ and  curve $P$ is the union of all free spaces of edges of 
$G$ with the polygonal curve $P$. Observe that the free space of one node $v$ with 
curve $P$ is a one-dimensional free space (denoted by $FD_v$), 
and the individual free spaces of all 	edges 
incident to node $v$ with curve $P$ share a one-dimensional free space at $v$. 
Thus, we can 
glue together the two-dimensional free-space diagrams along the one-dimensional free space 
they have in common, according to the adjacency information of the graph. 
The resulting three-dimensional structure is called \emph{free-space surface} of graph $G$ and curve $P$ 
in~\cite{AltERW03a} (see Figure \ref{fig:freespacesurface}). 

\begin{figure}[t]
	\centering
	\includegraphics[width=0.75\columnwidth]{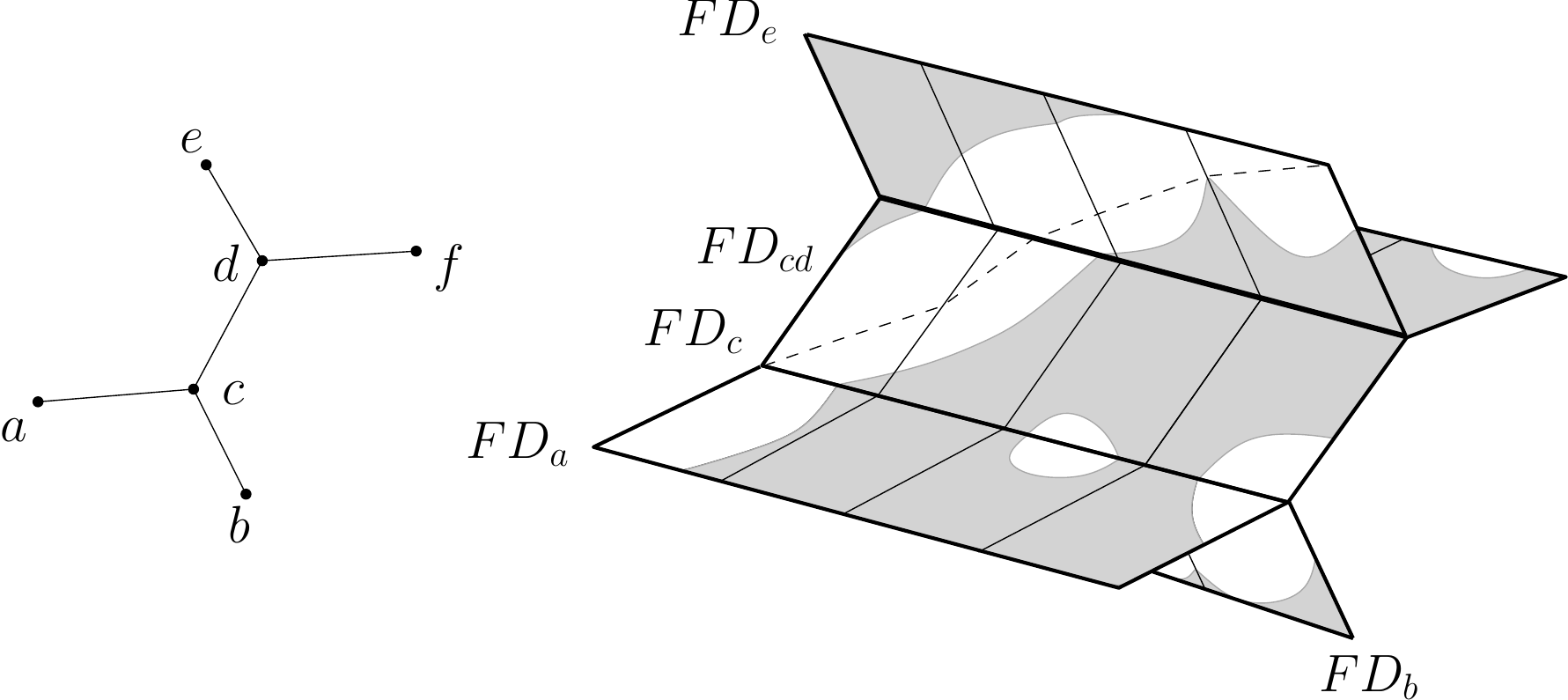}
	\caption{Free-space surface consists of free-space diagrams glued
together according to the topology of graph $G$. Grey dashed path is
a monotone path in the free space.}
	\label{fig:freespacesurface}
\end{figure}

Let $(i,j) \in E$ be an edge of graph $G$. Furthermore,  
let $FD_{ij}$ be an edge-curve free-space diagram 
corresponding to edge $(i,j)$, curve $P$, and distance $\epsilon$.
$FD_{ij}$  consists of one dimensional free-space $FD_i$,
then $m$ (size of curve $P$) cells in a row , 
and another one dimensional free-space $FD_j$(see Figure \ref{fig:freespacesurface}
and Figure \ref{fig:pointers}).

In~\cite{AltERW03a} it has been shown that, after constructing 
a free surface corresponding 
to a planar graph $G$ and a polygonal curve $P$,  there exists a 
path $\pi$ in $G$ s.t. $\distF(\pi,P) \le \epsilon$ 
iff there is a monotone path 
in the free-space surface from a lower left corner of
some individual edge-curve free-space diagram to an upper
 right corner of some other individual edge-trajectory 
free-space diagram (e.g., see the gray dashed path 
in Figure \ref{fig:freespacesurface}).

For $I$  a continuous interval of white points 
in $FD_i$, let the reachability pointers 
$\ell_{i,j}(I)$ and $r_{i,j}(I)$ be the leftmost 
and the rightmost white points, respectively, 
of $FD_j$ that can be reached from some point
in $I$ by a monotonic path in $FD_{ij}$
(see Figure \ref{fig:pointers}).
As a first step of the decision algorithm in~\cite{AltERW03a}, 
all one-dimensional free-spaces $FD_i$ 
(for every vertex $i \in V$), and also reachability pointers
are computed.
Next, the algorithm sweeps a line from left to
right (in direction of $P$) over all free spaces
at the same time while maintaining the points on the
sweepline that are reachable by some monotone path in
the free space from some lower left corner. It then updates 
this reachability information (using the 
reachability pointers)  Dijkstra-style while
advancing the sweepline.

Given a planar graph $G$ with $n$ vertices, 
a polygonal curve $P$ 
with length $m$ and a distance $\epsilon$, 
the algorithm decides in $O(m n\log n)$ time whether there exists a path
$\pi$ in $G$ such that  
 $\distF(P,\pi) \le \epsilon$. 
One can  use  parametric search equipped with the decision algorithm, 
to find a path a $\pi$ in $G$ which minimizes $\distF(\pi,P)$, 
by spending  $O(mn\log(mn) \log n)$ time and using $O(mn)$ space.
The decision algorithm in~\cite{AltERW03a}
is only a log-factor slower than the standard \Frechet distance decision problem, although 
it accomplishes a  more complicated task of comparing curve $P$
to all possible curves in graph $G$.

\begin{figure}[t]
	\centering
	\includegraphics[width=0.7\columnwidth]{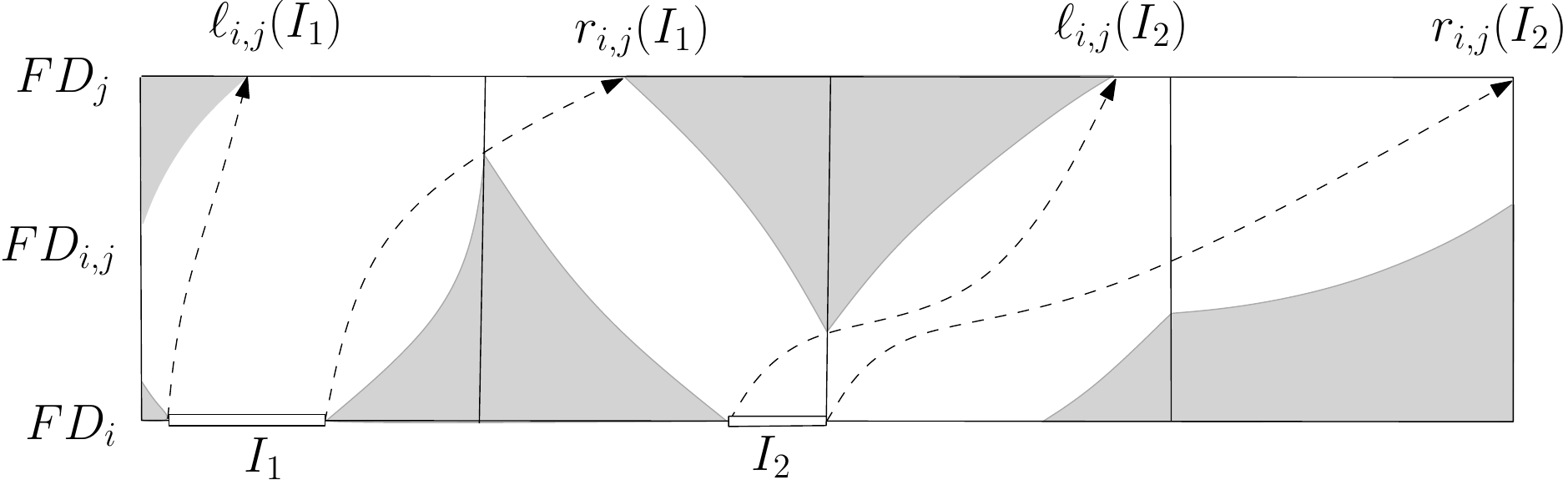}
	\caption{Reachability pointers}
	\label{fig:pointers}
\end{figure}

Map Matching based on the weak \Frechet distance has been also studied
by Brakatsoulas~\etal~\cite{VehicleTracking}, who give an $O(mn \log (mn))$ algorithm, where $m$ is the size of the curve and $n$ is the size of the graph. 
As explained before, the decision problem for the weak \Frechet distance
between two curves can be solved by testing if there
exists any path in the free space of the two curves from
lower left corner to upper right corner. This
can be done using any graph traversal algorithm such
as depth-first search in $O(mn)$ time.
In~\cite{VehicleTracking}, this approach 
is generalized to the map matching
problem by applying depth first search to
the free-space surface. They initialize the search with all
white lower left corners of individual edge-trajectory
free spaces, and stop the search if some upper
right white corner is found. Since the free-space surface
consists of $mn$ edge-segment cells, this algorithm runs
in $O(mn)$ time, which is a log-factor faster than the
algorithm based on the normal \Frechet distance. Applying
parametric search for optimization, in the same way as
in~\cite{AltG95}, adds an additional log-factor to the runtime
for a total of $O(mn \log(mn))$ to solve the optimization
problem. A new result~\cite{fasterMM-Weak1}
improves this running time to $O(mn)$. The method provided in that 
paper does not involve parametric search, 
and hence is also easier to implement. 
 Their algorithm also yields an $O(mn)$ algorithm for computing 
 the weak \Frechet distance between polygonal curves, 
 where one curve has size $m$ and the other has size $n$ (which
improves the $O(mn \log mn)$ result given by~\cite{AltG95}).

\subsection{Constrained Free-Space Diagram}

Spatio-temporal data is any information relating space and time. 
Recently, there has been considerable research 
in the area of analyzing and modeling spatio-temporal data~\cite{Cons-Free}. 
Movement patterns in such data refer to 
events and episodes expressed by a set of entities.
The problem of detecting movement patterns in spatio-temporal data has recently received considerable 
attention from several research communities, e.g., geographic information science, data mining, data bases and algorithms.

Buchin \etal~\cite{Cons-Free} propose a new and powerful tool, called constrained free space, 
for the analysis of trajectories, which, in  particular, allows for more temporally aware analyses. 

Their new tool provides an algorithm for detecting single file movement. 
A single file is a set of moving entities, which 
are following each other, one behind the other. 

Let a spatio-temporal trajectory $\CT$ of a moving entity $a$ be
given by $n$ time-space positions. That is, $\CT$ = $((t_1, p_1),...,(t_n, p_n))$, 
where $p_i \in \IR^2$ gives the position of entity $a$ 
at time $t_i$ for $i = 1,..., n$.
Assume that in between time stamps $t_i$ and $t_{i+1}$ the
entity $e$ moves with constant speed along a straight line
from $p_i$ to $p_{i+1}$ for $i = 1,..,n$ ~\cite{Cons-Free}. 

For detecting a single file behavior, we are given $m$ spatio-temporal trajectories
$\CT_1,..., \CT_m$ of entities $a_1,..., a_m$.
The entities $a_1,.., a_m$ are moving in single
file for a given time interval if during this time each entity
$a_{j+1}$ is following behind entity $a_j$ for $j=1,...,m-1$.
For the definition of following, fix three parameters
$T_{min}, T_{max}$, and $\delta \in \IR$ with $T_{min}<T_{max}$. The parameters
$T_{min}$ and $T_{max}$ specify minimum and maximum offsets in time, 
respectively, 
and $\delta$ specifies a maximum offset in space. 
One can detect whether one trajectory
is following behind another during a fixed time interval by
searching for a monotone path in the $[T_{min}, T_{max}]$-strip of
the free-space diagram of the trajectories. Let $k_{avg}$
and $k_{max}$ denote the average and maximum number of
cells intersected by the $[T_{min}, T_{max}]$-strip per row or column
of the free-space diagram. Then, for two trajectories of 
complexity $n$ each, it can be determined in $O(nk^2_{avg})$ time 
and $O(n + k_{max}^2)$ space during which time intervals one 
trajectory is following behind the other. 
Furthermore, for $m$ trajectories of complexity $n$ each, one can detect in $O(m^2nk_{avg})$ 
time and $O(nm+m^2 +k_{max})$ space all single file behaviors for a given time interval~\cite{Cons-Free} .

\section{\Frechet Distance in Different Metric Spaces}

In the \Frechet distance problem,
when the two curves are embedded in a general metric space, 
the distance between two points on the curves 
(i.e., the length of the shortest leash joining them) 
is not necessary the Euclidean distance, but sometimes it is a geodesic distance
due to existence of obstacles in the space.

\subsection{Geodesic \Frechet Distance}
In~\cite{WenkC08a}, Cook and Wenk  described an algorithm for the geodesic
\Frechet distance between two polygonal curves $P$ and $Q$ 
inside a simple polygon $K$. 
To solve the decision version, they used the free-space diagram 
structure introduced by Alt and Godau~\cite{AltG95}. 
The main observation here is that when two curves are located inside a 
simple polygon, the free space inside a cell is $x$-monotone, $y$-monotone, 
and connected~\cite{WenkC08a}. As such, only the boundaries of 
a cell need to be computed to propagate reachability 
in the free-space diagram.
There are $O(n^2)$ cells in the free-space diagram. 
Computing the boundary of each cell takes $O(\log k)$ time by the algorithm of Guibas 
and Hershberger~\cite{Guibas86}. 
Then, the reachability information is 
propagated through all cells in a dynamic programming manner as~\cite{AltG95}. 
Since the free space inside each cell is monotone, 
propagating reachability though each cell takes constant time. Therefore, 
if $P$ and $Q$ have total complexity $n$ and polygon $K$ has complexity $k$, 
after a one-time preprocessing step of $O(k)$ time, the geodesic \Frechet decision
problem can be solved for
any $\epsilon\ge 0 $  in $O(n^2 \log k)$ time and $O(k + n)$ space.
The space bounds follow because $O(1)$ space is needed per cell
and dynamic programming only requires that two rows of cells reside
in memory at any one time. The $O(k)$ term comes from storing the
preprocessing structures of~\cite{Guibas86} throughout the algorithm's execution.
Using parametric search, the exact 
geodesic \Frechet distance can be computed in 
$O(n^2 \log k \log n)$. 
Cook and Wenk~\cite{WenkC08a} proposed a randomized algorithm using 
a {\em red-blue } intersection approach which  
finds the exact geodesic \Frechet distance in 
$O(k + n^2 \log kn \log n)$ expected time and $O(k +n^3 \log kn)$ worst case
time.


Although the exact standard \Frechet distance is normally found in
$O(n^2 \log n)$ time using parametric search, parametric search
is often regarded as impractical because it is difficult to implement and
involves large constant factors~\cite{Cole87}. 
The randomized algorithm in~\cite{WenkC08a} 
is the first practical alternative to parametric search
for solving the exact \Frechet optimization problem.
Using the red-blue intersection approach as ~\cite{WenkC08a}, 
one can compute the exact \Frechet distance in $O(n^2 \log^2 n)$ expected
time and $O(n^2)$ space, where $n$ is the larger of the complexities of $P$
and $Q$~\cite{WenkC08a}.

%
%

\subsection{Homotopic \Frechet distance}

The definition of the classical \Frechet distance allows the leash to switch discontinuously, 
without penalty, from one side of an obstacle or a mountain to another.
Chambers  \etal~\cite{Chambers10} study the \Frechet distance between two 
polygonal curves $P$ and $Q$, located in the punctuated plane  consisting of  
$k$ points. They introduce a continuity requirement on the motion of the leash, i.e. the leash cannot switch, discontinuously, from one geodesic to another; in particular, the leash can not jump over obstacles  and can sweep over a mountain only if it is long enough
(see Figure~\ref{fig:Hom}).
This new similarity metric is called {\em homotopic} \Frechet distance.
It finds applications in morphing and robotics.
In spaces where shortest paths vary continuously
as their endpoints move, such as the Euclidean plane, the
\Frechet distance and homotopic \Frechet distance are
identical. In general, however, homotopic \Frechet distance
could be larger (but never smaller) than the classical
\Frechet distance.
Given two polygonal curves $P$ and $Q$ with complexity $n$ and $m$, respectively 
and $k$ points in the plane,
the homotopic \Frechet distance between $P$ and $Q$ in the plane can be computed
in $O(m^2n^2 k^3 \log(m n k))$ time~\cite{Chambers10}.

The algorithm for computing the geodesic \Frechet distance between two curves within a simple
polygon due to Cook and Wenk~\cite{WenkC08a}, 
is faster than the homotopic \Frechet computation algorithm in~\cite{Chambers10} by roughly a factor of $n$. This is because they use a randomized strategy in place of parametric search.

\begin{figure}[t]
	\centering
	\includegraphics[width=0.45\columnwidth]{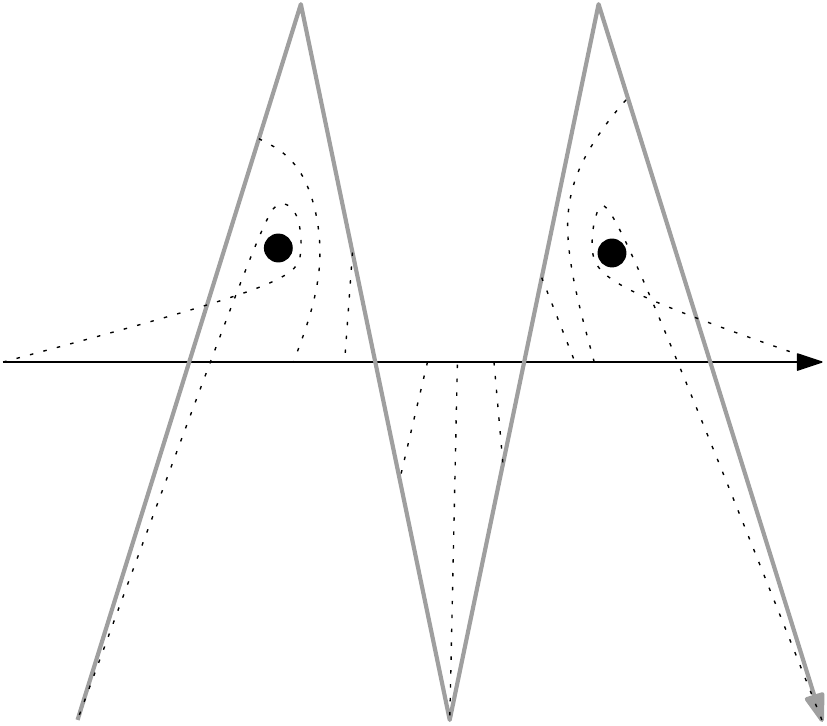}
	\caption{The dashed lines show the leash between two objects 
	while they are moving on their corresponding curves. 
	The leash can not jump over the obstacles.}
	\label{fig:Hom}
\end{figure}


%
%
%
Cook \etal~\cite{Cook2009} develop algorithms to compute the \Frechet distance of 
two curves  on convex and non-convex polyhedral surface.
Let $M$ be the total complexity of a problem space that contains
a polyhedral surface and auxiliary objects on the surface such as points,
line segments, and polygonal curves.
Then, \Frechet distance can be computed in $O(M^6 \log^2 M)$ time and $O(M^2)$ space in a convex polyhedral surface and 
$O(M^7 \log^2 M)$ time and $O(M^3)$ space in a non-convex polyhedral surface.

Cheung \etal~\cite{Cheung2009} consider two versions of 
the \Frechet distance problem in weighted planar subdivisions. In the first one, the distance between two points is the weighted length of the line segment joining the points. In the second one, the distance between two points is the length of the shortest path between the points. In both cases they give algorithms for finding a $(1+\epsilon)$-factor approximation of the \Frechet distance between two polygonal curves.

\section{Approximate \Frechet Distance}
A considerable amount of work 
has been done to improve running time
of computing \Frechet distance. 
Since improving the quadratic-time solution for general curves seems
to be hard, many researchers investigated \Frechet distance in restricted class 
of curves, rather than general curves. 
Also many works have been done 
to compute approximate \Frechet distance.

In~\cite{DisceteFD}, Eiter and Mannila 
introduced a 
close approximation and slightly simpler version of the \Frechet distance, called {\em discrete 
\Frechet distance}, which
only considers vertices of polygonal curves.
They showed that given two polygonal curves of $n$ and $m$ vertices, their discrete 
\Frechet distance can be computed  in $O(mn)$ time by a dynamic programming algorithm. 
Figure \ref{fig:DistFig} demonstrated the relationship between discrete 
and continuous \Frechet distance. 
It has been shown in~\cite{DisceteFD} that the discrete \Frechet distance is an 
upper bound for the \Frechet distance and the difference between these measures 
is bounded by the length of the longest edge of the polygonal 
curves. 
In a very recent work, 
Agarwal \etal~\cite{DiscreteSUBQuadractic} 
show how to break the quadratic barrier for the discrete \Frechet distance. 
They propose sub-quadratic 
$O(\frac{mn \log \log n}{\log n})$ time
algorithm for computing discrete \Frechet distance 
using $O(n+m)$ space, 
where $n$ and $m$ are the complexity of 
two polygonal curves.

\begin{figure}[t]
	\centering
	\includegraphics[width=0.65\columnwidth]{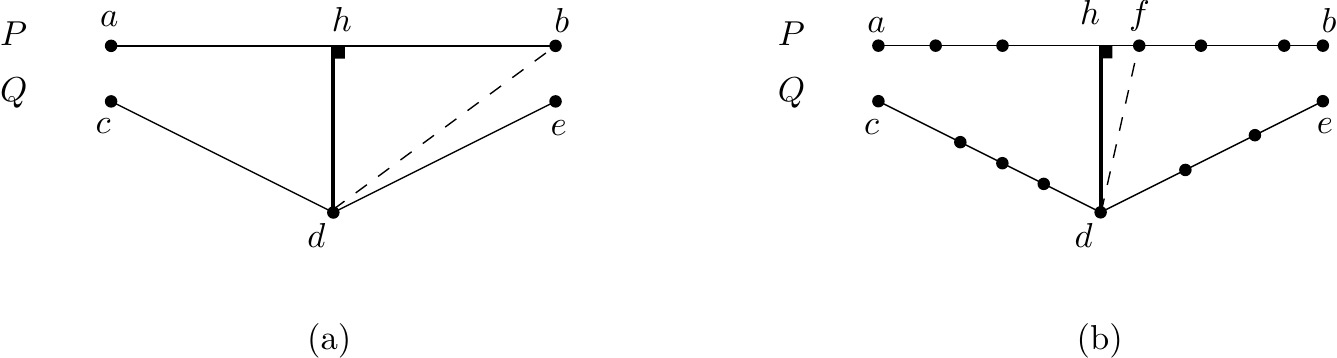}
	\caption{ (a) The discrete \Frechet distance could be arbitrarily larger than the continuous distance, e.g., $\distF(P,Q) = |\Ov{dh}|$, $\distDisF(P,Q) = |\Ov{db}|$.  (b) If we
put enough sample points on the two polygonal chains, 
then the resulting discrete \Frechet
distance, that is, $|\Ov{df}|$, closely approximates $|\Ov{dh}|$.
}
	\label{fig:DistFig}
\end{figure}



In~\cite{LowerBound-FD}, Buchin \etal~gave a lower bound of $\Omega(n \log n)$ time for the decision version of the \Frechet problem.
They also showed that
this bound holds for 
the discrete version of the problem as well.
The only subquadratic algorithms known are for 
restricted classes of curves such as for closed
convex curves and for $k$-bounded curves~\cite{Alt2003Comparison}. 
A curve is
called $k$-bounded when
for any two points on the curve, 
the portion of the curve in between them cannot be further
away from either point than $k/2$ times the distance between the two points. For closed convex
curves the \Frechet distance equals the Hausdorff distance and for $k$-bounded curves the \Frechet
distance is at most $(1+k)$ times the Hausdorff distance, and hence the $O(n \log n)$ algorithm for
the Hausdorff distance applies~\cite{Alt2003Comparison}.
Aronov et al.~\cite{FDRevisited}  proposed a near linear time 
$(1+\eps)$-approximation algorithm for the
discrete \Frechet distance, which only considers distances between vertices of the curves. Their
algorithm works for backbone curves, which are used to model protein backbones in molecular
biology. Backbone curves are required to have, 
unit edge length and a minimal distance
between any pair of vertices. 

In~\cite{Agarwal2002}, Agarwal \etal~consider the problem of approximating a polygonal curve 
$P$ under a given error criterion by another polygonal curve $P'$ whose vertices are a subset of the vertices of $P$. 
The goal is to minimize the number of vertices of $P'$ while ensuring that the error between $P'$ and $P$ is 
below a certain threshold.

In another recent work~\cite{c-packedFD},
the \Frechet distance has been studied between 
\emph{$c$-packed} curves. 
A curve $P$ is $c$-packed if the total length of $P$ inside any ball is bounded by $c$ times the radius of the ball. 
While not all curves are $c$-packed,the most real life curves are $c$-packed~\cite{c-packedFD}.
Given two polygonal $c$-packed curves $P$ and $Q$ 
with a total of $n$ vertices, and
a parameter $0 < \epsilon < 1$, 
they show that one can $(1+\eps)$-approximate the \Frechet distance between $P$ and $Q$ in
$O(\frac{cn}{\epsilon} + cn \log n)$ time.

\chapter{\Frechet Distance with Speed Limits}
\label{ch:speedFD}
In the classical \Frechet distance problem, 
the speed of motion on the two polygonal curves is unbounded.
 in which  motion speeds are bounded, both from below and from above. 
More precisely, associated to each segment of the curves, 
is a speed range
that specifies the minimum and the maximum speed allowed for travelling along that segment.
We say that a point object traverses a curve with \emph{permissible speed},
if it traverses the polygonal curve from start to end
so that the speed used on each segment falls within its permissible range.

The decision version of the \Frechet distance problem with speed limits is formulated as follows:
Let $P$ and $Q$ be two polygonal curves with minimum and maximum permissible speeds
assigned to each segment of $P$ and $Q$.
For a given $\eps \gee 0$, is there an assignment of speeds so that two point objects
can traverse $P$ and $Q$ with permissible speed and, throughout the entire
traversal, remain at distance at most $\eps$ from each other?
The objective in the optimization problem is to find the smallest such $\eps$.

In this chapter, we present a new algorithm that solves the 
decision version of the \Frechet distance problem with speed limits in $O(n^2 \log n)$ time. 
Our main approach is to compute a free-space diagram 
similar to the one used in the standard \Frechet distance algorithm (Section \ref{sec:classicalFD}).
However, since the complexity of the free-space diagram  in our problem
is cubic, in contrast to the standard free-space diagram that has quadratic complexity,  
we use a ``lazy computation'' technique to
avoid computing unneeded portions of the free space, 
and still be able to solve the decision problem correctly.
We then use our algorithm for the decision problem to
solve the optimization problem exactly in $O(n^3 \log n)$ time.

The \Frechet distance with speed limits we consider here
is a natural generalization of the classical \Frechet distance.
It has potential applications in GIS, when the speed of moving objects is considered
in addition to the geometric structure of the trajectories.
For a practical application of this metric, 
consider the case where 
trajectory of a vehicle is given 
to us, and we want to find the closest path 
in the road network to that trajectory. The good thing about the standard \Frechet metric 
is that we can use it here and find the closest path in the road network to the trajectory. 
Using our metric however, we can consider speed limits in the road network as well, 
and find a path in the road network which is more realistic.

This chapter is organized as follows. 
The problem is formally defined in the next section.
In Section~\ref{sec:decisionSpeed}, we describe a simple algorithm that solves the decision problem in $O(n^3)$ time.
In Section~\ref{sec:improvedSpeed}, we provide an improved algorithm for the decision problem 
that runs in $O(n^2\log n)$ time.
Section~\ref{sec:optimization} describes how the optimization problem can be solved efficiently.
Finally, we summarize in Section~\ref{sec:conclusion} and outline directions for future work.


\section{Preliminaries} \label{sec:preliminaries-Speed1}


\paragraph{\Frechet Distance with Speed Limits.}

Consider two point objects $\CO_P$ and $\CO_Q$ that traverse $P$ and $Q$, respectively 
from start to end.
If we think of the parameter $t$ in the parametrizations $\alpha$ and $\beta$
as ``time'', then $P(\alpha(t))$ and $Q(\beta(t))$ 
specify the positions of 
$\CO_P$ and $\CO_Q$ on $P$ and $Q$ respectively at time $t$.
The preimages of $\CO_P$ and $\CO_Q$ can be viewed as
two point objects $\COB_P$ and $\COB_Q$ traversing
$[0,n]$ and $[0,m]$, respectively,
with their positions at time $t$ being specified by $\alpha(t)$ and $\beta(t)$
($n$ is the length of $P$, $m$ is the length of $Q$).

In the classical definition of \Frechet distance,
the parametrizations $\alpha$ and $\beta$ are
arbitrary non-decreasing functions, 
meaning that $\COB_P$ and $\COB_Q$ (and therefore, $\CO_P$ and $\CO_Q$)
can move with arbitrary speeds in the range $[0,\infty]$.
In our variant of the \Frechet distance with speed limits,
each segment $S$ of the curves $P$ and $Q$
is assigned a pair of non-negative real numbers $(\vmin{S}, \vmax{S})$ 
that specify the minimum and the maximum permissible speed for moving along $S$.
The speed limits on each segment is independent of the limits of other segments.
When $\CO_P$ moves along a segment $S$ with speed $v$, 
$\COB_P$ moves along the preimage of $S$ (which is a unit segment) with speed $v/\|S\|$.
Therefore, the speed limit $(\vmin{S}, \vmax{S})$ on a segment $S$,
forces a speed limit on the preimage of $S$ which is bounded by the following two values:
\[
	\umin{S} = {\vmin{S} \over \|S\|} \ \ \mbox{and} \ \ \umax{S} = {\vmax{S} \over \|S\|}.
\]

We define a {\em speed-constrained parametrization of $P$\/}
to be a continuous surjective function $f: [0,T] \rightarrow [0,n]$ with $T > 0$
such that for any $i \in \set{1, \ldots, n}$,
the slope
of $f$ at all points $t \in [f^{-1}(i-1),f^{-1}(i)]$ 
is within $[\umin{P_i}, \umax{P_i}]$. 
Here, we define the slope of a function $f$ at a point $t$ to be 
$\lim_{h \rightarrow 0^+} f(t+h)/h$, where $h$ approaches 0 only from above (right).
By this definition, if $f$ is a continuous function, 
then the slope of $f$ at any point $t$ in its domain is well-defined,
even if $f$ is not differentiable at $t$. 

Given two polygonal curves $P$ and $Q$ of lengths $n$ and $m$, respectively 
with speed limits on their segments, 
the {\em speed-constrained \Frechet distance\/} between $P$ and $Q$ 
is defined as:
\[
	\distFS(P,Q) = \inf_{\alpha, \beta} \max_{t \in [0,T]} d( P(\alpha(t)), Q(\beta(t)) ),
\]
where $\alpha: [0,T] \rightarrow [0,n]$ ranges over all speed-constrained parametrizations of $P$ 
and $\beta:[0,T] \rightarrow [0,m]$ ranges over all speed-constrained parametrizations of $Q$.
Note that this new formulation of \Frechet distance is similar to the classical one,
with the only difference that the parametrizations here are restricted to have
limited slopes, reflecting the speed limits on the segments of the input polygonal curves.



\paragraph{Notation.}
We introduce some notation used throughout this chapter.

Let $\BNM = [0,n] \times [0,m]$ be an $n$ by $m$ rectangle  in the plane.
Each point $(s,t) \in \BNM$ uniquely represents a pair of points
$(P(s),Q(t))$ on the polygonal curves $P$ and $Q$.
We decompose $\BNM$ into
$n\cdot m$ unit grid cells $\cell{ij} = [i-1,i] \times [j-1,j]$
for $(i,j) \in \set{1, \ldots, n} \times \set{1, \ldots, m}$,
where each cell $\cell{ij}$ corresponds to
a segment $P_i$ on $P$ and a segment $Q_j$ on $Q$.
Given a parameter $\eps \gee 0$,
the {\em free space\/} $\Feps$ is defined as
\[
	\Feps = \set{(s,t) \in \BNM \ | \ d(P(s),Q(t)) \lee \eps }.
\]
We call any point $p \in \Feps$ a {\em feasible\/} point.
An example of the free-space diagram for two curves $P$ and $Q$ 
is given in Figure~\ref{fig:diagram}.a.

Each line segment bounding a cell in $\BNM$ is called an {\em edge\/} of $\BNM$.
We denote by $L_{ij}$ (resp., by $B_{ij}$) the left (resp., bottom) line segment bounding $\cell{ij}$.
For a cell $\cell{ij}$, we define the {\em entry side\/} of $\cell{ij}$ to be $\entry{ij} = L_{ij} \cup B_{ij}$,
and its {\em exit side\/} to be $\exit{ij} = B_{i,j+1} \cup L_{i+1,j}$.
Throughout this chapter, we process the cells in a {\em cell-wise\/} order,
in which a cell $\cell{ij}$ precedes a cell $\cell{k\ell}$ if either $i<k$
or $(i=k$  and $j<\ell)$ (this corresponds to the row-wise order of the cells,
from the first cell, $\cell{0,0}$, to the last cell, $\cell{nm}$). 

For an easier manipulation of the points and intervals on the boundary of the cells,
we define the following orders:
Given two points $p$ and $q$ in the plane, we say that 
$p$ is {\em before\/} $q$, 
and denote it by $p \lei q$, if either $p_x < q_x$ or $(p_x = q_x $  and  $ p_y > q_y)$.
For an interval $I$ of points in the plane, the {\em left endpoint\/} of $I$, denoted by $\Left(I)$,
is a point $p$ such that $p \lei q$ for all $q \in I$, $q \not= p$.
The {\em right endpoint\/} of $I$, denoted by $\Right(I)$, is defined analogously.
Given two intervals $I_1$ and $I_2$ in the plane,
we say that $I_1$ is {\em before\/} $I_2$, and denote it by $I_1 \lei I_2$, if
$\Left(I_1) \lei \Left(I_2)$ and $\Right(I_1) \lei \Right(I_2)$.
Note that  $I_1 \lei I_2$ implies that none of the intervals $I_1$ and $I_2$
can be properly contained in the other.


\section{The Decision Problem} \label{sec:decisionSpeed}

In this section, we provide an algorithm for solving the following decision problem:
Given two polygonal curves $P$ and $Q$ of lengths $n$ and $m$ respectively ($n \gee m$)
with speed limits on their segments, 
and a parameter $\eps \gee 0$,
decide whether $\distFS(P,Q) \lee \eps$.
We use a free-space diagram approach, similar to the one used in the standard \Frechet distance problem (Section \ref{sec:classicalFD}).
However, the complexity of the ``reachable portion'' on the cell boundaries is different 
in our problem; namely, each cell boundary in our problem has a complexity of $O(n^2)$,
while in the original problem cell boundaries have $O(1)$ complexity. 
This calls for a more detailed construction of the free space. 

Consider two point objects, $\CO_P$ and $\CO_Q$, traversing $P$ and $Q$, 
with their preimages, $\COB_P$ and $\COB_Q$, traversing $[0,n]$ and $[0,m]$, respectively.
When $\CO_P$ and $\CO_Q$ traverse $P$ and $Q$ from beginning to the end,
the trajectories of $\COB_P$ and $\COB_Q$ on $[0,n]$ and $[0,m]$
specify a path $\CP$ in $\BNM$ from $(0,0)$ to $(n,m)$.
Suppose that $\CP$ passes through a point $(s,t) \in \cell{ij}$.
The slope of $\CP$ at point $(s,t)$ is
equal to the ratio of the speed of $\COB_Q$ at point $t$
to the speed of $\COB_P$ at point $s$.
Therefore, the minimum slope at $(s,t)$ is obtained 
when $\COB_Q$ moves with its minimum speed at point $t$, 
and $\COB_P$ moves with its maximum speed at point $s$.
Similarly, the maximum slope is obtained 
when $\COB_Q$ moves with its maximum speed, 
and $\COB_P$ moves with its minimum speed.
We define 
\[
	\minS{ij} = {\umin{Q_j} \over \umax{P_i}} \ \ \mbox{and} \ \ \maxS{ij} = {\umax{Q_j} \over \umin{P_i}},
\]
where $\umin{\cdot}$ and $\umax{\cdot}$ are the speed limits for $\COB_P$ and $\COB_Q$ as
defined in Section~\ref{sec:preliminaries-Speed1}.
Indeed, $\minS{ij}$ and $\maxS{ij}$ specify
the minimum and the maximum ``permissible\rq{}\rq{} slopes 
for $\CP$ at any point inside $\cell{ij}$.
A path $\CP \subset \BNM$ is called {\em slope-constrained\/}
if for any point $(s,t) \in \CP \cap \cell{ij}$,
the slope of $\CP$ at $(s,t)$ is within $[\minS{ij}, \maxS{ij}]$.
A point $(s,t) \in \Feps$ is called {\em reachable\/} 
if there is a slope-constrained path from $(0,0)$ to $(s,t)$ in $\Feps$.

\begin{lemma} \label{lemma:reachable}
	$\distFS(P,Q) \lee \eps$ iff $(n,m)$ is reachable. 
\end{lemma}

\begin{proof}
	The $(\Rightarrow)$ part is straightforward. 
	For $(\Leftarrow)$,
	we need to show that if $(n,m)$ is reachable, then  
	there exist a speed-constrained parametrization
	$\alpha:[0,T] \rightarrow [0,n]$ of $P$ (for some $T>0$), 
	and  a speed-constrained parametrization 
	$\beta:[0,T] \rightarrow [0,m]$ of $Q$ such that
	$d( P(\alpha(t)), Q(\beta(t) ) \lee \eps$ for all $t \in [0,T]$.
	If $(n,m)$ is reachable, then by definition there is a
	slope-constrained path $\CP$ from $(0,0)$ to $(s,t)$ in $\Feps$.
	We construct two parametrizations $\alpha$ and $\beta$ from $\CP$ as follows.
	Let $\cell{i_1j_1}, \cell{i_2j_2}, \ldots, \cell{i_Nj_N}$ be the sequence of cells 
	that $\CP$ passes through, where $(i_1,j_1) = (0,0)$ and $(i_N,j_N) = (n,m)$.
	We can assume w.l.o.g. that for any $k$ ($1 \lee k \lee N$),
	the path portion $\CP_k = \CP \cap \cell{i_kj_k}$
	is a line segment. Otherwise, we could replace 
	$\CP_k$ by a line segment connecting the two endpoints of $\CP_k$
	which lies completely inside $\Feps$ (because $\Feps \cap \cell{i_{k}j_{k}}$ is convex),
	and whose slope remains within $[\minS{i_{k}j_{k}}, \maxS{i_{k}j_{k}}]$.

	Let $(p_{k-1},q_{k-1})$ and $ (p_k,q_k)$ be the two endpoints of $\CP_k$.
	The sequence $\sigma = (p_0,q_0), \ldots, (p_N,q_N)$ uniquely represents $\CP$
	(see Figure~\ref{fig:parametrizations}.a).
	We incrementally construct two point sequences $A$ and $B$ from $\sigma$ to represent
	$\alpha$ and $\beta$, respectively.
	Let $t_0 = 0$, $a_0 = 0$, and $b_0 = 0$.
	We start with $A = \set{(t_0, a_0)}$, and $B = \set{(t_0, b_0)}$.
	At each subsequent step $k$ from $1$ to $N$, we update $A$ and $B$ as follows.
	Let $s$ be the slope of $\CP_k$.
	Since $s \in [\minS{i_{k}j_{k}}, \maxS{i_{k}j_{k}}]$,
	there exist a $v_P \in [\umin{P_{i_k}},\umax{P_{i_k}}]$
	and a $v_Q \in [\umin{Q_{j_k}},\umax{Q_{j_k}}]$ 
	such that $s = v_Q / v_P$.
	Let $t = (p_k - p_{k-1})/ v_P = (q_k - q_{k-1})/ v_Q$, and set $t_k = t_{k-1}+ t$.
	We add to $A$ the point $(t_k, p_k)$, and to $B$ the point $(t_k, q_k)$
	(see Figure~\ref{fig:parametrizations}.b).
	The slope of  segment $(t_{k-1},a_{k-1})(t_{k},a_{k})$ is $v_P$, and 
	the slope of segment $(t_{k-1},b_{k-1})(t_{k},b_{k})$ is $v_Q$.
	Therefore, both these newly created segments 
	satisfy the corresponding speed constraints in $\alpha$ and $\beta$.
	Therefore, after the $N$th step, we obtain two point sets $A$ and $B$ of size $N+1$
	that fully define the speed-constrained parametrizations $\alpha$ and $\beta$, respectively.
\end{proof}

\begin{figure}[t]
	\centering
	\includegraphics[width=0.75\columnwidth]{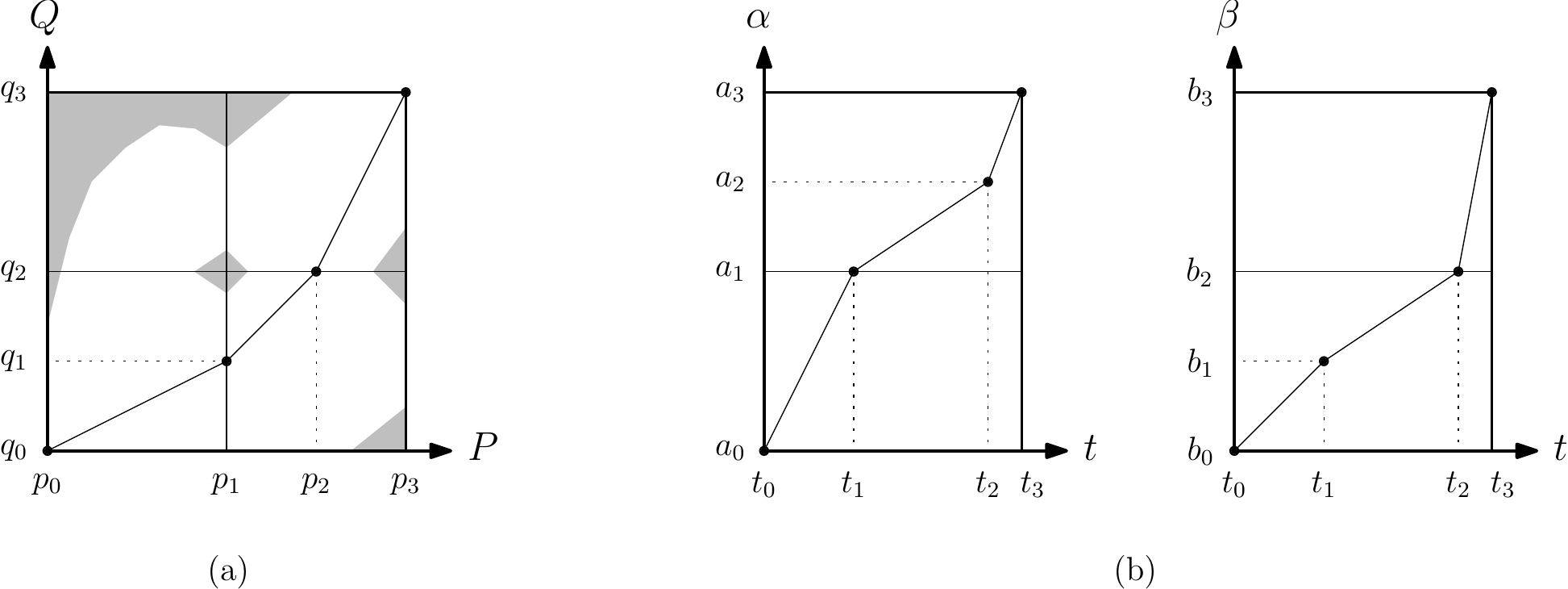}
	\caption{
	(a) A \SC path $\CP$ in the free space of $P$ and $Q$;
	(b) Two speed-constrained parametrizations of $P$ and $Q$,
	corresponding to the path $\CP$.}
	\label{fig:parametrizations}
\end{figure}


\paragraph{A Simple Algorithm.}
We now describe a simple algorithm for the decision problem.
As a preprocessing step, the free space, $\Feps$, is computed by the algorithm.
Let $\LF_{ij} = L_{ij} \cap \Feps$ and $\BF_{ij} = B_{ij} \cap \Feps$.
Since $F_\eps$ is convex within $\cell{ij}$(Section \ref{sec:classicalFD}),
each of $\LF_{ij}$ and $\BF_{ij}$ is a line segment. 
The preprocessing step therefore involves computing line segments
$\LF_{ij}$ and $\BF_{ij}$ for all feasible pairs $(i,j)$, 
which can be done in $O(n^2)$ time.
We then compute the reachability information on the boundary of each cell.
Let $\LR_{ij}$ be the set of reachable points in $L_{ij}$,
and $\BR_{ij}$ be the set of reachable points in $B_{ij}$.
We process the cells in cell-wise order, from $\cell{0,0}$ to $\cell{nm}$,
and at each cell $\cell{ij}$, we propagate the reachability information 
from the entry side of the cell to its exit side, using the following projection function. 
Given a point $p \in \entry{ij}$, the {\em projection} of $p$ onto the exit side of $\cell{ij}$ is defined as
\[
	\proj{ij}(p) = \set{ q \in \exit{ij} \ |  \  \mbox{the slope of $\overline{pq}$ is within } [\minS{ij}, \maxS{ij}]  }.
\]
For a point set $S \subseteq \entry{ij}$, we define 
$ \proj{ij}(S) =\bigcup_{p\in S} \proj{ij}(p)$
(see Figure~\ref{fig:project}.a).
To compute the set of reachable points on the exit side of a cell $\cell{ij}$,
the algorithm first projects $\LR_{ij} \cup \BR_{ij}$ to the exit side of $\cell{ij}$,
and takes its intersection with $\Feps$.
More precisely, the algorithm computes 
 $\LR_{i+1,j}$ and $\BR_{i,j+1}$ from
$\LR_{ij}$, $\BR_{ij}$, $\LF_{i+1,j}$, and $\BF_{i,j+1}$, using the following formula:
$\BR_{i,j+1} \cup \LR_{i+1,j} = \proj{ij}(\LR_{ij} \cup \BR_{ij}) \cap (\BF_{i,j+1} \cup \LF_{i+1,j})$
(see Figure~\ref{fig:project}.b).
Details are provided in Algorithm \ref{alg:main}.

\begin{figure}[t]
	\centering
	\includegraphics[width=0.58\columnwidth]{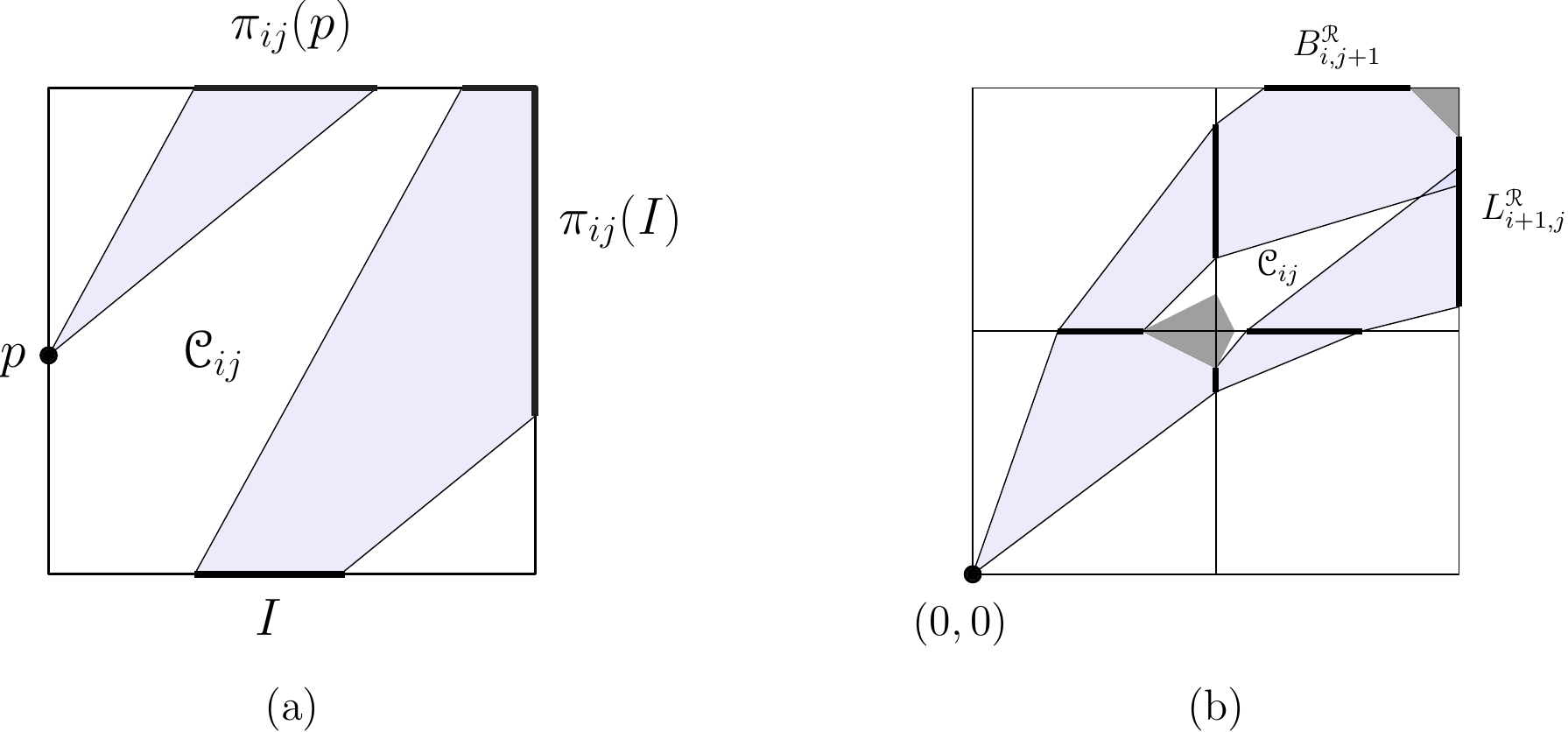}
	\caption{ (a) Projecting a point $p$ and an interval $I$ onto the exit side of $\cell{ij}$;
	(b) Computing reachable intervals on the exit side of a cell $\cell{ij}$.
	Dark gray areas represent infeasible (obstacles) regions.
	Reachable intervals are shown with bold line segments.
	}
	\label{fig:project}
\end{figure}


\vspace{0.5em}
\begin{algorithm} [h]
\caption {\sc Decision Algorithm} \label{alg:main}
\algsetup{indent=1.5em}
\begin{algorithmic}[1]
	\vspace{0.5em}
	\baselineskip=1\baselineskip

	\STATE Compute the free space, $\Feps$ \\
	\STATE Set $\LR_{0,0} = \BR_{0,0} = \set{(0,0)}$, \ 
		$\LR_{i,0} = \emptyset$ for $i \in \set{1, \ldots, n}$, \ 
		$\BR_{0,j} = \emptyset$ for $j \in \set{1, \ldots, m}$  
	\FOR {$i = 0$ to $n$} 
	 	\FOR {$j = 0$ to $m$}
			\STATE\label{lp:1} $\sigma = \LR_{ij} \cup \BR_{ij}$  
			\STATE\label{lp:2} $\lambda = \proj{ij}(\sigma)$  
			\STATE\label{lp:3} $\BR_{i,j+1} = \lambda \cap \BF_{i,j+1}$ 
			\STATE\label{lp:4} $\LR_{i+1,j} = \lambda \cap \LF_{i+1,j}$ 
		\ENDFOR
	\ENDFOR
	\STATE\label{line:last} Return 
``{\sc yes}" if $(n,m) \in \LR_{n+1,m}$, 
``{\sc no}", otherwise. 

\end{algorithmic}
\end{algorithm}
\vspace{0.5em}


\begin{lemma} \label{lemma:cell-process}
	After the execution of Algorithm~\ref{alg:main}, 
	a point $q \in \exit{ij}$ is reachable 	
	iff $q \in \BR_{i,j+1} \cup  \LR_{i+1,j}$.
\end{lemma}

\begin{proof}
	We prove the lemma by induction on the cells in cell-wise order.
	$(\Leftarrow)$ 
	Let $q \in \BR_{i,j+1} \cup  \LR_{i+1,j}$.
	Then, by our construction, there is a point $p \in \LR_{ij} \cup \BR_{ij}$
	such that $q \in \proj{ij}(p)$.
	By induction hypothesis, $p$ is reachable, and therefore,
	there is a \SC path $\CP$ in $\Feps$ connecting $(0,0)$ to $p$.
	Now, $\CP$ concatenated with $\overline{pq}$ is a slope-constrained path from $(0,0)$ to $q$, 
	implying that $q$ is reachable.
	$(\Rightarrow)$ 
	We show that any point $q \in \exit{ij}$ which is not in $\BR_{i,j+1} \cup  \LR_{i+1,j}$ is unreachable.
	Suppose the contrary, i.e., $q$ is reachable.
	Then, there exists a \SC path $\CP$ in $\Feps$ that connects $(0,0)$ to $q$.
	Because the slope of $\CP$ cannot be negative, $\CP$ must
	cross $\entry{ij}$ at some point $p$.
	Now, $p$ is reachable from $(0,0)$, because it is on a \SC path from $(0,0)$ to $p$.
	Therefore, $p \in \LR_{ij} \cup \BR_{ij}$ by induction.
	Consider two line segments $s_1$ and $s_2$ that connect $p$ to $\exit{ij}$
	with slopes $\minS{ij}$ and $\maxS{ij}$, respectively.
	Since $q \not\in \proj{ij}(p)$,
	the portion of $\CP$ that lies between $p$ and $q$ must cross either $s_1$ or $s_2$.
	But, it implies that the slope of $\CP$ at the cross point falls out of the permissible range $[\minS{ij}, \maxS{ij}]$,
	and thus, $\CP$ cannot be slope-constrained: a contradiction.
\end{proof}

\begin{corollary} \label{cor:correctness}
	Algorithm~\ref{alg:main} returns 
`` {\sc yes}" iff $\distFS(P,Q) \lee \eps$.
\end{corollary}
\begin{proof}
	This follows immediately from Lemmas~\ref{lemma:reachable} and \ref{lemma:cell-process}.
\end{proof}

We now show how Algorithm~1 can be implemented efficiently.
Let a {\em reachable interval\/} be a maximal contiguous subset of reachable points on 
the entry side (or the exit side) of a cell.
Therefore, each of $\LR_{ij}$ and $\BR_{ij}$ can be represented as a sequence of reachable intervals.
We make two observations:

\begin{obs} \label{obs:newInterval}
	For each cell $\cell{ij}$, the number of reachable intervals on $\exit{ij}$ is at most one more than
	the number of reachable intervals on $\entry{ij}$.
\end{obs}

\begin{proof}
	\REM{We show that the projection of reachable intervals from the entry side of a cell $\cell{ij}$
	to its exit side can 
	produce at most one new reachable interval.}
	Let $\sigma = \LR_{ij} \cup \BR_{ij}$ be the set of reachable points on $\entry{ij}$,
	and let $\lambda = \proj{ij}(\sigma)$ be the projection of $\sigma$ onto $\exit{ij}$.
	Since the projection on each reachable interval on the exit side is contiguous,
	no reachable interval in $\sigma$ can contribute to more than one reachable interval in $\lambda$.
	Therefore, the number of intervals in $\lambda$ is at most equal to the number of intervals in $\sigma$.
	(Note that projected intervals can merge.)
	However, after splitting $\lambda$ between $L_{i+1,j}$ and $B_{i,j+1}$,
	at most one of the intervals in $\lambda$ (the one containing $L_{i+1,j} \cap B_{i,j+1}$) may split into two, 
	which increases the number of intervals by at most one.
\end{proof}

\REM{
\begin{corollary} \label{cor:numIntervals}
	The number of reachable intervals on the entry side of each cell $\cell{ij}$ is $O(i \times j)$.
\end{corollary}

\begin{proof}
	Since there are at most $i \times j$ cells before $\cell{ij}$ 
	that can contribute to the number of intervals on $\entry{ij}$,
	and we start with one reachable interval (the one containing only (0,0))
	at the beginning, the number of intervals on $\entry{ij}$ can be at most~$i \times j$,
	bu Observation~\ref{obs:newInterval}
\end{proof}
}

\begin{corollary} \label{cor:numIntervals}
	The number of reachable intervals on the entry side of each cell is $O(n^2)$.
\end{corollary}

The above upper bound of $O(n^2)$ is indeed tight as proved in Section~\ref{sec:improvedSpeed}.

\begin{obs} \label{obs:order}
	Let $\seq{I_1, I_2, \ldots, I_k}$ be a sequence of intervals on the entry side of a cell $\cell{ij}$. 
	If $I_1 \lei I_2 \lei \cdots \lei I_k$ then $\proj{ij}(I_1) \lei \proj{ij}(I_2) \lei \cdots \lei \proj{ij}(I_k)$.
\end{obs}

\begin{proof}
	For all $t \in \set{1,\ldots,k}$,
	let $\ell_t$ be the line segment connecting $\Left(I_t)$ to $\Left(\proj{ij}(I_t))$,
	and $r_t$ be the line segment connecting $\Right(I_t)$ to $\Right(\proj{ij}(I_t))$.
	The observation immediately follows from the fact that
	all segments in the set $\set{\ell_t}_{1 \lee t \lee k}$ have slope $\maxS{ij}$ (and thus are parallel),
	and all segments in $\set{r_t}_{1 \lee t \lee k}$ have slope $\minS{ij}$.
	Note that this proof holds even if the intervals in the original sequence and/or
	intervals in the projected sequence overlap each other.
\end{proof}

\REM{
Note that for any two reachable intervals $I_1$ and $I_2$ on the entry side of a cell ($I_1 \not= I_2$),
we have either $I_1 \lei I_2$ or $I_2 \lei I_1$.
Therefore, $\lei$ defines a total order
on the set of reachable intervals on the entry side (and the exit side) of a cell.
As a result, Observation~\ref{obs:order} is applicable to the sequence of reachable intervals as well.
}

\begin{theorem} \label{thm:naive}
	Algorithm~\ref{alg:main} solves the decision problem in $O(n^3)$ time.
\end{theorem}

\begin{proof}
	The correctness of the algorithm follows from Corollary~\ref{cor:correctness}.
	For the running time, we first compute the time needed for processing a cell $\cell{ij}$.
	Let $r_{ij}$ be the number of reachable intervals on the entry side of $\cell{ij}$.
	We use a simple data structure, like a linked list, 
	to store each $\LR_{ij}$ and $\BR_{ij}$ as a sequence of its reachable intervals (sorted in $\lei$ order).
	We show that Lines~\ref{lp:1}--\ref{lp:4} can be performed in $O(r_{ij})$ time.
	In particular,  Line~\ref{lp:1} can be performed by a simple concatenation of two lists in O(1) time;
	and Lines \ref{lp:3} and \ref{lp:4} involve an easy intersection test for each of the intervals in $\lambda$, 
	which takes $O(r_{ij})$ time.
	The crucial part is Line~\ref{lp:2}
	at which reachable intervals are projected. 
	Computing the projection of each interval takes constant time.
	However, we need to merge intersecting intervals afterwards. 
	By Observation~\ref{obs:order}, the merge step can be performed via a linear scan, 
	which takes $O(r_{ij})$ time.
	The overall running time of the algorithm is therefore $O(\sum_{i,j}r_{ij})$.

	Since $r_{ij} = O(n^2)$ by Corollary~\ref{cor:numIntervals}, and there are $O(n^2)$ cells,
	a running time of $O(n^4)$ is immediately implied.
	We can obtain a tighter bound by computing $\sum_{i,j}r_{ij}$ explicitly.
	Define $R_k = \sum_{i+j=k}r_{ij}$, for $0 \lee k \lee 2n$.
	$R_k$ denotes the number of reachable intervals on the 
	entry side of all cells $\cell{ij}$ with $i+j = k$.
	By Observation~\ref{obs:newInterval}, 
	each of the $k+1$ cells contributing to $R_k$ can produce at most 1 new interval.
	Therefore, $R_{k+1} \lee R_{k} + k+1$. 
	Starting with $R_0 = 1$, we get $R_k \lee \sum_{\ell=0}^{k} (\ell + 1) = O(k^2)$. 
	Thus,
	\[
		\sum_{0\lee i,j \lee n} r_{ij} \ \lee \sum_{0 \lee k \lee 2n} R_k \ = \sum_{0 \lee k \lee 2n} O(k^2) \ = O(n^3).
	\]
\end{proof}


\section{An Improved Algorithm} \label{sec:improvedSpeed}

In the previous section,
we provided an algorithm that solves the decision problem in $O(n^3)$ time.
It is not difficult to see that any algorithm which is based on computing the reachability information on
all cells cannot be better than $O(n^3)$ time.
This is proved in the following lemma.

\begin{figure}[t]
	\centering
	\includegraphics[width=0.95\columnwidth]{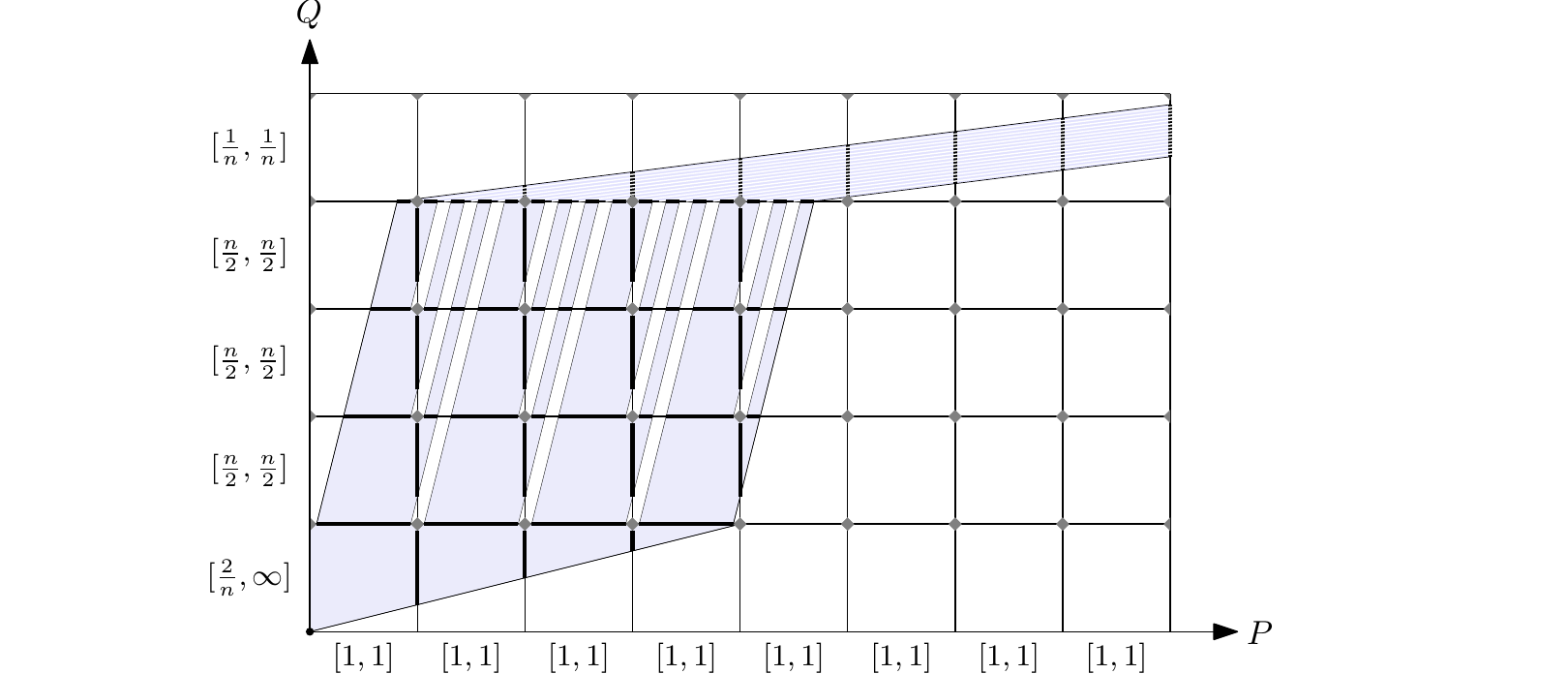}  
	\caption{ A lower bound example.
	The small gray diamonds represent obstacles in the free-space diagram.
	Reachable intervals are shown with bold black line segments.
	The numbers shown at each row and column represent speed limits on the corresponding segment.} 
	\label{fig:lowerbd}
\end{figure}

\begin{lemma} \label{lem:lowerbd}
	For any $n>0$, there exist two polygonal curves $P$ and $Q$ of size $O(n)$ 
	such that in the  free-space diagram corresponding to $P$ and $Q$, 
	there are $\Theta(n)$ cells each having $\Theta(n^2)$ reachable intervals on its entry side.
\end{lemma}

\begin{proof}
	Let $P$ be a polygonal curve consisting of $n$ horizontal segments of unit length
	centered at $(0,0)$,
	and let $Q$ be a polygonal curve consisting of $n/2+1$ vertical segments, 
	where each segment $Q_{2}$ to $Q_{n/2+1}$ has unit length centered at the origin,
	and $Q_1$ has length $1-{\delta}$, for a sufficiently small $\delta \ll 1/n$.
	Let $\eps = \sqrt{1/2- \delta+ \delta^2}$.
	The free-space diagram $\Feps$ for the two curves has a shape like 
	Figure~\ref{fig:lowerbd} (the gray diamond-shape regions show obstacles in the free space
	each having a width of $2\delta$ in $x$ direction).
	We assign the following speed limits to the segments of $P$ and $Q$.
	All segments of $P$ have speed limits $[1,1]$, 
	$Q_1$ has speed limits $[2/n, \infty]$, $Q_2$ to $Q_{n/2}$ have limits $[n/2, n/2]$, and
	$Q_{n/2+1}$ has limits $[1/n,1/n]$.
	The number of reachable intervals on each horizontal line $y=i$ is increased by $n/2$ 
	at each row $i$, for $i$ from 1 to $n/2$, 
	yielding a total number of $\Theta(n^2)$ reachable intervals on the line $y=n/2$.
	Since all these reachable intervals are projected to the right side in the last row, 
	each cell $\cell{i,n/2+1}$ for $i \in \set{n/2+1, \ldots, n}$ has $\Theta(n^2)$ 
	reachable intervals on its entry side.
\end{proof}

While the complexity of the free space is cubic by the previous lemma,
we show in this section that it is possible to
eliminate some of the unneeded computations, and
obtain an improved algorithm that solves the decision problem in $O(n^2 \log n)$ time.
The key idea behind our faster algorithm is to use a ``lazy computation'' technique: 
we delay the computation of reachable intervals until they are actually required.
In our new algorithm, instead of computing the projection of all reachable intervals 
one by one from the entry side of each cell to its exit side, 
we only keep a sorted order of projected intervals,
along with some minimal information
that enables us to compute the exact location of the intervals whenever necessary.

To this end, we distinguish between two types of reachable intervals.
Given a reachable interval $I$ in $\exit{ij}$, 
we call $I$ an \emph{interior interval}
if there is a reachable interval $I'$ in $\entry{ij}$ such that $I = \proj{ij}(I')$,
and we call $I$ a \emph{boundary interval} otherwise.
The main gain, as we see later in this section, is that the exact location of interior intervals
can be computed efficiently based on the location of the boundary intervals.
The following iterated projection is a main tool that we will use.

\paragraph{Iterated Projections.}
Let $I_1$ be a reachable interval on the entry side of a cell $\cell{i_1j_1}$,
and $I_k$ be an interval on the exit side of a cell $\cell{i_kj_k}$.
We say that $I_k$ is an \emph{iterated projection} of $I_1$, 
if there is a sequence of cells $\cell{i_2j_2}, \ldots, \cell{i_{k-1}j_{k-1}}$
and a sequence of intervals $I_2, \ldots, I_{k-1}$ such that for all $1 \lee t \lee k-1$,
$I_t \subseteq \entry{i_tj_t}$ and $I_{t+1} = \proj{i_tj_t}(I_t)$
(see Figure~\ref{fig:projection}).
In the following, we show that $I_k$ can be computed efficiently from $I_1$.

Given two points $p \in \cell{ij}$ and $q \in \cell{i'j'}$,
we say that $q$ is the {\em min projection\/} of $p$, 
if there is a polygonal path $\CP$ from $p$ to $q$ passing through 
a sequence of cells $\cell{i_1j_1}, \cell{i_2j_2}, \ldots, \cell{i_kj_k}$ ($k \gee 1$),
such that $(i_1,j_1) = (i,j)$, $(i_k,j_k) = (i',j')$,
and $\CP \cap \cell{i_tj_t}$ is a line segment whose slope is $\minS{i_tj_t}$,
for all $1 \lee t \lee k$.
The {\em max projection\/} of a point $p$ 
is defined analogously. 

\begin{figure}[t]
	\centering
	\includegraphics[width=0.35\columnwidth]{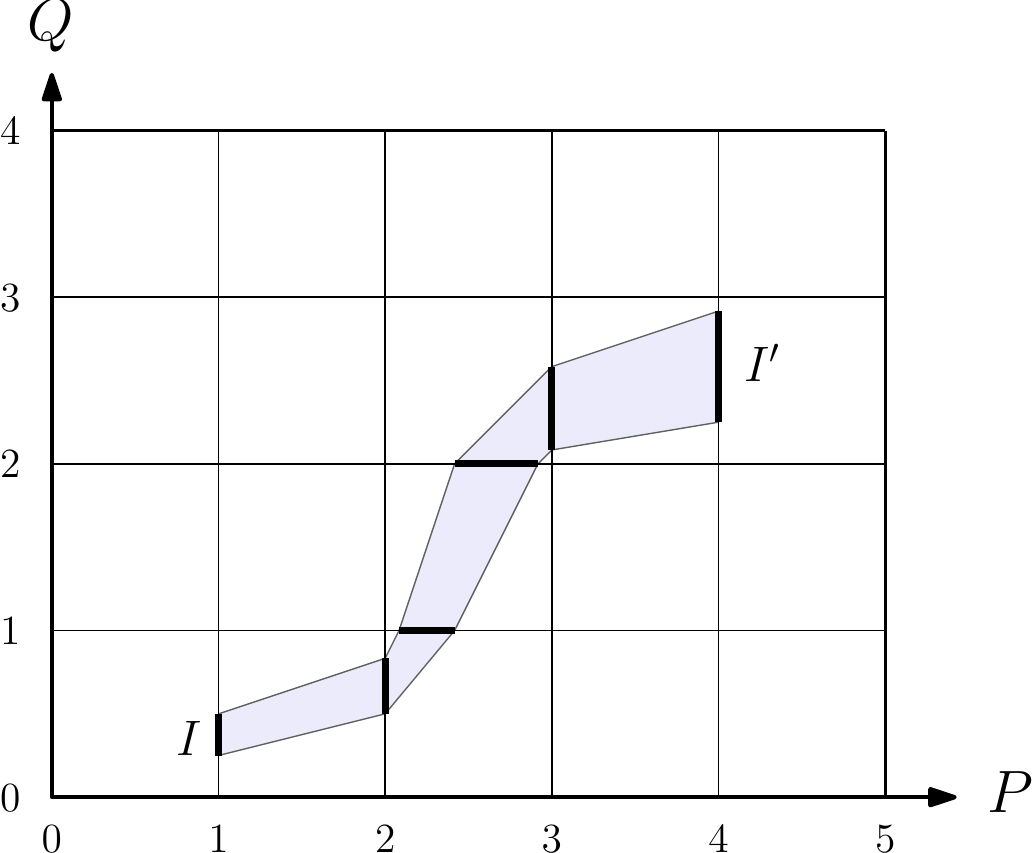}  
	\caption{ $I'$ is an iterated projection of $I$.}
	\label{fig:projection}
\end{figure}

\begin{lemma} \label{lemma:project}
	Using $O(n)$ preprocessing time and space, we can build a data structure 
	that for any point $p \in \BNM$ and any edge $e$ of $\BNM$,
	determines in $O(1)$ time if the min (or the max) projection of $p$ 
	onto the line containing $e$ lies before, after, or on $e$;
	and in the latter case, computes the exact projection of $p$ onto $e$ in constant time.
\end{lemma}

\begin{proof}
	Suppose, w.l.o.g., that $e$ is a vertical edge of $\BNM$, 
	corresponding to a vertex $P(i)$ of $P$ and 
	a segment $Q({j-1})Q(j)$ of Q.
	Then $e = \set{i} \times [j-1,j]$. 
	Let $q$ be the min projection of $p$ on the line $x=i$.
	Let $p=(p_x, p_y)$ and $q=(q_x, q_y)$.
	The path connecting $p$ to $q$ in the definition of the min projection
	has slope $\minS{ij}$ in each cell $\cell{ij}$ it passes through.
	Such a path corresponds to the traversals of two point objects
	$\COB_P$ and $\COB_Q$, where $\COB_P$ traverses $[p_x,q_x]$
	with its maximum permissible speed, and $\COB_Q$ traverses $[p_y,q_y]$
	with its minimum permissible speed.
	\REM{
	Since the length of $[p_x,q_x]$ is known, and the 
	maximum permissible speed for  $\COB_P$ along each interval $[k-1,k]$ 
	is given as $\umax{P_k}$ (see Section~\ref{sec:preliminaries-Speed1}),
	we can easily compute the time $t$ needed for $\COB_P$ 
	to traverse $[p_x,q_x]$ with its maximum speed.
	Now, having time $t$, and the minimum permissible speeds for $\COB_Q$,
	we can easily compute the distance $d$ that $\COB_Q$ walks in $t$ time, 
	if it starts from $p_y$ and always uses its minimum speed.
	We can then simply compute $q_y = p_y +d$.
	If $q_y \in [j-1, j]$, then we conclude that $q$
	lies in $e$ and report its exact location.
	Otherwise, we report that $q$ is before or after $e$,
	depending on whether $q_y < j-1$ or $q_y > j$.
	An analogous method can be used when $e$ is a horizontal edge and/or 
	when max projection is required.
	}%
	Since each of the point objects $\COB_P$ and $\COB_Q$ can traverse $O(n)$ segments,
	computing the min projection can be easily done in $O(n)$ time. 
	However, we can speedup the computation using a simple table lookup technique.
	For $\COB_P$, we keep two arrays $T^P_{\min}$ and  $T^P_{\max}$ of size $n$,
	where for each $i \in \set{1, \ldots, n}$, 
	$T^P_{\min}[i]$ (resp., $T^P_{\max}[i]$) represents the minimum (resp., maximum) time 
	needed for $\COB_P$ to traverse the interval $[0,i]$.
	Similarly, we keep two arrays $T^Q_{\min}$ and  $T^Q_{\max}$ for $\COB_Q$. 
	These four tables can be easily constructed in $O(n)$ time.
	To find time $t$ needed for $\COB_P$ 
	to traverse $[p_x,q_x]$ with its maximum speed, we do the following:
	we first lookup $a=T^P_{\max}[\ceil{p_x}]$ and $b = T^P_{\max}[q_x]$ in $O(1)$ time.
	Clearly, $b-a$ is equal to the time needed for $\COB_P$ to traverse $[\ceil{p_x},q_x]$ 
	(note that $q_x$ is an integer).
	We also compute the time $t'$ needed for $\COB_P$ to traverse $[p_x, \ceil{p_x}]$
	directly from the length of the interval, 
	and the maximum speed of $\COB_P$ in interval $[\ceil{p_x}-1, \ceil{p_x}]$.
	Therefore, $t = t' + b - a$  can be computed in $O(1)$ time total.
	By similar table lookups, we compute the times $t_1$ and $t_2$
	needed for $\COB_Q$ to traverse $[p_y, j-1]$ and $[p_y, j]$, respectively,
	with its minimum speed.
	If $t_1 \lee t \lee t_2$, then we conclude that $q_y$ lies in $e$, 
	and we can easily compute its exact location on $e$ by computing the distance that
	$\COB_Q$ traverses in $t- t_1$ time using its minimum speed on interval $[j-1,j]$.
	Otherwise, we output that $q$ is before or after $e$,
	depending on whether $t < t_1$ or $t > t_2$, all in $O(1)$ time.
\end{proof}

\begin{corollary} \label{cor:project}
	If $I'$ is an iterated projection of $I$, 
	then $I'$ can be computed from $I$ in $O(1)$ time, after $O(n)$ preprocessing time. 
\end{corollary}

\begin{proof}
	This is a direct corollary of Lemma~\ref{lemma:project}
	and the fact that if $I'=[a',b']$ is an iterated projection of $I=[a,b]$,
	then $a'$ is the max projection of $a$, and $b'$ is the min projection of $b$.
\end{proof}

\paragraph{The Data Structure.}
The main data structure that we need in our algorithm is a dictionary 
for storing a sorted sequence of intervals. 
A balanced binary search tree can be used for this purpose.
Let $T$ be the data structure that stores a sequence 
$\seq{I_1, I_2, \ldots, I_k}$ of intervals in $\lei$ order.
We need the following operations to be supported by $T$.

\begin{itemize}
	\item[] {\sc Search:} 
		Given a point $x$, find the leftmost interval $I$ in $T$ such that $x \lee \Left(I)$.
	\item[] {\sc Insert:} 
		Insert a new interval $I$ into $T$, right before $T.\mbox{\sc Search}(\Left(I))$,
		or at the end of $T$ if $I$ is to the right of all existing intervals in $T$. 
		In our algorithm, inserted intervals 
		are not properly contained in any existing interval of $T$, and therefore, the
		resulting sequence is always sorted.
	\item[] {\sc Delete:} 
		Delete an existing interval $I$ from $T$. 
	\item[] {\sc Split:} 
		Given an interval $I = I_j$, $1 < j \lee k$, split $T$ into 
		two data structures $T_1$ and $T_2$, containing 
		$\seq{I_1, \ldots, I_{j-1}}$ and  $\seq{I_j, \ldots, I_k}$, respectively. 
	\item[] {\sc Join:} 
		Given two data structures with interval sequences $\CI_1$ and $\CI_2$,
		where each interval in $\CI_1$ is before any interval in $\CI_2$,
		join the two structures to obtain a single structure $T$ 
		containing the concatenated sequence $\CI_1 \cdot \CI_2$.
\end{itemize}

It is straightforward to 
modify a standard balanced binary search tree
to perform all the above operations in $O(\log |T|)$ time
(for example, see Chapter 4 in \cite{Tarjan83}).
Note that the exact coordinates of the interior intervals are not explicitly stored in the data structure.
Rather, we compute the coordinates on the fly whenever a comparison is made, 
in $O(1)$ time per comparison,  using Corollary~\ref{cor:const}.

\paragraph{The Algorithm.}

Let $\LT_{ij}$ (resp., $\BT_{ij}$) denote the balanced search tree
storing the sequence of reachable intervals on $L_{ij}$ (resp., on $B_{ij}$).
The reachable intervals stored in the trees are not necessarily disjoint.
In particular, we allow interior intervals to have overlaps with each other, 
but not with boundary intervals.
Moreover, the exact locations of the interior intervals are not explicitly stored.
However, we maintain the invariant that 
each interior interval can be computed in $O(1)$ time, and that 
the union of the reachable intervals 
stored in $\LT_{ij}$ (resp., in $\BT_{ij}$) 
at each time is equal to  $\LR_{ij}$ (resp., $\BR_{ij}$).

The overall structure of the algorithm is similar to that of Algorithm~1.
We process the cells in cell-wise order,
and propagate the reachability information through each cell
by projecting the reachable intervals from the entry side to the exit side.
However, to get a better performance, 
cells are processed in a slightly different manner, as presented in Algorithm~2.
In this algorithm, $\exit{ij}$ is considered
as a single line segment whose points are ordered by $\lei$ relation.
For a set $S$ of intervals, we define $\union(S) = \bigcup_{I \in S} I$.
Given a data structure $T$ as defined in the previous subsection,
we use $T$ to refer to both the data structure and the 
set of intervals stored in $T$.
Given a point set $S$ on a line, by an interval (or a segment) of $S$ we mean 
a maximal continuous subset of points contained in $S$.

\vspace{0.5em}
\begin{algorithm} [h]
\caption {\sc Improved Decision Algorithm} \label{alg:improved}
\algsetup{indent=1.5em}
\begin{algorithmic}[1]
	\vspace{0.5em}
	\baselineskip=1\baselineskip

	\STATE\label{l:0} Compute the free space, $\Feps$
	\STATE\label{l:1} \For $i \in \set{0, \ldots, n}$ \Do $\LT_{i,0} = \emptyset$ 
	\STATE\label{l:2} \For $j \in \set{0, \ldots, m}$ \Do $\BT_{0,j} = \emptyset$
	\STATE\label{l:3} $\LT_{0,0}.\mbox{\sc Insert}([o,o])$ where $o = (0,0)$
	\FOR {$i = 0$ to $n$} \label{l:loop}
	 	\FOR {$j = 0$ to $m$}
			\STATE\label{l:cat} $T = \mbox{\sc Join}(\LT_{ij}, \BT_{ij})$  
			\STATE\label{l:cop} Project $T$ to the exit side of $\cell{ij}$
			\STATE\label{l:s} $S = \set{I \in T \ | \ I \not\subseteq \BF_{i,j+1} \mbox{ and } I \not\subseteq  \LF_{i+1,j} }$
			\STATE\label{l:del} \Foreach $I \in S$ \Do $T.\mbox{\sc Delete}(I)$
			\STATE\label{l:spl} $(\BT_{i,j+1}, \LT_{i+1,j}) = T.\mbox{\sc Split}(T.\mbox{\sc Search}((i,j)))$
			\STATE\label{l:ins1} \Foreach $I \subseteq (\union(S) \cap \BF_{i,j+1})$ \Do $\BT_{i,j+1}.\mbox{\sc Insert}(I)$ 
			\STATE\label{l:ins2} \Foreach $I \subseteq (\union(S) \cap \LF_{i+1,j})$ \Do $\LT_{i+1,j}.\mbox{\sc Insert}(I)$ 
		\ENDFOR
	\ENDFOR
	\STATE\label{l:last} Return {\sc yes} if $(n,m) \in \LT_{n+1,m}$, {\sc no} otherwise. 

\end{algorithmic}
\end{algorithm}
\vspace{0.5em}

The algorithm works as follows.
We first compute $\Feps$ in Line~\ref{l:0}.
Lines~\ref{l:1}--\ref{l:3} initializes the data structures 
for the first row and the first column of $\BNM$.
Lines~\ref{l:loop}--\ref{l:ins2} process the cells in cell-wise order.
For each cell $\cell{ij}$, Lines~\ref{l:cat}--\ref{l:ins2} 
propagate the reachability information through $\cell{ij}$
by creating data structures $\BT_{i,j+1}$ and $\LT_{i+1,j}$ on the exit side of $\cell{ij}$, 
based on $\BT_{ij}$ and $\LT_{ij}$, and the feasible intervals $\BF_{i,j+1}$ and $\LF_{i+1,j}$.
In Line~\ref{l:cat}, a data structure $T$ is obtained by joining the interval sequences in
$\BT_{ij}$ and $\LT_{ij}$.
We then project $T$ to the exit side of $\cell{ij}$ in Line~\ref{l:cop}
by (virtually) transforming each interval $I \in T$ to an interval $\proj{ij}(I)$ on $\exit{ij}$.
Since the projection preserves the relative order of intervals,
 by Observation~\ref{obs:order},
and we do not need to explicitly update the location of interior intervals on the exit side,
the projection is simply done by copying $T$ to the exit side of $\cell{ij}$ 
(boundary intervals will be fixed later in Lines~\ref{l:ins1}--\ref{l:ins2}).
Furthermore, since $\BT_{ij}$ and $\LT_{ij}$ are not needed afterwards in the algorithm,
we do not actually duplicate $T$. 
Instead, we simply assign $T$ to the exit side, without making a new copy.
In Line~\ref{l:s}, we determine a set $S$
of intervals that are not completely contained in $\BF_{i,j+1}$ or in $\LF_{i+1,j}$.
All such intervals are deleted from $T$ in Line~\ref{l:del}
(see Figure~\ref{fig:cell-intervals} for an illustration).
The remaining intervals in $T$ have no intersection with the corner point $(i,j)$.
Therefore, we can easily split $T$ in Line~\ref{l:spl} 
into two disjoint data structures, $\BT_{i,j+1}$ and $\LT_{i+1,j}$,
each corresponding to one edge of the exit side.
In Lines~\ref{l:ins1}--\ref{l:ins2} we insert the boundary intervals
to $\BT_{i,j+1}$ and $\LT_{i+1,j}$,
which are computed as those portions of $\union(S)$ that lie inside $\Feps$.
Note that whenever a boundary interval $I$ is inserted into a data structure,
its coordinates are stored along with the interval.
After processing all cells, the decision problem is easily answered in Line~\ref{l:last}
of the algorithm by checking if the target point $(n,m)$ is reachable.

\begin{figure}[t]
	\centering
	\includegraphics[width=0.35\columnwidth]{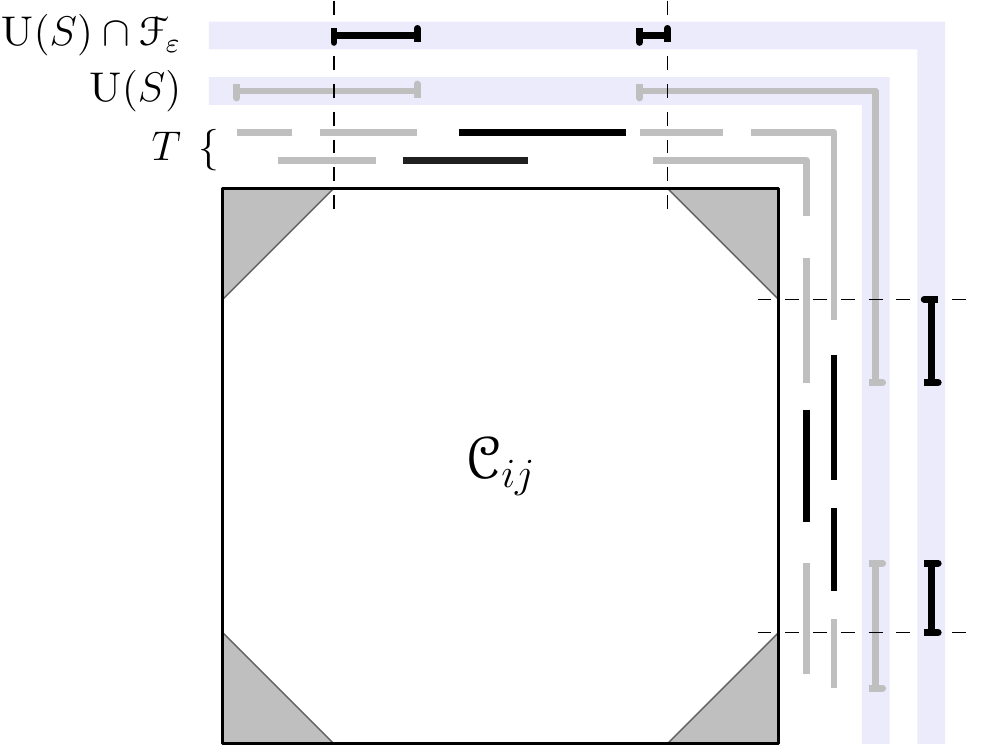}  
	\caption{An example of the execution of Algorithm 2 on a cell $\cell{ij}$. 
		The intervals of $S \subseteq T$ are shown in gray. 
		The black intervals in $T$ represent the interior intervals.
		The intervals in 	$\union(S) \cap \Feps$ are boundary intervals 
		which are inserted in Lines~\ref{l:ins1}--\ref{l:ins2}.
	}
	\label{fig:cell-intervals}
\end{figure}

\begin{lemma} \label{lemma:prop}
	After processing each cell $\cell{ij}$, 
	the following statements hold true:
	\begin{enumerate}
		\item[\em (i)] any interval inserted into $\exit{ij}$ in Lines~\ref{l:ins1}--\ref{l:ins2} is a boundary interval,
		\item[\em (ii)] each interior interval on $\exit{ij}$	can be expressed as an iterated projection of a boundary interval.
	\end{enumerate}
\end{lemma}
\begin{proof}
	(i) This is easily shown by observing that no interior interval is added to $S$ in Line~\ref{l:s},
	and therefore, $\union(S)$ cannot completely contain any interior interval.
	(ii) The proof is by induction on the cells in cell-wise order.
	Let $I$ be an interior interval on $\exit{ij}$. 
	Then, $I$ is a direct projection of an interval $I' \subseteq \entry{ij}$ obtained in  Line~\ref{l:cop}.
	If $I'$ is a boundary interval, then we are done.
	Otherwise, $I'$ is an interior interval, and therefore, 
	it is by induction an iterated projection of another boundary interval $I''$.
	Since $I = \proj{ij}(I')$ and $I' \subseteq \entry{ij}$,
	$I$ is in turn an iterated projection of~$I''$.
\end{proof}

\begin{corollary} \label{cor:const}
	After processing each cell $\cell{ij}$, 
	the exact location of each reachable interval on $\exit{ij}$ is accessible in $O(1)$ time.
\end{corollary}

\begin{proof}
	Fix a reachable interval $I$ on $\exit{ij}$.
	If $I$ is a boundary interval, then by Lemma~\ref{lemma:prop}(i), 
	it is inserted into a data structure by Lines~\ref{l:ins1}--\ref{l:ins2},
	and hence, its coordinates are stored in the data structure upon insertion.
	If $I$ is an interior interval, then by Lemma~\ref{lemma:prop}(ii),
	it is an iterated projection of a boundary interval, and hence, its location
	can be computed in $O(1)$ time using Corollary~\ref{cor:project}.
\end{proof}

\begin{lemma} \label{lemma:invar}
	After processing each cell $\cell{ij}$,
	$\BR_{i,j+1} \cup \LR_{i+1,j} = \union(\BT_{i,j+1} \cup \LT_{i+1,j})$.
\end{lemma}

\begin{proof}
	We prove the statement by induction on the cells in cell-wise order.
	Recall from Section~\ref{sec:decisionSpeed} (Algorithm~1) 
	that $\BR_{i,j+1} \cup \LR_{i+1,j} = \proj{ij}(\LR_{ij} \cup \BR_{ij}) \cap (\BF_{i,j+1} \cup \LF_{i+1,j})$.
	Therefore, it satisfies to show that 
	$\union(\BT_{i,j+1} \cup \LT_{i+1,j}) = \proj{ij}(\LR_{ij} \cup \BR_{ij}) \cap (\BF_{i,j+1} \cup \LF_{i+1,j})$.
	By Line~\ref{l:cat}, $\union(T) = \union(\LT_{ij} \cup \BT_{ij})$.
	Let $T_1$ be the set of intervals in $T$ right after the execution of Line~\ref{l:cop},
	$S$ be the set of intervals deleted in Line~\ref{l:del},
	$N$ be the set of new intervals inserted in Lines~\ref{l:ins1}--\ref{l:ins2},
	and $T_2 = (T_1 \bslash S) \cup N$.
	Fix a point $p \in \union(T_1)$, and let $K$ be the set of intervals in $T_1$ containing $p$.
	We distinguish between two cases: 
	\begin{itemize}
		\item $p \in \Feps$: There are two possibilities:	
		(1) $K \not\subseteq S$:
			Here, there is an interval in $K$ that remains in $T_1$ after
			deletion of $S$ in Line~\ref{l:del}. 
			Therefore, $p \in \union(T_2)$.
		(2) $K \subseteq S$:
			Here, all intervals of $K$ are removed in Line~\ref{l:del}.
			However, since $p \in \Feps$, there is an interval $I \in N$ such that $p \in I$.
			Therefore, after insertion of $I$ in Lines~\ref{l:ins1}--\ref{l:ins2}, we have
			$p \in \union(T_2)$.

		\item  $p \not\in \Feps$: 
		In this case, $K \subseteq S$,
		and hence $p \not\in \union(T_1 \bslash S)$. 
		Moreover, no interval in $N$ can contain $p$.
		Therefore, $p \not\in \union(T_2)$.
	\end{itemize}
	The above two cases together show that $\union(T_2) = \union(T_1) \cap \Feps$.
	Note that, $\union(T_1) = \proj{ij}(\union(\LT_{ij} \cup \BT_{ij}))$ (by Lines~\ref{l:cat} and \ref{l:cop}),
	and $T_2 = \BT_{i,j+1} \cup \LT_{i+1,j}$.
	Therefore, $\union(\BT_{i,j+1} \cup \LT_{i+1,j}) = \proj{ij}(\union(\LT_{ij} \cup \BT_{ij})) 
	 \cap (\BF_{i,j+1} \cup \LF_{i+1,j})$,
	which completes the proof, because $\LR_{ij} \cup \BR_{ij} = \union(\LT_{ij} \cup \BT_{ij})$
	by induction.
\end{proof}

\begin{theorem}
	Algorithm~\ref{alg:improved} solves the decision problem in $O(n^2 \log n)$ time.
\end{theorem}

\begin{proof}
	The correctness of the algorithm follows from Lemma~\ref{lemma:invar},
	combined with Lemma~\ref{lemma:cell-process}.
	For the running time, we compute the number of operations
	needed to process each cell $\cell{ij}$ in Lines~\ref{l:cat}--\ref{l:ins2}.
	Let $\OP$ denote the time needed for each data structure operation.
	Line~\ref{l:cat} needs one join operation that takes $O(\OP)$ time.
	Line~\ref{l:cop} consists of a simple assignment taking only $O(1)$ time.
	To compute the subset $S$ in Line~\ref{l:s}, we start walking from 
	the two sides of $T$, and add intervals to $S$ until we reach 
	the first intervals from both sides that do not belong to $S$.
	Moreover, we find the interval $I = T.\mbox{\sc Search}((i,j))$, 
	and start walking around $I$ in both directions until we find all consecutive intervals around $I$ 
	that lie in $S$ (see Figure~\ref{fig:cell-intervals}).
	To check if an interval lies in $S$ or not, we need to compute the coordinates
	of the interval that can be done in $O(1)$ time.
	Therefore, computing $S$ takes $O(|S|+\OP)$ time in total.
	Line~\ref{l:del} requires $|S|$ delete operation that takes $O(|S| \times \OP)$ time.
	Line~\ref{l:spl} consists of a split operation taking $O(\OP)$ time.
	The set $\union(S)$ used in Lines~\ref{l:ins1}--\ref{l:ins2}
	can be computed in $O(|S|)$ time by a linear scan over the set $S$.
	Since $\union(S)$ consists of at most three segments (see Figure~\ref{fig:cell-intervals}),
	computing $\union(S) \cap \Feps$ in Lines~\ref{l:ins1}--\ref{l:ins2} takes constant time.
	Moreover, there are at most four insertion operations in Lines~\ref{l:ins1}--\ref{l:ins2}
	to insert boundary intervals.
	Therefore, Lines~\ref{l:ins1}--\ref{l:ins2} takes $O(|S|+\OP)$  time.
	Thus, letting $s_{ij} = |S|$,
	processing each cell $\cell{ij}$ takes $O((s_{ij}+1) \times \OP)$ time in total.
	Since at most four new intervals are created at each cell,
	the total number of intervals created over all cells is $O(n^2)$.
	Note that any of these $O(n^2)$ intervals can be deleted at most once,
	meaning that $\sum_{i,j} s_{ij} = O(n^2)$. 
	Moreover, 	each comparison made in the data structures 
	takes $O(1)$ time by Corollary~\ref{cor:const},
	and hence, $\OP = O(\log n)$.
	Therefore, the total running time of the algorithm is 
	O($\sum_{i,j} (s_{ij}+1) \log n) = O( n^2 \log n )$.
\end{proof}


\section{Optimization Problem} \label{sec:optimization}

In this section, we describe how our decision algorithm can be used to 
compute the exact value of the \Frechet distance with speed limits between two polygonal curves. 
Let $\LF_{ij} = [a_{ij}, b_{ij}]$ and  $\BF_{ij} = [c_{ij}, d_{ij}]$.
Notice that the free space, $\Feps$, is an increasing function of $\eps$.
That is, for $\eps_1 \lee \eps_2$, we have $\CF_{\eps_1} \subseteq \CF_{\eps_2}$.
It is not hard to see that:

\begin{obs}
\label{obs:criticaltypes}
To find the exact value of $\delta = \distFS(P,Q)$,
we can start from $\eps = 0$, and continuously increase $\eps$ until
we reach the first point at which $\Feps$ contains a \SC path from $(0,0)$ to $(n,m)$.
This occurs at only one of the following ``critical values'':
\begin{itemize}
	\item[(A)] smallest $\eps$ for which $(0,0) \in \Feps$ or $(n,m) \in \Feps$,
	\item[(B)] smallest $\eps$ at which $\LF_{ij}$ or $\BF_{ij}$  becomes non-empty for some pair $(i,j)$,
	\item[(C)] smallest $\eps$ at which $b_{k\ell}$ is the min projection of $a_{ij}$, or 
		$d_{ij}$ is the max projection of $c_{k\ell}$, for some $i,j,k$, and $\ell$ .
\end{itemize}

\end{obs}

\begin{figure}[t]
	\centering
	\includegraphics[width=0.8\columnwidth]{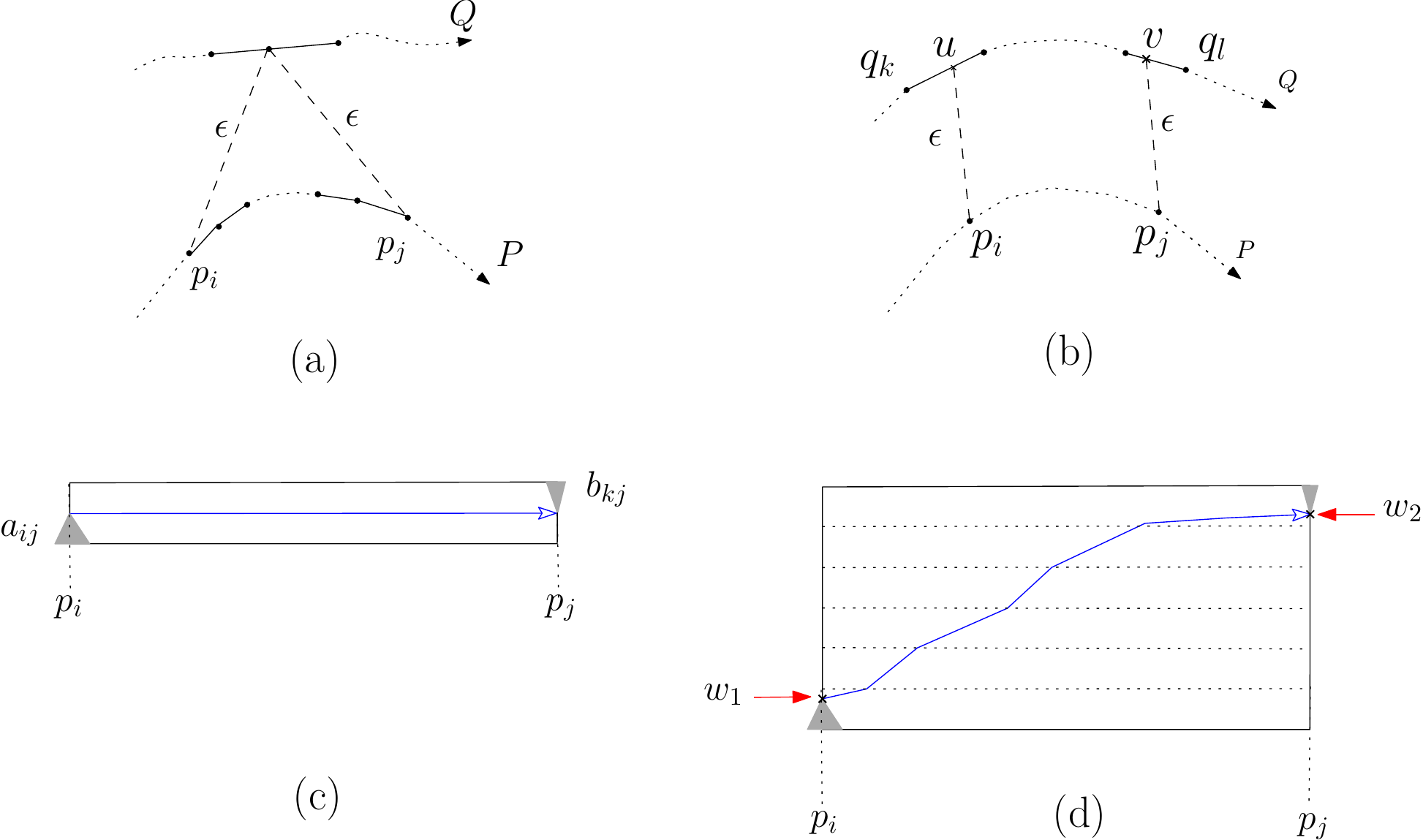}  
	\caption{(a,c) Type (C) critical distances in the standard \Frechet distance problem vs. 
	(b,d) type (C) critical distances in our instance of the problem.}
	\label{fig:typeCCC}
\end{figure}

Notice that here
type (A) and (B) of critical values are similar to 
the type (A) and (B) critical values in the standard \Frechet 
distance problem
(see Section \ref{sec:classicalFD} on Page \pageref{alg:StandardFDec}).
There are two critical distances of type (A) and 
$O(n^2)$ critical distance of type (B). All of these critical values 
can be computed in  $O(n^2)$  time.

Here, type (C) critical distances are slightly different from 
those distances in the standard 
\Frechet distance problem. 
Figure \ref{fig:typeCCC} illustrates that difference. 
In the standard \Frechet problem, a type (C) critical distance
corresponds to the common distance of two vertices of one curve 
to the intersection point of their bisector 
with an edge of the other curve (see Figure \ref{fig:typeCCC}a). This happens 
when a new horizontal or vertical passage opens within the diagram
(see Figure \ref{fig:typeCCC}c).
All  type (C) critical values in the standard \Frechet distance problem 
can be computed in $O(n^3)$ time. 

In our instance of the problem, computing type (C) critical distances
has further complications.  
Those distances arise when a new slope-constrained path opens within $\Feps$
which consists of a sequence of min-slopes (or max-slopes)
of the cells through which the path goes.
If $\eps$ is reduced, this path will seize to exist 
(for an instance, see Figure \ref{fig:typeCCC}d).

The geometric meaning of type (C) critical distances is as follows
(see Figure \ref{fig:typeCCC}b for an illustration).
Consider two vertices $p_i$ and $p_j$ from $P$ and
let $t_{p_ip_j}$ denote the time it takes for $\CO_p$
to travel from $p_i$ to $p_j$ on $P$ when the speed of $\CO_p$ on each 
segment is its corresponding maximum allowed speed.
Furthermore, let $q_k$ and $q_\ell$
be two vertices of $Q$ where 
$q_\ell$ is the first vertex after $q_k$
where $t_{p_ip_j} \le t_{q_\ell q_k}$ ($\CO_Q$ walks always 
with minimum allowed speed assigned to the segments of $Q$).
Now, let $u$ and $v$  be two points on $Q$,
where: 

(a) $u$ is before $v$ on $Q$,

(b) both are located between $q_k$ and $q_\ell$,

(c) $\|  u p_i\| = \|  v p_j\|$,

(d) and the time of travel from $u$ to $v$ is equal to $t_{p_ip_j}$. 

Then, among all such pairs of points $(u,v)$, 
let $(u_0,v_0)$ be the one which has the smallest distance. 
Since we are looking for the smallest distance where 
slope-constrained path opens up in the free-space diagram, 
$\| u_0p_i\|$ is a critical distance of type (C).

Next, we show how to compute all type (C) critical distances.
We first introduce a function, called {\sc Compute-Potential-Chains$(R, t)$}
provided in Algorithm \ref{alg:potentialChains}.
Input to that function consists of a curve $R$ and a fixed time $t \neq 0$.
The function computes a set $A$ which includes all 
the subcurves of $R$ from vertex $r_i$ to vertex $r_j$, $i < j$, where $r_j$ 
is the first vertex after $r_i$ such that $t \le t_{r_ir_j}$.
Algorithm \ref{alg:potentialChains} accomplishes this by using two pointers, 
called $\mu_1$ and 
$\mu_2$. At the start of the algorithm, $\mu_1$ points to the 
first vertex and $\mu_2$ points to the second vertex of $R$. 
Then $R$ is scanned once to report the set $A$ as described in 
Algorithm \ref{alg:potentialChains}.
In this algorithm, $chain_Q(\mu_1,\mu_2)$ means the polygonal chain 
of $Q$ which starts at $\mu_1$ and ends at $\mu_2$.

In Algorithm \ref{alg:typeCFinal}, we use the function stated 
in Algorithm \ref{alg:potentialChains}
to compute all the critical distances of type (C), 
for two curves $P$ and $Q$. 
For every pair of vertices $p_i$ and $p_j$ of $P$,
we call function {\sc Compute-Potential-Chains}$(Q, t_{p_ip_j})$ 
to compute subcurves $\alpha$ of $Q$ which start
at some vertex $q_k$ and end at some vertex $q_\ell$,
$k < \ell$, such that $q_\ell$ is the first vertex 
after $q_k$ where $t_{p_ip_j} \le t_{q_kq_\ell}$.
Then, for each curve $\alpha$, 
we do the calculation in Line 4 to compute critical distances of type (C).
We repeat the above for each pair of vertices $q_i$ and $q_j$ of $Q$
and curve $P$, in Line \ref{l:reprep}.
See Algorithm \ref{alg:typeCFinal} for more details.

\begin{algorithm} [t]
\caption {\sc Compute-Potential-Chains$(R, t)$ } 
\label{alg:potentialChains}
\algsetup{indent=1.5em}
\begin{algorithmic}[1]
	\vspace{0.5em}
	\baselineskip=1\baselineskip
	\STATE $A = \emptyset	$
	\STATE Let $(r_1,r_2,\dots,r_m)$ be the vertices of $R$

	\IF{$ t \le t_{r_1r_m} $}

	\STATE\label{l:0} $i = 1, j = 2$
	
	\STATE\label{l:0} $\mu_1 =r_i, \mu_2 = r_j$
	\WHILE{$\mu_1 \neq r_m$} \label{l:loopPot}

	\IF{$ t \le t_{\mu_1\mu_2}$}
	\STATE $A = A \cup chain_Q(\mu_1,\mu_2)$ 
	\STATE $i = i + 1$, $\mu_1 =r_i$
	\ELSE 
	\STATE $j = j + 1$, $\mu_2 =r_j$
	\ENDIF
	\ENDWHILE
	\ENDIF
	\STATE  return $A$

\end{algorithmic}
\end{algorithm}

 \begin{algorithm} [t]
\caption {\sc Compute type(c) critical distances} \label{alg:typeCFinal}
\algsetup{indent=1.5em}
\begin{algorithmic}[1]

\FOR { each pair $(p_i, p_j)$,  $0 \le i <  j \le n$}
\STATE A = {\sc Compute-Potential-Chains$(Q, t_{p_ip_j})$}
\FOR{each curve $\alpha$ in A}
\STATE let $(e_1, e_2, \dots, e_k)$ be the list of edges of $\alpha$, \\
 determine if there exists pairs of points $u \in e_1$, $v \in e_k$, \\
such that $\| up_i  \| = \| vp_j\|$ and $t_{uv} =  t_{p_ip_j}$
\\
among all such pairs, 

add  minimum of the distances $\| up_i \|$ to the critical distances of type (C).
\ENDFOR	
\ENDFOR

\FOR { each pair $(q_i, q_j)$,  $0 \le i <  j \le m$}
\STATE A = {\sc Compute-Potential-Chains$(P, t_{q_iq_j})$}
\STATE Repeat Lines 3 and 4 for each curve $\alpha$ in $A$ \label{l:reprep} 
\ENDFOR	

\REM{
\FOR { each $q_i, q_j, i, j \le m$}
\STATE B = $Compute-Potential-Chains(P, t_{q_iq_j})$
\FOR{each curve $\beta$ in B}
\STATE determine if a point $u$ exists on the first edge of $\beta$, \\
and a point $v$ exists on the last edge of $\beta$, \\
such that $\| uq_i  \| = \| vq_j\|$. \\
if exists, add distance $\| uq_i \|$ to critical distance of type (C)

\ENDFOR	
\ENDFOR		
}

\end{algorithmic}
\end{algorithm}

\begin{lemma}
\label{TimeOfTypeC}
Algorithm \ref{alg:typeCFinal} computes all critical values of type (C) 
in $O(n^3)$ total time. 
\end{lemma}

\begin{proof}
The correctness of Algorithm~\ref{alg:typeCFinal} follows from 
 Observation~\ref{obs:criticaltypes} and 
the geometric nature of type (C) critical distances
as described above. 

Algorithm \ref{alg:typeCFinal}  calls the function stated as Algorithm
\ref{alg:potentialChains}, $O(n^2)$ times
in Line 2.
Thus, to prove the lemma, it is sufficient to show that 
the running time of Algorithm  \ref{alg:potentialChains} is linear in the size 
of curve $R$.

Notice that the speed of travel on curve $R$ in Algorithm \ref{alg:potentialChains}
is equal to the minimum allowed speed assigned to each segment of $R$.
Thus, using the same approach as 
in Lemma \ref{lemma:project}, 
after linear time preprocessing, 
we can compute, in constant time, the time 
of travel from a vertex to another one.

The loop in Line \ref{l:loopPot} terminates  
when pointer $\mu_1$ reaches the last vertex of $R$.
Notice that pointer $\mu_1$ always moves forward in direction $R$ and 
points to vertices of $R$ one by one, in order. Also,
pointer $\mu_2$  always moves forward in direction $R$ and is never 
before $\mu_1$. Therefore, with one linear scan, 
Algorithm \ref{alg:potentialChains} computes and returns set $A$.

Next, we show that the computation in Line 4 of 
Algorithm \ref{alg:typeCFinal} can be 
done in $O(1)$ time.
Let 
$e_1 = ab$ and $e_k = cd$ be the first 
and last edges of $\alpha$
(see Figure \ref{fig:TimeOfTypeC}).
Suppose that the coordinate of the points in that figure are:
\begin{displaymath}
	a=(a_x,a_y), b=(b_x,b_y), c =(c_x,c_y), d=(d_x,d_y), p_i=(p_x,p_y), p_j=(q_x,q_y)
\end{displaymath}
Then, any point $u$ on segment $\Seg{ab}$ can be written as:
\begin{displaymath}
u = (b_x,b_y) + \frac{\|ub\|}{\|ab\|} (a_x-b_x,a_y -b_y)
\end{displaymath}

and any point $v$ on segment $\Seg{cd}$ can be written as:
\begin{displaymath}
v = (c_x,c_y) + \frac{\|cv\|}{\|cd\|} (d_x-c_x,d_y -c_y)
\end{displaymath}

We are looking for pairs of points $u$  and $v$ such that:

\begin{displaymath}
	\|p_iu\|^2 = \|p_jv\|^2 	
\end{displaymath}
\begin{displaymath}
	 \frac{\|ub\|}{v_{e_1}} + \frac{\|cv\|}{v_{e_k}} = t_{p_ip_j} - t_{bc}	
\end{displaymath}

Thus,

\begin{displaymath}
	(b_x-p_x+\frac{\|ub\|}{\|ab\|}(a_x-b_x))^2 + (b_y-p_y+	\frac{\|ub\|}{\|ab\|}(a_y-b_y))^2 		
\end{displaymath}

\begin{displaymath}
= 
\end{displaymath}

\begin{displaymath}
	(c_x-q_x+\frac{\|cv\|}{\|cd\|}(d_x-c_x))^2 + (c_y-q_y+	\frac{\|cv\|}{\|cd\|}(d_y-c_y))^2	
\end{displaymath}

\begin{displaymath}
	 \frac{\|ub\|}{v_{e_1}} + \frac{\|cv\|}{v_{e_k}} = t_{p_ip_j} - t_{bc}	
\end{displaymath}

\REM{
Replacing $\|ub\|$ by $S$ and $\|cv\|$ by $T$, we get:

\begin{displaymath}
	(b_x-p_x+\frac{S}{\|ab\|}(a_x-b_x))^2 + (b_y-p_y+	\frac{S}{\|ab\|}(a_y-b_y))^2 		
\end{displaymath}

\begin{displaymath}
= 
\end{displaymath}

\begin{displaymath}
	(c_x-q_x+\frac{T}{\|cd\|}(d_x-c_x))^2 + (c_y-q_y+	\frac{T}{\|cd\|}(d_y-c_y))^2	
\end{displaymath}

\begin{displaymath}
	 \frac{S}{v_{e_1}} + \frac{T}{v_{e_k}} = t_{p_ip_j} - t_{bc}	
\end{displaymath}
}

\vspace{0.2in}
Note that above equations can be solved in constant time. The following 
cases arise: (I) no such pair $(u,v)$ is  found, or 
(II) only one pair $(u,v)$ is found. In this case, $\|up_i\|$ is a critical distance, or
(III) more than one pairs of point $(u,v)$ are found. In this case, we determine, 
in constant time, the pair 
$(u_0,v_0)$ which has the minimum distance $\|u_0p_i\| = \|v_0p_j\|$ 
and then, $\|u_0p_i\|$ is a critical distance of type (C).
Hence, the running  time of Algorithm \ref{alg:typeCFinal} is $O(n^3)$.

\end{proof}

\begin{figure}[t]
	\centering
	\includegraphics[width=0.4\columnwidth]{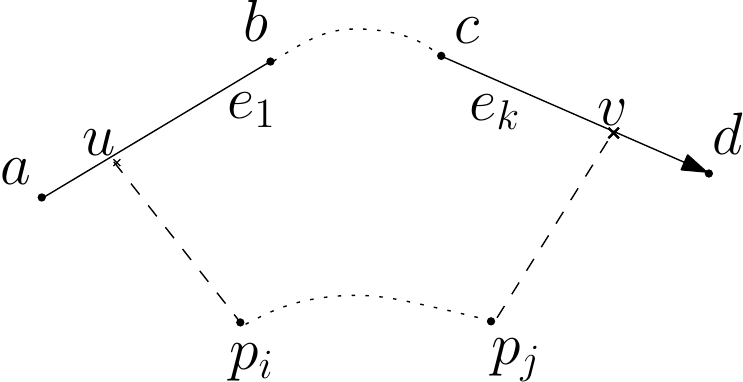}  
	\caption{Proof of Lemma \ref{TimeOfTypeC}}
	\label{fig:TimeOfTypeC}
\end{figure}


\begin{theorem}
	The exact \Frechet distance with speed limits can be computed in $O(n^3 \log n)$ time.
\end{theorem}
\begin{proof}
To find the exact value of $\distFS(P,Q)$,
we first compute all $O(n^3)$ critical distances of type (A), (B) and (C), 
and then we sort them. After sorting these values,
we do a binary search
(equipped with our decision algorithm)
to find the smallest $\eps$ for which $\distFS(P,Q) \lee \eps$.
In each search step, we solve the 
decision problem, if it has a positive answer, 
we continue with the half which contains 
smaller values. Otherwise, we continue with the half containing larger values.
The running time is dominated by the time of sorting $O(n^3)$ values, 
which is $O(n^3 \log n)$.
\end{proof}

In the standard \Frechet distance problem, 
parametric search based approach is used to compute the exact value of \Frechet distance. 
Next, we outline that approach and show that we cannot apply it to our instance of the problem.

Alt and Godau~\cite{AltG95} observed that
any comparison-based sorting algorithm that sorts
$a_{ij}, b_{ij}, c_{ij}$, and $d_{ij}$ 
(defined as functions of $\eps$)
has critical values that include those of type (C).
This is because the critical values of type (C) in the
standard \Frechet distance problem occur if 
$a_{ij} = b_{kj}$ or $c_{ij} = d_{ik}$,
for some $i,j$, and $k$. 
Therefore, 
to compute type (C) critical values, 
they   
used parametric search technique as follows.
First, compute all critical values of types (A) and (B), sort them
and then, perform binary search, and find two consecutive values  
$\eps_1$ and $\eps_2$ 
such that $\distF \in [\eps_1,\eps_2]$.
Let $S$ be the set of endpoints $a_{ij}, b_{ij}, c_{ij}$, $d_{ij}$ of intervals  
$\LF_{ij}$ and $\BF_{ij}$ that are nonempty for $\eps \in [\eps_1,\eps_2]$.
Then, Alt and Godau~\cite{AltG95} used Cole's parametric search method~\cite{Cole87}  
based on sorting the values in $S$ 
to find the exact value of $\distF$. 
Set $S$ consists of $O(n^2)$ polynomial functions 
$f_1(\eps) = a_{ij}, f_2 (\eps)=b_{kj}, f_3(\eps) = c_{ij}, ...$ 
of $\eps$. The values of these functions at 
$\delta$ will be given to a sorting network 
consists of parallel processors to get sorted (see Figure \ref{fig:PSproblem}).  
The crucial  requirement here is 
that at each stage,  
the transitivity of comparisons must hold, 
i.e.,  $f_1(\delta) \le f_2(\delta)$
and $f_2(\delta) \le f_4(\delta)$, implies that
$f_1(\delta) \le f_4(\delta)$.

That is not the case in our instance of the problem
because of the speed limit constraints.
Here, the critical values of type (C) occur if 
$b_{k\ell} = a_{ij}  +K_{ijk\ell}$ or $d_{ij} = c_{k\ell} + K'_{ijk\ell}$, 
for some $i,j,k$, and $\ell$. 
Although $K_{ijk\ell}$ or $K'_{ijk\ell}$
can be computed in $O(1)$ time
using Lemma \ref{lemma:project}, 
their value depends on  $i,j,k$ and $\ell$. 

Suppose that we use parametric search here. 
Assume that in the first stage of
parallel sorting, 
a processor compares
\REM{
e.g. $b_{k\ell}(\delta)$ with $a_{ij}(\delta) +K_{ijk\ell}$. 
Let $a_{ij}(\delta) +K_{ijk\ell} < b_{k\ell}(\delta)$. 
Furthermore, assume that another processor compares e.g.
$b_{gh}(\delta)$ with $a_{ef}(\delta) + K'_{efgh}$.
Let $ b_{gh}(\delta) < a_{ef}(\delta) + K'_{efgh}$. 
Then, assume in the next stage, 
$b_{k\ell}(\delta)$ is compared with $a_{ef}(\delta) + K''_{k\ell ef}$ and 
let $b_{k\ell}(\delta) < a_{ef}(\delta) + K''_{k\ell ef}$.
Unlike in the case of standard \Frechet distance problem, 
we cannot conclude that $b_{gh}(\delta)<b_{k\ell}(\delta)$ 
}
e.g. $f_{k\ell}(\delta)$ with $f_{ij}(\delta) +K_{ijk\ell}$. 
Let $f_{ij}(\delta) +K_{ijk\ell} < f_{k\ell}(\delta)$. 
Furthermore, assume that another processor compares e.g.
$f_{gh}(\delta)$ with $f_{ef}(\delta) + K'_{efgh}$.
Let $f_{gh}(\delta) < f_{ef}(\delta) + K'_{efgh}$. 
Then, assume in the next stage, 
$f_{k\ell}(\delta)$ is compared with $f_{ef}(\delta) + K''_{k\ell ef}$ and 
let $f_{ef}(\delta) + K''_{k\ell ef} < f_{k\ell}(\delta)$.
Unlike in the case of standard \Frechet distance problem, 
we cannot conclude that $f_{gh}(\delta)<f_{k\ell}(\delta)$ 
by transitivity since 
another $K'''_{ghk\ell}$ affects 
the comparison. 
Therefore, it seems unlikely that we can 
apply the parametric search technique
to compute $\distFS(P,Q)$, as pointed out by Alt~\cite{AltFinal}.

Recently, in~\cite{HarPeled11}, 
a  randomized algorithm is introduced that computes 
the \Frechet distance between two polygonal curves in $O(n^2 \log n)$ 
time with high probability, without using parametric search.
The key observation used in their algorithm is that
given a distance interval $I=[a,b]$, one can
find all type (C) critical distances in $I$ in
$O( (n^2 + k) \log n )$ time, where $k$ is number of these distances in range $I$.
They use a sweep line algorithm to achieve that running time.
In our instance of the problem, we have additional speed constraints, 
which makes it  hard to adopt the approach in \cite{HarPeled11}
to get a faster running time.
To be more precise,  consider the following sub-problem:

Suppose a curve $Q$, a time $t$, a distance interval $I=[a,b]$ 
and two vertices $p_i$ and $p_j$ from curve $P$ are given. 
Also assume that the object on $Q$ always walks 
with minimum speed associated to each edge. 
Now find all pairs of points $u$ and $v$ on $Q$ which satisfy the conditions:

(I) $\|p_i,u\|  = \|p_j, v \| = d$,  (II) $a \le d \le b$, and 
(III) time of travel from $u$ to $v$ on $Q$ is $t$.
It is unclear how to find such pairs efficiently. 

\REM{
The total preprocessing time must be less than $o(n^3)$ for all 
the vertices of $P$.
If the above sub-problem can be solved in 
$O(\log n + k)$ time,  
to compute $\distFS(P,Q)$,
one can use the randomized approach in \cite{HarPeled11}
to achieve an algorithm with expected running time less than $o(n^3)$.
}

\begin{figure}[t]
	\centering
	\includegraphics[width=0.5\columnwidth]{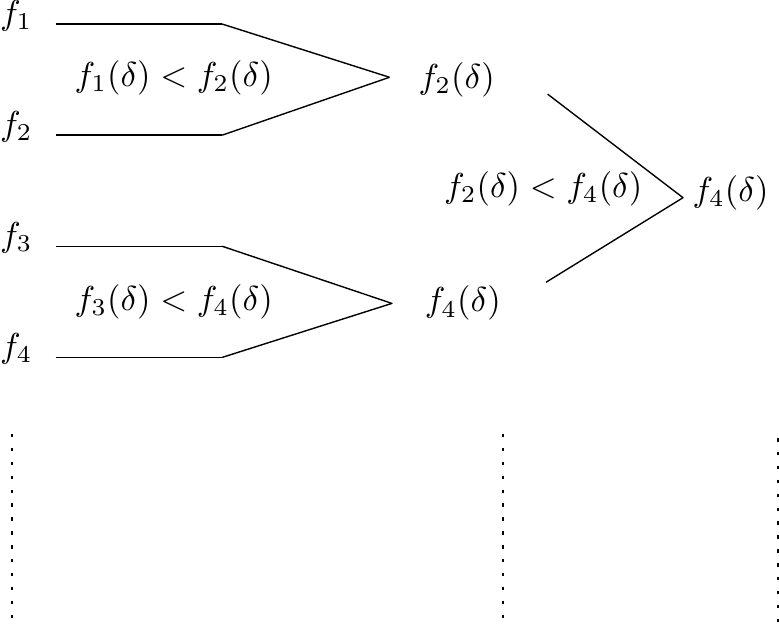}  
	\caption{Transitivity of comparisons must be kept during stages of parallel sorting in parametric search}
	\label{fig:PSproblem}
\end{figure}


\section{Conclusions} \label{sec:conclusion}

In this chapter, we introduced a variant of the \Frechet distance between two polygonal curves
in which 
the speed of traversal along each segment of the curves is restricted to be within a specified range.
We presented an efficient algorithm to solve the decision problem
in $O(n^2 \log n)$ time.
This led to a $O(n^3 \log n)$ time algorithm for finding 
the exact value of the \Frechet distance with speed limits.


Several open problems arise from our work.
In particular, it is interesting to consider 
speed limits in
other variants of the \Frechet distance studied in the literature. In the next chapter, we will study the
same problem  in the case where 
two curves lie inside a simple polygon.
Our result can be also useful 
in matching planar maps, where the objective is 
to find a curve in a road network that is as close as possible to a vehicle trajectory. 
In~\cite{AltERW03a}, the traditional \Frechet metric is used to match 
a trajectory to a road network. 
If the road network is very congested,  
the \Frechet distance with speed limits introduced here seems to find a more realistic path in the road network, 
close to the trajectory of the vehicle.
It is also interesting to extend our variant of the \Frechet distance 
to the setting where the speed limits on the segments of the curves change as functions over time.

Preliminary results of this chapter 
are presented in the 21st Canadian Conference on 
Computational Geometry~\cite{oursCCCG2009}. 
The full version of the paper is published in 
the special issue of Computational Geometry - Theory and Application~\cite{oursSpeedJournal}.
Alt~\cite{AltFinal} pointed out 
that due to the restrictions imposed by speed constraints, 
parametric search is not applicable.
It remains open whether there exists an algorithm 
that can solve the optimization 
problem faster than $O(n^3 \log n)$ time.

\chapter{Speed-constrained Geodesic \Frechet Distance }
\label{ch:speed-geodesic}

\REM{
Given two polygonal curves inside a simple polygon,
we study the problem of finding the 
\Frechet distance between the two curves
under the following two conditions
(i) the distance between two points on the curves is measured as the length of the shortest path
between them lying inside the simple polygon, and 
(ii) the traversal along each segment of the polygonal curves 
is restricted to be between a minimum and a maximum permissible speed
assigned to that segment.%

We provide an algorithm that decides in $O(n^2 (k + n))$ time 
whether the speed-constrained geodesic \Frechet distance between two polygonal curves inside a 
simple polygon is within a given value $\eps$, 
where $n$ is the number of segments in the curves, and $k$ is the complexity of the polygon. 
This leads to an algorithm for solving this variant of the \Frechet distance exactly
in $O(n^2(k + n)\log n)$ time and $O(n^2+k)$ space.
}


\section{Introduction} \label{sec:intro}


Several variants of the \Frechet distance have been studied in the literature.
Cook and Wenk~\cite{WenkC08a} studied  
the geodesic \Frechet distance inside a simple polygon.
In this variant, the leash is constrained to the interior of a simple polygon.  
Therefore,
a \emph{geodesic} distance is used to measure the length of the leash, which is 
the length of the shortest path inside the polygon connecting the two endpoints of the leash.
In~\cite{WenkC08a}, it is shown that the geodesic \Frechet distance 
between two polygonal curves of size $n$ inside a simple polygon of size $k$
can be computed in $O(n^2 \log (kn) \log n +k)$ expected time and $O(n^2+ k)$ space. 

In Chapter \ref{ch:speedFD}, we
introduced a generalization of the \Frechet distance,
in which users are allowed to set speed limits on each segment.
We showed that
for two polygonal curves of size $n$ with
speed limits assigned to their segments,
the speed-constrained \Frechet distance 
can be computed in $O(n^2 \log^2 n)$ time and $O(n^2)$ space. 
Note that in the problem instance of that chapter, there is no restriction for the leash to stay inside a simple polygon 
and thus, the leash lengths are measured using  the Euclidean distance.

In this chapter, we study the speed-constrained geodesic \Frechet distance inside a simple polygon
which is a simultaneous generalization of both 
\Frechet distances studied in~\cite{WenkC08a} and in the previous
chapter.
The decision version of the problem is formulated as follows:
Let $P$ and $Q$ be two polygonal curves inside a simple polygon, 
with minimum and maximum permissible speeds assigned to each segment of $P$ and $Q$.
For a given $\eps \gee 0$, 
can two point objects traverse
$P$ and $Q$ with permissible speeds (without backtracking) and, 
throughout the entire traversal, remain at geodesic distance at 
most $\eps$ from each other? The objective in the 
optimization problem is to find the smallest such $\eps$. 

We show that the 
decision version of the speed-constrained geodesic \Frechet distance problem 
can be solved in $O(n^2(k+n))$ time and $O(n^2+k)$ space,
where $n$ is the number of segments in the curves,
and $k$ is the complexity of the simple polygon.
This leads to a solution to the optimization problem
in $O(kn^3)$ time.

Algorithms for computing various variants of the \Frechet distance
are typically based on computing a free-space diagram consisting of $O(n^2)$ cells,
as we have seen in Chapters \ref{ch:related} and 
\ref{ch:speedFD},
and then propagating the reachability information 
one by one through the cells. 
While we adopt this general approach, 
the construction of the free-space diagram is more challenging in our problem
as we need to compute the whole free space inside each cell.
This is in contrast to other variants that 
only need to compute the free space on the boundaries of the cells.
A main contribution of our work is thus to 
fully describe the structure of the free space inside a cell,
establish its complexity,
and show how it can be computed efficiently.
Propagating the reachability information through the cells is also
more challenging  in our problem compared to the previous ones in previous chapters, 
as here, the shape of the free space inside a cell can
substantially affect the projection of the 
reachable intervals on its boundaries.

\REM{
The rest of the paper is organized as follows.  
We defines the problem formally in Section~\ref{sec:preliminaries} along with some notations.
Section~\ref{sec:ComputingFreeSpace} describes the structure of the free space 
and presents an algorithm for computing it efficiently. 
In Section \ref{sec:decision}, 
we propose an algorithm to compute our variant of the \Frechet distance. 
}

\section{Preliminaries} 
\label{sec:preliminaries}

A {\em polygonal curve\/} in $\IR^d$ is  a continuous function 
$P:[0,n] \rightarrow \IR^d$ with $n \in \IN$, 
such that for each $i \in \set{0, \ldots, n-1}$,
the restriction of $P$ to the interval $[i, i+1]$ 
is affine (i.e., forms a line segment).
The integer $n$ is called the {\em length\/} of $P$.
Moreover, the sequence ${P(0), \ldots, P(n)}$ represents the set of {\em vertices\/} of $P$.
For each $i \in \set{1, \ldots, n}$, 
we denote the line segment $P(i-1)P(i)$ by $P_i$. 
Given a simple polygon $K$ and two points $p,q \in K$,
the \emph{geodesic distance} of $p$ and $q$ with respect to $K$,
denoted by $d_K(p,q)$,
is defined as the length of the shortest path between $p$ and $q$
that lies completely inside $K$.

\paragraph{Speed-constrained geodesic \Frechet distance.}
Let $P$ be a polygonal curve such that
assigned to each segment $S$ of $P$,
there is a pair of non-negative real numbers $(\vmin{S}, \vmax{S})$ 
specifying the minimum and the maximum permissible speed for moving along $S$.
We define a {\em speed-constrained parametrization of $P$\/}
to be a continuous surjective function $f: [0,T] \rightarrow [0,n]$ with $T > 0$
such that for any $i \in \set{1, \ldots, n}$,
the slope of $f$ at all points $t \in [f^{-1}(i-1),f^{-1}(i)]$ 
is within $[\umin{P_i}, \umax{P_i}]$, where
$\umin{S} = {\vmin{S} / \|S\|}$ and $\umax{S} = {\vmax{S} / \|S\|}$.

Given a simple polygon $K$
and two polygonal curves $P$ and $Q$ inside $K$
of lengths $n$ and $m$ respectively 
with speed limits assigned to their segments, 
the {\em speed-constrained geodesic \Frechet distance\/} of $P$ and $Q$ inside $K$
is defined as
\[
	\distFS(P,Q) = \inf_{\alpha, \beta} \max_{t \in [0,T]} d_K( P(\alpha(t)), Q(\beta(t)) ),
\]
where $\alpha: [0,T] \rightarrow [0,n]$ ranges over all speed-constrained parametrizations of $P$ 
and $\beta:[0,T] \rightarrow [0,m]$ ranges over all speed-constrained parametrizations of $Q$.

\paragraph{Free-space diagram.}
Let $\BNM = [0,n] \times [0,m]$ be a $n$ by $m$ rectangle  in the plane.
Each point $(s,t) \in \BNM$ uniquely represents a pair of points
$(P(s),Q(t))$ on the polygonal curves $P$ and $Q$.
We decompose $\BNM$ into
$n\times m$ unit grid cells $\cell{ij} = [i-1,i] \times [j-1,j]$
for $(i,j) \in \set{1, \ldots, n} \times \set{1, \ldots, m}$,
where each cell $\cell{ij}$ corresponds to
a segment $P_i$ on $P$ and a segment $Q_j$ on $Q$.
Given  two polygonal curves $P$ and $Q$ inside a simple polygon $K$
and a parameter $\eps \gee 0$,
the {\em free space\/} $\Feps$ is defined as
$
	\Feps = \{(s,t) \in \BNM \ | \ d_K(P(s),Q(t)) \lee \eps \}.
$
We denote by $L_{ij}$ (resp., by $B_{ij}$) the left (resp., bottom) line segment bounding $\cell{ij}$. 
The {\em entry side\/} of $\cell{ij}$ is defined as $\entry{ij} = L_{ij} \cup B_{ij}$,
and its {\em exit side\/} as $\exit{ij} = B_{i,j+1} \cup L_{i+1,j}$. 
Given two points $p$ and $q$ on the boundary of a cell, we say that 
$p$ is {\em before\/} $q$, 
denoted by $p \lei q$, if either $p_x < q_x$ or 
$(p_x = q_x$ and $p_y > q_y)$.

\paragraph{Hourglass data structure.}
Fix a simple polygon $K$.
Given two points $p,q \in K$, 
we denote by $\SP(p,q)$ the shortest path 
between $p$ and $q$ that lies inside $K$,
and denote its length by  $\| \SP(p,q) \|$.
Let $\Ov{ab}$ and $\Ov{cd}$ 
be two 
line segments inside $K$.
The \emph{hourglass} $\Ho{ab}{cd}$
is defined as the maximal region bounded by 
the segments $\Ov{ab}$ and $\Ov{cd}$, 
and the shortest path chains 
$\SP(a,c)$, $\SP(a,d)$, $\SP(b,c)$ and $\SP(b,d)$.
Three examples of hourglasses are illustrated in Figure~\ref{fig:OpenH}.
(See~\cite{Guibas86} for applications of the hourglass.)
Note that for any two points $p \in \Ov{ab}$ and $q \in \Ov{cd}$,
the shortest path $\SP(p,q)$ is contained in $\Ho{ab}{cd}$.
The intersection of $\Ho{ab}{cd}$ and the boundary of $K$
consists of at most four polygonal curves,
each of which is called a \emph{chain} of $\Ho{ab}{cd}$.

\begin{figure}[h]
	\centering
	\includegraphics[width=0.8\columnwidth]{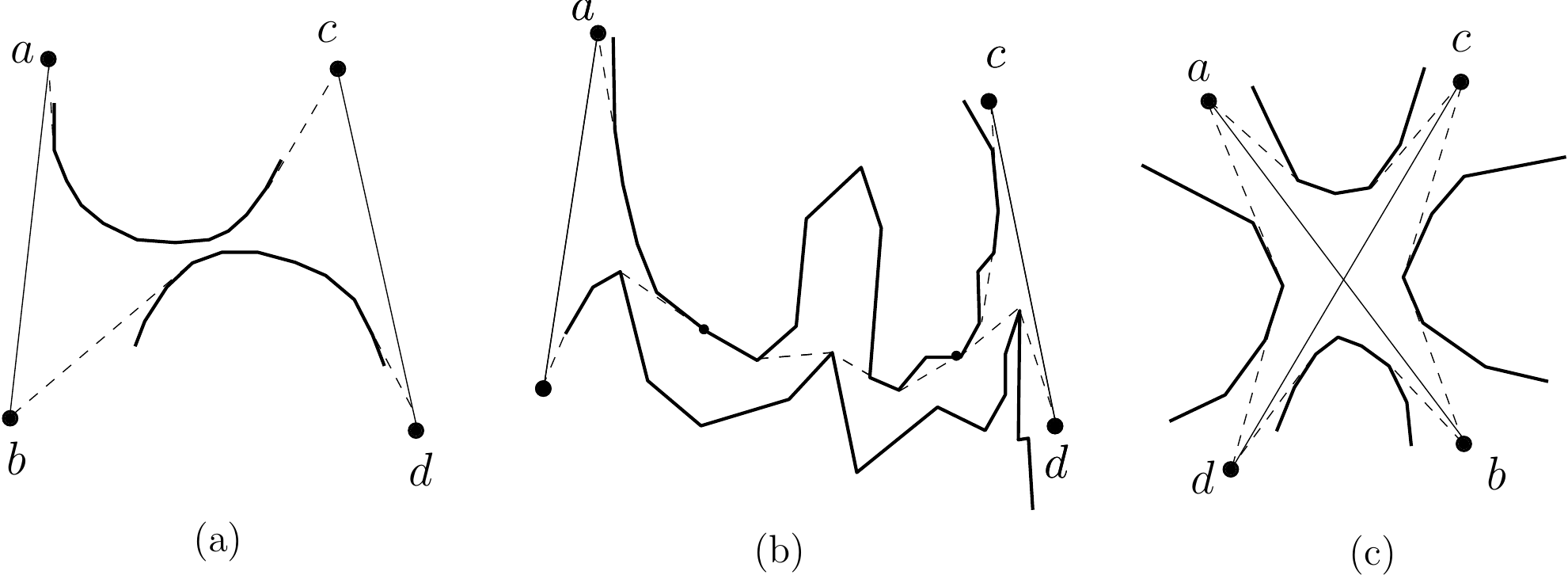}
	\caption{ 
	(a) An open hourglass (b) A closed hourglass (c) An intersecting hourglass.
	}
	\label{fig:OpenH}
\end{figure}

\REM{
If $\SP(a,c)$ and $\SP(b,d)$
are disjoint, the hourglass is called open, 
otherwise it is called closed. 
If $\Ov{ab}$ and $\Ov{cd}$ are crossing,
$\Ho{ab}{cd}$ is defined as the region bounded by 
$\SP(a,c)$, $\SP(a,d)$, $\SP(b,c)$ and $\SP(b,d)$.
}


\newcommand{\Sub}{{\sigma}}

\section{Computing the Free Space Inside a Cell} 
\label{sec:ComputingFreeSpace}

In the classical \Frechet distance problem
(Section \ref{sec:classicalFD}),
the free space inside each cell 
is convex and can be determined in $O(1)$ time. 
When distances are geodesic, 
the free space is not necessarily convex,
but it is still connected and $xy$-monotone 
(see~\cite{WenkC08a} for the proof).

%
%

Therefore, to solve the geodesic \Frechet distance (without speed limits),
one only needs to compute the free space on the boundaries of the cells.
 In~\cite{WenkC08a}, A. Cook \etal ~show how to compute  
the boundary of a cell in $O(\log k)$ time after
$O(k)$ time preprocessing, 
based on the algorithm of Guibas and Hershberger~\cite{Guibas86}.  
Also, one could use Chambers \etal's approach in~\cite{Chambers10}, 
to compute the boundary of the cells in $O(\log k)$ time.
In contrast to above works, in our generalized version where motion speeds are limited,
we need to compute the full description of the free space in the
interior of the cells as well in order to propagate the reachability 
information correctly. 

\REM{
To compute the full description of the free space in the interior of cells,
one could use Chambers \etal's approach in \cite{Chambers10}. 
That leads to $O(n^2k\log k)$ time solution, 
where $n$ is the total complexity of the curves and 
$k$ is the complexity of the polygon. Here, we 
propose an algorithm which 
computes the interior of all cells in $O(n^2k)$ total time. 
}

We use the hourglass data structure to compute the 
boundary of the free space inside a cell.
Consider an hourglass $\Ho{ab}{cd}$ 
and two points $p \in \Ov{ab}$ and $q \in \Ov{cd}$.
The shortest path $\SP(p,q)$ is either a straight segment 
(in case $p$ and $q$ see each other),
or consists of two tangents from $p$ and $q$ to the chains of $\Ho{ab}{cd}$
plus a subpath between the two tangent points.
We denote this subpath by $\Sub(p,q)$.
Note that $\Sub(p,q)$ consists of a sequence of vertices of $K$,
lying on at most two chains of the hourglass.

\begin{definition}
Consider an hourglass $\Ho{ab}{cd}$ and two intervals 
$\Ov{a' b'} \subseteq \Ov{ab} $ and $\Ov{c' d'} \subseteq \Ov{cd}$,
so that for any $p \in \Ov{a' b' }$ and any $q \in \Ov{c' d'}$, 
$\Sub(p,q)$ is the same.
The region bounded by the intervals  $\Ov{a'b'}$ and  $\Ov{c'd'}$ 
and the paths $\SP(a',c')$ and $\SP(b', d')$ is called a butterfly,
and is denoted by $\But{a' b'}{c' d'}$ (see Figure~\ref{fig:but}).
\end{definition}


\begin{lemma} \label{lemma:ButFhyperbolic}
	Given a butterfly $\But{a' b'}{c' d'}$,
	the function  $f(p,q) = \| \SP(p,q) \|$ over the domain $[a',b'] \times [c',d']$ 
	is a hyperbolic surface. 
\end{lemma}

\begin{proof}
	Fix a point $p \in \Ov{a'b'}$ and a point $q \in \Ov{c'd'}$.
	Let $k_1$ and $k_2$ be the two endpoints of $\Sub(p,q)$.
	Then $\| \SP(p,q) \| = \|pk_1\| + \| \Sub(p,q) \| + \|k_2q\|$.
	By the butterfly property, $k_1$, $k_2$, and $\| \Sub(p,q) \|$ are fixed 
	for all $p$ and $q$ in the domain.
	Therefore,  $\| \SP(p,q) \|$ is the sum of two $L_2$ distances plus a constant,
	which forms a hyperbolic surface.
\end{proof}

\begin{figure}[h]
	\centering
	\includegraphics[width=00.5\columnwidth]{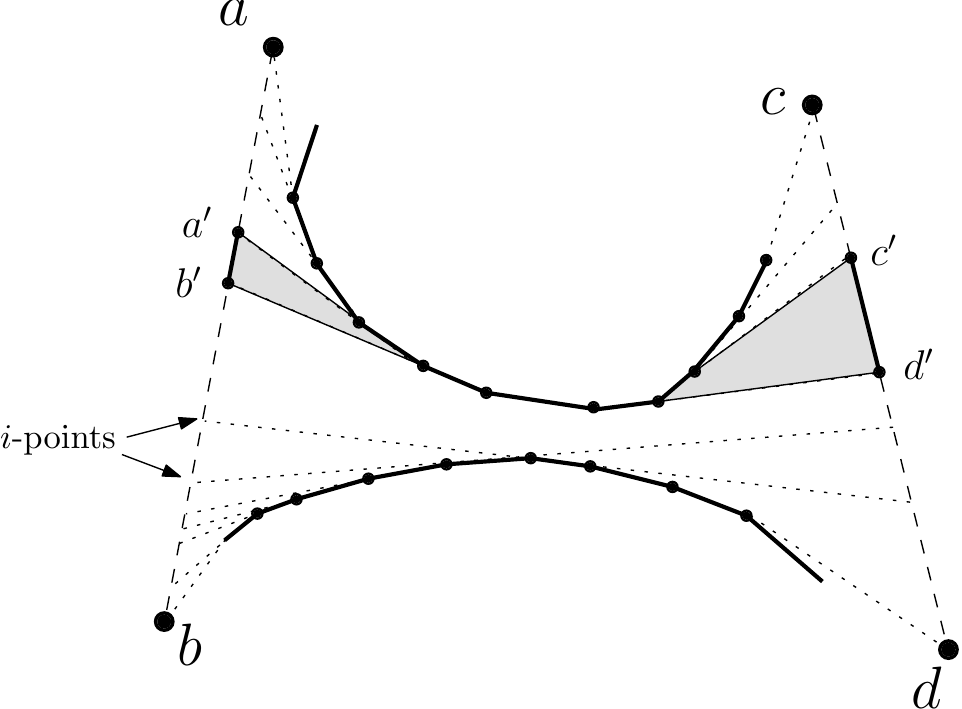}
	\caption{ An hourglass $\Ho{ab}{cd}$ with a butterfly $\But{a' b'}{c' d'}$.}
	\label{fig:but}
\end{figure}

Consider an edge $e$ on a chain of the hourglass $\Ho{ab}{cd}$. 
Extend $e$ to a line and find its intersection with $\Ov{ab}$ and $\Ov{cd}$ (as shown in Figure~\ref{fig:but}). 
We call such an intersection point an \emph{$i$-point}. 
Note that the number of $i$-points on each of the segments $\Ov{ab}$ and $\Ov{cd}$ is $O(k)$.

\begin{obs} \label{obs:$i$-pointbutt}
	Any two consecutive $i$-points $i_1,i_2 \in \Ov{ab}$ 
	and any two consecutive $i$-points $j_1,j_2 \in \Ov{cd}$
	form a butterfly $\But{i_1i_2}{j_1j_2}$.
\end{obs}

Consider two polygonal curves $P$ and $Q$ inside $K$.
Let $P_i = \Ov{ab}$ be a segment of $P$,
and $Q_j = \Ov{cd}$ be a segment of $Q$.
By dividing $\Ov{ab}$ and $\Ov{cd}$ at $i$-points,
the corresponding cell  $\cell{ij}$ in the free-space diagram 
is decomposed into $O(k^2)$ subcells,
where each subcell corresponds to a butterfly
(see Figure~\ref{fig:c-pointschains}).

Let $f(p,q) = \| \SP(p,q) \|$ be a function
defined over all $(p,q) \in [a,b] \times [c,d]$.
The intersection of the plane $z=\eps$ with the function $f$ determines the 
boundary of $\Feps$ inside the cell $\cell{ij}$.
The boundary of $\Feps$ crosses the boundary of each subcell 
in at most two points, each of which is called a \emph{$c$-point}.
The following two lemmas describe the structure of the free space inside  $\cell{ij}$.

\begin{lemma}
	Any two consecutive $c$-points on the boundary of $\Feps$ 
	are connected with a hyperbolic arc, 
	and the line segment connecting the two endpoints of the arc lies completely inside $\Feps$.
\end{lemma}
\begin{proof}
	This follows from Lemma~\ref{lemma:ButFhyperbolic}.
\end{proof}

\begin{lemma}
	The number of $c$-points inside a cell is $O(k)$. 
\end{lemma}
\begin{proof}
	This follows from the fact that 
	any $xy$-monotone curve intersecting an $n \times m$ (non-uniform) grid 
	can cross at most $2(n+m)$ cells of the grid.
\end{proof}

\begin{figure}[h]
	\centering
	\includegraphics[width=0.4\columnwidth]{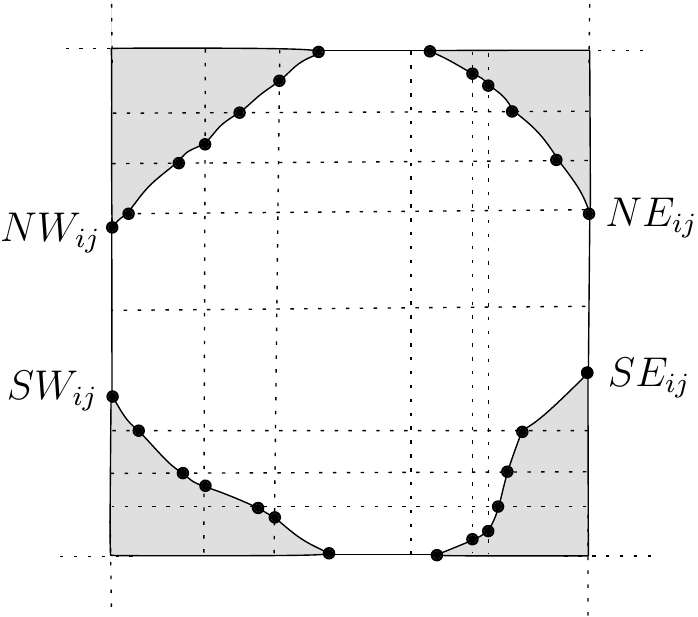}
	\caption{ The free space inside a cell. }
	\label{fig:c-pointschains}
\end{figure}

\paragraph{Computing c-points.}
Let $\Lin$ denote the line as a result of 
extending line segment $\Seg{cd}$.
Our algorithm for computing $c$-points is based on the following observation.

\begin{obs} \label{lemma:shrinkexpand}
	Consider an hourglass $\Ho{ab}{cd}$ and a fixed $\eps > 0$.
	Let $p$ be a point moving on $\Ov{ab}$,
	and let $q$ be a point that moves on the 
	line $\Lin$  to maintain 
	geodesic distance $\eps$ from $p$.
	When $p$ moves monotonically from $a$ to $b$,  
	$q$ has at most one directional change along $\Lin$.
\end{obs}

Observation~\ref{lemma:shrinkexpand} enables us to compute all $c$-points inside a cell
by two linear walks. Details are provided in Algorithm~\ref{alg:ComputeFD}.
In this algorithm, $h_c$ refers to a point on $\Ov{ab}$ which is closest to $c$,
$p_1 \prec_{\Dir{ab}} p_2$ means that 
$p_1$ is before $p_2$ in direction $\Dir{ab}$,
and $\PP(p,q)$ refers to the unique point in the free-space diagram 
corresponding to a point $p \in P$ and $q \in Q$.
The output of the algorithm is four connected $c$-point chains
as depicted in Figure~\ref{fig:c-pointschains}.

\renewcommand{\algorithmicrequire}{\textbf{Input:}}

\begin{algorithm} [h]
\caption {\sc computing c-points inside a cell } \label{alg:ComputeFD}
\algsetup{indent=1.5em}
\begin{algorithmic}[1]
	\vspace{0.5em}
	\baselineskip=1\baselineskip
	\REQUIRE  An hourglass $\Ho{ab}{cd}$ corresponding to a cell $\cell{ij}$ and a fixed $\eps > 0$.
	
	\STATE \label{lp:extend} Compute $i$-points on $\Ov{ab}$ and $\Ov{cd}$.
	\STATE \label{lp:q1q2} Find $q_1,q_2 \in \Lin$ s.t. $\| \SP(a,q_1) \| = \|\SP(a,q_2) \| = \eps $.
	\STATE \label{lp:initialization0} Set $\eta_1 = q_1$ and  $\eta_2 = q_2$, assuming that $q_1 \prec_{\Dir{cd}} q_2$.
	\STATE \label{lp:initialization} Set $\mu = a$.
	
	\WHILE { $\mu$ has not reached  $b$} 
		\STATE \label{lp:move} Move $\mu$ in direction $\Dir{ab}$, and 
			move $\eta_1$ on $\Lin$ s.t. $\| \SP(\mu,\eta_1) \| = \eps $ 
			until either $\mu$ or $\eta_1$ reaches an $i$-point.	
		\IF{$F(\mu,\eta_1) \in \cell{ij}$}
		\STATE Insert $F(\mu,\eta_1)$ into $SW_{ij}$ if  $\mu \prec_{\Dir{ab}} h_c$, 
			otherwise  insert $F(\mu,\eta_1)$ into $NW_{ij}$.
		\ENDIF
	\ENDWHILE
	
	\STATE Repeat lines 4--8 with $\eta_2$ instead of $\eta_1$ to obtain $NE_{ij}$ and $SE_{ij}$.
	\RETURN $NW_{ij}$, $SW_{ij}$, $NE_{ij}$, and $SE_{ij}$.
\end{algorithmic}
\end{algorithm}



\section{The Decision Problem} \label{sec:decision}

In this section, we show how the decision version of 
our \Frechet distance problem can be solved efficiently.
We use the notation of Chapter \ref{ch:speedFD}.
A path $\CP \subset \BNM$ is called {\em slope-constrained\/}
if for any point $(s,t) \in \CP \cap \cell{ij}$,
the slope of $\CP$ at $(s,t)$ is within 
$\minS{ij} = {{\umin{Q_j}}/{\umax{P_i}}}$ and $\maxS{ij} = {{\umax{Q_j}}/{\umin{P_i}}}$.
A point $(s,t) \in \Feps$ is called {\em reachable\/} 
if there is a slope-constrained path from $(0,0)$ to $(s,t)$ in $\Feps$.
As shown in previous chapter, 
 $\distFS(P,Q) \lee \eps$ if and only if the point $(n,m)$ is reachable.

Reachable points on the entry side of each cell
form a set of $O(n^2)$ disjoint intervals, 
each of which is called a \emph{reachable interval}(as in previous chapter).
To decide if  $(n,m)$ is reachable,
the general approach is to propagate the reachability information
one by one, in row-major order, 
from $\cell{0,0}$ to $\cell{nm}$.
The propagation in each cell $\cell{ij}$ involves
projecting the set of reachable intervals 
from the entry side of the cell to its exit side.

Since the free space inside a cell is not necessarily convex in our problem, 
the projection can be affected by the boundary of $\Feps$ inside a cell
(see Figure~\ref{fig:merging}). 
We use the $c$-point information computed in the previous section 
to compute projections.
Indeed, only $c$-points on the convex hull of $NW_{ij}$ and $SE_{ij}$ 
are needed to compute correct projections. 
Since $c$-points inside each chain are stored in a sorted $x$ (and $y$) order,
the convex hull of the chains can be computed using a Graham scan in $O(k)$ time.
We call the convex hull of $NW_{ij}$ (resp., $SE_{ij}$)
the \emph{left chain} (resp., the \emph{right chain}) of $\cell{ij}$.

Given a point $p \in \entry{ij}$, Algorithm~\ref{alg:projection} 
computes the projection of 
$p$ onto $\exit{ij}$ in $O(\log k)$ time.

\begin{lemma} \label{lemma:project-point}
	Given a point $p \in \entry{ij}$, Algorithm~\ref{alg:projection} 
	computes the projection of 
	$p$ onto $\exit{ij}$ in $O(\log k)$ time.
\end{lemma}

\begin{proof}
	Finding each of the two tangents in Line~1 takes $O(\log k)$ time using binary search.
	The rest of the algorithm takes constant time.
\end{proof}

\begin{algorithm} [h]
\caption {\sc projection function} \label{alg:projection}
\algsetup{indent=1.5em}
\begin{algorithmic}[1]
	\vspace{0.5em}
	\baselineskip=1\baselineskip
	
	\REQUIRE A point $p \in \entry{ij}$
	\STATE Let $t_\ell$ and $t_r$ be tangents (if they exist) from $p$ to the left and to the right chain of $\cell{ij}$, respectively.
	\STATE Let $a_1$ and $a_2$ be the projection of $p$ in directions $t_\ell$ and $\maxS{ij}$, respectively.	
	\STATE Let $b_1$ and $b_2$ be the projection of $p$ in directions $t_r$ and $\minS{ij}$, respectively.
	\RETURN $[\max(a_1, a_2), \min(b_1, b_2)]$
\end{algorithmic}
\end{algorithm}

\begin{figure}[h]
	\centering
	\includegraphics[width=0.7\columnwidth]{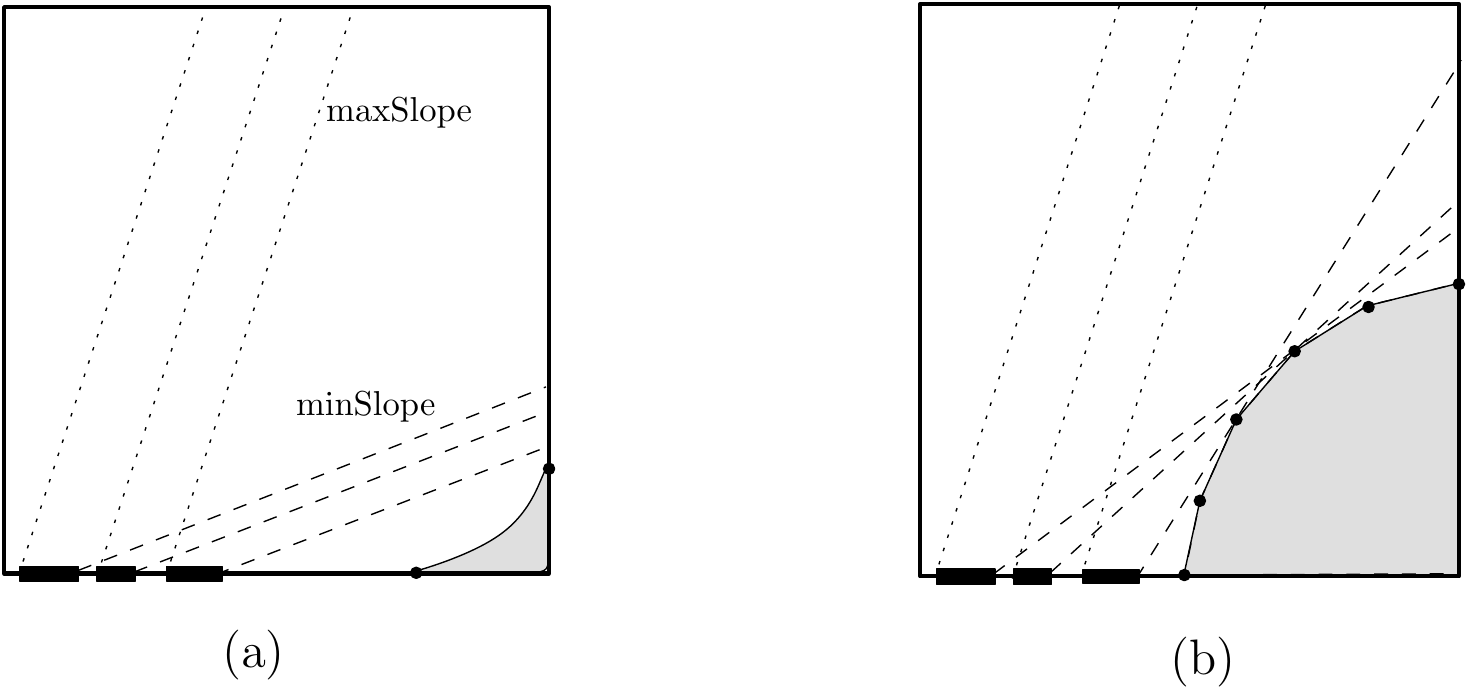}
	\caption{ 
		Projecting reachable intervals inside cells with convex and non-convex interior. 
	}
	\label{fig:merging}
\end{figure}

\begin{lemma} \label{thm:naive}
	Given a cell $\cell{ij}$ with $r_{ij}$ reachable intervals on its entry side, 
	Algorithm~\ref{alg:projection} projects all the reachable intervals onto the exit side of $\cell{ij}$
	in $O(k+r_{ij})$ time.
\end{lemma}

\begin{proof}
Let $t_1$ be a line in direction $\minS{ij}$ 
tangent to the left chain of $\cell{ij}$,
and let $t_2$ be a line in direction $\maxS{ij}$ 
tangent to the right chain of $\cell{ij}$.
Let $a_1$ and $a_2$ be the intersection points of $t_1$ and $t_2$
with $\entry{ij}$, respectively.
For any point $p \in \entry{ij}$
that lies outside $[a_1, a_2]$, the projection of $p$ is empty. 
Therefore, we delete those portions of reachable intervals
that lie outside $[a_1, a_2]$. 
Now, the projection of each of the remaining intervals 
can be simply computed 
by projecting its two endpoints.

To avoid spending $O(\log k)$ time for projecting each endpoint, 
we use a cross-ranking technique. 
This reduces the total time needed for computing the tangents in Algorithm~\ref{alg:projection}.
Let $T_1$ be the list of all endpoints of the reachable intervals on $\entry{ij}$ in $\lei$ order.
We construct another list $T_2$ as follows. 
Perform an edge traversal of the right chain, starting with the rightmost edge.
Each edge encountered is extended to a line until it intersects the entry side at a point
which is then added to $T_2$.
We merge $T_1$ and $T_2$ (in $\lei$ order) to create a list $T$.
Each item in $T$ has a pointer to its corresponding $c$-point or reachable interval endpoint, and vice versa.
Moreover, each item in $T$ which comes from $T_1$
keeps a pointer to its preceding item in $T$ which comes from $T_2$.
Now, given a reachable interval endpoint $p$, to compute the tangent from $p$ to the right chain,
we simply find the item $t \in T$ corresponding to $p$, 
and then find the item in $T_2$ preceding $t$ in $T$. This item
uniquely determines the $c$-point at which the tangent from $p$ to the right chain occurs.
We process the left chain in the same way.
This enables us to compute each tangent in constant time, after the cross-ranking step,
leading to $O(k+r_{ij})$ total time for projecting all endpoints.
\end{proof}

Combined with the fact that $\sum_{0\lee i,j \lee n} r_{ij} = O(n^3)$ as in Chapter \ref{ch:speedFD},
the decision problem can be solved in $O(n^2(k+ n))$ time and $O(n^2+ k)$ space.

\begin{figure}[t]
	\centering
	\includegraphics[width=0.6\columnwidth]{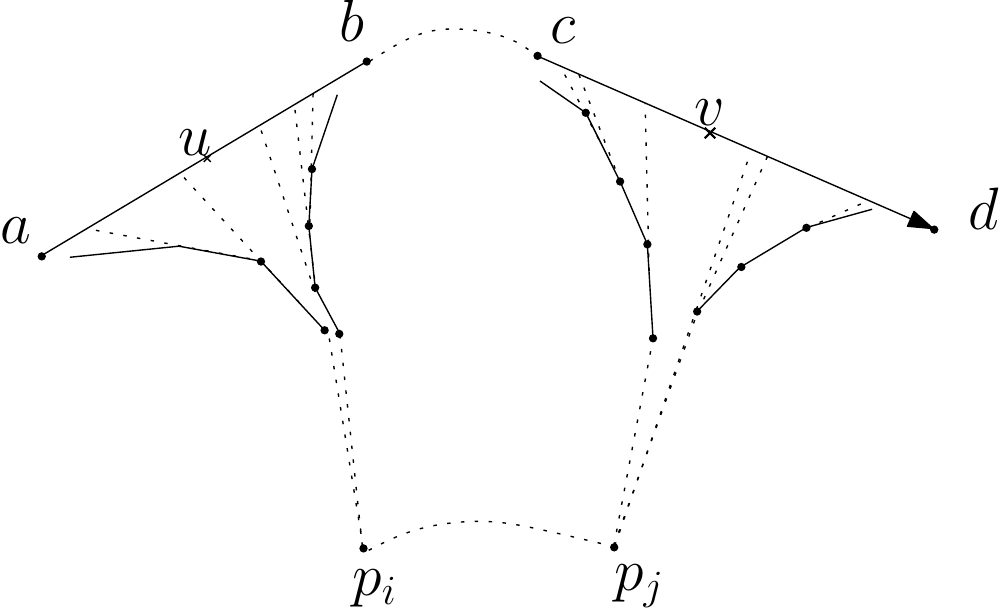}  
	\caption{Proof of Theorem \ref{thm:naivePoly}}
	\label{fig:PolyLast}
\end{figure}

\newcommand{\Funn}{\CF}

\begin{theorem} \label{thm:naivePoly}
   The exact value of $\distFS(P,Q)$ between curves $P$ and $Q$ inside polygon $K$
   can be computed in $O( kn^3)$ time.
\end{theorem}
\begin{proof}
We use the same technique as in Section  \ref{sec:optimization}
to compute $\distFS(P,Q)$.
There are two critical distances of type (A) and 
$O(n^2)$ critical distances of type (B).
Geodesic distances inside a simple polygon 
are computed by the algorithms of Guibas and 
Hershberger~\cite{Guibas86, Hershberger91}.
These algorithms preprocess the polygon in $O(k)$
so that 
the shortest path queries between two points  or
between a point and a line segment 
can be solved in $O(\log k)$ time. 
Thus, we can compute type (A) and type (B) distances 
in $O(n^2 \log k)$ time.


As in previous chapter, 
there are $O(n^3)$ critical distances of type (C). 
To compute them, we use Algorithm~\ref{alg:typeCFinal} on Page
\pageref{alg:typeCFinal}, 
after modifying Line 4 of that algorithm.
For the case where $P$ and $Q$ are in the plane
and distances are Euclidean, we showed 
that Line 4 can be done in $O(1)$ time in Lemma \ref{TimeOfTypeC}. 
Here, because distances are geodesic, the 
run-time is $O(k)$ as explained in the following. 

In the algorithm of Guibas and 
Hershberger, all shortest paths between a point $p_i$ and a line segment $\Seg{ab}$ are represented by a funnel, denoted by  $\Funn_{p_i,\Seg{ab}}$
(see Figure \ref{fig:PolyLast}). $\Funn_{p_i,\Seg{ab}}$
is the region bounded by the line segment $\Seg{ab}$
and the shortest path chains $\SP(p_i,a)$ and 
$\SP(p_i,b)$. 
Extend all line segments in the shortest path 
chains $\SP(p_i,a)$ and $\SP(p_j,b)$ of funnel $\Funn_{p_i,\Seg{ab}}$
to a line, and find the intersection of those lines with 
$\Seg{ab}$ (see Figure \ref{fig:PolyLast}). Do the same in funnel $\Funn_{p_j,\Seg{cd}}$ with 
respect to segment $\Seg{cd}$. Now, maintain
the list of points which starts at 
$a$,  the intersection points on $\Seg{ab}$ in order, 
and ends at $b$, 
in a list denoted by  $L_{\Seg{ab}}$. Similarly, compute $L_{\Seg{cd}}$.
Next, we create 
two lists $L'_{\Seg{ab}}$ and $L'_{\Seg{cd}}$ from lists $L_{\Seg{ab}}$ and $L_{\Seg{cd}}$ to apply 
the same technique in Algorithm \ref{alg:potentialChains} and then, compute
critical distances of type (C)  in Algorithm \ref{alg:typeCFinal}.

Let $L'_{\Seg{ab}} = L_{\Seg{ab}}$ and 
$L'_{\Seg{cd}} = L_{\Seg{cd}}$. 
For each point $u \in L_{\Seg{ab}}$, 
compute distance $\| up_i\|$ and
find point(s) $v$ on $\Seg{cd}$ where
$\| up_i\| = \| vp_j\|$, and insert $v$ in $L'_{\Seg{cd}}$. 
Likewise, 
for each point $v \in L_{\Seg{cd}}$, 
compute distance $\| vp_j\|$ and
find point(s) $u$ on $\Seg{ab}$ where
$\| up_i\| = \| vp_j\|$, and insert $u$ in $L'_{\Seg{ab}}$.

The run-time to create these two lists is $O(k)$ because
the function which represent distances from a 
point to a line segment inside a polygon is a bitonic function. 
Therefore, distances from point $p_i$ (resp., point $p_j$)  to points in $L_{\Seg{ab}}$ (resp., to points in  $L_{\Seg{cd}}$)
are increasing or decreasing or bitonic. 
Thus, $L'_{\Seg{ab}}$ and $L'_{\Seg{cd}}$ can be computed 
 in $O(k)$ time using the cross-ranking technique. 

Having computed these two lists $L'_{\Seg{ab}}$ and $L'_{\Seg{cd}}$, 
we can then use two pointers as in Algorithm 
\ref{alg:potentialChains} and 
by the same  calculation 
described in Lemma \ref{TimeOfTypeC}, 
compute critical distances of type (C).
Therefore, the run-time of Line 4 of Algorithm~\ref{alg:typeCFinal} 
is $O(k)$ in this case. Since that line is executed 
$O(n^3)$ times, type (C) 
critical distances can be computed in $O(k n^3)$ total time.

After computing all type (A), (B) and (C) critical distances, we sort them 
and then, we perform binary search equipped with our decision algorithm, 
to find the the exact value of 
speed-constrained geodesic \Frechet  distance. 
Hence, we obtain an $O(kn^3)$ time algorithm to compute $\distFS(P,Q)$
for two curves $P$ and $Q$ inside a simple polygon.
\end{proof}

\section{Conclusion}
In this chapter, we introduced a variant of the \Frechet distance between two polygonal curves inside a 
simple polygon,
in which 
the speed of traversal along each segment of the curves is restricted to be within a specified range.

We presented an algorithm that decides in $O(n^2 (k + n))$ 
time whether the speed-constrained geodesic \Frechet distance between two polygonal curves inside a 
simple polygon is within a given value $\eps$, 
where $n$ is the number of segments in the curves, and $k$ is the complexity of the polygon. 

Several open problems arise from our work.
In particular, it is interesting to consider speed limits in
other variants of the \Frechet distance studied in the literature,
such as the \Frechet distance between two curves lying 
on a convex polyhedron~\cite{AnilFrechet}, or on a polyhedral surface~\cite{Cook2009}.

Results of this chapter 
are presented in 22nd Canadian Conference on Computational Geometry~\cite{oursCCCG2010}.

\chapter{Improved Algorithms for Partial Curve Matching}
\label{ch:partial}

\section{Introduction} 
\label{sec:introduction}

\REM{
The \Frechet distance is a widely-used metric for measuring the similarity of the curves. 
It finds applications in morphing~\cite{EfratGHMM02},
handwriting recognition~\cite{SriraghavendraKB07}, protein structure alignment~\cite{JiangXZ08}, etc.
This measure is often illustrated as the minimum-length leash needed for a
person to walk a dog, while each of them is traversing
a pre-specified polygonal curve without backtracking.
}
As described in Section \ref{sec:classicalFD},
Alt and Godau~\cite{AltG95} showed how the 
\Frechet distance between two polygonal curves
with $n$ and $m$ vertices can be computed in $O(nm \log(nm))$ time.
For their solution, they introduced the \fs diagram.
\REM{
The \fs diagram and its variants have been proved to be useful in other applications
involving the \Frechet distance
(see e.g. \cite{AltERW03,BuchinBG10,WenkC10,HR11}).

Various extensions of the \Frechet distance have been studied in the literature,
including \Frechet distance between two curves 
inside a simple polygon~\cite{WenkC10}, 
on polyhedral surfaces~\cite{CookW09}, 
and on simplicial complexes~\cite{HR11}.

}

As discussed in Section
\ref{sec:RelatedPartial},
Alt and Godau~\cite{AltG95} 
in their seminal work, 
studied the partial curve matching problem.
Given two polygonal curves $P$ and $Q$ of size $n$ and $m$, respectively, 
they presented an algorithm that decides in $O(nm \log(nm))$ time whether
there is a subcurve $R$ of $P$ whose \Frechet distance to $Q$ is at most $\eps$,
for a given $\eps \gee 0$. 
Using parametric search, they solved the optimization problem of finding the minimum such $\eps$
in $O(nm \log^2(nm))$ time.

Later, Alt \etal~\cite{AltERW03a} proposed a generalization of the partial curve matching problem 
to measure the similarity of a curve to some part of a graph.
Given a polygonal curve $P$ and a graph $G$, 
they presented an $O(nm \log m)$-time algorithm to decide whether there is a path $\pi$ in $G$ whose 
\Frechet distance to $P$ is at most $\eps$, 
with $n$ and $m$ being the size of $P$ and $G$, respectively.
A variant of the partial curve matching
in the presence of outliers is studied by Buchin \etal~\cite{ExactPartial},
leading to an algorithm with $O(nm (n+m)\log(nm))$ running time.

\paragraph{\bf Our results.\ }
In this chapter, we present a simple data structure, which we call \emph{\fs map},
that enables us to solve several variants of the \Frechet distance problem efficiently.
The results we obtain using this data structure are summarized below.
In the following, $n$ and $m$ represent the size of 
the two given polygonal curves $P$ and $Q$, respectively, 
and $\eps \gee 0$ is a fixed input parameter.

\begin{itemize}
\setlength{\itemsep}{.3em}
\renewcommand{\labelitemi}{${\tiny \bullet}$}

	\item \emph{Partial curve matching}. \ 
	Given two polygonal curves $P$ and $Q$, 
	we present an algorithm to decide in $O(nm)$ time whether
	there is a subcurve $R \subseteq P$ whose \Frechet distance to $Q$ is at most $\eps$.
	This improves the best previous algorithm for this decision problem 
	due to Alt and Godau~\cite{AltG95}
	(described in Section \ref{sec:RelatedPartial}),
	that requires $O(nm \log (nm))$ time.
	This also leads to an $O(\log (nm))$ faster algorithm for solving
	the optimization version of the problem, using parametric search.
	
	\item \emph{Closed \Frechet metric}. \ 
	As described in Section \ref{sec:RelatedClosed},	
	Alt and Godau showed that 
	for two closed curves $P$ and $Q$, 	
	the decision problem of whether the closed \Frechet distance between $P$
	and $Q$ (as defined in Section~\ref{sec:appl})
	is at most $\eps$ can be solved in $O(nm \log (nm))$ time. 
	We improve this long-standing result by giving an algorithm that runs in $O(nm)$ time.
	As a result, we also improve by a $\log (nm)$-factor 
	the running time of the optimization algorithm for computing the minimum such $\eps$.

	\item \emph{Minimum/Maximum walk}. \ 
	We introduce two new variants of the \Frechet distance
	as generalizations of the partial curve matching problem.
	Given two curves $P$ and $Q$ and a fixed $\eps \gee 0$, 
	the \emph{maximum walk} problem asks for the
	maximum-length continuous subcurve of $Q$ whose \Frechet distance to $P$ is at most $\eps$.
	We show that this optimization problem can be solved efficiently in $O(nm)$ time,
	without additional $\log (nm)$ factor.
	The \emph{minimum walk} problem is analogously defined, and can be solved efficiently
	within the same time bound.

	\item \emph{Graph matching}. \ 
	Given a directed acyclic graph $G$ with a straight-line embedding in $\IR^d$,
	for fixed $d \gee 2$,
	we present an algorithm to decide in $O(nm)$ time whether
	a given curve $P$ matches some part of 
	$G$ under a \Frechet distance of~$\eps$, 
	with $n$ and $m$ being the size of $P$ and $G$, respectively.
	This improves the map matching algorithm of Alt \etal~\cite{AltERW03a}
	(described in Section \ref{sec:RelatedMapMatching})
	by an $O(\log m)$ factor for the particular case in which $G$ is a directed acyclic graph.
	 
	
\end{itemize}

\REM{
They also provided an
$O(n^2 \log^2 n)$ time algorithm to compute the \Frechet distance between two
closed curves. Since then, different variants of the \Frechet distance have
been studied in the literature which are all based on constructing
a data structure called free space diagram, introduced by Alt and Godau~\cite{AltG95}.
We present a data structure on top of the free space diagram
that besides its simplicity, allows us to compute the \Frechet distance between
two closed curves, $O(\log n)$ time faster than the algorithm by Alt and
Godau. In addition, our data structure can be employed in several
applications related to partial curve matching.
}

The above improved results are obtained using a novel simple approach for
propagating the reachability information ``sequentially'' 
from bottom side to the top side of the \fs diagram.
Our approach is different from and simpler than the 
divide-and-conquer approach used by Alt and Godau~\cite{AltG95}
(explained in Section \ref{sec:RelatedPartial}),
and also, 
than the approach taken by Alt \etal~\cite{AltERW03a}
(explained in Section \ref{sec:RelatedPartial}) which is
a mixture of line sweep, dynamic programming, and Dijkstra's algorithm.

The \fs map introduced in this thesis
encapsulates all the information available in the
standard \fs diagram, yet
it is capable of answering a more general type of queries efficiently.
Namely, for any query point on the bottom side of the \fs diagram, 
our data structure can efficiently report all points on the top side of the diagram 
which are reachable from that query point.
Given that our data structure has the same
size and construction time as the standard \fs diagram, 
it can be viewed as a powerful alternative 
or generalization. 

\REM{
The current lower bound
for deciding whether the \Frechet distance
between two polygonal curves with total $n$ vertices is at most a given value $\eps$,
is $\bigOmega(n \log (n))$ \cite{LowerBound-FD}.
However, no subquadratic algorithm is known for this decision problem,
and hence, it is widely accepted that the actual lower bound 
for this problem is $\bigOmega(n^2)$ (see e.g.~\cite{AltBook2009}).
If this holds, then the results
obtained in this thesis do not only represent improvements, 
but are also optimal.
}

The remainder of this chapter is organized as follows.
In Section~\ref{sec:preliminaries}, we provide basic definitions
and elementary algorithms, such as vertical ray shooting, 
which will be used in our construction. 
In Section~\ref{sec:main}, we define the \fs map and 
show how it can be efficiently constructed.
In Section~\ref{sec:appl}, we present some applications of the \fs map to problems
such as partial curve matching, maximum/minimum walk, and closed \Frechet metric.
In Section~\ref{sec:graph}, we provide an improved algorithm
for matching a curve in a DAG.
We conclude in Section~\ref{sec:concl} with some open problems.


\section{Preliminaries} 
\label{sec:preliminaries}

\REM{
A {\em polygonal curve\/} in $\IR^d$ is  a continuous function 
$P:[0,n] \rightarrow \IR^d$ 
such that for each $i \in \set{1, \ldots, n}$,
the restriction of $P$ to the interval $[i-1, i]$ 
is affine (i.e., forms a line segment).
The integer $n$ is called the {\em size\/} of $P$.
For each $i \in \set{1, \ldots, n}$, 
we denote the line segment $P|_{[i-1,i]}$ by $P_i$.  

A {\em monotone reparametrization} of $[0,n]$ 
is a continuous non-decreasing function $\alpha: [0,1] \rightarrow [0,n]$
with $\alpha(0)=0$ and $\alpha(1)=n$.
Given two polygonal curves $P$ and $Q$ of size $n$ and $m$, respectively, 
the {\em \Frechet distance\/} between $P$ and $Q$ is defined as
\[
	\distF(P,Q) = \inf_{\alpha, \beta} \max_{t \in [0,1]} \| P(\alpha(t)), Q(\beta(t)) \|,
\]
where $\|\!\cdot\!\|$ denotes the Euclidean metric, 
and $\alpha$ and $\beta$ range over all monotone reparameterizations of 
$[0,n]$ and $[0,m]$, respectively.
}
Here, we borrow some notations from previous 
chapters.
Given a parameter $\eps \gee 0$,
the {\em free space\/} of the two curves $P$ and $Q$ is defined as
$$
	\Feps(P,Q) = \set{(s,t) \in [0,n] \times [0,m] \ | \ \|P(s),Q(t)\| \lee \eps }.
$$

We call points in $\Feps(P,Q)$ \emph{feasible}.
The partition of the rectangle $[0,n] \times [0,m]$ 
into regions formed by feasible and infeasible points
is called the \emph{\fs diagram} of $P$ and $Q$, denoted by $\FD_\eps(P,Q)$
(see Figure~\ref{fig:free-space}.a).

\begin{figure}[t]
	\centering
	\includegraphics[height=0.4\columnwidth]{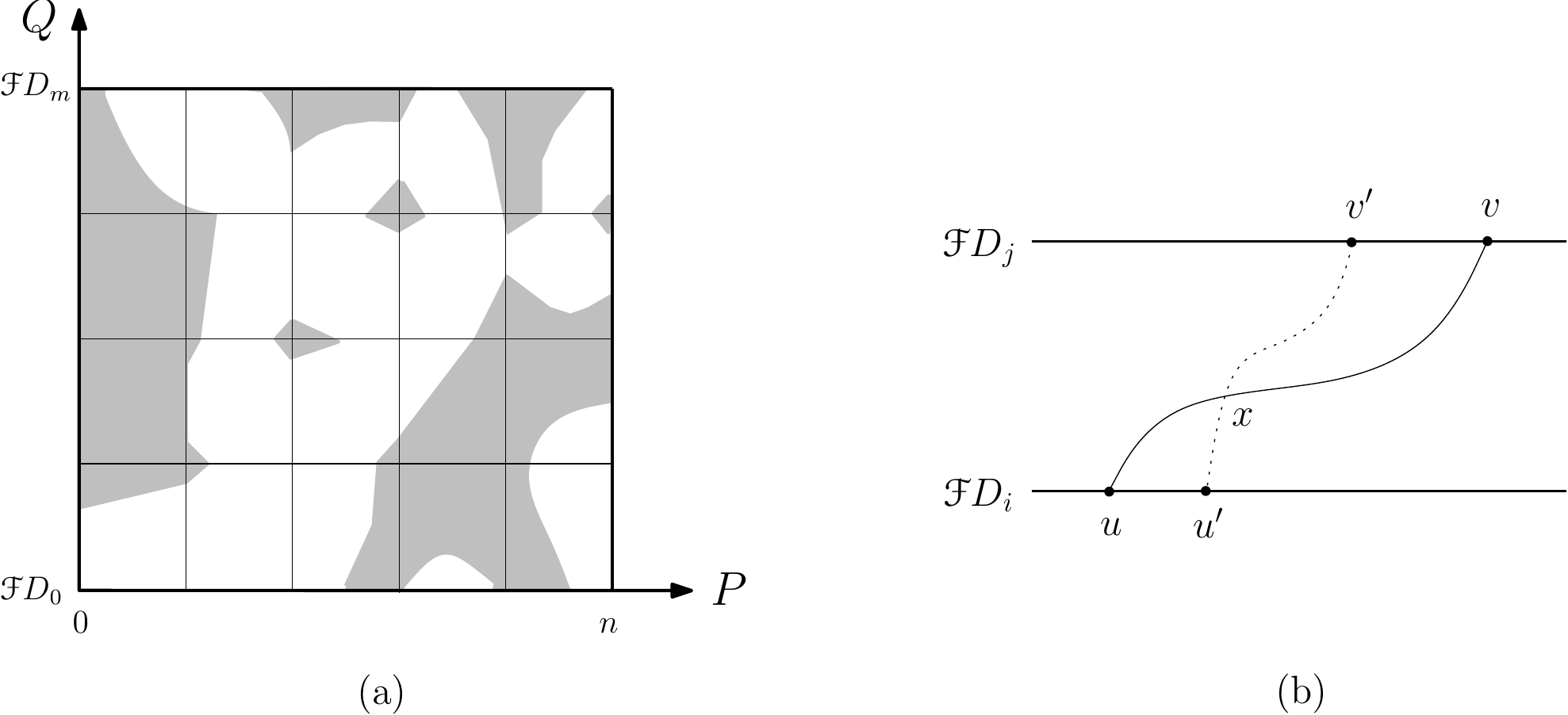}
	\caption{(a) An example of a \fs diagram.
	(b) Proof of the crossing lemma.}
	\label{fig:free-space}
	\vspace{-0.5em}	
\end{figure}

Let $P$ and $Q$ be two polygonal curves of size $n$ and $m$, respectively,
and $\eps \gee 0$ be a fixed parameter.
Following the notation used by Alt~\etal~\cite{AltERW03a},
we denote by $\FD_j$, for $0 \lee j \lee m$, 
the one-dimensional \fs diagram $\FD_\eps(P, Q) \cap ([0,n] \times \set{j})$,
corresponding to the curve $P$ and the point $Q(j)$. 
For each $(i,j) \in \set{1 \cdots n} \times \set{1 \cdots m}$,
the intersection of the free-space diagram with the square $[i-1,i] \times [j-1,j]$ 
is called a \emph{cell} of the diagram. 
Likewise, we call the intersection of $\FD_j$ with each interval $[i-1,i]$ a cell 
(or more precisely, the $i$-th cell) of $\FD_j$.

A curve is called \emph{feasible} if it lies completely within $\Feps(P,Q)$,
and is called {\em monotone} if it is monotone in both $x$- and $y$-coordinates.
Given two points $u$ and $v$ in the free space, we say that 
$v$ is \emph{reachable} from~$u$,
denoted by $u \reach v$,
if there is a monotone feasible curve in $\Feps(P,Q)$ from $u$ to~$v$.
Alt and Godau~\cite{AltG95} showed that $\distF(P,Q) \lee \eps$ 
if and only if 
$(0,0) \reach (n,m)$.
Clearly, reachability is ``transitive'':
if $u \reach v$ and $v \reach w$, then $u \reach w$.
Given two points $a$ and $b$ in the plane, 
we write $a < b$ if $a_x < b_x$,
and write $a \lee b$ if $a_x \lee b_x$.

\begin{lemma}[Crossing Lemma] \label{lemma:cross}
	Let $u,u' \in \F_i$ and $v,v' \in \F_j$ $(i < j)$ such that $u \lee u'$ and $v' \lee v$.
	If $u \reach v$ and $u' \reach v'$, then $u \reach v'$ and $u' \reach v$.
\end{lemma}

\begin{proof}
	Let $\pi$ be a monotone feasible curve that connects $u$ to $v$. 
	Since $u'$ and $v'$ are on different sides of $\pi$,
	any monotone curve that connects $u'$ to $v'$ in $\Feps(P,Q)$
	intersects $\pi$ at some point $x$ (see Figure~\ref{fig:free-space}.b).
	The concatenation of the subcurve from $u$ to $x$ 
	and the one from $x$ to $v'$ gives a monotone feasible curve from 
	$u$ to $v'$. 
	Similarly, $v$ is connected to $u'$ by a monotone feasible curve through $x$.
	\qed
\end{proof}

\REM{
\begin{figure}[t]
	\centering
	\includegraphics[width=0.4\columnwidth]{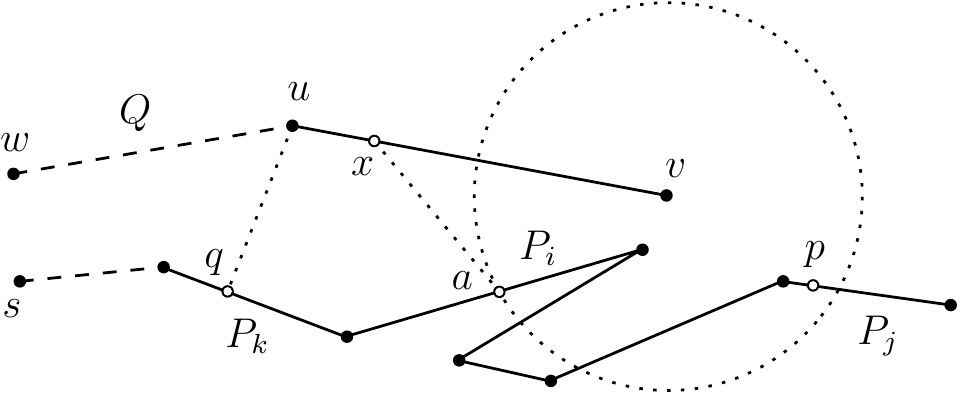}
	\caption{Crossing Lemma.}
	\label{fig:cross}
\end{figure}
}

For $0 \lee j \lee m$, we denote by $\F_j$ the set of feasible points in $\FD_j$. 
$\F_j$ consists of $O(n)$ feasible intervals,
where each feasible interval is a maximal continuous set of feasible points,
restricted to be within a cell.
For any feasible point set $S$ in $\Feps(P,Q)$,
we define the \emph{projection} of $S$ on $\FD_j$ as
$$
	\RE_{j}(S) := \set{v \in \F_j \provided \exists u \in S \,\mbox{ s.t. }\,  u \reach v}.
$$

For an interval $I$ on $\F_i$,
we define the \emph{left pointer} of $I$ on $\F_j$ $(i \lee j)$,
denoted by $\lp_{j}(I)$, to be the leftmost point in $\RE_{j}(I)$.
Similarly, the \emph{right pointer} of $I$ on $\F_j$, denoted by $\rp_{j}(I)$,
is defined to be the rightmost point in $\RE_{j}(I)$.
If $\RE_{j}(I)$ is empty, both pointers $\lp_{j}(I)$ and $\rp_{j}(I)$ are set to NIL.
These pointers were previously used in \cite{AltERW03a,AltG95},
and form a main ingredient of our data structure.
For a single point $u$, we simply use $\RE_{j}(u)$, $\lp_{j}(u)$, and $\rp_{j}(u)$
instead of $\RE_{j}(\set{u})$, $\lp_{j}(\set{u})$, and $\rp_{j}(\set{u})$, respectively. 
The following observation is an immediate corollary of Lemma~\ref{lemma:cross}.

\begin{obs} \label{obs:sorted}
	For any two points $u,v \in \F_i$ with $u \lee v$, and for any $j \gee i$,
	we have $\lp_{j}(u) \lee \lp_{j}(v)$ and $\rp_{j}(u) \lee \rp_{j}(v)$.	
\end{obs}

For an interval $I$ on a horizontal line, we denote by $\Left(I)$ and $\Right(I)$ 
the left and the right endpoint of $I$, respectively.
The following simple lemma 
will be used frequently throughout this chapter.

\begin{lemma} \label{lemma:scan}
	Given two sequences $A$ and $B$ of points on a horizontal line sorted from left to right, 
	we can compute for each point $a \in A$,  
	the leftmost point $b \in B$ with $a \lee b$
	in $O(|A|+|B|)$ total time.
\end{lemma}

\begin{proof}
	We scan the two sequences simultaneously from left to right using two pointers.
	Whenever we reach a point $b \in B$,
	we advance our pointer on $A$ until we reach the first point $a \in A$
	with $a > b$. We then make all points in $A$ scanned during this step up to
	(not including) $a$ to point to $b$.
	We then advance our pointer on $B$ by one, and repeat the above procedure.
	\qed
\end{proof}

\REM{
\begin{obs} \label{lemma:scan}
	Given two sequences $A$ and $B$ of points on a horizontal line sorted from left to right, 
	we can preprocess the two sequences into a data structure of size $O(|A|)$ in $O(|A|+|B|)$ time,
	such that for any query point $a \in A$, 
	the leftmost point $b \in B$ with $a < b$
	can be reported in $O(1)$ time,
	assuming random access to the elements of $A$.
\end{obs}

\begin{proof}
	We scan the two sequences simultaneously from left to right
	We scan the points in $B$ in order from left to right,
	and for each point $b$ in $B$, assign $b$ to all points $a \in A$ to the right of $b$
	which are not yet processed (i.e., are not yet assigned with any point of $B$).
	At each time, we keep a pointer to the leftmost non-processed point of $A$,
	and move the pointer to the right linearly as points of $A$ are processed.
	It is easy to verify that during this process, each point of $A$ 
	is assigned with the nearest point of $B$ to its right,
	and that each point of $A$ and $B$ is visited only once.
	\qed
\end{proof}
}

\subsection{Vertical Ray Shooting}
\label{sec:shooting}

The following special case of the vertical ray shooting problem
appears as a subproblem in our construction.
Consider a vertical slab $[0,1] \times [0,m]$ (see Figure~\ref{fig:shooting}).
For each $1 \lee i \lee m$, 
there are two (possibly empty) segments 
in the slab at height~$i$,
attached to the boundaries of the slab, one from left and the other from right.
Given a query point $q$,
the vertical ray shooting problem involves finding
the first segment in the slab directly above $q$.
If the query points are restricted to be among 
the endpoints of the segments,
we show below that the vertical ray shooting queries can be answered in $O(1)$
time, after $O(m)$ preprocessing time.

\begin{figure}[t]
	\centering
	\includegraphics[height=0.37\columnwidth]{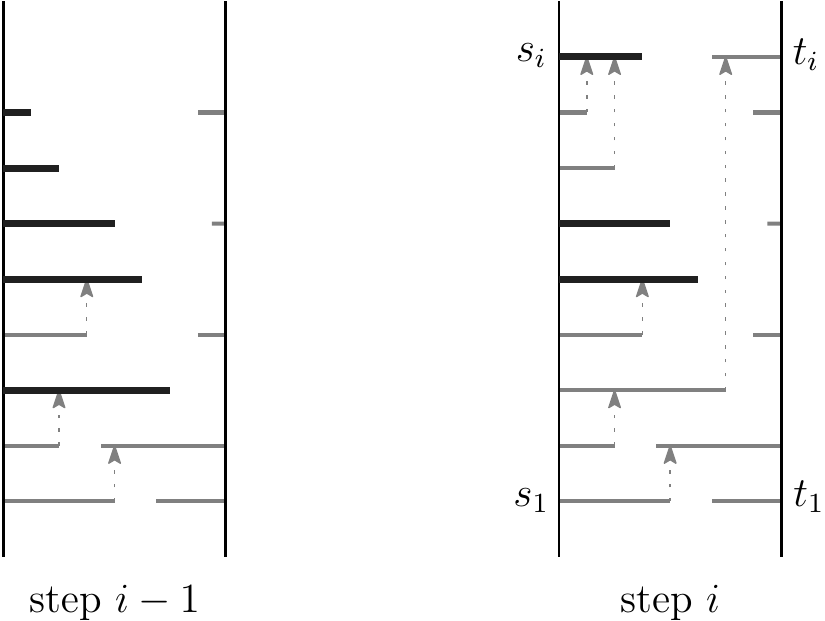}
	\caption{An example of the execution of Algorithm~\ref{alg:shooting}.
	Segments in the queue at the end of each step are shown in bold.
	}
	\label{fig:shooting}
\end{figure}

\begin{lemma} \label{lemma:shooting}
	Let $S$ be a set of segments
	$s_i = [0,a_i] \times \set{i}$, and 
	$T$ be a set of segments $t_i = [b_i, 1] \times \set{i}$
	with $0 \lee a_i \lee b_i \lee 1$, for $1 \lee i \lee m$. 
	We can find for each segment $s_i \in S$, 
	the first segment in $S \cup T$ 
	directly above $\Right(s_i)$ 
	in $O(m)$ total time.
\end{lemma}

\begin{proof}
	Algorithm~\ref{alg:shooting} 
	assigns to each segment $s_i$ of $S$,
	an \emph{up} pointer that points to the first segment directly above $\Right(s_i)$,
	if such a segment exists. 
	The algorithm makes use of a double-ended queue $Q$ 
	(a combination of a queue and a stack, commonly known as ``deque''),
	that supports the standard operations
	{\sc push}(), {\sc pop}(), and {\sc top}(),
	along with two additional operations
	{\sc bottom}() and {\sc bottom-pop}(), that are analogous to
	{\sc top}() and {\sc pop}(), respectively,
	but applied to the bottom of the queue.

\begin{algorithm} [h]
\caption {\sc Ray-Shooting$(S,T)$} \label{alg:shooting}
\algsetup{indent=1.5em}
\begin{algorithmic}[1]
	\vspace{0.5em}
	\baselineskip=1\baselineskip
	\STATE $Q \eq \emptyset$
	\STATE $Q.\mbox{\sc push}(s_1)$
	\FOR {$i$ from 2 to $m$} 
		\WHILE {$|Q.\mbox{\sc top}()| \lee |s_i| $}  \label{l:w1-1}
			\STATE $Q.\mbox{\sc pop}().\up \eq s_i$ \label{l:w1-2}
		\ENDWHILE 
		\WHILE {$|Q.\mbox{\sc bottom}()| \gee 1-|t_i| $}  \label{l:w2-1}
			\STATE $Q.\mbox{\sc bottom-pop}().\up \eq t_i$ \label{l:w2-2}
		\ENDWHILE 
		\STATE $Q.\mbox{\sc push}(s_i)$
	\ENDFOR 
\end{algorithmic}
\end{algorithm}
	
	We say that a segment $s \in S$
	is \emph{covered} by a segment $t \in S \cup T$,
	if a vertical ray from $\Right(s)$ intersects $t$.	
	For  $1 \lee i \lee m$,
	let $S_i = \set{s_1, \ldots, s_i}$ and $T_i = \set{t_1, \ldots, t_i}$.
	The following invariant is maintained by the algorithm:
	At the end of iteration~$i$, 
	$Q$ contains a subset of segments from $S_i$
	that are not covered by any segment from $S_i \cup T_i$,
	in a decreasing order of their lengths from bottom to the top of the queue.
	The invariant clearly holds for $i=1$.
	Suppose by induction that the invariant holds for $i-1$.
	In the $i$-th iteration,
	we first pop off from the top of the queue all segments
	covered by $s_i$, in Lines~\ref{l:w1-1}--\ref{l:w1-2}.
	Then, we remove from the bottom of the queue all
	segments covered by $t_i$, in Lines~\ref{l:w2-1}--\ref{l:w2-2}.
	Finally, we add $s_i$ to the top of the queue.
	(See Figure~\ref{fig:shooting} for an illustration.)
	It is easy to verify that after the insertion of $s_i$,
	the segments of $Q$ are still sorted in a decreasing order of their lengths
	(because we have already removed segments smaller than $s_i$ from $Q$),
	and that, no segment of $Q$ is covered by a segment in $S_i \cup T_i$
	(because we have removed covered segments from $Q$).
	Furthermore, it is clear that any segment $s$ removed from $Q$
	is assigned to the first segment that is directly above $\Right(s)$,
	because we are processing segments in order from bottom to the top.
	The correctness of the algorithm therefore follows.
	Note that after the termination of the algorithm,
	$Q$ still contains some uncovered segments from $S$, whose
	up pointers are assumed to be NIL, 
	as they are not covered by any segment in $S \cup T$.
	Since each segment of $S$ is inserted into and removed from the queue at most once,
	Lines~\ref{l:w1-2} and \ref{l:w2-2} of the algorithm are executed at most $m$ times,
	and hence, the whole algorithm runs in $O(m)$ time.
	\qed
\end{proof}

Consider the vertical slab $[0,1] \times [0,m]$, 
and the two sets of segments $S$ and $T$ as defined above. 
We call a segment $t_j \in T$ {reachable} from a segment $s_i \in S$,
if there is a monotone path from a point on $s_i$ to a point on $t_j$ 
not intersecting any other segment in $S \cup T$.
For a segment $s_i \in S$, the \emph{topmost reachable} segment in $T$ is 
a reachable segment $t_j$ with the maximum index $j$. 
In Figure~\ref{fig:shooting}, for example, 
the topmost reachable segments for $s_1$ and $s_2$
are $t_2$ and $t_i$, respectively.

\begin{lemma} \label{lemma:top-segment}
	Let $S$ and $T$ be the two sets of segments  defined in Lemma~\ref{lemma:shooting}.
	Then, for each segment $s_i \in S$, $1 \lee i \lee m$, 
	the topmost reachable segment in $T$
	can be computed in $O(m)$ total time.
\end{lemma}

\begin{algorithm} [h]
\caption {\sc Topmost-Reachable-Segments$(S,T)$} \label{alg:top-segment}
\algsetup{indent=1.5em}
\begin{algorithmic}[1]
	\vspace{0.5em}
	\baselineskip=1\baselineskip
	\FOR {$i$ from $m$ to $1$} 
		\IF {$s_i.\up = \nil$}
			\STATE $s_i.\topp \eq t_m$
		\ELSIF {$s_i.\up \in T$}
			\STATE $s_i.\topp \eq s_i.\up$
		\ELSE
			\STATE $s_i.\topp \eq s_i.\up.\topp$
		\ENDIF
	\ENDFOR 
\end{algorithmic}
\end{algorithm}

\begin{proof}
	Algorithm~\ref{alg:top-segment} scans all segments in $S$ from top to bottom,
	and assigns to each segment $s_i$ in $S$
	a \emph{top} pointer that points to the topmost segment in $T$ reachable from $s_i$.
	The algorithm works as follows.
	Suppose that the top pointers for all segments in $S$ above $s_i$ are computed.
	At $i$-th iteration, if $s_i$ is not covered by any other segment above it,
	(i.e., $s_i.\up$ is $\nil$), then the topmost reachable segment of $s_i$ is set to $t_m$.
	If $s_i$ is covered by a segment $t_j \in T$, 
	then the topmost reachable segment of $s_i$ is $t_j$.
	Otherwise, if $s_i$ is covered by a segment $s_j \in S$,
	then all segments in $T$ above $s_j$ that are reachable from $s_i$ are also reachable from $s_j$,
	and hence, the topmost such segment can be obtained from $s_j.\topp$ pointer,
	which is computed earlier.
	Therefore, computing all top pointers can be performed in $O(m)$ total time.
	\qed
\end{proof}

\noindent
Analogous to the previous lemma, 
a result can be stated for a horizontal slab.

\begin{figure}[t]
	\centering
	\includegraphics[height=0.18\columnwidth]{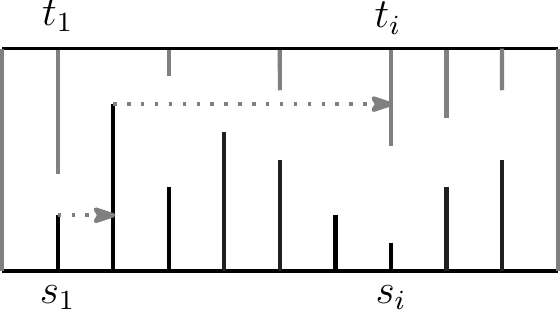}
	\caption{A horizontal slab with vertical segments. 
	The rightmost segment reachable from $s_1$ in this figure is $t_i$.
	}
	\label{fig:horizontal}
\end{figure}

\begin{corollary} \label{cor:rightmost}
	Consider a horizontal slab $[0,n] \times [0,1]$. 
	Let $S$ be a set of segments
	$s_i = \set{i} \times [0,a_i]$, and 
	$T$ be a set of segments $t_i = \set{i} \times [b_i, 1]$
	with $0 \lee a_i \lee b_i \lee 1$, for $1 \lee i \lee n$ (see Figure~\ref{fig:horizontal}).
	Then, for all segments $s_i \in S$, 
	the rightmost segment in $T$ reachable from $s_i$
	can be computed in $O(n)$ total time.
\end{corollary}


\REM{
It is easy to observe that $\lp_{j}(I) = \lp_{j}(\Left(I))$. 
However, this is not the case that $\rp_{j}(I) = \lp_{j}(\Right(I))$ (see Figure~\ref{fig:pointers}.b).
We define $\Last_j(I)$ to be the rightmost point $n \in I$ such that $\RE_{j}(n)$ is non-empty. 
Obviously, $\rp_{j}(I) = \lp_{j}(\Last_j(I))$.
}

\REM{
\vspace{.5em}
\begin{obs} \label{obs:main}
	Let $I$ be a feasible interval of $\F_i$, and $j$ be an index $(0 \lee i \lee j \lee m)$. Then
	\addtolength\leftmargini{0.8em}
	\begin{enumerate}
		\item[\rm (i)] any point on $\RE_{j}(I)$ is reachable from $\Left(I)$;
		\item[\rm (ii)] for any point $n \in I$ with non-empty $\RE_j(n)$, $\rp_{j}(n) = \rp_{j}(I)$;
		\item[\rm (iii)] for any two point $n,m \in I$ with $n \lex m$, we have $\lp_{j}(n) \lex \lp_{j}(m)$ and $\rp_{j}(n) \lex \rp_{j}(m)$.
	\end{enumerate}		
\end{obs}

\begin{proof}
	(i) and (ii) follow from transitivity of reachability.
	(iii) follows from Lemma~\ref{lemma:cross}.
	\qed
\end{proof}
}

\REM{
\noindent
For $0 \lee j \lee m$, 
we define the \emph{reachable set} $\R(j)$ recursively as follows: 
\begin{itemize}
	\item $\R(0) = \F_0$,
	\item $\R(j) = \RE_{j}(\R(j-1))$ for all $1\lee j \lee m$.
\end{itemize}
}%


\section{The Main Data Structure} 
\label{sec:main}

In this section, we describe our main data structure 
that yields improved algorithms for several variants of the \Frechet distance. 
For $0 \lee j \lee m$, 
we define the \emph{reachable set} $\R(j) := \RE_{j}(\F_0)$
to be the set of all points in $\F_j$ reachable from $\F_0$.  
We call each interval of $\R(j)$,
contained in a feasible interval of $\F_j$,
a \emph{reachable interval}.
By our definition, $\R(0) = \F_0$.
The following observation is immediate by the transitivity of reachability.

\begin{obs} \label{obs:reachable}
	For $0 \lee i < j \lee m$, $\R(j) = \RE_{j}(\R(i))$.
\end{obs}
An important property of the reachable sets is described in the following lemma.

\begin{lemma} \label{lemma:reach}
	For any two indices $i,j$ $(0 \lee i < j \lee m)$ and any point $u \in \R(i)$,
	$\RE_{j}(u) = \R(j) \cap [\lp_j(u), \rp_j(u)]$. 
\end{lemma}

\begin{proof}
	Let $S = [\lp_j(u), \rp_j(u)]$.
	By Observation~\ref{obs:reachable},
	$\R(j) = \RE_{j}(\R(i))$.
	Thus, it is clear by the definition of pointers 
	that $\RE_{j}(u) \subseteq \R(j) \cap S$.
	Therefore, it remains to be shown that $\R(j) \cap S \subseteq \RE_{j}(u)$.
	Suppose, by way of contradiction, that there is a point $v \in \R(j)  \cap S$
	such that $v \not\in \RE_{j}(u)$.
	Since $v \in \R(j)$, there exists some point $u' \in \R(i)$ such that
	$u' \reach v$.
	If $u'$ is to the left (resp., to the right) of $u$, then 
	the points $u, u', v$, and $\lp_j(u)$ (resp., $\rp_j(u)$)
	satisfy the conditions of Lemma~\ref{lemma:cross}.
	Therefore, by Lemma~\ref{lemma:cross}, $u \reach v$,
	which implies that $v \in \RE_{j}(u)$; a contradiction.
	\qed
\end{proof}

Lemma~\ref{lemma:reach} provides an efficient method for storing the sets
$\RE_{j}(I)$, for all feasible intervals $I$ on $\F_0$.
Namely, instead of storing each set $\RE_{j}(I)$ separately,
one set per feasible interval $I$, 
which takes up to $\bigTheta(n^2)$ space,
we only need to store a single set $\R(j)$, 
along with the pointers $\lp_{j}(I)$ and $\rp_{j}(I)$,
which takes only $O(n)$ space in total.
The set $\RE_{j}(I)$, for each interval $I$ on $\F_0$,
can be then obtained by $\R(j) \cap [\lp_j(I), \rp_j(I)]$.
For each interval $I$ on $\F_0$,
we call the set $\set{\lp_j(I), \rp_j(I)}$
a \emph{compact representation} of $\RE_{j}(I)$.
The following lemma is a main ingredient of our fast computation of reachable sets.

\begin{lemma} \label{lemma:propagate}
	For $0 < j \lee m$, if $\R(j-1)$ is given,
	then $\R(j)$ can be computed in $O(n)$ time. 
\end{lemma}

	\begin{figure}[t]
		\centering
		\includegraphics[width=0.8\columnwidth]{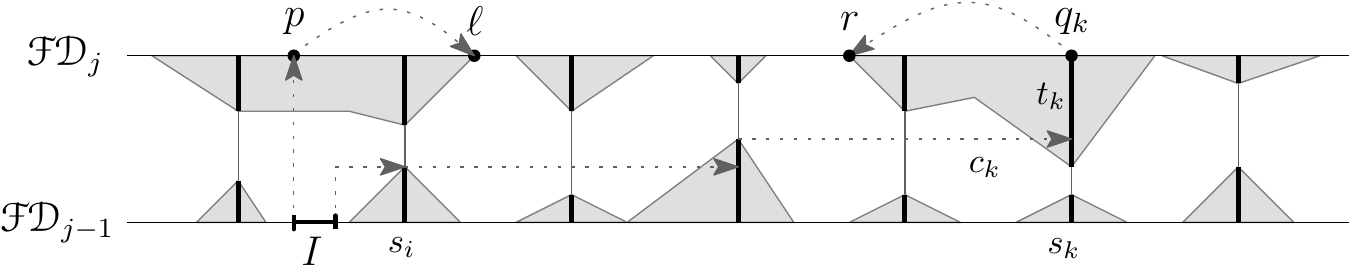}
		\caption{
			Computing $\R(j)$ from $\R(j-1)$.}
		\label{fig:propagate}
	\end{figure}
	
\begin{proof}
	Let $\CD$ be 
	the intersection of the \fs diagram
	with the rectangle $[0,n]\times[j-1,j]$.
	$\CD$ is composed of $n$ square cells, numbered from left to right by $c_1$ to $c_n$.
	For all reachable intervals $I$ on $\R(j-1)$, 
	we compute pointers $\lp_j(I)$ and $\rp_j(I)$
	in $O(n)$ time as follows.
	For each cell $c_k$ in $\CD$, the intersection of 
	the right boundary of $c_k$ with the infeasible part of the \fs diagram
	forms two (possibly empty) vertical segments, denoted by
	$s_k$ and $t_k$, respectively, 
	as in Figure~\ref{fig:propagate}. 
	For each cell $c_k$, 
	we denote the top-right corner of $c_k$ by $q_k$.
	We pre-compute for each point $q_i$, $1 \lee i \lee n$,
	a pointer $\Next(q_i)$ (resp., $\Prev(q_i)$) 
	that points to the first feasible point on or immediately
	after (resp., before) $q_i$ in $\FD_j$.
	Let $S$ be the set of all left and right endpoints of feasible intervals on $\FD_j$.
	Since for each point $q_i$, $\Next(q_i)$ and $\Prev(q_i)$, if not $\nil$, are
	included in $S$, we can compute all next/prev pointers using two linear scans
	in $O(n)$ time by Lemma~\ref{lemma:scan}.
	After computing $\Next(q_i)$ pointers,
	we can compute $\Next(q)$ for any point $q \in \FD_j$ in constant time.
	
	Now, fix an interval $I$ on $\R(j-1)$.
	We compute $\lp_j(I)$ and $\rp_j(I)$ as follows.
	Let $c_i$ be the cell containing $I$,
	let $p$ be the vertical projection of $\Left(I)$ onto $\FD_j$,
	and let $t_k$ be the rightmost segment reachable from $s_i$,
	computed by Corollary~\ref{cor:rightmost}
	(see Figure~\ref{fig:propagate}).
	We set $\ell = \Next(p)$ and $r = \Prev(q_k)$.
	(If $\Next(p) > \Prev(q_k)$, we set $\ell = r = \nil$.)
	It is easy to verify that no point before $\ell$ and no point after $r$ on $\FD_j$
	can be reachable from $I$, and that, every feasible point on $\FD_j$ between 
	$\ell$ and $r$ is reachable from $I$. 
	Therefore, $\lp_j(I) = \ell$ and $\rp_j(I) = r$.
	As a result, computing $\lp_j(I)$ and $\rp_j(I)$ 
	for each reachable interval $I$ on $\R(j-1)$
	takes $O(1)$ time,
	after $O(n)$ preprocessing time for computing the next/prev pointers.
	Thus, we can compute $\lp_j(I)$ and $\rp_j(I)$ for all reachable intervals $I$ on $\R(j-1)$
	in $O(n)$ total time.

	After computing the left and right pointers, 
	we can produce $\R(j) = \RE_j(\R(j-1))$
	by identifying those (portions of) intervals on $\F_j$
	that lie in at least one interval $[\lp_j(I), \rp_j(I)]$.
	Since for all intervals $I$ on $\R(j-1)$ sorted from left to right, 
	$\lp_j(I)$'s and $\rp_j(I)$'s are in sorted order 
	by Observation~\ref{obs:sorted},
	we can accomplish this step by a linear scan over the 
	left and right pointers in $O(n)$ time.
	\qed
\end{proof}

\REM{
\begin{algorithm} [h]
\caption {\sc Compute-Pointers$(I)$} 
\label{alg:pointers}
\algsetup{indent=1.5em}
\begin{algorithmic}[1]
	\vspace{0.5em}
	\STATE let $p = $ orthogonal projection of $\Left(I)$ onto $\FD_j$ 
	\STATE $\ell \eq \Next(p)$
	\STATE $r \eq \Prev(q_k)$
	\IF {$\ell = \nil \text{ or } r = \nil \text{ or } \ell > r$}
		\STATE $\ell, r \eq \nil$ 
	\ENDIF
	\STATE set $\lp_j(I) \eq \ell, \rp_j(I) \eq r$ 
\end{algorithmic}
\end{algorithm}
}
\newpage
\subsection{Data Structure}

We now describe our main data structure, 
which we call \fs map.
The data structure maintains reachability information 
on each row of the \fs diagram, using some additional
pointers that help answering reachability queries efficiently. 
The \emph{\fs map} of two curves $P$ and $Q$
consists of the following:
\begin{itemize}
	\item[(i)] the reachable sets $\R(j)$, for $0 \lee j \lee m$,
	\item[(ii)] the right pointer $\rp_{j}(I)$ for each reachable interval $I$ on $\R(j-1)$, $0 < j \lee m$,
	\item[(iii)] the leftmost reachable point after each cell in $\FD_j$, for $0 < j \lee m$, and
	\item[(iv)] the rightmost take-off point before each cell in $\FD_j$, for $0 \lee j < m$,
\end{itemize}
where a \emph{take-off} point on $\FD_j$ is a reachable point in $\R(j)$
from which a point on $\FD_{j+1}$ is reachable.
For example, in Figure~\ref{fig:query}, $\ell_j$ is the leftmost reachable point after $\ell'$,
and $r'$ is the rightmost take-off point before $r_{j-1}$.
For a cell $c$ in $\FD_j$, by \emph{after} $c$ we mean after $\Right(c)$,
and by \emph{before} $c$ we mean before $\Left(c)$.

\begin{lemma} \label{lemma:map}
	Given two polygonal curves $P$ and $Q$ of size $n$ and $m$, respectively,
	we can build the \fs map of $P$ and $Q$ in $O(nm)$ time.
\end{lemma}

\begin{proof}
	We start from $\R(0) = \F_0$, 
	and construct each $\R(j)$ iteratively from $\R(j-1)$, for $j$ from 1 to $m$,
	using Lemma~\ref{lemma:propagate}.
	The total time needed for this step is $O(nm)$.
	The construction of $\R(j)$,
	as seen in the proof of Lemma~\ref{lemma:propagate}, 
	involves computing all right (and left) pointers,
	for all reachable intervals on $\R(j-1)$. 
	Therefore, item~(ii) of the data structure can be obtained at no additional cost.
	Item~(iii) is computed as follows.
	Let $S$ be the set of all left pointers obtained upon constructing $\R(j)$.
	For each cell $c$ in $\FD_j$, the leftmost reachable point after $c$, if any, 
	is a member of $S$. 
	We can therefore compute item~(iii) for each row $\FD_j$
	by a linear scan over the cells and the set $S$ 
	using Lemma~\ref{lemma:scan} in $O(n)$ time.
	For each row, item~(iv) can be computed analogous to item~(iii), but in a reverse order.
	Namely, given the set $\R(j)$, we compute the set of points on $\FD_{j-1}$
	reachable from $\R(j)$ in the \fs diagram rotated by 180 degrees.
	Let $S$ be the set of all left pointers obtained in this reverse computation.
	For each cell $c$ in $\FD_{j-1}$, the rightmost take-off point before $c$, if there is any, 
	is a member of $S$. 
	We can therefore compute item~(iv) for each row 
	by a linear scan over the cells and the set $S$ 
	using Lemma~\ref{lemma:scan} in $O(n)$ time.
	The total time for constructing the \fs map is therefore $O(nm)$.
	\qed
\end{proof}

In the following, we show how the reachability queries can be 
efficiently answered, using the \fs map.
For the sake of describing the query algorithm, 
we introduce two functions as follows.
Given a point $u \in \FD_j$,
we denote by $\lr(u)$ the leftmost reachable point on or after $u$ on $\FD_j$.
Analogously, we denote by $\rl(u)$ the rightmost take-off point on or before $u$ on $\FD_j$.
Note that both these functions can be computed in $O(1)$
time using the pointers stored in the \fs map.

\begin{algorithm} [t]
\caption {{\sc Query($u$)}, where $u \in \F_0$} 
\label{alg:query}
\algsetup{indent=1.5em}
\begin{algorithmic}[1]
	\vspace{0.5em}
	\baselineskip=1\baselineskip
	\STATE let $\ell_0 = r_0 = u$ \label{l:init}
	\FOR {$j = 1$ to $m$} 
		\STATE let $\ell'$ be the orthogonal projection of $\ell_{j-1}$ onto $\FD_j$  \label{l:l1}
		\STATE $\ell_j \eq \lr(\ell')$   \label{l:l2}
		\STATE let $r' = \rl(r_{j-1})$
		\IF {$r' < \ell_{j-1}$ or $r' = \nil$} \label{l:c1}
			\STATE $r_j \eq \nil$ \label{l:c2}
		\ELSE
			\STATE $r_j \eq \rp_{j}(I)$, for $I$ being the reachable interval containing $r'$
		\ENDIF 	
		\IF {$\ell_j$ or $r_j$ is $\nil$} 
			\STATE return $\nil$
		\ENDIF 
	\ENDFOR 
	\STATE return $\ell_m, r_m$

\end{algorithmic}
\end{algorithm}

\begin{lemma} \label{lemma:query}
	Let the \fs map of $P$ and $Q$ be given.
	Then, for any query point $u \in \F_0$, $\lp_{m}(u)$ and $\rp_{m}(u)$
	can be computed in $O(m)$ time.
\end{lemma}

\begin{proof}
	The procedure for computing $\lp_{m}(u)$ and $\rp_{m}(u)$
	for a query point $u \in \F_0$ is described in Algorithm~\ref{alg:query}.
	The following invariant holds during the execution of the algorithm:
	After the $j$-th iteration, $\ell_j = \lp_{j}(u)$ and $r_j = \rp_{j}(u)$.
	We prove this by induction on $j$.
	The base case, $\ell_0 = r_0 = u$, trivially holds. 
	Now, suppose inductively that 
	$\ell_{j-1} = \lp_{j-1}(u)$ and $r_{j-1} = \rp_{j-1}(u)$.
	We show that after the $j$-th iteration, the invariant holds for $j$.
	We assume, w.l.o.g., that $\RE_j(u)$ is non-empty, i.e., $\lp_{j}(u) \lee \rp_{j}(u)$.
	Otherwise, the last take-off point from $\R(j-1)$ will be either $\nil$, or 
	smaller than $\ell_{j-1}$, which is then detected and handled 
	by Lines~\ref{l:c1}--\ref{l:c2}.
	
	\begin{figure}[b]
		\centering
		\includegraphics[width=0.7\columnwidth]{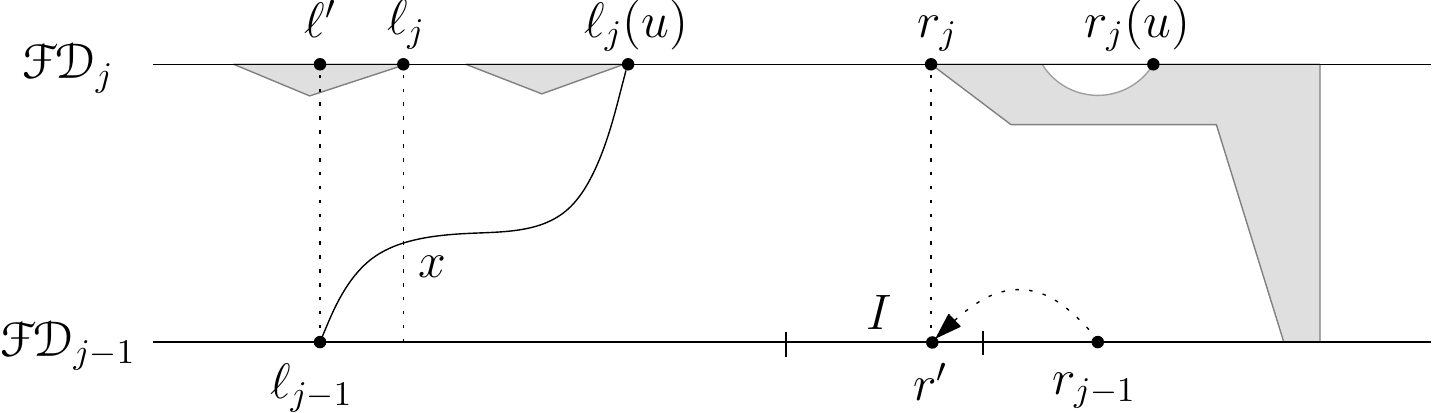}
		\caption{Proof of Lemma~\ref{lemma:query}.}
		\label{fig:query}
	\end{figure}

	We first show that $\ell_j = \lp_{j}(u)$.
	Suppose by contradiction that $\ell_j \not= \lp_{j}(u)$.
	If $\ell_j < \lp_{j}(u)$, then we draw a vertical line
	from $\ell_j$ to $\FD_{j-1}$ (see Figure~\ref{fig:query}).
	This line crosses any monotone path from
	$\ell_{j-1} = \lp_{j-1}(u)$ to $\lp_j(u)$ at a point $x$.
	The line segment $x\ell_j$ is completely in the free space,
	because otherwise, it must be cut by an obstacle, which 
	contradicts the fact that the free space inside a cell is convex~\cite{AltG95}.
	But then, $\ell_j$ becomes reachable from $\ell_{j-1}$ through $x$,
	contradicting the fact that $\lp_j(u)$ is the leftmost reachable point in $\R(j)$.
	The case, $\ell_j > \lp_{j}(u)$, cannot arise, 
	because then, $\lp_j(u)$ is a reachable point after $\ell'$ and
	before $\ell_j$, which contradicts
	our selection of $\ell_j$ as the leftmost reachable point of $\ell'$ in line~\ref{l:l2}.

	We can similarly show that $r_j = \rp_{j}(u)$.
	Suppose by contradiction that $r_j \not= \rp_{j}(u)$.
	The case $r_j > \rp_{j}(u)$ is impossible,
	because otherwise, $r_j$ is a point on $\R(j)$
	reachable from $\R(j-1)$ which appears after $\rp_j(u)$.
	This contradicts the fact that $\rp_j(u)$ is the rightmost point on $\R(j)$.
	If $r_j < \rp_{j}(u)$ (see Figure~\ref{fig:query}),
	then $\rp_{j}(u)$ is reachable from a point $x \in \R(j-1)$ with $x < r'$, 
	because $r'$ is the rightmost take-off point on or before $r_{j-1}$.
	But then, by Lemma~\ref{lemma:cross}, $\rp_{j}(u)$ is reachable from $r'$,
	which contradicts the fact that $r_j$ is the left pointer of the reachable interval $I$
	containing $r'$.
	\qed
\end{proof}

\REM{
The result of this section is summarized below.

\begin{theorem} \label{thm:mainFreeSpaceMap}
	Given two polygonal curves $P$ and $Q$ of size $n$ and $m$, respectively,
	we can build in $O(nm)$ time a data structure of size $O(nm)$, such that
	for any query point $u \in \F_0$, 
	a compact representation of $\RE_m(u)$
	can be reported in $O(m)$ time. 
\end{theorem}
}


\subsection{Improved Query Time} \label{sec:improved}

In this section, we show how the query time
in the \fs map can be improved
by keeping some additional information in our data structure,
without increasing either the preprocessing time or space complexity.
This improved query time is crucial for 
applications such as the minimum walk problem.

We use our vertical ray shooting data structure from Section~\ref{sec:shooting}.
For each feasible interval $I$ on $\F_0$, 
we partition $I$ into $O(m)$ subintervals, such that
for all points $u$ in a subinterval, 
the first segment directly above $u$ in the ray shooting data structure is the same.
Such a partitioning can be obtained by a simple scan on each column of
the \fs map from bottom to the top.
The total number of subintervals obtained this way is $O(nm)$.
	
\begin{theorem} \label{thm:mainFreeSpaceMap}
	Given two polygonal curves $P$ and $Q$ of size $n$ and $m$, respectively,
	we can build in $O(nm)$ time a data structure of size $O(nm)$, such that
	for any query point $u \in \F_0$, 
	a compact representation of $\RE_m(u)$
	can be reported in $O(\log m)$ time.
	Furthermore, if the subinterval containing $u$ is given as part of the query,
	then a compact representation of $\RE_m(u)$ can be reported in $O(1)$ time.
\end{theorem}

\begin{proof}
	We first build the \fs map in $O(nm)$ time as per Lemma~\ref{lemma:map}.
	Let $I$ be a feasible interval on $\F_0$.
	For each $u \in I$, we have $\rp_m(u) = \rp_m(I) = \rp_m(\Right(I))$.
	Therefore, by storing $\rp_m(I)$ for all feasible intervals $I$ on $\F_0$,
	we can report $\rp_m(u)$ for each query point $u \in \F_0$ in $O(1)$ time.
	Since there are $O(n)$ feasible intervals on $\F_0$,
	and computing each right pointer takes $O(m)$ time by Lemma~\ref{lemma:query},
	this step takes $O(nm)$ time in total.
	To report $\lp_m(u)$ quickly, we store for each reachable interval $I \in \R(j)$,
	$0 < j < m$, the pointer $\lp_m(I)$ in the data structure.
	We can compute all these left pointers in $O(nm)$ time as follows.
	We first preprocess each column of the \fs map
	for vertical ray shooting as in Lemma~\ref{lemma:shooting},
	by assuming horizontal segments 
	to be non-reachable intervals on each row $\FD_j$. 
	To compute left pointers, we inductively process 
	the \fs map from top to bottom.
	Suppose that the left pointers are computed and stored for 
	all reachable intervals above $\FD_j$,
	and let $I$ be a reachable interval on $\FD_j$,
	with $q = \Left(I)$.
	We can find the first non-reachable segment $s$
	above $q$ using our ray shooting data structure in $O(1)$ time.
	If no such $s$ exists, $\lp_m(q)$ is directly above $q$ on $\R(m)$.
	Otherwise, as in Algorithm~\ref{alg:query},
	we project $q$ directly to a point $q' \in s$, 
	and then, find the first reachable point $p$ after $q'$.
	If such a point $p$ exists, it should be the left endpoint of a reachable interval $I'$,
	for which we have already stored the pointer $\lp_m(\Left(I'))$.
	Therefore, $\lp_m(q) = \lp_m(\Left(I'))$ can be computed 
	in $O(1)$ time. 
	As a result, finding all left pointers takes $O(n)$ time for each $\FD_j$,
	and $O(nm)$ time for the whole \fs map.

	Now, for each subinterval $J$ on $\F_0$,
	we compute $\lp_m(J)$ in the same way described above 
	in $O(1)$ time. Namely, we find the unique segment $s$ above $J$,
	find the first reachable point $p$ after $s$, 
	and take the pointer $\lp_m(p)$, which is stored in the data structure.
	The total time and space needed for this step is therefore $O(nm)$.
	For any query point $u \in \F_0$,
	we first locate the subinterval $J$ containing $u$ in $O(\log m)$ time.
	Now, $\lp_m(u) = \lp_m(J)$ and
	$\rp_m(u) = \rp_m(I)$ for the feasible interval $I$ containing subinterval $J$,
	both accessible in $O(1)$ time.
	Note that the only expensive operation in our query algorithm is
	to locate the subinterval containing the query point.
	If the subinterval is given, 
	then the query can be answered in $O(1)$ time.
	\qed
\end{proof}


\section{Applications} \label{sec:appl}

In this section, we provide some of the applications of our \fs map data structure.

\subsection{Partial Curve Matching\ }
Given two polygonal curves $P$ and $Q$, and an $\eps \gee 0$,
the partial curve matching problem involves
deciding whether there exists a subcurve $R \subseteq P$ such that $\distF(R,Q) \lee \eps$.
As noted in~\cite{AltG95}, 
this is equivalent to deciding whether 
there exists a monotone path in the free space from $\FD_0$ to $\FD_m$. 
This decision problem can be efficiently solved using the \fs map.
For each feasible interval $I$ on $\FD_0$,
we obtain a compact representation of $\R_m(\Left(I))$ using 
Theorem~\ref{thm:mainFreeSpaceMap} in $O(1)$ time.
Observe that $\R_m(I) = \emptyset$ if and only if 
$\R_m(\Left(I)) = \emptyset$.
Therefore, we can decide in $O(nm)$ time 
whether there exists a point on $\FD_m$ reachable from $\FD_0$.
Furthermore, we can use parametric search as in~\cite{AltG95}
to find the smallest $\eps$ for which the answer to the above decision problem is ``YES" in 
$O(nm \log (nm))$ time.
Therefore, we obtain:

\begin{theorem} \label{thm:partial}
	Given two polygonal curves $P$ and $Q$ of size $n$ and $m$, respectively,
	we can decide in $O(nm)$ time whether there exists a subcurve $R \subseteq P$ 
	such that $\distF(R,Q) \lee \eps$, for a given $\eps \gee 0$.
	A subcurve $R \subseteq P$ minimizing $\distF(R,Q)$ can be computed in $O(nm \log (nm))$ time.
\end{theorem}

\subsection{Closed Curves}
Given two closed curves $P$ and $Q$,
define
\[
	\distC(P,Q) = \inf_{s_1,s_2 \in \IR} \distF(P \mbox{ shifted by } s_1, Q \mbox{ shifted by } s_2)
\]
to be the closed \Frechet metric
between $P$ and $Q$.

Consider a diagram $\CD$ of size $2n \times m$
obtained from concatenating two copies of the standard \fs diagram of $P$ and $Q$.
Alt and Godau showed that $\distC(P,Q) \lee \eps$ if and only if
there exists a monotone feasible path in $\CD$ from $(t,0)$ to $(n+t,m)$,
for a value $t \in [0,n]$.
We show how such a value $t$, if any exists,
can be found efficiently using a \fs map built on top of $\CD$.

\begin{obs} \label{obs:t}
	Let $i$ be a fixed integer $(0 < i \lee n)$,
	$I_i=[a,b]$ be the feasible interval on the $i$-th cell of $\FD_0$, and 
	$J_i=[c,d]$ be the feasible interval on the $(i+n)$-th cell of $\FD_m$.
	Then there exists a value $t \in [i-1,i]$ with $(t,0) \reach (n+t,m)$ if and only if 
	$\max((\lp_m(I_i))_x, c) \lee b+n$ and $\min((\rp_m(I_i))_x, d) \gee a+n$.
\end{obs}

We  iterate on $i$ from 1 to $n$,
and check for each $i$ if a desired value $t \in [i-1,i]$ exists using Observation~\ref{obs:t}. 
Each iteration involves examining $\lp_m(I_i)$ and $\rp_m(I_i)$,
which are accessible in $O(1)$ time using Theorem~\ref{thm:mainFreeSpaceMap}.
The total time is therefore $O(nm)$, required for building the \fs map.

\begin{theorem} \label{thm:closed}
	Given two closed polygonal curves $P$ and $Q$ of size $n$ and $m$, respectively,
	we can decide in $O(nm)$ time whether $\distC(P,Q) \lee \eps$, for a given $\eps \gee 0$.
	Furthermore, $\distC(P,Q)$ can be computed in $O(nm \log (nm))$ time.
\end{theorem}

\subsection{Maximum Walk}

Another variant of the \Frechet distance problem is 
the following:
Given two curves $P$ and $Q$ and a fixed $\eps \gee 0$, 
find a maximum-length continuous subcurve of $Q$ whose 
\Frechet distance to $P$ does not exceed $\eps$.
In the dog-person illustration, this problem
corresponds to finding the best starting point on $P$, 
such that when the person walks the whole curve $Q$,
his or her dog can walk the maximum length on $P$,
without exceeding a leash of length $\eps$. 
We show that this optimization problem,
which is a generalized version of the partial curve matching problem,
can be solved efficiently in $O(nm)$ time
using the \fs map. 
The following observation is the main ingredient.

\begin{obs} \label{obs:max}
	Let $R$ be a maximum-length subcurve of $P$ such that $\distF(R,Q) \lee \eps$.
	The starting point of $R$ corresponds to the left endpoint of a feasible interval $I$ on $\FD_0$,
	and its ending point corresponds to $\rp_m(I)$. 
\end{obs}

By Observation~\ref{obs:max}, we only need to test $n$
feasible intervals on $\FD_0$, and their right pointer on $\FD_m$
to find the best subcurve $R$. 
If we keep the length of $P$ from its beginning to each of its $n$ segments in a table,
we can compute the length of each subcurve $R$ of $P$ 
in $O(1)$ time using two table lookups as 
it is explained in Chapter \ref{ch:speedFD}.
Computing the maximum-length subcurve $R$ will
therefore take $O(n)$ time for computing the lengths, 
plus $O(mn)$ time for constructing the \fs map.

\begin{theorem} \label{thm:max}
	Given two polygonal curves $P$ and $Q$ of size $n$ and $m$, respectively,
	and a parameter $\eps \gee 0$,
	we can find in $O(nm)$ time a maximum-length subcurve $R \subseteq P$ such that 
	$\distF(R,Q) \lee \eps$.
\end{theorem}

\subsection{Minimum Walk}

Given two curves $P$ and $Q$ and a fixed $\eps \gee 0$, 
the \emph{minimum walk} problem asks for the minimum-length continuous subcurve of
$P$ that a person can walk while his/her dog walks the whole curve $Q$
without exceeding a leash of length~$\eps$. 
This optimization problem can be again solved efficiently 
using our \fs map. 

\begin{theorem} \label{thm:min}
	Given two polygonal curves $P$ and $Q$ of size $n$ and $m$, respectively,
	and a parameter $\eps \gee 0$,
	we can find in $O(nm)$ time a minimum-length subcurve $R \subseteq P$ such that 
	$\distF(R,Q) \lee \eps$.
\end{theorem}

\begin{proof}
	Let $R$ be a minimum-length subcurve of $P$ such that $\distF(R,Q) \lee \eps$.
	Observe that the starting point of $R$ corresponds to 
	the right endpoint of a subinterval $J$ on $\F_0$,
	and its ending point corresponds to $\lp_m(J)$. 
	Therefore, to find the best subcurve $R$, 
	we only need to check the right endpoints of $O(nm)$
	subintervals on $\FD_0$ and their corresponding left pointers.
	By Theorem~\ref{thm:mainFreeSpaceMap}, this takes $O(1)$ time per subinterval.
	The total time needed is therefore $O(nm)$.
	\qed
\end{proof}


\section{Matching a Curve in a DAG} \label{sec:graph}

Let $P$ be a polygonal curve of size $n$, 
and $G$ be a connected geometric graph with $m$ straight-line edges.
Alt \etal~\cite{AltERW03a} presented an $O(nm \log m)$-time algorithm 
to decide whether
there is a path $\pi$ in $G$ with \Frechet distance at most $\eps$ to $P$, 
for a given $\eps \gee 0$.
In this section, 
we improve this result for the particular case when $G$ is a directed acyclic graph (DAG), 
by giving an algorithm that runs in only $O(nm)$ time.
The idea is to use a sequential reachability propagation approach 
similar to the one used in Section~\ref{sec:main}.
Our approach is structurally different from
the one used by Alt \etal~\cite{AltERW03a}.

We first borrow some notation from~\cite{AltERW03a}.
Let $G=(V,E)$ be a connected DAG with $m$ edges,
such that $V = \set{1, \ldots, \nu}$
corresponds to points $\set{v_1, \ldots, v_\nu} \subseteq \IR^d$, for $\nu \lee m+1$.
We assume, w.l.o.g., that the elements of $V$ are numbered
according to a topological ordering of the vertices of $G$.
Such a topological ordering can be computed in $O(m)$ time.
We embed each edge $(i,j) \in E$ as
an oriented line segment $s_{ij}$ from $v_i$ to $v_j$. 
Each $s_{ij}$ is continuously parametrized 
by values in $[0,1]$ according to its natural parametrization,
namely, $s_{ij}: [0,1] \rightarrow \IR^d$.

For each vertex $j \in V$, let $\FD_{j} := \FD_\eps(P, v_{j})$ 
be the one-dimensional \fs diagram corresponding 
to the path $P$ and the vertex $j$.
We denote by $L_j$ and $R_j$ the left endpoint and the right endpoint
of $\FD_j$, respectively.
Moreover, we denote by $\F_j$ the set of feasible points on $\FD_{j}$.
For each $(i,j) \in E$, let $\FD_{ij} := \FD_\eps(P, s_{ij})$
be a two-dimensional \fs diagram, which consists
of a row of $n$ cells.
We glue together the two-dimensional \fs diagrams
according to the adjacency information of $G$,
as shown in Figure~\ref{fig:graph}.
The resulting structure is called the \emph{\fs surface} of $P$ and $G$,
denoted by $\FS_\eps(P,G)$.
We denote the set of feasible points in $\FS_\eps(P,G)$ by $\F_\eps(P,G)$. 

\begin{figure}[t]
	\centering
	\includegraphics[width=0.5\columnwidth]{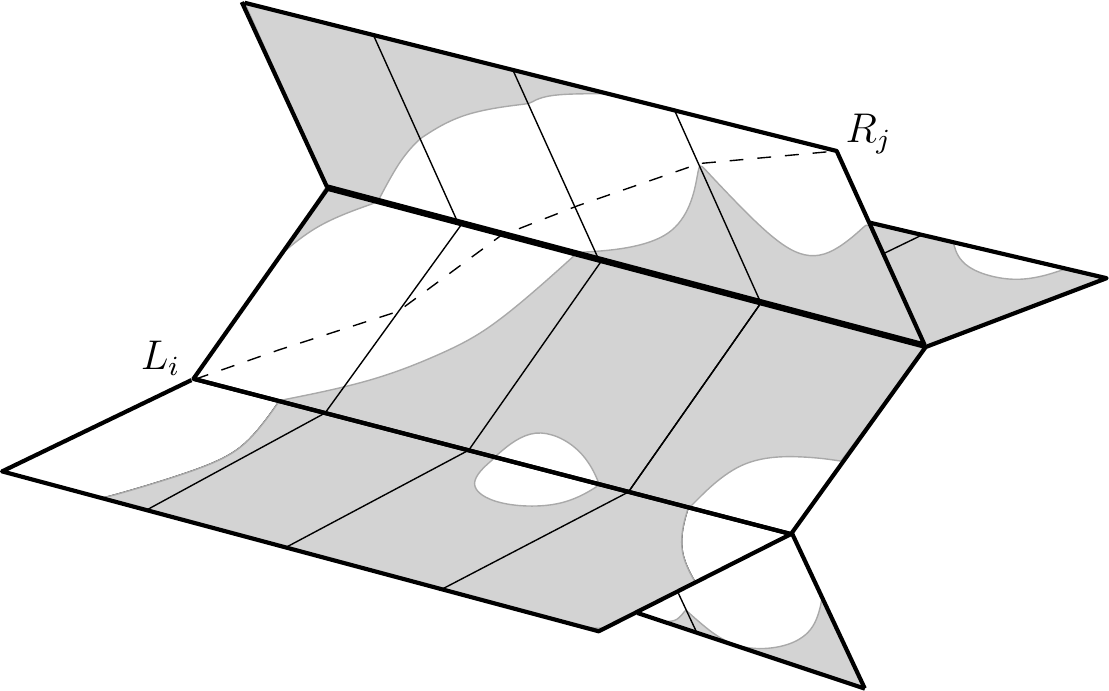}
	\caption{An example of a \fs surface.}
	\label{fig:graph}
\end{figure}

Given two points $u,v \in \F_\eps(P,G)$, 
we say that $v$ is \emph{reachable} from~$u$,
denoted by $u \reach v$,
if there is a monotone feasible curve from $u$ to~$v$ in $\F_\eps(P,G)$,
where monotonicity in each cell of the surface
is with respect to the orientation of the edges of $P$ and $G$
defining that cell.
Given a set of points $S \subseteq \F_\eps(P,G)$,
we define
$
	\RE_{j}(S) := \set{v \in \F_j \ | \ \exists u \in S \,\mbox{ s.t. }\,  u \reach v}.
$
Let $\CL = \cup_{j \in V} (L_j \cap \F_j)$.
For each $j \in V$, 
we define the \emph{reachable set} $\R(j) := \RE_{j}(\CL)$.
Observe that there is a path $\pi$ in $G$ with $\distF(P, \pi) \lee \eps$
if and only if there is a vertex $j \in V$ with $R_j \in \R(j)$.

\begin{algorithm} [h]
\caption {\sc DAG-Matching-Decision$(P, G, \eps)$} \label{alg:graph}
\algsetup{indent=1.5em}
\begin{algorithmic}[1]
	\vspace{0.5em}
	\baselineskip=1\baselineskip
	\FORALL {$j \in V$ in a topological order} 
		\STATE $\R(j) \eq \RE_j(L_j \cap \F_j) \cup (\cup_{\substack{(i,j) \in E}} \RE_j(\R(i)))$  \label{l:main}
	\ENDFOR 
	\STATE let $S = \cup_{\substack{j \in V}} (R_j \cap \R(j))$
	\STATE return {\sc true} if $S \not= \emptyset$, otherwise return {\sc false}
\end{algorithmic}
\end{algorithm}

\begin{theorem} \label{thm:graph}
	Given a polygonal curve $P$ of size $n$ and a directed acyclic graph $G$ of size $m$, 
	we can decide in $O(nm)$ time whether there is a path $\pi$ in $G$ 
	with $\distF(P, \pi) \lee \eps$, for a given $\eps \gee 0$.
	A path $\pi$ in $G$ minimizing $\distF(P, \pi)$ can be computed in $O(nm \log (nm))$ time.
\end{theorem}

\begin{proof}
	Algorithm~\ref{alg:graph} computes, for each vertex $j \in V$,
	the reachable set $\R(j)$ in a topological order.
	It then returns true only if there is a vertex $j \in V$ such that $R_j$ is reachable
	which indicates the existence of a path $\pi$ in $G$ with $\distF(P, \pi) \lee \eps$.
	To prove the correctness, we only need to show that 
	for every vertex $j \in V$, the algorithm computes $\R(j)$ correctly.
	We prove this by induction on $j$.
	Suppose by induction that the set $\R(i)$
	for all $i < j$ is computed correctly.
	Now consider a point $u \in \F_j$.
	If $u \in \R(j)$, then there exists a vertex $k < j$ such that
	$L_k$ is connected to $u$ by a monotone feasible curve $\CC$ in $\FS_\eps(P,G)$.
	If $k=j$, then $u \in \R(j)$ because
	$\RE_j(L_j \cap \F_j)$ is added to $\R(j)$ in line~\ref{l:main}.
	If $k<j$, then the curve $\CC$ must pass through a vertex $i$ with $(i,j) \in E$.
	Since 
	the vertices of $V$ are sorted in a topological order,
	we have $i < j$, and hence, $\R(i)$ is computed correctly by the induction hypothesis.
	Hence, letting $x = \CC \cap \F_i$, we have $x \in \R(i)$.
	Furthermore, we know that $x$ is connected to $u$ using the curve $\CC$.
	Therefore, the point $u$ is in $\RE_j(\R(i))$,
	and hence, is added to $\R(j)$ in Line~\ref{l:main}.
	Similarly, we can show that if $u \not\in \R(j)$,
	then $u$ is not added to $\R(j)$ by the algorithm.
	Suppose by contradiction that $u$ is added to $\R(j)$ in line~\ref{l:main}.
	Then either $u \in \RE_j(L_j \cap \F_j)$ or
	$u \in \RE_j(\R(i))$, for some $i < j$.
	But by the definition of reachability, both cases imply that 
	$u$ is reachable from a point in $\CL$,
	which is a contradiction.
	
	For the time complexity, 
	note that each $\RE_j(\R(i))$ in Line~\ref{l:main} can be computed in 
	$O(n)$ time using Lemma~\ref{lemma:propagate}. 
	Moreover, $\RE_j(L_j \cap \F_j)$, for each $j \in V$, can be computed
	by finding the largest feasible interval on $\F_j$ containing $L_j$
	in $O(n)$ time.
	Therefore, processing each edge $(i,j)$ takes $O(n)$ time,
	and hence, the whole computation takes $O(nm)$ time.
	Once the algorithm finds a reachable left endpoint $v$,
	we can construct a feasible monotone path connecting a right endpoint $u \in \CL$ to $v$
	by keeping, for each reachable interval $I$ on $R(j)$, 
	a back pointer to a reachable interval $J$ on $R(i)$, $(i, j) \in E$,
	from which $I$ is reachable. 
	The path $u \reach v$ can be constructed 
	by following the back pointers from $v$ to $u$, in $O(m)$ time.
	For the optimization problem, we use parametric search as in~\cite{AltERW03a,AltG95}
	to find the value of $\distF(P, \pi)$ by an extra $\log(nm)$-factor,
	namely, in $O(nm \log (nm))$ time.
	\qed
\end{proof}
Note that Algorithm~\ref{alg:graph} only works if the input graph is a DAG,
because it needs a topological ordering on the vertices in order to 
sequentially propagate reachability information.  
By the way, it is straight-forward 
to modify the algorithm to allow paths in $G$
to start and end anywhere inside edges of the graph, 
not necessarily at the vertices. 
This can be easily done by allowing 
the feasible path found by our algorithm
to start and end at any feasible point 
on the left and right boundary of $\FD_{ij}$,
for each edge $(i,j) \in E$.


\section{Conclusions} \label{sec:concl}

In this chapter, we presented improved algorithms
for several variants of the \Frechet distance problem.
Our improved results are based on a new data structure, called \fs map,
that might be applicable to other problems involving the \Frechet metric.
It remains open whether the same improvements obtained here
can be achieved for matching curves inside general graphs (see the next section where for 
complete graphs, we present some 
improvement).  
Proving a lower bound better than $\bigOmega(n \log n)$ 
is another major problem left open.

Preliminary results of this chapter 
are presented in the 
19th Annual European Symposium on Algorithms (ESA 2011)~\cite{oursESA2011}.
The full version of the paper is accepted 
for publication in Algorithmica~\cite{oursPartialAlgorithmica}.

\newcommand{\re}{\ell}

\chapter{Curve-Pointset Matching Problem (CPM) }
\label{ch:StayClose}

	Given a point set $S$ and a polygonal curve $P$ in $\IR^d$,
	we study the problem of finding a polygonal curve $Q$ whose vertices are from $S$
	and has minimum \Frechet distance to $P$. 
	Not all points in $\pset$ are required to be on 
	$Q$. Furthermore, a point in $\pset$ may be present multiple times on  $Q$. 
	We refer to this problem as  
	Curve-Pointset Matching (CPM) Problem.
	We present an efficient algorithm to solve the decision version 
	of this problem in $O(nk^2)$ time,
	where $n$ and $k$ represent the sizes of $P$ and $S$, respectively.
	Furthermore, if the answer to the decision problem is affirmative, 
	our algorithm can compute the curve with minimum number of segments 
	in $\eps$-~\Frechet distance to $P$.	
	In addition, we show that a curve minimizing the \Frechet distance can be computed in 
	$O(nk^2 \log(nk))$ time.
	As a by-product, we improve the map matching algorithm of Alt \etal\ 
	by an $O(\log k)$ factor for the case when the map is a complete graph.


\section{Introduction}

\REM{
Matching two geometric patterns 
is a fundamental problem in pattern recognition,
protein structure prediction, computer vision, geographic information systems, etc.
Usually these patterns consist of line segments and polygonal curves. 

One of the most popular ways to measure the similarity of two curves is
to use the \Frechet distance. 
An intuitive way to illustrate the \Frechet distance is as follows.
Imagine a person walking his/her dog, where the person and the dog, 
each travels a pre-specified curve, from beginning to the end, 
without ever letting go of the leash or backtracking.
The \Frechet distance between the two curves is the minimum length of a leash which is necessary.
The leash length determines how similar the two curves are to each other:
a short leash means the curves are similar,
and a long leash means that the curves are different from each other.

Two problem instances naturally arise:  decision and optimization.
In the {\em decision problem}, one wants to decide whether two polygonal curves $P$  and $Q$
are within~$\eps$ \Frechet distance to each other. 
In the {\em optimization problem}, one wishes to determine the minimum such~$\eps$.
Alt and Godau~\cite{AltG95} presented an $O(n^2)$-time algorithm for the decision problem,
where $n$ denotes the total number of segments in the curves.
They also solved the corresponding optimization problem in $O(n^2\log n)$ time.
}
In this chapter, we address the following variant of the \Frechet distance problem.
Given a point set $S$ and a polygonal curve $P$ in $\IR^d$ ($d \gee 2)$,
find a polygonal curve $Q$, with its vertices chosen from $S$, 
such that the \Frechet distance between $P$ and $Q$ is minimum.
Note that in our problem definition, 
not all points in $\pset$ need to be chosen as well as 
a point in $\pset$ can appear more than 
once as a vertex  in $Q$. 
In the decision version of the problem,
we want to decide if there is polygonal curve $Q$ through $S$
whose \Frechet distance to $P$ is at most $\eps$, for a given $\eps \gee 0$.
An instance of the decision problem is illustrated in Figure~\ref{fig:instance}.

One can use the map matching algorithm of Alt \etal~\cite{AltERW03a} (described in 
Section \ref{sec:RelatedMapMatching})
to solve the decision version of this problem 
by constructing a complete graph $G$ on top of $S$,
and then running Alt \etal's algorithm on $G$ and $P$.
If $n$ and $k$ represent the sizes of $P$ and $S$, respectively,
this leads to a running time of $O(nk^2 \log k)$ 
for solving the decision problem.

In this chapter, we present a simple algorithm to solve the decision version 
of the above problem in $O(nk^2)$ time.
This improves upon the algorithm of Alt \etal~\cite{AltERW03a}
by a $O(\log k)$ factor for the case when a curve is matched in a complete graph.
Our approach is different from and simpler
than the approach taken by Alt \etal\ which is a mixture of line sweep, dynamic
programming, and Dijkstra's algorithm.

\begin{figure}[t]
	\centering
	\includegraphics[height=0.5\columnwidth]{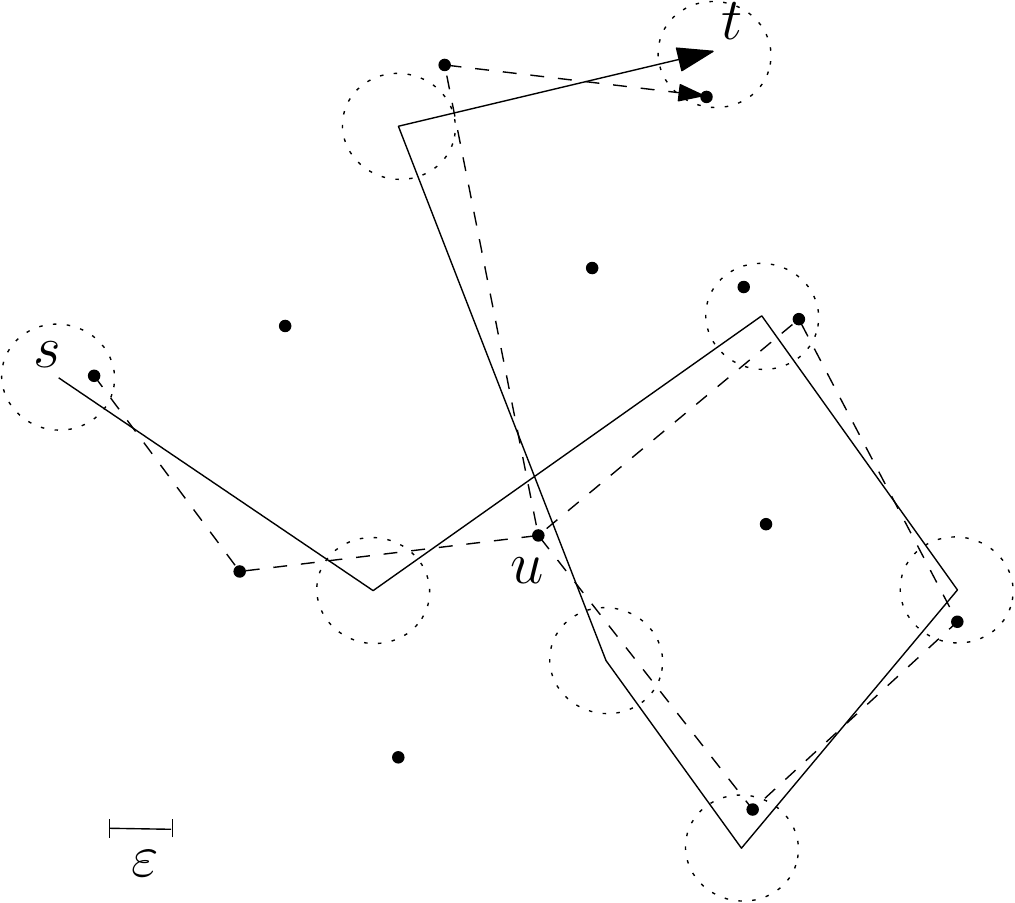}
	\caption{A problem instance. The dashed curve is in $\eps$-\Frechet distance to the solid curve. Point $u$ is used multiple times in the dashed curve.}
	\label{fig:instance}
\end{figure}


\section{Preliminaries}
\label{sec:preliminariesCPM}

Let $\eps \gee 0$ be a real number, and $d \gee 2$ be a fixed integer.
For any point $p \in \IR^d$,
we define $\CB(p,\eps) \equiv \{q \in \IR^d : \|pq\| \lee \eps\}$
to be a \emph{ball} of radius $\eps$ centered at $p$,
where $\|\cdot\|$ denotes the Euclidean distance.
Given a line segment $L \subset \IR^d$,
we define $\CC(L, \eps) \equiv \cup_{p\in L} \CB(p,\eps)$
to be a \emph{cylinder} of radius~$\eps$ around $L$
(see Figure~\ref{fig:cylinder}).

A curve  in $\IR^d$ can be represented as  a continuous function 
$P:[0,1] \rightarrow \IR^d$.
Given two points $u,v \in P$,  
we write $u \lei v$, if $u$ is located before $v$ on $P$.
The relation~$\lex$ is defined analogously.
For a subcurve $R \subseteq P$,
we denote by $\Left(R)$ and $\Right(R)$
the first and the last point of $R$ along $P$, respectively.

\begin{figure}[h]
	\centering
	\includegraphics[width=0.6\columnwidth]{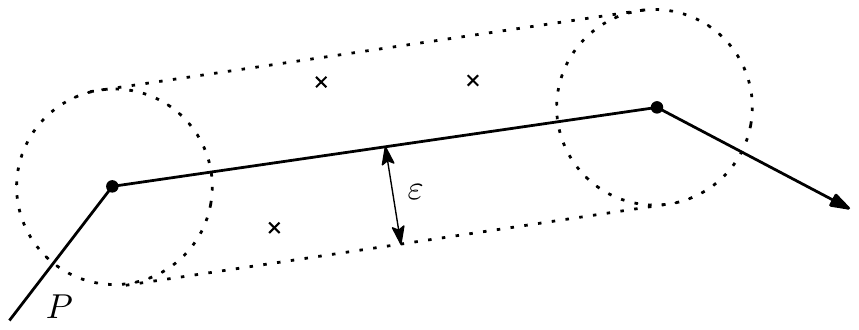}
	\caption{A cylinder of radius $\eps$ around segment $L$.}
	\label{fig:cylinder}
\end{figure}

Given two curves $\alpha, \beta: [0,1] \rightarrow \IR^d$,
the {\em \Frechet distance\/} between $\alpha$ and $\beta$ is defined as
$
	\distF(\alpha,\beta) = \inf_{\sigma, \tau} \max_{t \in [0,1]} \| \alpha(\sigma(t)), \beta(\tau(t)) \|,
$
where $\sigma, \tau: [0,1] \rightarrow [0,1]$ 
range over all continuous non-decreasing surjective functions.
The following two observations are immediate.

\begin{obs}\label{obs:simple}
	Given four points $a, b,c,d \in \IR^d$, if
	$\| ab\| \lee \eps$ and $\| cd\| \lee \eps$, then
	$\distF(\Dir{ac},\Dir{bd}) \lee \eps$. 
\end{obs}

\begin{obs}\label{obs:concat}
	Let $\alpha_1$, $\alpha_2$, $\beta_1$, and $\beta_2$ 
	be four curves 
	such that $\distF(\alpha_1,\beta_1) \lee \eps$ and
	$\distF(\alpha_2,\beta_2) \lee \eps$. 
	If the ending point of $\alpha_1$ (resp., $\beta_1$), 
	is the same as 
	the starting point of $\alpha_2$  (resp., $\beta_2$),
	then $\distF(\alpha_1 + \alpha_2, \beta_1 + \beta_2) \lee \eps$,
	where $+$ denotes the concatenation of two curves.
\end{obs}


\section{The Decision Algorithm}

Let $P$ be a polygonal curve composed of $n$ line segments $P_1, \ldots, P_n$,
and let $S$ be a set of $k$ points in $\IR^d$.
In this section, we provide an algorithm to
decide whether there exists a polygonal curve $Q$ whose vertices are chosen from $S$,
such that $\distF(P,Q) \lee \eps$,
for a given $\eps \gee 0$. 

We denote by $s$ and $t$ the starting and the ending point of $P$, respectively.
For each segment $P_i$ of $P$,
we denote by $C_i$ the cylinder $\CC(P_i, \eps)$,
and by $S_i$ the set $S \cap C_i$.
Furthermore, for each point $u \in C_i$,
we denote by $P_i[u]$ the line segment $P_i \cap \CB(u,\eps)$.

We call a polygonal curve $Q$ \emph{feasible} if 
all its vertices are from $S$, and
$\distF(Q, P') \lee \eps$ for a subcurve $P' \subseteq P$ starting at $s$.
If $Q$ ends at a point $v \in S$ and $P'$ ends at a point $p \in P$, 
we call the pair $(v,p)$ a \emph{feasible pair}.
A point $v \in S_i$ is called \emph{reachable} (at cylinder $C_i$)
if there is a feasible curve ending at $v$ in $C_i$.

Consider a feasible curve $Q$ starting at a point $u \in S_1$ and
ending at a point $v \in S_i$.
Since no backtracking is allowed in the definition of \Frechet distance,
$Q$ traverses all cylinders $C_1$ to $C_i$ in order, 
until it reaches $v$.
Moreover, by our definition of reachability, 
each vertex of $Q$ 
is reachable at some cylinder $C_j$, $1 \lee j \lee i$.

Our approach for solving the decision problem is to process the cylinders
one by one from $C_1$ to $C_n$, and identify at each cylinder $C_i$
all points of $S$ which are reachable at $C_i$.
The decision problem will be then reduced (by Observation~\ref{obs:concat}) 
to checking whether there is a reachable point in the ball $\CB(t,\eps)$.

To propagate the reachability information through the cylinders, 
we need a primitive operation described below.
Let $u \in S_i$ be a point reachable at cylinder $C_i$, 
and let $Q$ be a feasible curve ending at $u$.
For each point $v \in S$, 
we denote by $\ri_i(u,v)$ the index of the furthest cylinder 
we can reach by the curve $Q+\Dir{uv}$.
In other words,
$\ri_i(u,v)$ is the largest index $\ell \gee i$ such that 
$v \in S_\ell$ is reachable via $u \in S_i$.
If $Q+\Dir{uv}$ is not feasible, we set $\ri_i(u,v) = 0$.
The following lemma 
is a direct corollary of a similar one proved in~\cite{AltERW03a}  (Lemma~3)
for computing the so-called right pointers.

\begin{lemma}[\cite{AltERW03a}] \label{lemma:linear}
	Given two points $u,v \in S$,
	we can compute $\ri_i(u,v)$ for all $1 \lee i \lee n$ in $O(n)$ total time.
\end{lemma}

We use the following lemma for our algorithm.

\begin{lemma} \label{lemma:cross}
	Let $\ri_i(u,v) = \ell$.
	For all $i \lee j \lee \ell$,  if $v \in S_j$, then $v$ is reachable at $C_j$.
\end{lemma}

\begin{proof}
	Let $Q$ be a feasible curve starting at a point $w \in S \cap \CB(s,\eps)$ and ending at $u$,
	and let $Q' = Q+\Dir{uv}$.
	Since $v$ is reachable at $C_\ell$ via $Q'$,
	there is a subcurve $P'$ of $P$ starting at $s$ and ending at a point $p \in P_\ell[v]$
	(see Figure~\ref{fig:cross}).
	Consider two point objects $\CO_P$ and $\CO_Q$ 
	traversing $P'$ and $Q'$, respectively, from beginning to end,
	while keeping $\eps$ distance to each other.
	Since $v$ is reachable via $u \in S_i$,
	$\CO_P$ is at a point $a \in P_i$ when $\CO_Q$ is at $u$.
	Fix a cylinder $C_j$,  $i < j \lee \ell$, such that  $v \in C_j$.
	When $\CO_P$ reaches the point $b=\Left(P_j[v])$, 
	$\CO_Q$ is at a point $x \in \Seg{uv}$ such that $\|bx\| \lee \eps$.
	The subcurve of $Q'$ from $w$ to $x$ has \Frechet distance at most $\eps$
	to the subcurve of $P$ from $s$ to $b$,
	and the segment $\Seg{xv}$ has \Frechet distance at most $\eps$ to the point $b$
	by Observation~\ref{obs:simple}.
	Therefore, by Observation~\ref{obs:concat}, 
	the whole curve $Q'$ has \Frechet distance at most $\eps$
	to the subcurve $P'$ from $s$ to $b$,
	meaning that $v$ is reachable at $C_j$.
\end{proof}

\begin{figure}[t]
	\centering
	\includegraphics[width=0.6\columnwidth]{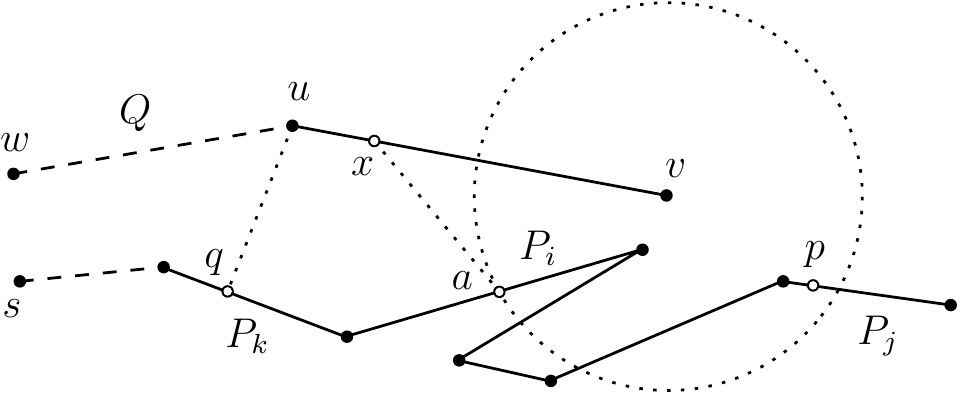}
	\vspace{0.5em}
	\caption{Proof of Lemma~\ref{lemma:cross}}
	\label{fig:cross}
\end{figure}

The above proof, not only shows that 
$v$ is reachable at $C_j$, 
but also that 
the pair $(v, \Left(P_j[v]))$ is feasible.
The following lemma is therefore immediate.

\begin{lemma} \label{lemma:left-point}
	If $\ri_i(u,v) = \ell$ and $v \in S_j$, $i < j \lee \ell$,
	then $(v, \Left(P_j[v]))$ is a feasible pair.
\end{lemma}


\begin{algorithm} [h]
\caption {{\sc Decision}$(S, P, \eps)$} 
\label{alg:decCCCG11}
\algsetup{indent=1.5em}
\begin{algorithmic}[1]
	\vspace{0.5em}
	\baselineskip=1\baselineskip
	\STATE {\bf Initialize:}  \label{l:init}
		\STATE \hspace{1.5em}compute $\ri_i(u,v)$ for all $u,v \in S$ and $1 \lee i \lee n$ \label{l:init-1} 
		\STATE \hspace{1.5em}set $\re_{v} = 0$ for all $v \in S$ \label{l:init-2} 
		\STATE \hspace{1.5em}let $\R_0 = S \cap \CB(s,\eps)$ \label{l:R0} \label{l:init-3} 
		\STATE \hspace{1.5em}set $\re_v = 1$ for all $v \in \R_0$  \label{l:init-4} 
	\FOR {$i = 1$ to $n$}   \label{l:mainstart}
		\STATE let $\R_i = \set{v \in S_{i} : \re_{v}\gee i}$  \label{l:direct}
		\STATE let $q = \min_{v \in \R_i}{\Left(P_i[v])}$  \label{l:min}
		\FORALL {$v \in S_i \setminus \R_i$} \label{l:loop1start}
			\IF {$q \lex \Right(P_i[v])$} \label{l:cond}
				\STATE add $v$ to $\R_{i}$ \label{l:reach}
			\ENDIF
		\ENDFOR \label{l:loop1end}
		\FORALL {$(u,v) \in \R_{i} \times S$}  \label{l:rel-1}
				\STATE $\re_{v} \eq \max\set{\re_v, \ri_i(u,v)}$  \label{l:rel-2}
		\ENDFOR
	\ENDFOR 
	\STATE {\bf return} {\sc yes} if $\R_n \cap \CB(t,\eps) \not= \emptyset$ 
	\label{l:final}

\end{algorithmic}
\end{algorithm}

\paragraph{The Algorithm}
Our algorithm for solving the decision problem is provided in Algorithm~\ref{alg:decCCCG11}.
It maintains,
for each cylinder $C_i$, 
a set $\R_i$ of all points in $S_i$ which are reachable at $C_i$.
To handle the base case more easily,
we assume, w.l.o.g., that the curve $P$ starts with a segment $P_0$  
consisting of a single point $\set{s}$.
Every point of $S$ inside the cylinder $C_0 = \CB(s,\eps)$ is reachable by definition.
Therefore, we initially set $\R_0 = S \cap \CB(s,\eps)$ (in line~\ref{l:R0}).

For each point $v \in S$, the algorithm maintains an index $\re_v$, 
whose value at the beginning of each iteration $i$ is the following:
$\re_v = \max_{0 \lee j < i, u \in \R_j} \ri_j(u,v)$.
In other words, $\re_v$ 
points to the largest index $\ell$ for which $v$ is reachable at $C_\ell$
via a reachable point $u$ in some earlier cylinder $C_j$, $j < i$.
Initially, we set $\re_v  = 1$ for all points in $\R_0$,
because all points in $\R_0$ are also reachable in $C_1$, as $C_0 \subseteq C_1$.
For all other points, $\re_v$ is set to 0 in the initialization step.
The following invariant holds during the execution of the algorithm.

\begin{lemma}\label{lemma:main}
	After the $i$-th iteration of Algorithm~\ref{alg:decCCCG11}, 
	the set $\R_{i}$ consists of all points in $S_i$ which are reachable at cylinder $C_i$.
\end{lemma}

\begin{proof}
	We prove the lemma by induction on $i$.
	The base case $i=0$ trivially holds.
	Suppose by induction that, for each $0 \lee j < i$,
	the set $\R_i$ is computed correctly.
	In the $i$-th iteration,
	we first add to $R_i$ (in Line~\ref{l:direct}) all points in $S_i$
	which are reachable through a point in a set $\R_j$, for $1 \lee j < i$.
	We call these points \emph{entry points} of cylinder $C_i$.
	We then add to $\R_i$ in lines~\ref{l:min}--\ref{l:loop1end} all points in $S_i$
	which are reachable through the entry points of $C_i$
	(see Figure~\ref{fig:reach} for an example).

	We first show that all points added to $R_i$ are reachable at $C_i$.
	For each point $v \in S_i$ added to $R_i$ in Line~\ref{l:direct},
	we have $\re_v \gee i$. 
	It means that there is a point $u \in R_j$, for some $j < i$, such that $\ri_j(u,v) \gee i$. 
	Therefore, Lemma~\ref{lemma:cross} implies that $v$ is reachable at $C_i$.
	Now, consider a point $v$ added to $\R_i$ in line~\ref{l:reach}.
	According to the condition in line~\ref{l:cond}, 
	there is an entry point $w$ in $C_i$ such that
	$\Left(P_i[w]) \lex \Right(P_i[v])$.
	By Observation~\ref{obs:simple}, the segment $\Dir{wv}$ is 
	within $\eps$ \Frechet distance to the line segment
	from $\Left(P_i[w])$ to $\Right(P_i[v])$. 
	Moreover, by Lemma~\ref{lemma:left-point}, $(w, \Left(P_i[w])$ is a feasible pair.
	Therefore, by Observation~\ref{obs:concat},
	$v$ is reachable.
	
	Next, we show that
	any reachable point at $C_i$ is added to $\R_i$ by the algorithm.
	Suppose that there is 
	a point $v \in S_i$ which is reachable at $C_i$,
	but is not added to $\R_i$.
	Let $Q$ be a feasible curve ending at $v$,
	and $w$ be the first point on $Q$ which is reachable at $C_i$.
	By our definition, $w$ is an entry point of $C_i$.
	If $w = v$, then $v$ must be added to $R_i$ in Line~\ref{l:direct},
	which is a contradiction.
	If $w$ is before $v$ on $Q$,
	then we have $\Left(P_i[w]) \lex \Right(P_i[v])$.
	Now, by our selection of $q$ in Line~\ref{l:min},
	we have $q \lex \Left(P_i[w]) \lex \Right(P_i[v])$,
	and hence, $v$ is added to $R_i$ in line~\ref{l:reach},
	which is again a contradiction.
\end{proof}

\begin{figure}[t]
	\centering
	\includegraphics[width=0.9\columnwidth]{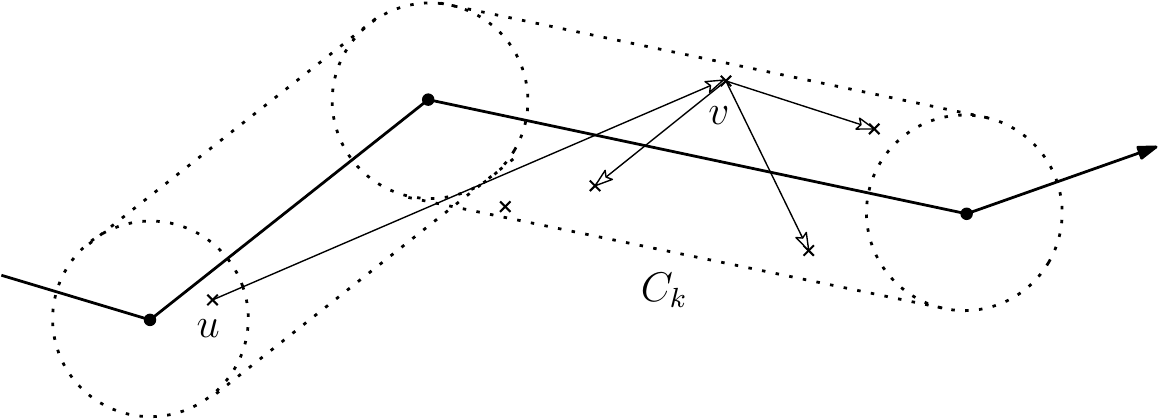}
	\caption{Point $v$ is an entry point of $C_i$.}
	\label{fig:reach}
\end{figure}

\begin{theorem} \label{thm:main}
	Given a polygonal curve $P$ of $n$ segments and a set $S$ of $k$ points in $\IR^d$, 
	we can decide in $O(nk^2)$ time whether there is a polygonal curve $Q$ through $S$ 
	such that $\distF(P, Q) \lee \eps$, for a given $\eps \gee 0$.
	A polygonal curve $Q$ through $S$ of size $O(\min\set{n,k})$ minimizing $\distF(P, Q)$ 
	can be computed in $O(nk^2 \log (nk))$ time.
\end{theorem}

\begin{proof}
	The correctness of the decision algorithm (Algorithm~\ref{alg:decCCCG11})
	directly follows from Lemma~\ref{lemma:main}.
	Line~\ref{l:init-1} of the algorithm takes $O(nk^2)$ time by Lemma~\ref{lemma:linear}.
	The other three lines in the initialization step (lines~\ref{l:init-2}--\ref{l:init-4}) take only $O(k)$ time.
	In the main loop, lines~\ref{l:direct}--\ref{l:reach} take $O(k)$ time, and
	lines~\ref{l:rel-1}--\ref{l:rel-2} require $O(k^2)$ time.
	Therefore, the whole loop takes $O(nk^2)$ time in total.

	Once the algorithm finds a reachable point $v \in S_n \cap \CB(t,\eps)$,
	we can construct a feasible curve $Q$ ending at $v$ by keeping, 
	for each reachable point $u$ at a cylinder $C_i$, 
	a back-pointer to a reachable point $w$ at $C_j$, $j \lee i$, 
	from which $u$ is reachable. 
	The feasible curve $Q$ can be then constructed 
	by following the back pointers from $v$ to a point in $S_1 \cap \CB(s,\eps)$. 
	Since at most two points from each cylinder are selected in this process,
	the curve $Q$ has $O(\min\set{n,k})$ segments.
	For the optimization problem, we use parametric search as in~\cite{AltERW03a,AltG95}, 
	to find a curve minimizing $\distF(P, Q)$ by an extra $\log(nk)$-factor
	in $O(nk^2 \log (nk))$ time.
\end{proof}


\section{Conclusions}

In this chapter, we presented a simple efficient algorithm 
for finding a polygonal curve through a given point set $S$ in $\IR^d$
such that its \Frechet distance to a given polygonal curve $P$ is minimized.
Several interesting problems remain open.
For a fixed $\eps$, one can easily modify the algorithm provided
here to find a curve with a minimum number of segments, 
having \Frechet distance at most $\eps$ to $P$.
It can be done by keeping reachable points in a priority queue,
and propagating the reachability information in a Dijkstra-like manner.
However, we cannot see any easy adaptation of our algorithm to 
find a curve passing through a maximum number of points for a fixed $\eps$.

The algorithm presented in this chapter improves
the map matching algorithm of Alt \etal~\cite{AltERW03a} 
for the case of matching a curve in a complete graph.
The current lower bound available for the problem is $\bigOmega((n+k) \log (n+k))$
due to Buchin \etal~\cite{LowerBound-FD}.
It is therefore open whether a better algorithm is available,
or whether the algorithm obtained in this chapter is optimal.

Results of this chapter 
are presented in 23rd Canadian Conference on 
Computational Geometry~\cite{oursCCCG2011}.


\chapter{All-Points CPM Problem is NP-complete}
\label{ch:NP-Complete}

\section{Introduction}
In this chapter, we study 
a variant of the problem discussed in the previous chapter. 
We refer to this variant as the All-Points CPM problem. We address the following:
Consider a pointset $S \subseteq \IR^d$ and a polygonal curve $P$ in $\IR^d$, 
for $d \gee 2$ being a fixed dimension.
The objective is to decide whether there exists a polygonal curve $Q$ in  $\eps$-\Frechet
distance to $P$ such that the vertices of $Q$ are all chosen from the 
pointset $\pset$. Moreover, curve $Q$ has to visit every point of $\pset$
and it can visit a point multiple times. 
We  prove  that  this problem  is  NP-complete  by  reducing from 3CNF-SAT problem.
In an independent work~\cite{NPComplete-Pointset} (which is done after   
our NP-completeness result), the authors have shown that
the version of this problem where points in $\pset$ has to 
be visited only once, is NP-complete too. 
Their proof is obtained via reduction
from a restricted version of the 3SAT problem, 
called (3,B2)-SAT problem,
where the input to formulas is restricted
in which each literal occurs exactly twice.
In \cite{DiscretelyFollowing},  
Wylie and Zhu studied 
All-points CPM problem from the 
perspective of discrete \Frechet distance
and they showed that it is solvable in 
$O(nk)$ time ($n$ is the size of curve $P$ and $k$ is the size of  pointset $\pset$). 
Furthermore, they showed that the version of the problem in which 
each point of $\pset$ can only used once in $Q$ is  NP-complete.

\section{General Case is NP-complete}
\label{sec:NPComp}

\subsection{Preliminaries}
\noindent{\bf Notation.}
We denote by $P = \langle p_1p_2p_3...p_n \rangle$, a polygonal curve $P$
with vertices $p_1 p_2 \dots p_n$ in order 
and by $start(P)$ and $end(P)$, we denote 
the starting and ending point of $P$, respectively.
For a curve $P$ and a point $x$, by $P \ap x$, 
we mean connecting $end(P)$ to point $x$.
We use the same notation $P \ap Q$ 
to show the concatenation of 
two curves $P$ and $Q$ (which means connecting $end(P)$ to $start(Q)$).
Let $M(\Seg{ab})$ denote the  midpoint of line segment $\Seg{ab}$. 
For a point $q$ in the plane, let $x(q)$ and $y(q)$
denote the $x$ and $y$ coordinate of $q$, respectively.

For two intersecting line segments $\Seg{ab}$ and $\Seg{cd}$, let $ \Seg{ab} \dashv	\Seg{cd}$ denote the intersection point of them.
Let $\overleftrightarrow{bc}$ denote the line as a result of 
extending line segment $\Seg{bc}$.
For a point $p$ and a line segment $\Seg{bc}$,
let $p \perp \Seg{bc}$ denote the point on
line $\overleftrightarrow{bc}$,
located on the perpendicular from $p$ to $\overleftrightarrow{bc}$.

\begin{definition} \label{def:curves}
Given a pointset $\pset$ in the plane, let $Curves(\pset)$
be a set of polygonal curves $Q = \langle q_1 q_2 \dots q_n \rangle$ where: 
$$  \forall{q_i} : q_i \in \pset \mbox{  and  } 
  \forall{a} \in \pset: \exists{q_i} \mbox{  s.t. }  q_i = a. $$
\end{definition}

\begin{definition} \label{def:feasibleNPC}
Given a pointset $\pset$, a polygonal curve $P$ and a distance $\eps$, 
a polygonal curve $Q$ is called {\em feasible} if: 
$Q \in Curves(S)$ and $  \distF(P,Q) \le \eps$.
\end{definition}

We show that the problem of deciding whether a 
feasible curve exists or not  is NP-complete.
It is easy to see that this problem is in NP, since 
one can polynomially check whether $Q \in Curves(S)$
and also $\distF(P,Q) \le \eps$, using the algorithm in \cite{AltG95} 
(explained in Section \ref{sec:classicalFD}).

%

\subsection{Reduction Algorithm}

We reduce in Algorithm~\ref{alg:reduction},
an instance of 3CNF-SAT formula $\phi$ 
to an instance of our problem.
The input is a boolean formula 
$\phi$ with $k$ clauses $C_1, C_2, \dots, C_k$ and $n$ variables $x_1, x_2, \dots,x_n$ 
and the output is a pointset $\pset$,
a polygonal curve $P$ in the plane and 
a distance $\eps = 1$.

We construct the pointset $S$ as follows.
For each clause $C_j$, $1 \le j \le k$, in the formula $\phi$, 
we place three points $\{s_j, g_j,c_j\}$ in the plane, which are computed
in the $j$-th iteration of Algorithm \ref{alg:reduction} (from Line \ref{l:makeSLoop} to Line \ref{l:EndLoopPointSet}).
We define $o_j$ to be $M(\Seg{\sma_j\gre_j})$.
By $\sq_j$, $1\le j \le k$, we denote
a square in the plane, centered at $o_j$, 
with diagonal $\Seg{\sma_j \gre_j}$. 
We refer to $\sq_j$, $1\le j \le k$, as
{\em c-squares}. 
For an example of a pointset $S$ corresponding to a formula, 
see Figure \ref{fig:pathAExample}a.

Our reduction algorithm constructs the polygonal curve $P$ 
 through $n$ iterations. In  the $i$-th iteration, $1 \le i \le n$,
it builds a subcurve $\cfev_i$ corresponding to  a variable 
$x_i$ in the formula $\phi$ and appends that curve to $P$.
In addition to those $n$ subcurves, two curves 
$\cfev_{n+1}$ and $\cfev_{n+2}$ are appended to $P$. 
We will later discus the reason we add those 
two curves. Every subcurve  $\cfev_i$ of $P$ starts at point $u$ 
and ends at point $v$.
Furthermore, each  $\cfev_i$ goes through
$\sq_1$ to $\sq_k$ in order, enters each c-square $\sq_j$ from
the side $\Seg{c_j s_j}$  and exists that square from 
the side $\Seg{c_j g_j}$  (for an illustration, see Figure \ref{fig:pathAExample}a). 
Curve $\cfev_i$ itself is built incrementally  
through iterations of the loop at line \ref{l:looptoMakeL} 
of Algorithm \ref{alg:reduction}. 
In the $j$-th iteration, when $\cfev_i$ goes through $\sq_j$,
three points, which are within  $\sq_j$, are added to  $\cfev_i$
(these three points are computed through Lines 
\ref{l:makeclausestart} to \ref{l:makeclauseend}).
Next, before $\cfev_i$ reaches  $\sq_{j+1}$,
two points,  denoted by $\alpha_j$ and $\beta_j$, are added to that curve 
(these two points are computed in Lines \ref{l:alpha} 
and \ref{l:beta}).

Each $\cfev_i$ corresponds to variable $x_i$ in our approach.
We simulate $1$ or $0$ values of $x_i$
as follows.
Consider a point object $\CO_L$ 
traversing $\cfev_i$, from starting point $u$ to ending point $v$. 
Consider 
another point object $\CO_2$ which wants to 
walk from $u$ to $v$
on a path whose vertices are from points in $S$ and it wants to stay in distance one 
to $\CO_L$. We will show that 
by our construction, object $\CO_2$ has two options, either taking 
the path $A = \langle us_1g_2s_3 \dots v \rangle$ or the path $B = \langle ug_1s_2g_3 \dots v \rangle$ 
(See Figure \ref{fig:pathAExample}a and \ref{fig:pathAExample}b for an illustration). 
Choosing path $A$ by $\CO_2$ means $x_i = 1$ and choosing path $B$ means $x_i = 0$.
We first prove in Lemma \ref{lemma:PathA} that $\distF(\cfev_i,A) \le 1 $ and 
in Lemma  \ref{lemma:PathB} that $\distF(\cfev_i,B) \le 1 $.
Furthermore, in
Lemma \ref{lemma:NoSwitchFromAtoB}, we prove that as soon as $\CO_2$ chooses 
path $A$ at point $u$ to walk towards $v$, 
it can not switch to any vertex on path $B$.
Analogously, we show that as soon as $\CO_2$ chooses path $B$ at point $u$ to walk towards $v$, 
it can not switch to any vertex on path $A$.
In addition, in Lemmas \ref{lemma:ABCanSeeC} and \ref{lemma:NOTABCanSeeC}, we prove that 
if $x_i$ appears in clause $C_j$,
$\CO_2$ could visit point $c_j$ via the path $A$ and not $B$. In contrast, 
when $\neg x_i$ appears in the clause $C_j$,
$\CO_2$ could visit point $c_j$ via the path $B$ and not $A$.
However, when none of  $x_i$ or $\neg x_i$ appear 
in $C_j$, $\CO_2$ can take neither $A$ nor $B$ to visit $c_j$.
Thus, $c_j$ can be visited, 
if and only if there is 
an $i$ such that 
either $x_i$ or 
$\neg x_i$ are in clause $C_j$.

\begin{algorithm} 
\caption {{\sc Reduction Algorithm}} 
\label{alg:reduction}
\algsetup{indent=1.5em}
\begin{algorithmic}[1]
	\baselineskip=0.5\baselineskip
	\REQUIRE  3SAT formula $\phi$ with $k$ clauses $C_1 \dots C_k$ and $n$ variables $x_1 \dots x_n$

	\vspace{0.1in}
		
	\hspace{-0.2in} {\bf Construct pointset $S$:}  

	\STATE $\pset \leftarrow \emptyset$ \label{l:init}


	\STATE $\gre_1 = (1,1) $ \label{l:makeSStart}

	\FOR {$j = 1$ to $k$}   \label{l:makeSLoop}

	\STATE $\sma_j \leftarrow \big(x(\gre_j)-2,y(\gre_j)-2\big) $
		 \STATE $o_j   \eq  M(\Seg{\sma_j\gre_j})$

		\IF {($j$ is odd)	}

	\STATE $c_j \leftarrow \big(x(s_j), y(\gre_j) \big)$, $w_j  \eq \big(x(o_j)+\frac{1}{4},  y(o_j)-\frac{1}{4} \big)$

	\STATE  $\gre_{j+1} \eq \big (   x(s_{j}) + \frac{1}{4} + 8, y(s_{j}) + \frac{7}{4} +15 \big)$     		\label{l:ComputeNextEven}
    
		\ELSE
	\STATE $c_j \leftarrow \big(x(g_j), y(s_j) \big)$, $w_j  \eq \big(x(o_j)-\frac{1}{4},  y(o_j)+\frac{1}{4} \big)$ 
    
	\STATE  $\gre_{j+1} \eq \big ( x(s_{j}) + \frac{7}{4} + 15, y(s_{j}) + \frac{1}{4} + 8 \big)$     		\label{l:ComputeNextOdd}


	\ENDIF


\STATE $z_j = M(\Seg{c_jw_j})$

    \STATE $S = S \cup \{\sma_j,\gre_j, c_j\}$   \label{l:EndLoopPointSet}

	\ENDFOR

	\IF {($k$ is odd)}  \label{l:ComputeV}
	  \STATE $\eta \eq \big( x(o_k)+1, y(o_k)+4\big)$  
      \STATE $v \eq \big( x(o_k)+1, y(o_k)+9\big)$  
		\ELSE
	   \STATE $\eta \eq \big( x(o_k)+4, y(o_k)+1\big)$  
      \STATE $v \eq \big( x(o_k)+9, y(o_k)+1\big)$  
	\ENDIF

	\STATE $u = (-9,-1)$

	\STATE $t \eq \big(x(v),y(u) -20\big)$

    \STATE $\pset = \pset \cup \{u, v, t\}$ \label{l:makeSEnd}

\vspace{0.15in}

	\hspace{-0.25in} {\bf Construct polygonal curve $P$:}

	\STATE $P \eq \emptyset$ \label{l:makeP}

	\STATE $P \eq P \ap t$

	\FOR { $i = 1$ to $n+2$  }   \label{l:mainstart}

		\STATE  $\cfev_i \eq \emptyset$ \label{l:startofL}

		\STATE  $\cfev_{i} \eq \cfev_{i} \ap u$ 
		\STATE $\cfev_{i} \eq \cfev_{i} \ap (-4,-1) $  \label{l:Adduh1toell}
		\FOR {$j = 1$ to $k$} \label{l:looptoMakeL}

			\IF { ($x_i \in C_j$ and $j$ is odd ) or ($\neg x_i \in C_j$ and $j$ is even ) } \label{l:makeclausestart}	
			\STATE $\cfev_{i} \eq \cfev_{i} \ap M(\Seg{s_jc_j}) \ap c_j \ap w_j  $
			\ELSIF {{ ($\neg x_i \in C_j$ and $j$ is odd ) or ($ x_i \in C_j$ and $j$ is even ) }}
			\STATE $\cfev_{i} \eq \cfev_{i} \ap w_j \ap c_j \ap M(\Seg{g_jc_j} ) $
			\ELSE		
			\STATE $\cfev_i \eq  \cfev_i \ap w_j \ap c_j \ap w_j$
			\ENDIF \label{l:makeclauseend}

			\IF {$j \neq k$}

			\STATE $\alpha_j =  \frac{4}{5} g_j + \frac{1}{5} g_{j+1}$ \label{l:alpha}

			\STATE $\beta_j = \frac{1}{5} s_j + \frac{4}{5} s_{j+1} $ \label{l:beta}
			
			\STATE $\cfev_i \eq \cfev_i \ap \alpha_j \ap \beta_j$

			\ENDIF

		\ENDFOR
	
		\STATE $\cfev_i \eq \cfev_i \ap \eta \ap v$ \label{l:subcurve}

		\STATE  $P \leftarrow P \ap \cfev_i$ 
		\STATE  $P \eq P \ap t$

	\ENDFOR

	\vspace{0.05in}
\RETURN  pointset $\pset$, polygonal curve $P$ and distance $\eps = 1$

\end{algorithmic}
\end{algorithm}

\begin{table}[h]
\centering
\begin{tabular}{ r | l | l  }
if $x_i \in C_1$   & location of $\CO_A$ & location of $\CO_L$  
 \\
\hline
    
&  $u$ & $u$  \\
&  				 $h_1$ s.t.  $\| h_1\mu_1 \| \le \eps$ & $\mu_1 = (-4,-1)$\\

& $s_1 $  & $M(\Seg{s_1c_1})$	      \\

\hline
if $\neg x_i \in C_1$ 
&  $u$ & $u$  \\
&  				 $h_1$ s.t.  $\| h_1\mu_1 \| \le \eps$ & $\mu_1 = (-4,-1)$\\

& $s_1$ &  $ \Dir{\mu_1w_1} \dashv	 \Seg{s_1c_1} $\\

\hline
if $x_i \notin C_{1} \& \neg x_i \notin C_{1}$ & $u$ & $u$  \\
&  				 $h_1$ s.t.  $\| h_1\mu_1 \| \le \eps$ & $\mu_1 = (-4,-1)$\\
& $s_1$ &  $ \Dir{\mu_1w_1} \dashv	 \Seg{s_1c_1} $\\

\end{tabular}
\vspace{0.2 in}
\caption{Proof of Lemma \ref{lemma:PathA}, the base case of induction}
\label{tab:BaseCasePathA}
\end{table}

\begin{lemma}\label{lemma:PathA}
Consider any subcurve $\cfev_i$, $1\le i \le n+2$,  
which is built through Lines \ref{l:mainstart} to \ref{l:subcurve} 
of Algorithm \ref{alg:reduction}. Let $A$ be the polygonal curve  $\langle u\sma_1\gre_2\sma_3\gre_4..v\rangle$. 
Then, $\distF(\cfev_i,A) \le 1$.
\end{lemma}

\begin{proof}

We prove the lemma by induction on the number of segments along $A$. 
Consider two point objects $\CO_L$ and $\CO_A$ 
traversing $\cfev_i$ and $A$, respectively (Figure \ref{fig:pathAExample}a depicts an instance of $\cfev_i$ and $A$).
We show that $\CO_L$ and $\CO_A$ can walk
their respective curve, from the beginning to
 end, while keeping distance $1$ to each other. 

The base case of induction trivially holds as follows 
(see Figure \ref{fig:PathAClause1} for an illustration).
Table \ref{tab:BaseCasePathA} lists  pairwise locations of 
$\CO_L$ and $\CO_A$, where the distance of each pair is at most $1$.
Hence, $\CO_A$ can walk from $u$ to $s_1$ on the 
first segment of $A$ (segment $\Dir{us_1}$), 
while keeping distance $\le 1$ to $\CO_L$.

Assume inductively that $\CO_L$ and $\CO_A$ have feasibly walked along 
their respective curves, until $\CO_A$ reached $s_j$.
Then, as the induction step, 
we 
show that
$\CO_A$ can walk to $g_{j+1}$ and then to $s_{j+2}$, while keeping distance $1$ to $\CO_L$.
Table \ref{tab:PathA} lists pairwise locations
of $\CO_A$ and $\CO_L$ such that $\CO_A$ could reach  $s_{j+2}$.
One can easily check that the distance between the pair of points 
in that table is at most one.
 (For an illustration, see Figure \ref{fig:PathA}).

\begin{table}[t]
\centering
\begin{tabular}{ r | l | l  }
  & location of $\CO_A$ & location of $\CO_L$  
 \\
\hline
   if $x_i \in C_j$  & $s_j$ & $M(\Seg{c_js_j})$\\
	& $z_j$ & $c_j$\\ 
	&  & $w_j$\\ 

	& $\Seg{c_jg_j} \dashv \Seg{s_jg_{j+1}} $ & $\Seg{w_{j}\alpha_{j}} \dashv \Seg{c_{j}g_{j}}$ \\

   if $\neg x_i \in C_j$  & $s_j$ & $\Seg{\beta_{j-1}w_j}	\dashv \Seg{c_js_j}$ \\
	& $w_j \perp \Seg{s_jg_{j+1}}$ &$w_j$\\
	& $z_j$ &$z_j$\\
	& &$c_j$\\
	& $\Seg{c_jg_j} \dashv \Seg{s_jg_{j+1}} $ &$ M(\Seg{c_jg_j}) $\\

   if $x_i \notin C_j \& \neg x_i \notin C_j$  & $s_j$ & $\Seg{\beta_{j-1}w_j}	\dashv \Seg{c_js_j}$\\
	&$w_j \perp \Seg{s_jg_{j+1}}$  & $w_j$\\
	&$z_j$  & $z_j$\\

	& &$c_j$\\
	& &$w_j$\\
& $\Seg{c_jg_j} \dashv \Seg{s_jg_{j+1}} $ & $\Seg{w_{j}\alpha_{j}} \dashv \Seg{c_{j}g_{j}}$ \\

\hline
	&  $h_1$ s.t.  $\| h_1\alpha_j \| \le \eps$ & $\alpha_j$\\
	&  	$h_2$ s.t.  $\| h_2 \beta_j \| \le \eps$ & $\beta_j$\\

\hline
if $x_i \in C_{j+1}$ &  $\Seg{s_{j+1}c_{j+1}}	\dashv \Seg{s_{j}g_{j+1}}$    & $\Seg{\beta_jw_{j+1}}	\dashv \Seg{c_{j+1}s_{j+1}}$\\
& $z_{j+1}$ & $w_{j+1}$\\
&  & $z_{j+1}$\\
&  & $c_{j+1}$ \\
&  $g_{j+1}$ & $M(\Seg{c_{j+1}g_{j+1}})$\\

if $\neg x_i \in C_{j+1}$ &  			$\Seg{s_{j+1}c_{j+1}}	\dashv \Seg{s_{j}g_{j+1}}$  & $M(\Seg{s_{j+1}c_{j+1}})$\\
 & $z_{j+1}$ & $c_{j+1}$ \\
 &  & $w_{j+1}$ \\
 &  $g_{j+1}$ & $\Seg{g_{j+1}c_{j+1}}	\dashv \Seg{w_{j+1}\alpha_{j+1}}$ \\

if $x_i \notin C_{j+1} \& \neg x_i \notin C_{j+1}$ &  	$\Seg{s_{j+1}c_{j+1}}	\dashv \Seg{s_{j}g_{j+1}}$  & $\Seg{\beta_jw_{j+1}}	\dashv \Seg{c_{j+1}s_{j+1}}$\\
& $z_{j+1}$ & $w_{j+1}$\\
 & & $c_{j+1}$ \\
 & & $w_{j+1}$ \\
&  $g_{j+1}$ & $\Seg{g_{j+1}c_{j+1}}	\dashv \Seg{w_{j+1}\alpha_{j+1}}$ \\

\hline

	&  $h_3$ s.t.  $\| h_3\alpha_{j+1} \| \le \eps$ & $\alpha_{j+1}$\\
	&  	$h_4$ s.t.  $\| h_4 \beta_{j+1} \| \le \eps$ & $\beta_{j+1}$\\
\hline

 if $\neg x_i \in C_{j+2}$  & $s_{j+2}$ & $\Dir{\alpha_{j+1}w_{j+2}} 	\dashv \Seg{c_{j+2}s_{j+2}}$ \\
   if $ x_i \in C_{j+2}$  & $s_{j+2}$ & $M(\Seg{c_{j+2}s_{j+2}})$\\
   if $x_i \notin C_{j+2} \& \neg x_i \notin C_{j+2}$  & $s_{j+2}$ & $\Dir{\alpha_{j+1}w_{j+2}}	\dashv \Seg{c_{j+2}s_{j+2}}$\\

\end{tabular}
\caption{Distance between pair of points is less or equal to one}
\label{tab:PathA}
\end{table}

Finally, if $k$ is an odd number, then  
$\Dir{s_kv}$ is the last segment along $B$, otherwise, 
$\Dir{g_kv}$ is the last one. In 
either case, 
that edge crosses  the circle $\CB(\eta,1)$, where $\eta$ is the last vertex of 
$\cfev_i$ before $v$ (point $\eta$ is computed in line \ref{l:ComputeV} of 
Algorithm \ref{alg:reduction}). Therefore, 
 $\CO_A$ can walk to $v$, while keeping distance $1$ to $\CO_L$.

\qed
\end{proof}

\begin{figure}[t]
	\centering
	\includegraphics[width=1\columnwidth]{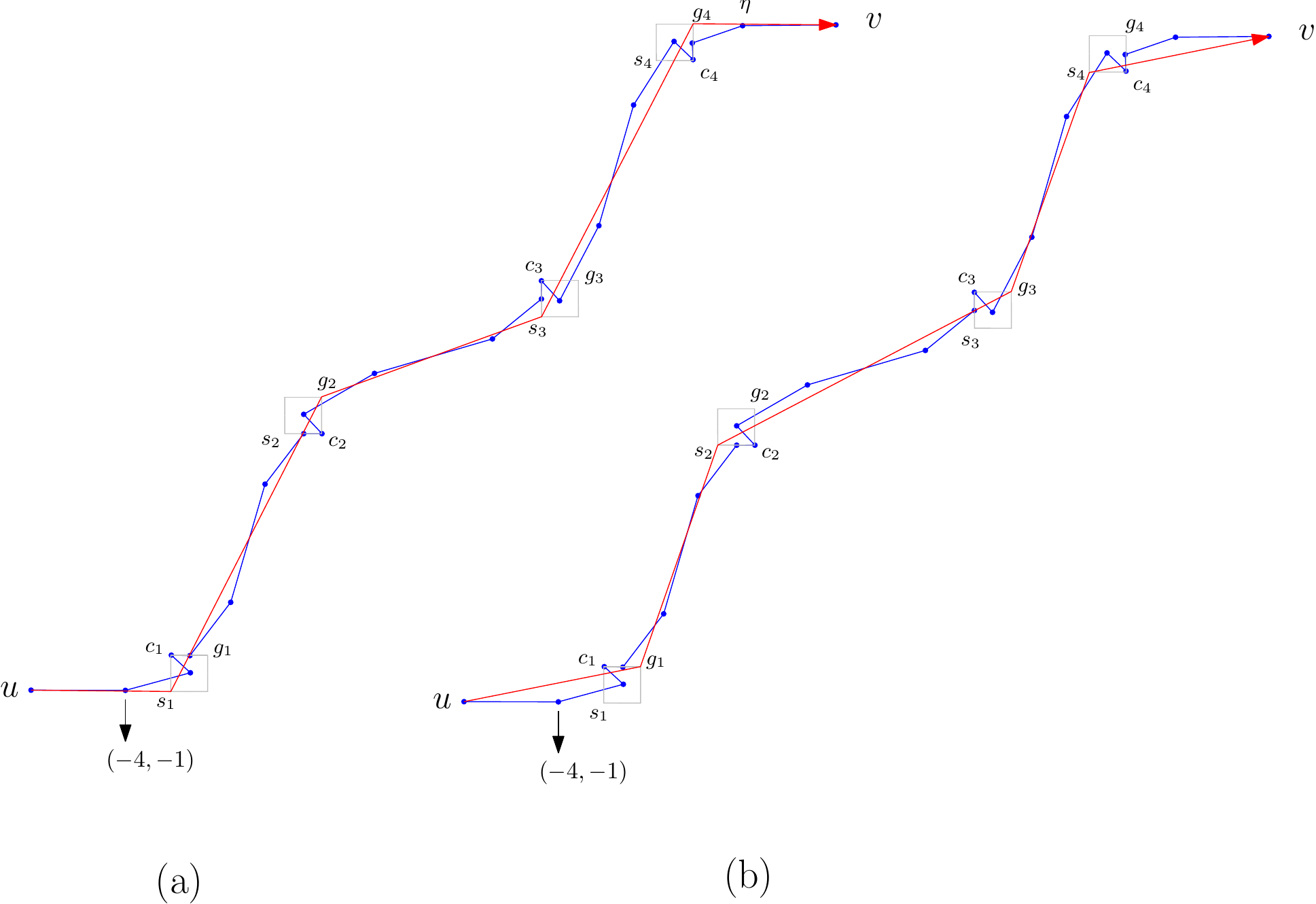}
	\caption{
Blue curve is an example of curve $\cfev_i$ which corresponds to variable $x_i$ in  formula $\phi$. The formula has four clauses $C_1, C_2, C_3$ and $C_4$,	 where the occurrence of variable $x_i$  in those clauses is:
$\neg x_i \in C_1$, $ \neg x_i \in C_2$, $x_i \in C_3$ and $x_i \in C_4$.  
For each clause $C_i$, the reduction algorithm places three point $s_i,g_i$ and $c_i$ in the plane. (a) Red curve is curve $A$. 
 (b) Red curve is curve $B$.  }
	\label{fig:pathAExample}
\end{figure}

\begin{figure}[h]
	\centering
	\includegraphics[width=0.9\columnwidth]{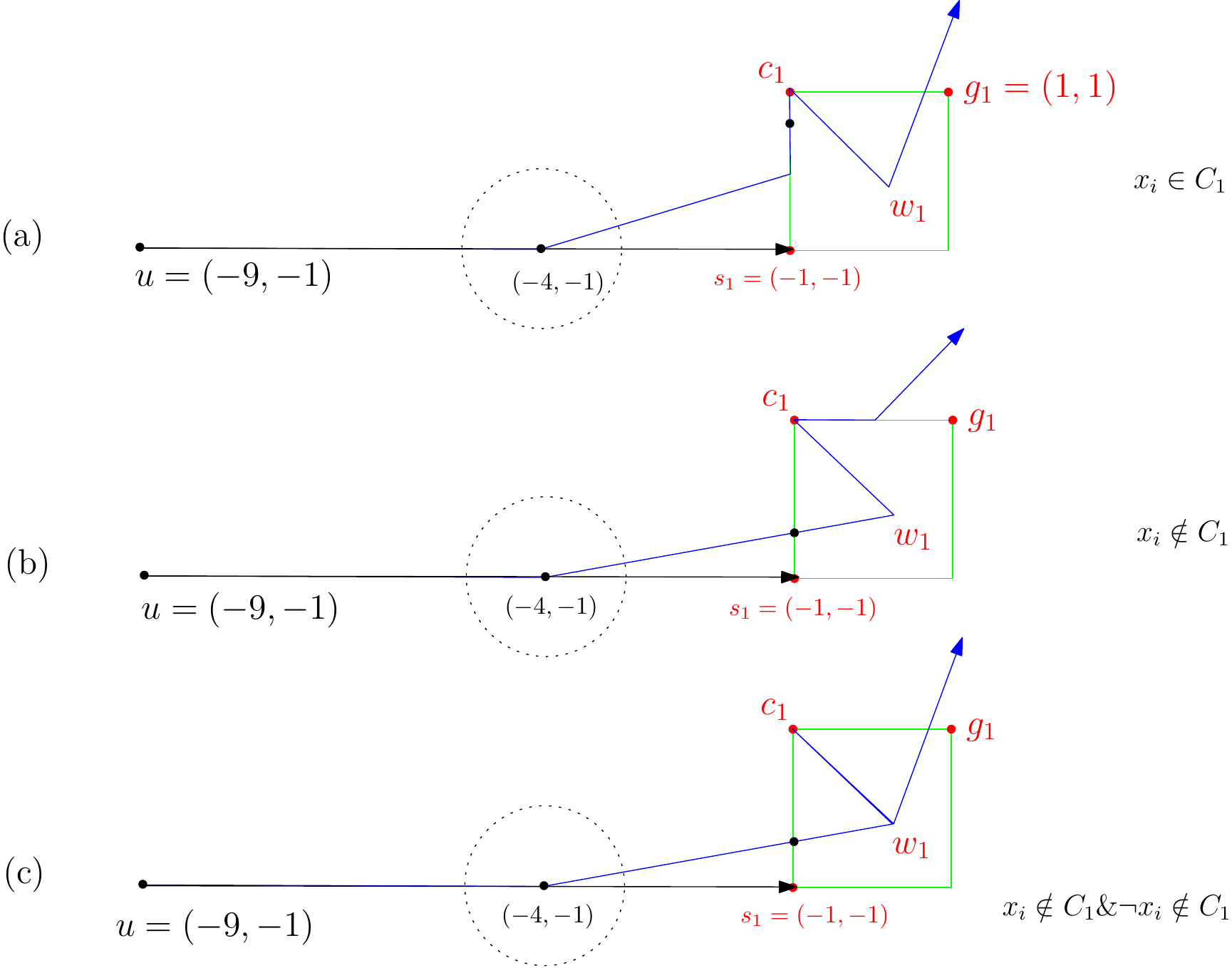}
	\caption{Base case of induction in the proof of Lemma \ref{lemma:PathA}}
	\label{fig:PathAClause1}
\end{figure}

\begin{figure}
	\centering
	\includegraphics[width=0.9\columnwidth]{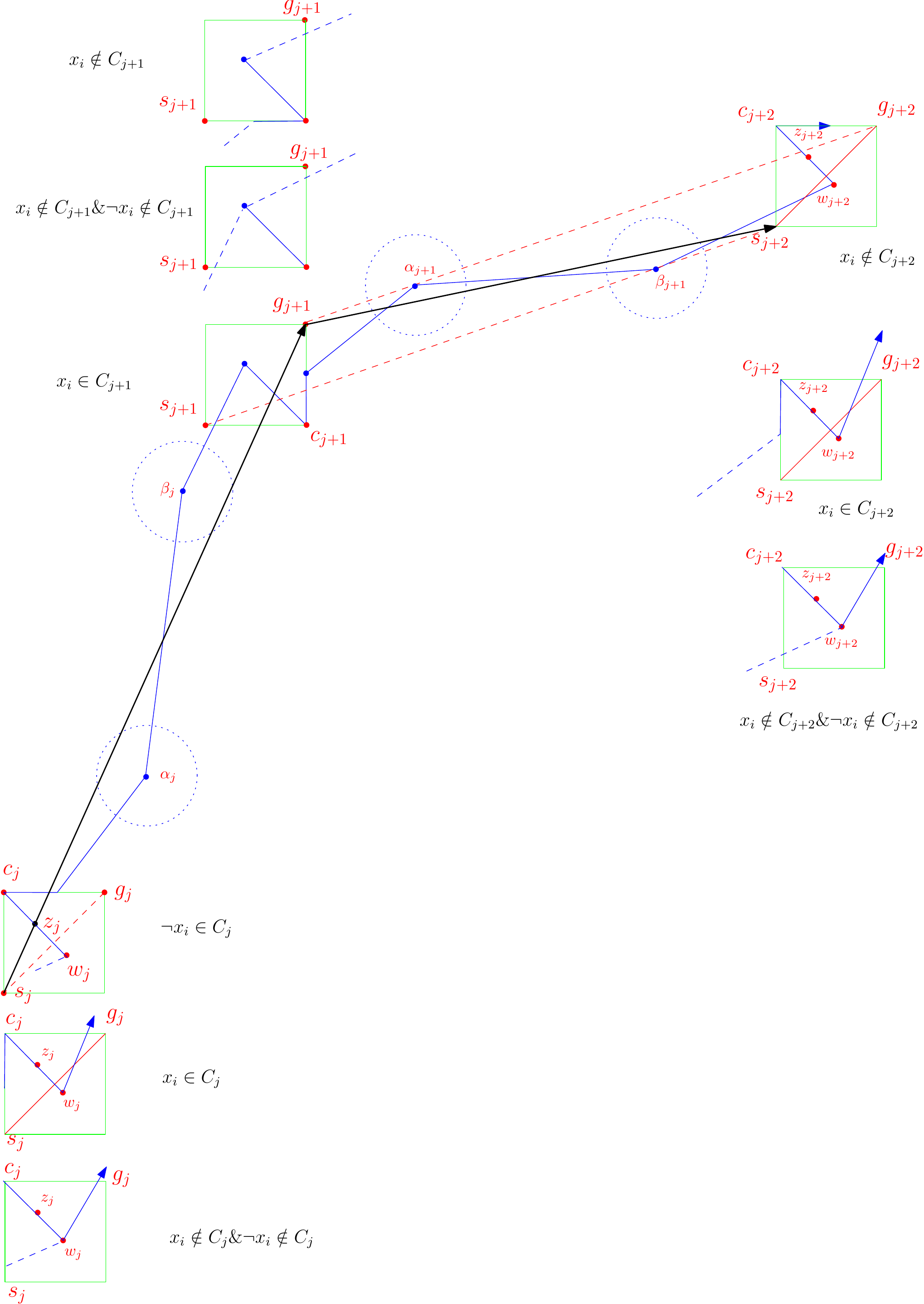}
	\caption{Proof of Lemma \ref{lemma:PathA}}
	\label{fig:PathA}
\end{figure}

\begin{figure}[h]
	\centering
	\includegraphics[width=0.9\columnwidth]{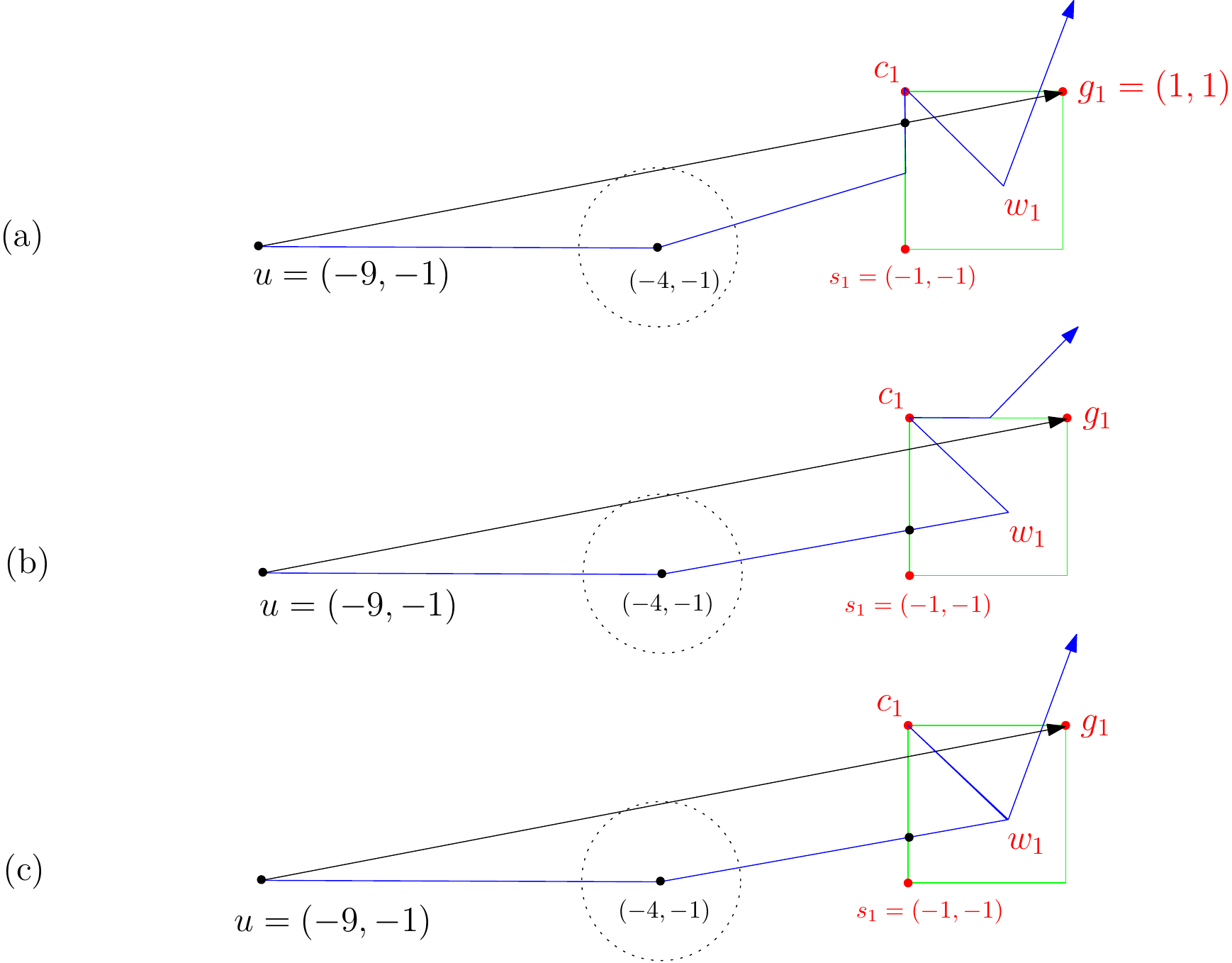}
	\caption{Base case of induction in the proof of Lemma \ref{lemma:PathB} }
	\label{fig:PathBBaseCase}
\end{figure}

\begin{figure}[h]

	\centering
	\includegraphics[width=0.9\columnwidth]{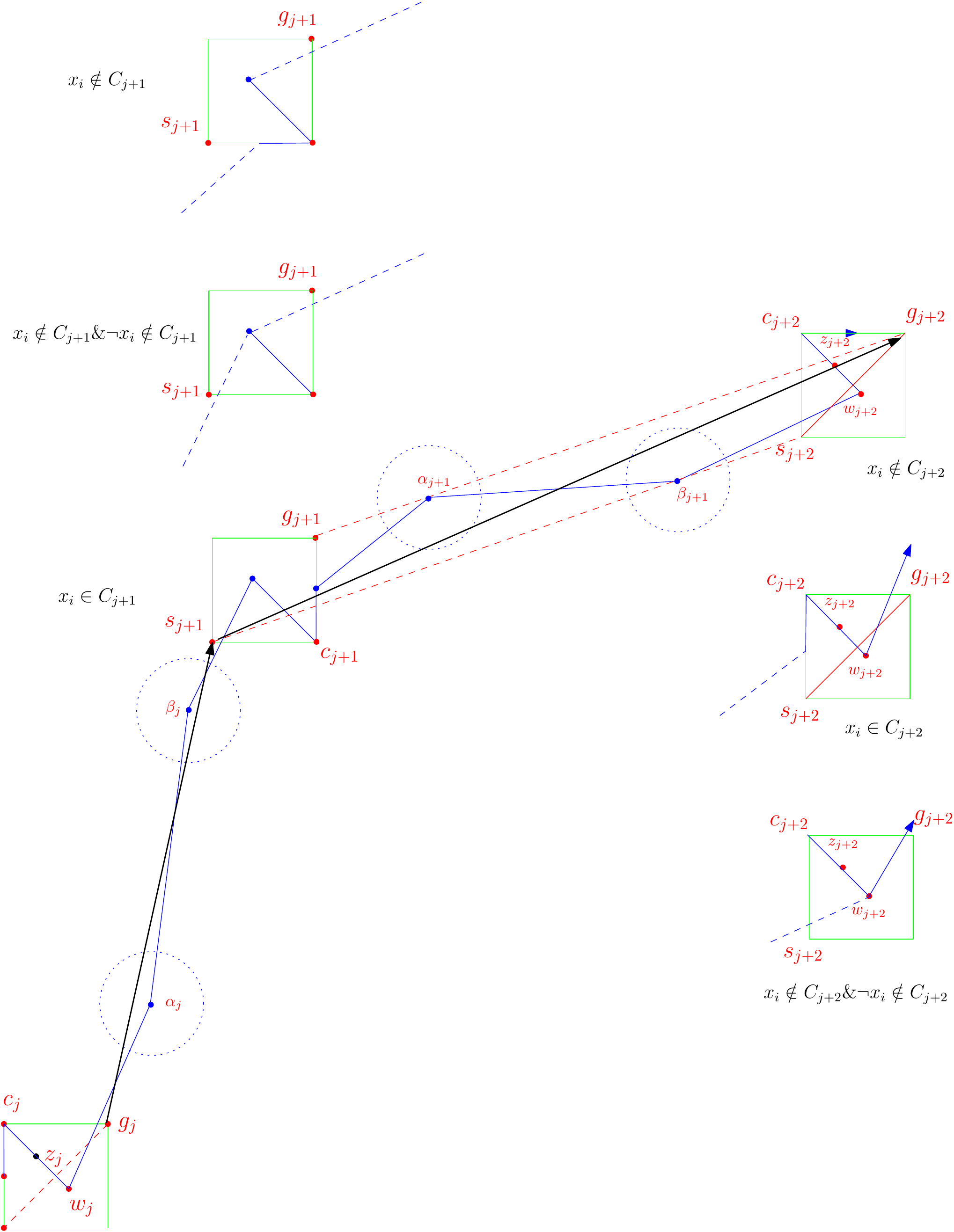}
	\caption{Proof of Lemma \ref{lemma:PathB}}
	\label{fig:PathB}
\end{figure}

\begin{lemma}\label{lemma:PathB}
Consider any subcurve $\cfev_i$, $1\le i \le n+2$,  
constructed through Lines \ref{l:mainstart} to \ref{l:subcurve} 
of Algorithm \ref{alg:reduction}. Let $B$ be the polygonal curve  $\langle u\gre_1\sma_2\gre_3\sma_4..v\rangle$. Then, $\distF(\cfev_i,B) \le 1$.
\end{lemma}

\begin{proof}

Consider two point objects $\CO_L$ and $\CO_B$ 
traversing $\cfev_i$ and $B$, respectively (Figure \ref{fig:pathAExample}b depicts an instance of $\cfev_i$ and $B$).
To prove the lemma, we show that $\CO_L$ and $\CO_B$ can walk along
their respective curves, from beginning to
the end, while keeping distance $1$ to each other. 

The base case of induction holds as follows 
(see Figure \ref{fig:PathBBaseCase} for an illustration).
Table \ref{tab:BaseCasePathB} lists  pairwise locations of 
$\CO_L$ and $\CO_B$, where the distance of each pair is less or equal to $1$.
Therefore, $\CO_B$ can walk from $u$ to $g_1$ while
keep distance one to $\CO_L$.

\begin{table}[h]
\centering
\begin{tabular}{ r | l | l  }
if $x_i \in C_1$   & location of $\CO_B$ & location of $\CO_L$  
 \\
\hline
    
&  $u$ & $u$  \\
&  				 $h_1$ s.t.  $\| h_1\mu_1 \| \le \eps$ & $\mu_1 = (-4,-1)$\\

				&    $h_2 = \Dir{ug_1} \dashv	 \Seg{s_1c_1} $ & $\mu_2 = M(s_1c_1)$   \\
&  				 			      $h_2 $ & $c_1$  \\
&    $\Seg{ug_1} \dashv \Seg{c_1w_1}$ &  $\Seg{ug_1} \dashv \Seg{c_1w_1}$ \\
&  				 		     $w_1 \perp ug_1 $ &   $w_1$  \\
& $g_1 $  & $\Seg{w_1\alpha_1}	\dashv \Seg{c_1g_1}$	      \\

\hline
if $\neg x_i \in C_1$ 
&  $u$ & $u$  \\
&  				 $h_1$ s.t.  $\| h_1\mu \| \le \eps$ & $\mu_1 = (-4,-1)$\\
& $h_2 = \Dir{ug_1} \dashv	 \Seg{s_1c_1} $ &  $\mu_2 = \Dir{\mu_1w_1} \dashv	 \Seg{s_1c_1} $\\
& $\Seg{ug_1} \dashv \Seg{c_1w_1}$ & $w_1$ \\
&  &  $c_1$\\
& $g_1$ &  $M(c_1g_1)$\\

\hline
if $x_i \notin C_{1} \& \neg x_i \notin C_{1}$ & $u$ & $u$  \\
&  				 $h_1$ s.t.  $\| h_1\mu \| \le \eps$ & $\mu_1 = (-4,-1)$\\
& $h_2 = \Dir{ug_1} \dashv	 \Seg{s_1c_1} $ &  $\mu_2 = \Dir{\mu_1w_1} \dashv	 \Seg{s_1c_1} $\\
& $\Seg{ug_1} \dashv \Seg{c_1w_1}$ & $w_1$ \\
&  &  $c_1$\\
&  &  $w_1$\\
& $g_1 $  & $\Seg{w_1\alpha_1}	\dashv \Seg{c_1g_1}$	      \\

\end{tabular}
\vspace{0.2 in}
\caption{Pairwise location of $\CO_B$ and $\CO_L$, to prove the base case of induction in Lemma \ref{lemma:PathB} }
\label{tab:BaseCasePathB}
\end{table}

Assume inductively that $\CO_L$ and $\CO_B$ have feasibly walked along 
their respective curves, until $\CO_B$ reached $g_j$.
Then, as the induction step, 
we 
show that
$\CO_B$ can walk to $s_{j+1}$ and then to $g_{j+2}$ 
, while keeping distance $1$ to $\CO_L$.
This is shown in Table \ref{tab:PathB}
 (see Figure \ref{fig:PathB} for an illustration).

\begin{table}[h]
\centering
\begin{tabular}{ r | l | l  }
  & location of $\CO_B$ & location of $\CO_L$  
 \\
\hline
   if $x_i \in C_j$  & $g_j$ & $\Dir{\alpha_{j-1}w_j} 	\dashv \Seg{c_jg_j}$ \\
   if $\neg x_i \in C_j$  & $g_j$ & $M(\Seg{c_jg_j})$\\
   if $x_i \notin C_j \& \neg x_i \notin C_j$  & $g_j$ & $\Dir{\alpha_{j-1}w_j}	\dashv \Seg{c_jg_j}$\\

\hline
	&  $h_3$ s.t.  $\| h_3\alpha_j \| \le \eps$ & $\alpha_j$\\
	&  	$h_4$ s.t.  $\| h_4 \beta_j \| \le \eps$ & $\beta_j$\\

\hline
if $x_i \in C_{j+1}$ &  				 $s_{j+1}$  & $\Seg{\beta_jw_{j+1}}	\dashv \Seg{c_{j+1}s_{j+1}}$\\
& $w_{j+1} \perp \Seg{s_{j+1}g_{j+2}}$ & $w_{j+1}$\\
& $z_{j+1}$ & $z_{j+1}$\\

&  $\Seg{g_{j+1}c_{j+1}}	\dashv \Seg{s_{j+1}g_{j+2}}$ & $c_{j+1}$ \\
&  & $M(\Seg{c_{j+1}g_{j+1}})$
 \\

if $\neg x_i \in C_{j+1}$ &  				 $s_{j+1}$  & $M(\Seg{s_{j+1}c_{j+1}})$\\

 & $z_{j+1}$ & $c_{j+1}$ \\
 &  & $w_{j+1}$ \\

 &  $\Seg{g_{j+1}c_{j+1}}	\dashv \Seg{s_{j+1}g_{j+2}}$ & $\Seg{g_{j+1}c_{j+1}}	\dashv \Seg{w_{j+1}\alpha_{j+1}}$ \\

if $x_i \notin C_{j+1} \& \neg x_i \notin C_{j+1}$ &  				 $s_{j+1}$  & $\Seg{\beta_jw_{j+1}}	\dashv \Seg{c_{j+1}s_{j+1}}$\\
& $z_{j+1}$ & $w_{j+1}$\\
 & & $c_{j+1}$ \\
 & & $w_{j+1}$ \\
&  $\Seg{g_{j+1}c_{j+1}}	\dashv \Seg{s_{j+1}g_{j+2}}$ & $\Seg{g_{j+1}c_{j+1}}	\dashv \Seg{w_{j+1}\alpha_{j+1}}$ \\

\hline

	&  $h_5$ s.t.  $\| h_5\alpha_{j+1} \| \le \eps$ & $\alpha_{j+1}$\\
	&  	$h_6$ s.t.  $\| h_6 \beta_{j+1} \| \le \eps$ & $\beta_{j+1}$\\
\hline

if $ x_i \in C_{j+2}$ &  		$\Seg{s_{j+2}c_{j+2}}	\dashv \Seg{s_{j+1}g_{j+2}}$ 		   & $M(\Seg{s_{j+2}c_{j+2}})$\\

 & $z_{j+2}$ & $c_{j+2}$ \\
 &  & $w_{j+2}$ \\

 &  $g_{j+2}$ & $\Seg{g_{j+2}c_{j+2}}	\dashv \Seg{w_{j+2}\alpha_{j+2}}$ \\

if $\neg x_i \in C_{j+2}$ &  				$\Seg{s_{j+2}c_{j+2}}	\dashv \Seg{s_{j+1}g_{j+2}}$ 		   &   $\Seg{s_{j+2}c_{j+2}}	\dashv \Seg{\beta_{j+1}w_{j+2}}$   \\
 & $z_{j+2}$ & $w_{j+2}$ \\
 &  & $c_{j+2}$ \\

 &   $g_{j+2}$ & $M(\Seg{c_{j+2}g_{j+2}})$ \\

if $x_i \notin C_{j+2} \& \neg x_i \notin C_{j+2}$ &  				
$\Seg{s_{j+2}c_{j+2}}	\dashv \Seg{s_{j+1}g_{j+2}}$ 		   &   $\Seg{s_{j+2}c_{j+2}}	\dashv \Seg{\beta_{j+1}w_{j+2}}$  \\
& $z_{j+2}$ & $w_{j+2}$\\
 & & $c_{j+2}$ \\
 & & $w_{j+2}$ \\
&  $g_{j+2}$ & $\Seg{g_{j+2}c_{j+2}}	\dashv \Seg{w_{j+2}\alpha_{j+2}}$ \\

\end{tabular}
\caption{Distance between pair of points is less or equal to one}
\label{tab:PathB}
\end{table}

Finally, if $k$ is an odd number, then  
$\Dir{g_kv}$ is the last segment along $B$, otherwise, 
$\Dir{s_kv}$ is the last one. In any case, 
that edge crosses  circle $\CB(\eta,1)$, where $\eta$ is the last vertex of 
$\cfev_i$ before $v$ (point $\eta$ is computed after the condition checking 
in line \ref{l:ComputeV} of 
Algorithm \ref{alg:reduction}). Therefore, 
 $\CO_B$ can walk to $v$, while keeping distance $1$ to $\CO_L$.

\qed

\end{proof}

\begin{figure}
	\centering
	\includegraphics[width=0.9\columnwidth]{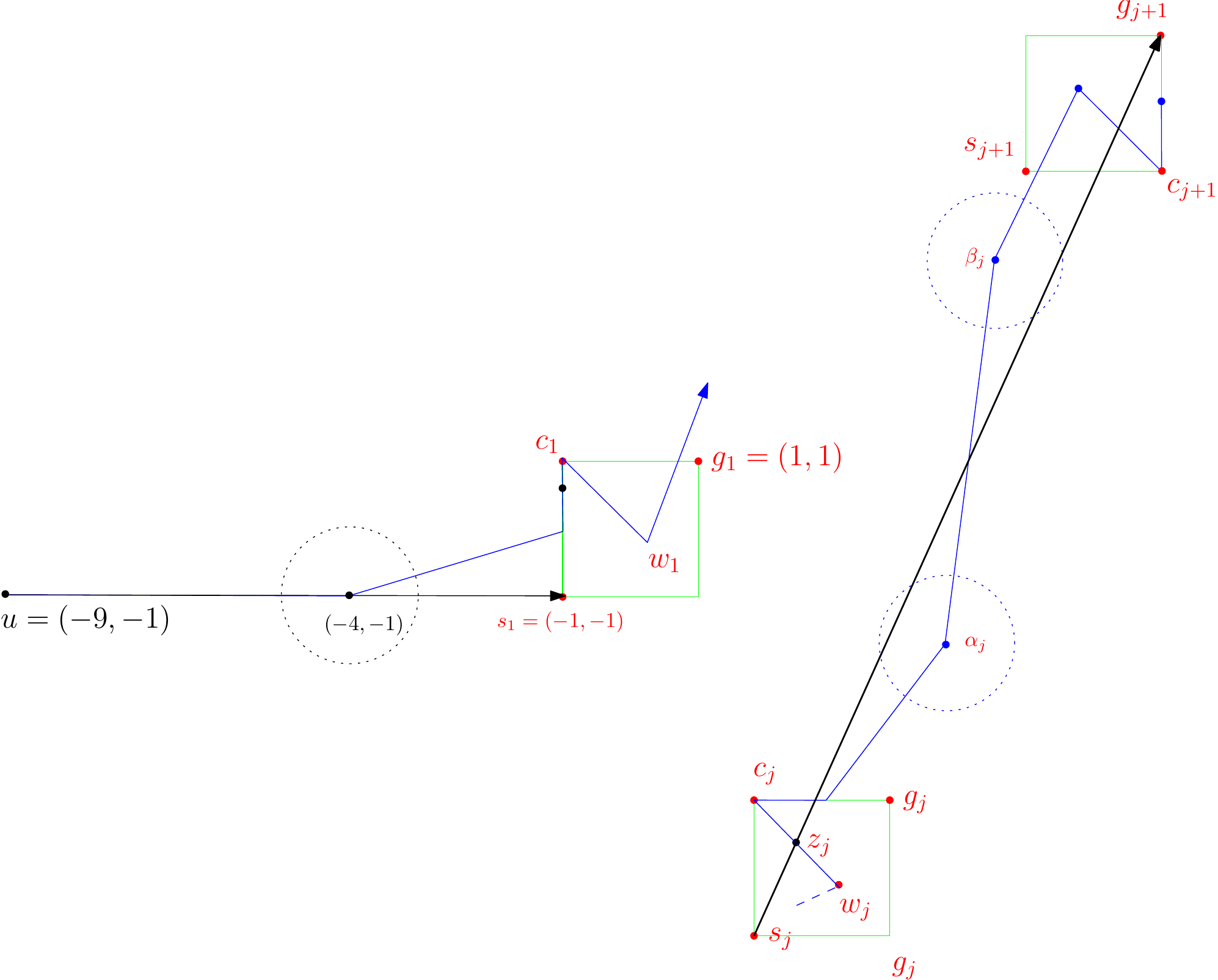}
	\caption{Proof of Lemma \ref{lemma:NoSwitchFromAtoB}}
	\label{fig:noswitch}
\end{figure}

\begin{lemma}\label{lemma:NoSwitchFromAtoB}
Consider any curve $\cfev_i \subset P $, $1\le i \le n+2$. Imagine that  
a point object $\CO_L$ is walking from $u$ to $v$ on $\cfev_i$. Furthermore, imagine
two point objects $\CO_A$ and $\CO_B$ which are walking on curves $A$ and $B$ 
(from Lemmas \ref{lemma:PathA} and \ref{lemma:PathB}), respectively,
while keeping distance $1$ to $\CO_L$. 
If $\CO_A$ goes to any vertex of $B$ or $\CO_B$ goes to any vertex of $A$, 
then they loose distance $\le 1$ to $\CO_L$.

\end{lemma}

\begin{proof}
Let 
$cl_i$ refer to 
points $\{s_i,g_i,c_i\}$.
Notice that we have placed the $cl_{i+1}$ points far enough from 
the $cl_{i}$ points so that 
no curve can go to $cl_{i+1}$
and come back to $cl_i$ and stay 
in \Frechet distance 1 to $\cfev_i$.
Therefore, to prove the lemma, 
we only focus on two consecutive c-squares.
We show that no subcurve $l' \subseteq \cfev_i$ exists such 
that (for an illustration, see Figure \ref{fig:noswitch}) :

\begin{itemize}

\item $\distF(l',\Dir{\sma_j\gre_j}) \le 1$ because:

for all $j$, $1 \le j \le k$, point $c_j$ is always a vertex of $\cfev_i$. 
A point on $\cfev_i$ at distance 1 
to $\sma_j$ lies before $c_j$ 
in direction $\Dir{\cfev_i}$, 
while a point on $\cfev_i$ at distance 1 
to point $\gre_j$ lies after $c_j$ in direction $\Dir{\cfev_i}$.
Since $dist(c_j, \Seg{\sma_j\gre_j}) >1$, 
no subcurve $l' \subseteq \cfev_i$ exists such that 
$\distF(l',\Dir{\sma_j\gre_j}) \le 1$.

\item $\distF(l',\langle \sma_jc_j\gre_{j}\rangle) \le 1$ or $\distF(l',\langle\gre_jc_j\sma_{j}\rangle) \le 1$, because:

For all $j$, $1 \le j \le k$, $w_j$
is a vertex of $\cfev_i$. 
A point on $\cfev_i$ at distance 1 
to $\sma_j$ lies before $w_j$ 
in direction $\Dir{\cfev_i}$, 
while a point on $\cfev_i$ at distance 1 
to point $\gre_j$ lies after $w_j$ in direction $\Dir{\cfev_i}$.
Since $dist(w_j, \Seg{\sma_jc_j}) >1$ and 
$dist(w_j, \Seg{\gre_j\gre_j}) >1$,
no subcurve $l' \subseteq \cfev_i$ exists such that 
$\distF(l', \langle \sma_jc_j\gre_{j}\rangle ) \le 1$.
Similarly,  no subcurve $l' \subseteq \cfev_i$ exists such that 
$\distF(l', \langle\gre_jc_j\sma_j\rangle ) \le 1$.

\item $\distF(l',\langle \sma_j\sma_{j+1} \rangle) \le 1$ or  $\distF(l',\langle\gre_j\gre_{j+1} \rangle) \le 1$ because:

Vertex $\alpha_{i}$ of $\cfev_{i}$
guarantees the first part as $dist( \alpha_{i},\Seg{\sma_j\sma_{j+1}}) > 1 $, 
and vertex $\beta_{i}$ of $\cfev_{i}$
guarantees the second part, 
as $dist( \beta_{i},\Seg{\gre_j\gre_{j+1}}) > 1$.

\item $\distF(l',\langle c_jc_{j+1} \rangle ) \le 1$,  because $dist( \alpha_{i},\Seg{c_jc_{j+1}}) > 1$

\item $\distF(l',\langle uc_1 \rangle) \le 1$, because $dist( (-4,-1),\Seg{uc_1}) >1$

\item  $\distF(l',\langle c_j\gre_{j+1} \rangle) \le 1$, because $dist( \alpha_i,\Seg{c_j\gre_{j+1}}) >1$

\item  $\distF(l',\langle c_j\sma_{j+1} \rangle ) \le 1$, because $dist( \alpha_i,\Seg{c_j\sma_{j+1}}) >1$

\item  $\distF(l',\langle c_kv \rangle) \le 1$, because $dist( \eta,\Seg{c_kv}) >1$ 
\end{itemize}

\end{proof}

\noindent To establish the correctness of
our reduction algorithm, 
from now on, we define:
$(a_i = \sma_i, b_i = \gre_i)$,  
 when $i$ is an odd number, 
and
$(a_i = \gre_i, b_i = \sma_i)$, 
when $i$ is an even number, for $1 \le i \le k$.

\begin{lemma}\label{lemma:ABCanSeeC}
Consider the curve $A = \langle ua_1a_2a_3 \dots a_k v \rangle$ from Lemma \ref{lemma:PathA}. Let 
$A_1$ be a subcurve of $A$ which starts at $u$ and ends at $a_j$, $1 \le j \le k$.
Furthermore, let $A_2$ be a subcurve of $A$ which starts at $a_j$ and ends at $v$.
For any curve $\cfev_i$ , $1\le i \le n+2$,
if $x_i \in C_j$, 
$\distF(A_1 \ap c_j \ap A_2, \cfev_i) \le \eps$. 
Similarly, consider the curve $B = \langle ub_1b_2b_3 \dots b_k v \rangle$ from Lemma \ref{lemma:PathB}. Let 
$B_1$ be a subcurve of $B$ which starts at $u$ and ends at $b_j$, $1 \le j \le k$.
Furthermore, let $B_2$ be a subcurve of $B$ which starts at $b_j$ and ends at $v$.
For any curve $\cfev_i$ , $1\le i \le n+2$,
if $\neg x_i \in C_j$,   
$\distF(B_1 \ap c_j \ap B_2, \cfev_i) \le \eps$. 
\end{lemma}

\begin{proof}
When $x_i$ appears in clause $C_j$, point $z = M(\Seg{c_ja_j})$ is 
a vertex of $\cfev_i$. 
Since $\| c_ja_j \| = 2$ and $z$ is the  midpoint of 
 $\Seg{c_ja_j}$,
$\CO_L$ can wait at $z$ while $\CO_A$ visits $c_j$.
Therefore, as the lemma states, 
we can cut curve $A$ at vertex $a_j$,
add two edges $\Dir{a_jc_j}$
and then $\Dir{c_ja_j}$ to $A$,
and continue with the same 
curve $A$
from $a_j$ to $A$'s endpoint. 
For the 
modified $A$, still $\distF(A, \cfev_i) \le \epsilon$ holds. 

When $\neg x_i$ appears in clause $C_j$, point $z = M(\Seg{c_jb_j})$ is 
a vertex of $\cfev_i$. 
Since $\| c_jb_j \| = 2$ and $z$ is the  midpoint of 
 $\Seg{c_jb_j}$,
$\CO_L$ can wait at $z$ while $\CO_A$ visits $c_j$
and comes back to $b_j$.
Therefore, as the lemma says, 
we can cut curve $B$ at vertex $b_j$,
add two edges $\Dir{b_jc_j}$
and then $\Dir{c_jb_j}$ to $B$,
and continue with the same 
curve $B$
from $b_j$ to $B$'s endpoint. 
For the 
modified $B$, still $\distF(B, \cfev_i) \le \epsilon$ holds.

\end{proof}

\begin{lemma}\label{lemma:NOTABCanSeeC}
Consider curve $A$ (respectively, $B$) from previous lemma.
For any curve $\cfev_i$, $1\le i \le n+2$,
when $ x_i \notin C_j$ and $\neg x_i \notin C_j$,   
curve $A$ (resp., $B$) can not be modified to visit $c_j$.
\end{lemma}
\begin{proof}
This holds because $dist(w_j,\Seg{a_jc_j}) >1$
and $dist(w_j,\Seg{b_jc_j}) >1$.

\end{proof}


\vspace{0.1 in}

\begin{theorem}
Given a formula $\phi$ with $k$ clauses $C_1, C_2, \dots, C_k$ and $n$ variables $x_1, x_2 \dots, x_n$,
as input, let curve $P$ and pointset $\pset$ be the output of Algorithm \ref{alg:reduction}. 
Then, $\phi$ is satisfiable iff a 
curve $Q \in Curves(S)$ exists such that 
$\distF(P,Q) \le 1$.
\end{theorem}

\begin{proof}


For $(\Rightarrow)$: 
Assume that  formula $\phi$ is satisfiable. 
In Algorithm \ref{alg:buildQ}, we show that 
knowing the truth value of the literals in $\phi$, 
we can build a curve $Q$ which 
visits every point in $\pset$ and $\distF(P,Q) \le 1$.

\begin{algorithm} [h]
\caption {{\sc Build a feasible curve $Q$ }} 
\label{alg:buildQ}
\algsetup{indent=1.5em}
\begin{algorithmic}[1]	
		\baselineskip=0.9\baselineskip
	\REQUIRE  Truth table of variables $x_1, x_2, \dots, x_n$ in $\phi$

	\STATE $Q \eq \emptyset$
	\STATE $Q \eq Q \ap t$ \label{l:startPoint}
	 
	\FOR {$i=1$ to $n$}   
	\IF {$(x_i = 1)$}
	\STATE $\pi \eq \langle ua_1a_2a_3 \dots a_kv \rangle$
	\FORALL {$C_j$ clauses, if $x_i \in C_j$ }
	\STATE let $\pi_1 $ be  subcurve of $\pi$ from $u$	 to $a_j$
	\STATE let $\pi_2 $ be  subcurve of $\pi$ from $a_j$	 to $v$
	\STATE $\pi \eq \pi_1 \ap c_j \ap \pi_2$  \label{l:visitCone}
	\ENDFOR
	\STATE $Q \eq Q \ap \pi$ \label{l:x1}
	\ELSE 	
	\STATE $\pi \eq \langle ub_1b_2b_3 \dots b_kv \rangle$
	\FORALL {$C_j$ clauses, if $\neg x_i \in C_j$ }
	\STATE let $\pi_1 $ be  subcurve of $\pi$ from $u$	 to $b_j$
	\STATE let $\pi_2 $ be  subcurve of $\pi$ from $b_j$	 to $v$
	\STATE $\pi \eq \pi_1 \ap c_j \ap \pi_2$ \label{l:visitCzero}
	\ENDFOR
	\STATE $Q \eq Q \ap \pi$ \label{l:x0}

	\ENDIF
	\STATE $Q \eq Q \ap t$
	\ENDFOR
	\STATE $Q \eq Q \ap \langle ua_1a_2a_3 \dots a_kv \rangle$\label{l:nplusone}
	\STATE $Q \eq Q \ap t$

	\STATE $Q \eq Q \ap \langle ub_1b_2b_3 \dots b_kv \rangle$\label{l:nplustwo}
	\STATE $Q \eq Q \ap t$  \label{l:endPoint}

	\STATE {\bf return} {\sc Q}  
\end{algorithmic}
\end{algorithm}


First, we show $\distF(P,Q) \le 1$, where $Q$
is the output curve of Algorithm \ref{alg:buildQ}.
Recall that by Algorithm \ref{alg:reduction}, 
curve $P$ includes $n$ subcurves $\cfev_{i}$ each corresponds 
to a variable $x_i$. 
Both curves $P$ and $Q$ start and end at a same point $t$.
For each curve $\pi$ which is appended to $Q$ 
in the $i$-th iteration of Algorithm \ref{alg:buildQ} 
(Line \ref{l:x1} or Line \ref{l:x0}), 
$\distF(\pi,\cfev_i) \le 1$  by Lemma \ref{lemma:ABCanSeeC}. 
Notice that $P$ also includes two additional subcurves $\cfev_{n+1}$ and $\cfev_{n+2}$ 
whereas there is no variable $x_{n+1}$ and $x_{n+2}$ in formula $\phi$. 
These two curves are to resolve two special cases: 
when all variables $x_i$ are 1,  no $\neg x_i$ appears in $\phi$,
and when all variables $x_i$ are 0,  
no $x_i$ appears in $\phi$.
Because of these two curves, 
we added two additional curves in line \ref{l:nplusone}
and \ref{l:nplustwo} to $Q$. Finally, by  Observation 
\ref{obs:concat}, $\distF(P,Q) \le 1$.

Next, we show that curve $Q$ visits every point in $S$. First of 
all, by the curves added to $Q$ 
in Line \ref{l:nplusone} and \ref{l:nplustwo}, 
all $a_j$ and $b_j$, $1\le j \le k$, in $S$ will be visited. 
It is sufficient to show that $Q$ will visit all $c_j$ points in $S$  as well.
Since  formula $\phi$ is satisfied, every clause $C_i$ in $\phi$ must be satisfied 
too. Fix clause $C_j$. At least one of the literals in $C_j$
must have a truth value $1$. If $x_i \in C_j$ and $x_i = 1$, 
then by line \ref{l:visitCone}, curve $Q$ visits $c_j$.
On the other hand, if $\neg x_i \in C_j$ and $x_i = 0$, 
by Line \ref{l:visitCzero}, curve $Q$ visits $c_j$. We conclude that 
curve $Q$ is feasible.

Now $(\Leftarrow)$ part:

Let $Q$ be a feasible curve with respect to $P$ and pointset $\pset$.
Notice that curve $P$ consists of $n$ subcurves $\cfev_i$, 
$1 \le i \le n $, where each corresponds to one variable $x_i$. 
From the configuration of each $\cfev_i$ in c-squares, 
one can easily construct formula $\phi$ with 
all of its clauses and literals. 


Imagine two point objects $\CO_Q$ 
and $\CO_P$ walk on $P$ and $Q$, respectively. 
We find the truth value of variable $x_i$ in the formula
by looking at the path that $\CO_Q$ takes to stay in \Frechet distance 1 to $\CO_P$, 
when $\CO_P$ walks on curve $\cfev_i$ corresponding to $x_i$.
If $\CO_Q$ takes path $A$ from Lemma \ref{lemma:PathA} 
while $\CO_P$ is walking on  $\cfev_i$, then $x_i = 1$. But if $\CO_Q$ takes path $B$ from Lemma \ref{lemma:PathB} 
while $\CO_P$ is walking on  $\cfev_i$, then $x_i = 0$. 
Object $\CO_Q$ decides between path $A$ or $B$,  when both $\CO_Q$ and $\CO_P$ are at point $u$. 
Lemma \ref{lemma:NoSwitchFromAtoB} ensures that  
once they start walking, 
$\CO_Q$ can not change its path from $A$ to $B$ 
or from $B$ to $A$. 
Therefore, the truth value of a variable $x_i$ is consistent.

The only thing left to show is the reason that formula $\phi$ is satisfiable. 
It is sufficient to show every clause of $\phi$ is satisfiable. 
Consider any clause $C_j$.
Since curve $Q$ is feasible, 
it uses every point in $\pset$.  
Assume w.l.o.g. that $\CO_Q$ visits $c_j$ 
when $\CO_P$ is walking along curve $\cfev_i$.  
By Lemmas 
\ref{lemma:NoSwitchFromAtoB} and \ref{lemma:ABCanSeeC},
this only happens when either ($x_i$ appears in $C_i$ and $x_i = 1$)
or ($\neg x_i$ appears in $C_i$ and $x_i =0$). 
Therefore, $C_j$ is satisfiable.


%
%

%
%
%
%

The last ingredient of the NP-completeness proof is
to show that the reduction takes polynomial time.  
One can easily see that Algorithm \ref{alg:reduction}
has running time $O(nk)$, 
where $n$ is the number of variables in 
the input formula with $k$ clauses.

\end{proof}

\subsection{Implementation Results}
To show the simplicity of our reduction algorithm, we have implemented it
in Java.  The figures in this chapter are all generated by our program. 
Our test case, as 
an input to the program, is a formula $\phi$ with four clauses. The output is 
three sets $S$, $L$ and $C$ as follows.

Set $S$ is a pointset computed by 
Algorithm \ref{alg:reduction}.
Since $\phi$ has four clauses, 
$S$
contains the following points:

$S = \{ s_1,g_1,c_1, s_2,g_2,c_2, s_3,g_3,c_3, s_4,g_4,c_4 , u, v, t \}$.

Set $L$ is a set of curves, 
where each of it is a configuration of 
$\cfev_i$ in the reduction algorithm.
Choosing  a formula with four clauses as an input, 
enables us to check all possible configurations of curve $\cfev_i$ built by Algorithm \ref{alg:reduction}. 
Let  $x_i$ be a variable in  formula
$\phi$.
Since $x_i$ or $\neg x_i$ or none could appear in a clause, and the formula has four clauses,  set $L$ contains 81 curves $\cfev_i$. 

Set $C$ contains all possible curves $\alpha$, 
each built in this way: 
$\alpha$ starts from point $u$, 
goes through arbitrary points from $\{g_1,c_1,s_1\}$, 
then to arbitrary points in $\{g_2,c_2,s_2\}$, next
to arbitrary points from $\{g_3,c_3,s_3\}$, and lastly 
from $\{g_4,c_4,s_4\}$
and at the end, 
$\alpha$ ends at $v$.
Therefore, set $C$ has almost 1,000,000,000 polygonal curves.

Let $\alpha$ be any curve in $C$ 
and $\ell$ be any curve in $L$. 
We compute in our program, the \Frechet distance between every curve $\alpha$ and $\ell$.
Notice that $C$ has huge amounts of curve data. We implemented our 
program in an efficient way so 
that we could do this computation in 
a fair amount of time. First, 
all 81 curves $\ell$ are computed and then, 
by computing each $\alpha$ in $C$, 
we compute 81 \Frechet distances 
$\distF(\alpha,\ell)$. Therefore, in total, 
almost $81 \times 1,000,000,000$ 
\Frechet distances have been computed by our program. 
The experiment is 
performed on four machines in parallel, each 
has an Intel(R) Core(TM) i7 CPU 2.67GHz and 12GB RAM.

The results show that in all cases,
$\distF(\alpha,\ell) > 1$ 
except for the following cases:
\vspace{0.1 in}

Case I:  $\alpha = \langle u,s_1,g_2,s_3,g_4,v \rangle$, then 
$\distF(\alpha,\ell) \le  1$, for any curve $\ell$ in $L$.

Case II:  $\alpha = \langle u,g_1,s_2,g_3,s_4,v \rangle$, then 
$\distF(\alpha,\ell) \le  1$, for any curve $\ell$ in $L$.
 
Case III:  $\alpha = \langle u,g_1,c_1,g_1,s_2,g_3,s_4,v \rangle$, then 
$\distF(\alpha,\ell) \le  1$, for $\ell$ corresponding to the case 
where $\neg x_i$ appeared in the 
first clause.

Case IV: $\alpha = \langle u,s_1,c_1,s_1,g_2,s_3,g_4,v \rangle$, 
then 
$\distF(\alpha,\ell) \le  1$, for $\ell$ corresponding to the case 
where $x_i$ appeared in the 
first clause 
and so on for other occurrence of 
variable $x_i$ in other clauses.

\vspace{0.1 in}
Case I confirms Lemma \ref{lemma:PathA},
case II confirms Lemma \ref{lemma:PathB},
cases III and IV confirm Lemma \ref{lemma:ABCanSeeC},
and all together confirm Lemma \ref{lemma:NoSwitchFromAtoB}.

\section{Conclusions}
\label{sec:conc}
In this chapter, we investigated the problem of deciding whether a polygonal curve through a given pointset $\pset$ exists, which 
visits every point in $\pset$ and
 is in $\eps$-\Frechet distance 
to a curve $P$. We showed that this problem is NP-complete. 


\REM{
\chapter{All-Points CPM Problem is NP-complete}
\label{ch:NP-Complete}

\section{Introduction}
In this chapter, we study 
a variant of the problem discussed in the previous chapter. 
We refer to this variant as the All-Points CPM problem. We address the following:
Consider a pointset $S \subseteq \IR^d$ and a polygonal curve $P$ in $\IR^d$, 
for $d \gee 2$ being a fixed dimension.
The objective is to decide whether there exists a polygonal curve $Q$ in  $\eps$-\Frechet
distance to $P$ such that the vertices of $Q$ are all chosen from the 
pointset $\pset$. Moreover, curve $Q$ has to visit every point of $\pset$
and it can visit a point multiple times. 
We  prove  that  this problem  is  NP-complete  by  reducing from 3CNF-SAT problem.
In an independent work~\cite{NPComplete-Pointset} (which is done after   
our NP-completeness result), the authors have shown that
the version of this problem where points in $\pset$ has to 
be visited only once, is NP-complete too. 
Their proof is obtained via reduction
from a restricted version of the 3SAT problem, 
called (3,B2)-SAT problem,
where the input to formulas is restricted
in which each literal occurs exactly twice.
In \cite{DiscretelyFollowing},  
Wylie and Zhu studied 
All-points CPM problem from the 
perspective of discrete \Frechet distance
and they showed that it is solvable in 
$O(nk)$ time ($n$ is the size of curve $P$ and $k$ is the size of  pointset $\pset$). 
Furthermore, they showed that the version of the problem in which 
each point of $\pset$ can only used once in $Q$ is  NP-complete.

\section{General Case is NP-complete}
\label{sec:NPComp}

\subsection{Preliminaries}
\noindent{\bf Notation.}
We denote by $P = \langle p_1p_2p_3...p_n \rangle$, a polygonal curve $P$
with vertices $p_1 p_2 \dots p_n$ in order 
and by $start(P)$ and $end(P)$, we denote 
the starting and ending point of $P$, respectively.
For a curve $P$ and a point $x$, by $P \ap x$, 
we mean connecting $end(P)$ to point $x$.
We use the same notation $P \ap Q$ 
to show the concatenation of 
two curves $P$ and $Q$ (which means connecting $end(P)$ to $start(Q)$).
Let $M(\Seg{ab})$ denote the  midpoint of line segment $\Seg{ab}$. 
For a point $q$ in the plane, let $x(q)$ and $y(q)$
denote the $x$ and $y$ coordinate of $q$, respectively.

For two intersecting line segments $\Seg{ab}$ and $\Seg{cd}$, let $ \Seg{ab} \dashv	\Seg{cd}$ denote the intersection point of them.
Let $\overleftrightarrow{bc}$ denote the line as a result of 
extending line segment $\Seg{bc}$.
For a point $p$ and a line segment $\Seg{bc}$,
let $p \perp \Seg{bc}$ denote the point on
line $\overleftrightarrow{bc}$,
located on the perpendicular from $p$ to $\overleftrightarrow{bc}$.

\begin{definition} \label{def:curves}
Given a pointset $\pset$ in the plane, let $Curves(\pset)$
be a set of polygonal curves $Q = \langle q_1 q_2 \dots q_n \rangle$ where: 
$$  \forall{q_i} : q_i \in \pset \mbox{  and  } 
  \forall{a} \in \pset: \exists{q_i} \mbox{  s.t. }  q_i = a. $$
\end{definition}

\begin{definition} \label{def:feasibleNPC}
Given a pointset $\pset$, a polygonal curve $P$ and a distance $\eps$, 
a polygonal curve $Q$ is called {\em feasible} if: 
$Q \in Curves(S)$ and $  \distF(P,Q) \le \eps$.
\end{definition}

We show that the problem of deciding whether a 
feasible curve exists or not  is NP-complete.
It is easy to see that this problem is in NP, since 
one can polynomially check whether $Q \in Curves(S)$
and also $\distF(P,Q) \le \eps$, using the algorithm in \cite{AltG95} 
(explained in Section \ref{sec:classicalFD}).

%

\subsection{Reduction Algorithm}

We reduce in Algorithm~\ref{alg:reduction},
an instance of 3CNF-SAT formula $\phi$ 
to an instance of our problem.
The input is a boolean formula 
$\phi$ with $k$ clauses $C_1, C_2, \dots, C_k$ and $n$ variables $x_1, x_2, \dots,x_n$ 
and the output is a pointset $\pset$,
a polygonal curve $P$ in the plane and 
a distance $\eps = 1$.

We construct the pointset $S$ as follows.
For each clause $C_j$, $1 \le j \le k$, in the formula $\phi$, 
we place three points $\{s_j, g_j,c_j\}$ in the plane, which are computed
in the $j$-th iteration of Algorithm \ref{alg:reduction} (from Line \ref{l:makeSLoop} to Line \ref{l:EndLoopPointSet}).
We define $o_j$ to be $M(\Seg{\sma_j\gre_j})$.
By $\sq_j$, $1\le j \le k$, we denote
a square in the plane, centered at $o_j$, 
with diagonal $\Seg{\sma_j \gre_j}$. 
We refer to $\sq_j$, $1\le j \le k$, as
{\em c-squares}. 
For an example of a pointset $S$ corresponding to a formula, 
see Figure \ref{fig:pathAExample}a.

Our reduction algorithm constructs the polygonal curve $P$ 
 through $n$ iterations. In  the $i$-th iteration, $1 \le i \le n$,
it builds a subcurve $\cfev_i$ corresponding to  a variable 
$x_i$ in the formula $\phi$ and appends that curve to $P$.
In addition to those $n$ subcurves, two curves 
$\cfev_{n+1}$ and $\cfev_{n+2}$ are appended to $P$. 
We will later discus the reason we add those 
two curves. Every subcurve  $\cfev_i$ of $P$ starts at point $u$ 
and ends at point $v$.
Furthermore, each  $\cfev_i$ goes through
$\sq_1$ to $\sq_k$ in order, enters each c-square $\sq_j$ from
the side $\Seg{c_j s_j}$  and exists that square from 
the side $\Seg{c_j g_j}$  (for an illustration, see Figure \ref{fig:pathAExample}a). 
Curve $\cfev_i$ itself is built incrementally  
through iterations of the loop at line \ref{l:looptoMakeL} 
of Algorithm \ref{alg:reduction}. 
In the $j$-th iteration, when $\cfev_i$ goes through $\sq_j$,
three points, which are within  $\sq_j$, are added to  $\cfev_i$
(these three points are computed through Lines 
\ref{l:makeclausestart} to \ref{l:makeclauseend}).
Next, before $\cfev_i$ reaches  $\sq_{j+1}$,
two points,  denoted by $\alpha_j$ and $\beta_j$, are added to that curve 
(these two points are computed in Lines \ref{l:alpha} 
and \ref{l:beta}).

Each $\cfev_i$ corresponds to variable $x_i$ in our approach.
We simulate $1$ or $0$ values of $x_i$
as follows.
Consider a point object $\CO_L$ 
traversing $\cfev_i$, from starting point $u$ to ending point $v$. 
Consider 
another point object $\CO_2$ which wants to 
walk from $u$ to $v$
on a path whose vertices are from points in $S$ and it wants to stay in distance one 
to $\CO_L$. We will show that 
by our construction, object $\CO_2$ has two options, either taking 
the path $A = \langle us_1g_2s_3 \dots v \rangle$ or the path $B = \langle ug_1s_2g_3 \dots v \rangle$ 
(See Figure \ref{fig:pathAExample}a and \ref{fig:pathAExample}b for an illustration). 
Choosing path $A$ by $\CO_2$ means $x_i = 1$ and choosing path $B$ means $x_i = 0$.
We first prove in Lemma \ref{lemma:PathA} that $\distF(\cfev_i,A) \le 1 $ and 
in Lemma  \ref{lemma:PathB} that $\distF(\cfev_i,B) \le 1 $.
Furthermore, in
Lemma \ref{lemma:NoSwitchFromAtoB}, we prove that as soon as $\CO_2$ chooses 
path $A$ at point $u$ to walk towards $v$, 
it can not switch to any vertex on path $B$.
Analogously, we show that as soon as $\CO_2$ chooses path $B$ at point $u$ to walk towards $v$, 
it can not switch to any vertex on path $A$.
In addition, in Lemmas \ref{lemma:ABCanSeeC} and \ref{lemma:NOTABCanSeeC}, we prove that 
if $x_i$ appears in clause $C_j$,
$\CO_2$ could visit point $c_j$ via the path $A$ and not $B$. In contrast, 
when $\neg x_i$ appears in the clause $C_j$,
$\CO_2$ could visit point $c_j$ via the path $B$ and not $A$.
However, when none of  $x_i$ or $\neg x_i$ appear 
in $C_j$, $\CO_2$ can take neither $A$ nor $B$ to visit $c_j$.
Thus, $c_j$ can be visited, 
if and only if there is 
an $i$ such that 
either $x_i$ or 
$\neg x_i$ are in clause $C_j$.

\begin{algorithm} 
\caption {{\sc Reduction Algorithm}} 
\label{alg:reduction}
\algsetup{indent=1.5em}
\begin{algorithmic}[1]
	\baselineskip=0.5\baselineskip
	\REQUIRE  3SAT formula $\phi$ with $k$ clauses $C_1 \dots C_k$ and $n$ variables $x_1 \dots x_n$

	\vspace{0.1in}
		
	\hspace{-0.2in} {\bf Construct pointset $S$:}  

	\STATE $\pset \leftarrow \emptyset$ \label{l:init}


	\STATE $\gre_1 = (1,1) $ \label{l:makeSStart}

	\FOR {$j = 1$ to $k$}   \label{l:makeSLoop}

	\STATE $\sma_j \leftarrow \big(x(\gre_j)-2,y(\gre_j)-2\big) $
		 \STATE $o_j   \eq  M(\Seg{\sma_j\gre_j})$

		\IF {($j$ is odd)	}

	\STATE $c_j \leftarrow \big(x(s_j), y(\gre_j) \big)$, $w_j  \eq \big(x(o_j)+\frac{1}{4},  y(o_j)-\frac{1}{4} \big)$

	\STATE  $\gre_{j+1} \eq \big (   x(s_{j}) + \frac{1}{4} + 8, y(s_{j}) + \frac{7}{4} +15 \big)$     		\label{l:ComputeNextEven}
    
		\ELSE
	\STATE $c_j \leftarrow \big(x(g_j), y(s_j) \big)$, $w_j  \eq \big(x(o_j)-\frac{1}{4},  y(o_j)+\frac{1}{4} \big)$ 
    
	\STATE  $\gre_{j+1} \eq \big ( x(s_{j}) + \frac{7}{4} + 15, y(s_{j}) + \frac{1}{4} + 8 \big)$     		\label{l:ComputeNextOdd}


	\ENDIF


\STATE $z_j = M(\Seg{c_jw_j})$

    \STATE $S = S \cup \{\sma_j,\gre_j, c_j\}$   \label{l:EndLoopPointSet}

	\ENDFOR

	\IF {($k$ is odd)}  \label{l:ComputeV}
	  \STATE $\eta \eq \big( x(o_k)+1, y(o_k)+4\big)$  
      \STATE $v \eq \big( x(o_k)+1, y(o_k)+9\big)$  
		\ELSE
	   \STATE $\eta \eq \big( x(o_k)+4, y(o_k)+1\big)$  
      \STATE $v \eq \big( x(o_k)+9, y(o_k)+1\big)$  
	\ENDIF

	\STATE $u = (-9,-1)$

	\STATE $t \eq \big(x(v),y(u) -20\big)$

    \STATE $\pset = \pset \cup \{u, v, t\}$ \label{l:makeSEnd}

\vspace{0.15in}

	\hspace{-0.25in} {\bf Construct polygonal curve $P$:}

	\STATE $P \eq \emptyset$ \label{l:makeP}

	\STATE $P \eq P \ap t$

	\FOR { $i = 1$ to $n+2$  }   \label{l:mainstart}

		\STATE  $\cfev_i \eq \emptyset$ \label{l:startofL}

		\STATE  $\cfev_{i} \eq \cfev_{i} \ap u$ 
		\STATE $\cfev_{i} \eq \cfev_{i} \ap (-4,-1) $  \label{l:Adduh1toell}
		\FOR {$j = 1$ to $k$} \label{l:looptoMakeL}

			\IF { ($x_i \in C_j$ and $j$ is odd ) or ($\neg x_i \in C_j$ and $j$ is even ) } \label{l:makeclausestart}	
			\STATE $\cfev_{i} \eq \cfev_{i} \ap M(\Seg{s_jc_j}) \ap c_j \ap w_j  $
			\ELSIF {{ ($\neg x_i \in C_j$ and $j$ is odd ) or ($ x_i \in C_j$ and $j$ is even ) }}
			\STATE $\cfev_{i} \eq \cfev_{i} \ap w_j \ap c_j \ap M(\Seg{g_jc_j} ) $
			\ELSE		
			\STATE $\cfev_i \eq  \cfev_i \ap w_j \ap c_j \ap w_j$
			\ENDIF \label{l:makeclauseend}

			\IF {$j \neq k$}

			\STATE $\alpha_j =  \frac{4}{5} g_j + \frac{1}{5} g_{j+1}$ \label{l:alpha}

			\STATE $\beta_j = \frac{1}{5} s_j + \frac{4}{5} s_{j+1} $ \label{l:beta}
			
			\STATE $\cfev_i \eq \cfev_i \ap \alpha_j \ap \beta_j$

			\ENDIF

		\ENDFOR
	
		\STATE $\cfev_i \eq \cfev_i \ap \eta \ap v$ \label{l:subcurve}

		\STATE  $P \leftarrow P \ap \cfev_i$ 
		\STATE  $P \eq P \ap t$

	\ENDFOR

	\vspace{0.05in}
\RETURN  pointset $\pset$, polygonal curve $P$ and distance $\eps = 1$

\end{algorithmic}
\end{algorithm}

\begin{table}[h]
\centering
\begin{tabular}{ r | l | l  }
if $x_i \in C_1$   & location of $\CO_A$ & location of $\CO_L$  
 \\
\hline
    
&  $u$ & $u$  \\
&  				 $h_1$ s.t.  $\| h_1\mu_1 \| \le \eps$ & $\mu_1 = (-4,-1)$\\

& $s_1 $  & $M(\Seg{s_1c_1})$	      \\

\hline
if $\neg x_i \in C_1$ 
&  $u$ & $u$  \\
&  				 $h_1$ s.t.  $\| h_1\mu_1 \| \le \eps$ & $\mu_1 = (-4,-1)$\\

& $s_1$ &  $ \Dir{\mu_1w_1} \dashv	 \Seg{s_1c_1} $\\

\hline
if $x_i \notin C_{1} \& \neg x_i \notin C_{1}$ & $u$ & $u$  \\
&  				 $h_1$ s.t.  $\| h_1\mu_1 \| \le \eps$ & $\mu_1 = (-4,-1)$\\
& $s_1$ &  $ \Dir{\mu_1w_1} \dashv	 \Seg{s_1c_1} $\\

\end{tabular}
\vspace{0.2 in}
\caption{Proof of Lemma \ref{lemma:PathA}, the base case of induction}
\label{tab:BaseCasePathA}
\end{table}

\begin{lemma}\label{lemma:PathA}
Consider any subcurve $\cfev_i$, $1\le i \le n+2$,  
which is built through Lines \ref{l:mainstart} to \ref{l:subcurve} 
of Algorithm \ref{alg:reduction}. Let $A$ be the polygonal curve  $\langle u\sma_1\gre_2\sma_3\gre_4..v\rangle$. 
Then, $\distF(\cfev_i,A) \le 1$.
\end{lemma}

\begin{proof}

We prove the lemma by induction on the number of segments along $A$. 
Consider two point objects $\CO_L$ and $\CO_A$ 
traversing $\cfev_i$ and $A$, respectively (Figure \ref{fig:pathAExample}a depicts an instance of $\cfev_i$ and $A$).
We show that $\CO_L$ and $\CO_A$ can walk
their respective curve, from the beginning to
 end, while keeping distance $1$ to each other. 

The base case of induction trivially holds as follows 
(see Figure \ref{fig:PathAClause1} for an illustration).
Table \ref{tab:BaseCasePathA} lists  pairwise locations of 
$\CO_L$ and $\CO_A$, where the distance of each pair is at most $1$.
Hence, $\CO_A$ can walk from $u$ to $s_1$ on the 
first segment of $A$ (segment $\Dir{us_1}$), 
while keeping distance $\le 1$ to $\CO_L$.

Assume inductively that $\CO_L$ and $\CO_A$ have feasibly walked along 
their respective curves, until $\CO_A$ reached $s_j$.
Then, as the induction step, 
we 
show that
$\CO_A$ can walk to $g_{j+1}$ and then to $s_{j+2}$, while keeping distance $1$ to $\CO_L$.
Table \ref{tab:PathA} lists pairwise locations
of $\CO_A$ and $\CO_L$ such that $\CO_A$ could reach  $s_{j+2}$.
One can easily check that the distance between the pair of points 
in that table is at most one.
 (For an illustration, see Figure \ref{fig:PathA}).

\begin{table}[t]
\centering
\begin{tabular}{ r | l | l  }
  & location of $\CO_A$ & location of $\CO_L$  
 \\
\hline
   if $x_i \in C_j$  & $s_j$ & $M(\Seg{c_js_j})$\\
	& $z_j$ & $c_j$\\ 
	&  & $w_j$\\ 

	& $\Seg{c_jg_j} \dashv \Seg{s_jg_{j+1}} $ & $\Seg{w_{j}\alpha_{j}} \dashv \Seg{c_{j}g_{j}}$ \\

   if $\neg x_i \in C_j$  & $s_j$ & $\Seg{\beta_{j-1}w_j}	\dashv \Seg{c_js_j}$ \\
	& $w_j \perp \Seg{s_jg_{j+1}}$ &$w_j$\\
	& $z_j$ &$z_j$\\
	& &$c_j$\\
	& $\Seg{c_jg_j} \dashv \Seg{s_jg_{j+1}} $ &$ M(\Seg{c_jg_j}) $\\

   if $x_i \notin C_j \& \neg x_i \notin C_j$  & $s_j$ & $\Seg{\beta_{j-1}w_j}	\dashv \Seg{c_js_j}$\\
	&$w_j \perp \Seg{s_jg_{j+1}}$  & $w_j$\\
	&$z_j$  & $z_j$\\

	& &$c_j$\\
	& &$w_j$\\
& $\Seg{c_jg_j} \dashv \Seg{s_jg_{j+1}} $ & $\Seg{w_{j}\alpha_{j}} \dashv \Seg{c_{j}g_{j}}$ \\

\hline
	&  $h_1$ s.t.  $\| h_1\alpha_j \| \le \eps$ & $\alpha_j$\\
	&  	$h_2$ s.t.  $\| h_2 \beta_j \| \le \eps$ & $\beta_j$\\

\hline
if $x_i \in C_{j+1}$ &  $\Seg{s_{j+1}c_{j+1}}	\dashv \Seg{s_{j}g_{j+1}}$    & $\Seg{\beta_jw_{j+1}}	\dashv \Seg{c_{j+1}s_{j+1}}$\\
& $z_{j+1}$ & $w_{j+1}$\\
&  & $z_{j+1}$\\
&  & $c_{j+1}$ \\
&  $g_{j+1}$ & $M(\Seg{c_{j+1}g_{j+1}})$\\

if $\neg x_i \in C_{j+1}$ &  			$\Seg{s_{j+1}c_{j+1}}	\dashv \Seg{s_{j}g_{j+1}}$  & $M(\Seg{s_{j+1}c_{j+1}})$\\
 & $z_{j+1}$ & $c_{j+1}$ \\
 &  & $w_{j+1}$ \\
 &  $g_{j+1}$ & $\Seg{g_{j+1}c_{j+1}}	\dashv \Seg{w_{j+1}\alpha_{j+1}}$ \\

if $x_i \notin C_{j+1} \& \neg x_i \notin C_{j+1}$ &  	$\Seg{s_{j+1}c_{j+1}}	\dashv \Seg{s_{j}g_{j+1}}$  & $\Seg{\beta_jw_{j+1}}	\dashv \Seg{c_{j+1}s_{j+1}}$\\
& $z_{j+1}$ & $w_{j+1}$\\
 & & $c_{j+1}$ \\
 & & $w_{j+1}$ \\
&  $g_{j+1}$ & $\Seg{g_{j+1}c_{j+1}}	\dashv \Seg{w_{j+1}\alpha_{j+1}}$ \\

\hline

	&  $h_3$ s.t.  $\| h_3\alpha_{j+1} \| \le \eps$ & $\alpha_{j+1}$\\
	&  	$h_4$ s.t.  $\| h_4 \beta_{j+1} \| \le \eps$ & $\beta_{j+1}$\\
\hline

 if $\neg x_i \in C_{j+2}$  & $s_{j+2}$ & $\Dir{\alpha_{j+1}w_{j+2}} 	\dashv \Seg{c_{j+2}s_{j+2}}$ \\
   if $ x_i \in C_{j+2}$  & $s_{j+2}$ & $M(\Seg{c_{j+2}s_{j+2}})$\\
   if $x_i \notin C_{j+2} \& \neg x_i \notin C_{j+2}$  & $s_{j+2}$ & $\Dir{\alpha_{j+1}w_{j+2}}	\dashv \Seg{c_{j+2}s_{j+2}}$\\

\end{tabular}
\caption{Distance between pair of points is less or equal to one}
\label{tab:PathA}
\end{table}

Finally, if $k$ is an odd number, then  
$\Dir{s_kv}$ is the last segment along $B$, otherwise, 
$\Dir{g_kv}$ is the last one. In 
either case, 
that edge crosses  the circle $\CB(\eta,1)$, where $\eta$ is the last vertex of 
$\cfev_i$ before $v$ (point $\eta$ is computed in line \ref{l:ComputeV} of 
Algorithm \ref{alg:reduction}). Therefore, 
 $\CO_A$ can walk to $v$, while keeping distance $1$ to $\CO_L$.

\qed
\end{proof}

\begin{figure}[t]
	\centering
	\includegraphics[width=1\columnwidth]{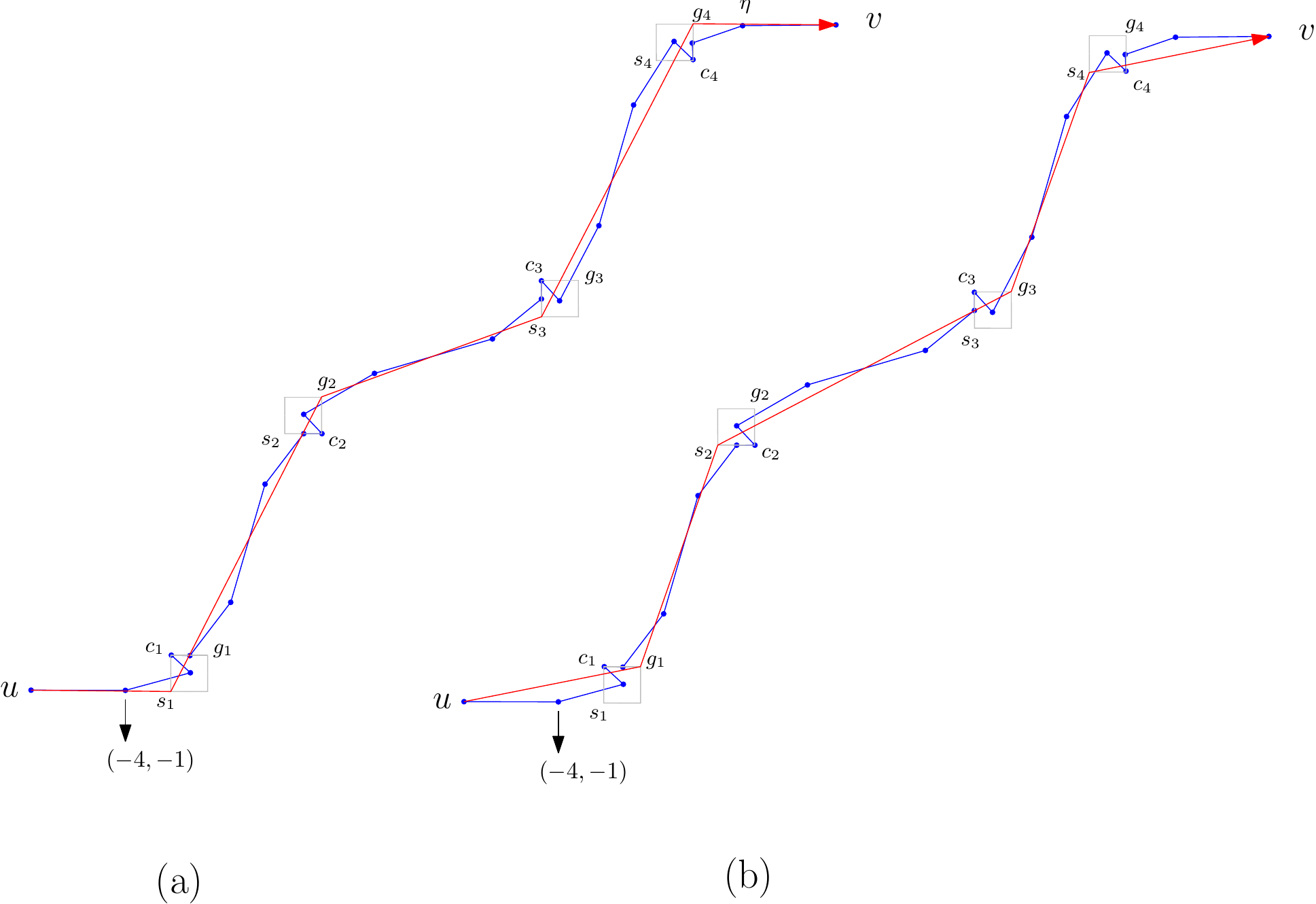}
	\caption{
Blue curve is an example of curve $\cfev_i$ which corresponds to variable $x_i$ in  formula $\phi$. The formula has four clauses $C_1, C_2, C_3$ and $C_4$,	 where the occurrence of variable $x_i$  in those clauses is:
$\neg x_i \in C_1$, $ \neg x_i \in C_2$, $x_i \in C_3$ and $x_i \in C_4$.  
For each clause $C_i$, the reduction algorithm places three point $s_i,g_i$ and $c_i$ in the plane. (a) Red curve is curve $A$. 
 (b) Red curve is curve $B$.  }
	\label{fig:pathAExample}
\end{figure}

\begin{figure}[h]
	\centering
	\includegraphics[width=0.9\columnwidth]{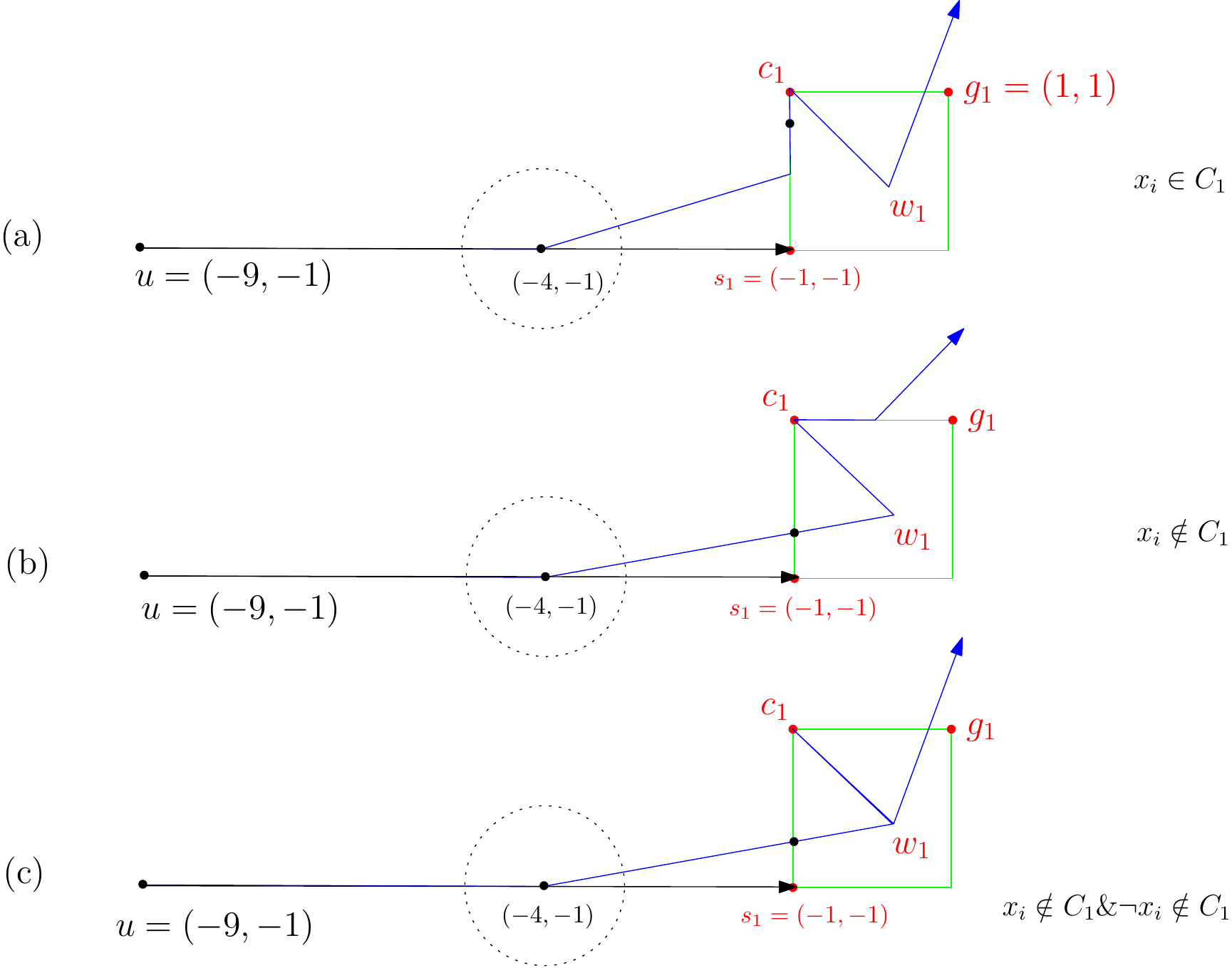}
	\caption{Base case of induction in the proof of Lemma \ref{lemma:PathA}}
	\label{fig:PathAClause1}
\end{figure}

\begin{figure}
	\centering
	\includegraphics[width=0.9\columnwidth]{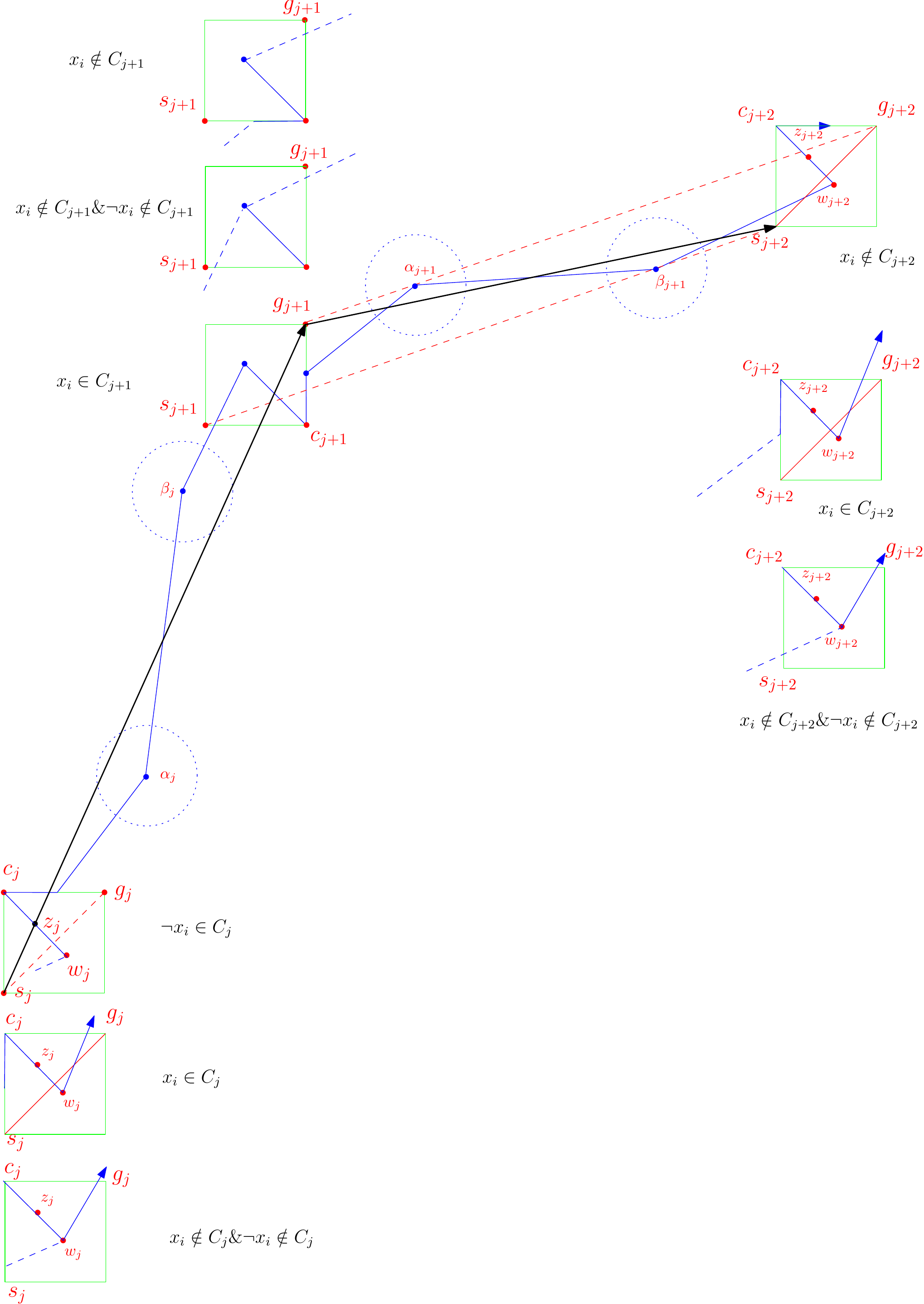}
	\caption{Proof of Lemma \ref{lemma:PathA}}
	\label{fig:PathA}
\end{figure}

\begin{figure}[h]
	\centering
	\includegraphics[width=0.9\columnwidth]{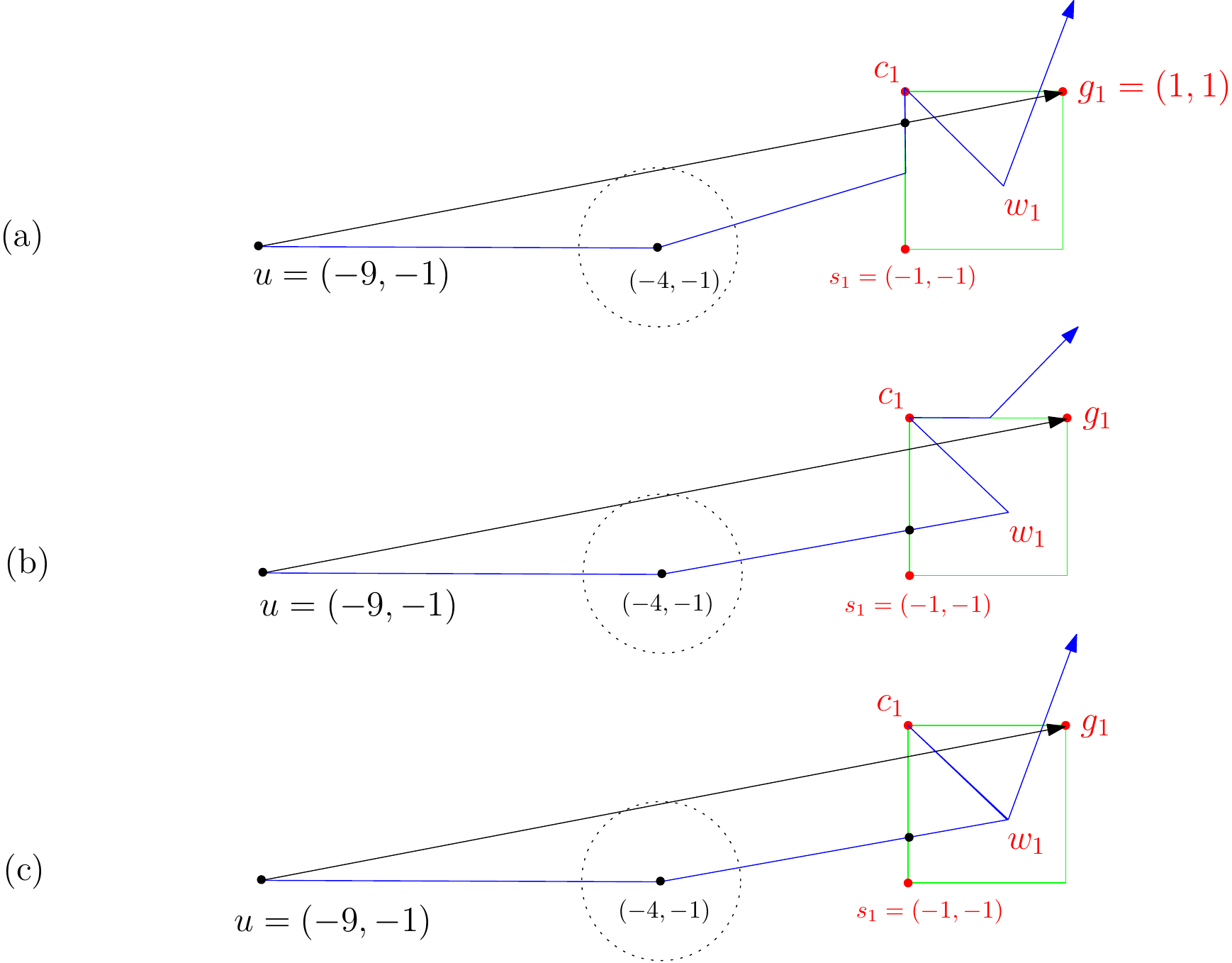}
	\caption{Base case of induction in the proof of Lemma \ref{lemma:PathB} }
	\label{fig:PathBBaseCase}
\end{figure}

\begin{figure}[h]

	\centering
	\includegraphics[width=0.9\columnwidth]{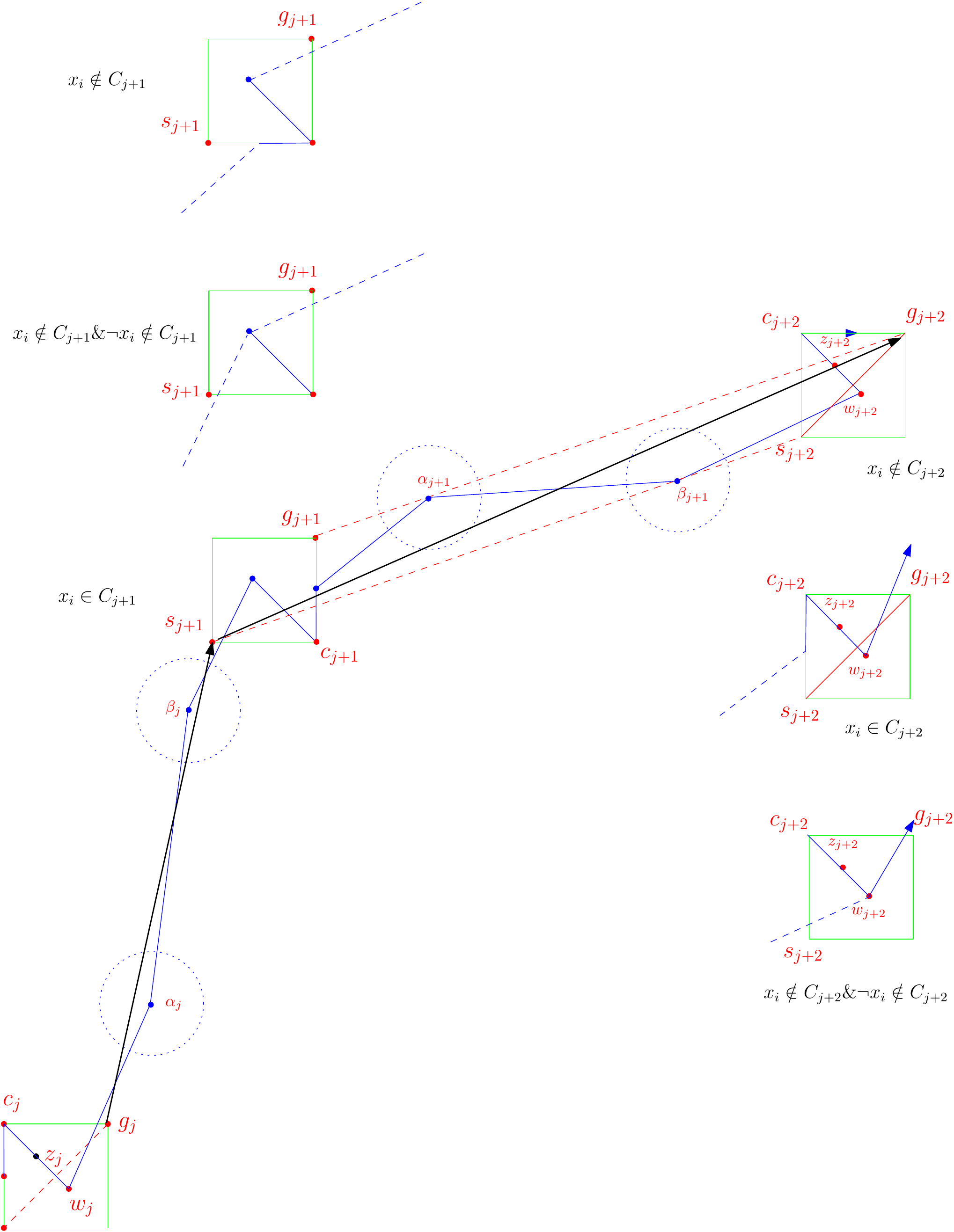}
	\caption{Proof of Lemma \ref{lemma:PathB}}
	\label{fig:PathB}
\end{figure}

\begin{lemma}\label{lemma:PathB}
Consider any subcurve $\cfev_i$, $1\le i \le n+2$,  
constructed through Lines \ref{l:mainstart} to \ref{l:subcurve} 
of Algorithm \ref{alg:reduction}. Let $B$ be the polygonal curve  $\langle u\gre_1\sma_2\gre_3\sma_4..v\rangle$. Then, $\distF(\cfev_i,B) \le 1$.
\end{lemma}

\begin{proof}

Consider two point objects $\CO_L$ and $\CO_B$ 
traversing $\cfev_i$ and $B$, respectively (Figure \ref{fig:pathAExample}b depicts an instance of $\cfev_i$ and $B$).
To prove the lemma, we show that $\CO_L$ and $\CO_B$ can walk along
their respective curves, from beginning to
the end, while keeping distance $1$ to each other. 

The base case of induction holds as follows 
(see Figure \ref{fig:PathBBaseCase} for an illustration).
Table \ref{tab:BaseCasePathB} lists  pairwise locations of 
$\CO_L$ and $\CO_B$, where the distance of each pair is less or equal to $1$.
Therefore, $\CO_B$ can walk from $u$ to $g_1$ while
keep distance one to $\CO_L$.

\begin{table}[h]
\centering
\begin{tabular}{ r | l | l  }
if $x_i \in C_1$   & location of $\CO_B$ & location of $\CO_L$  
 \\
\hline
    
&  $u$ & $u$  \\
&  				 $h_1$ s.t.  $\| h_1\mu_1 \| \le \eps$ & $\mu_1 = (-4,-1)$\\

				&    $h_2 = \Dir{ug_1} \dashv	 \Seg{s_1c_1} $ & $\mu_2 = M(s_1c_1)$   \\
&  				 			      $h_2 $ & $c_1$  \\
&    $\Seg{ug_1} \dashv \Seg{c_1w_1}$ &  $\Seg{ug_1} \dashv \Seg{c_1w_1}$ \\
&  				 		     $w_1 \perp ug_1 $ &   $w_1$  \\
& $g_1 $  & $\Seg{w_1\alpha_1}	\dashv \Seg{c_1g_1}$	      \\

\hline
if $\neg x_i \in C_1$ 
&  $u$ & $u$  \\
&  				 $h_1$ s.t.  $\| h_1\mu \| \le \eps$ & $\mu_1 = (-4,-1)$\\
& $h_2 = \Dir{ug_1} \dashv	 \Seg{s_1c_1} $ &  $\mu_2 = \Dir{\mu_1w_1} \dashv	 \Seg{s_1c_1} $\\
& $\Seg{ug_1} \dashv \Seg{c_1w_1}$ & $w_1$ \\
&  &  $c_1$\\
& $g_1$ &  $M(c_1g_1)$\\

\hline
if $x_i \notin C_{1} \& \neg x_i \notin C_{1}$ & $u$ & $u$  \\
&  				 $h_1$ s.t.  $\| h_1\mu \| \le \eps$ & $\mu_1 = (-4,-1)$\\
& $h_2 = \Dir{ug_1} \dashv	 \Seg{s_1c_1} $ &  $\mu_2 = \Dir{\mu_1w_1} \dashv	 \Seg{s_1c_1} $\\
& $\Seg{ug_1} \dashv \Seg{c_1w_1}$ & $w_1$ \\
&  &  $c_1$\\
&  &  $w_1$\\
& $g_1 $  & $\Seg{w_1\alpha_1}	\dashv \Seg{c_1g_1}$	      \\

\end{tabular}
\vspace{0.2 in}
\caption{Pairwise location of $\CO_B$ and $\CO_L$, to prove the base case of induction in Lemma \ref{lemma:PathB} }
\label{tab:BaseCasePathB}
\end{table}

Assume inductively that $\CO_L$ and $\CO_B$ have feasibly walked along 
their respective curves, until $\CO_B$ reached $g_j$.
Then, as the induction step, 
we 
show that
$\CO_B$ can walk to $s_{j+1}$ and then to $g_{j+2}$ 
, while keeping distance $1$ to $\CO_L$.
This is shown in Table \ref{tab:PathB}
 (see Figure \ref{fig:PathB} for an illustration).

\begin{table}[h]
\centering
\begin{tabular}{ r | l | l  }
  & location of $\CO_B$ & location of $\CO_L$  
 \\
\hline
   if $x_i \in C_j$  & $g_j$ & $\Dir{\alpha_{j-1}w_j} 	\dashv \Seg{c_jg_j}$ \\
   if $\neg x_i \in C_j$  & $g_j$ & $M(\Seg{c_jg_j})$\\
   if $x_i \notin C_j \& \neg x_i \notin C_j$  & $g_j$ & $\Dir{\alpha_{j-1}w_j}	\dashv \Seg{c_jg_j}$\\

\hline
	&  $h_3$ s.t.  $\| h_3\alpha_j \| \le \eps$ & $\alpha_j$\\
	&  	$h_4$ s.t.  $\| h_4 \beta_j \| \le \eps$ & $\beta_j$\\

\hline
if $x_i \in C_{j+1}$ &  				 $s_{j+1}$  & $\Seg{\beta_jw_{j+1}}	\dashv \Seg{c_{j+1}s_{j+1}}$\\
& $w_{j+1} \perp \Seg{s_{j+1}g_{j+2}}$ & $w_{j+1}$\\
& $z_{j+1}$ & $z_{j+1}$\\

&  $\Seg{g_{j+1}c_{j+1}}	\dashv \Seg{s_{j+1}g_{j+2}}$ & $c_{j+1}$ \\
&  & $M(\Seg{c_{j+1}g_{j+1}})$
 \\

if $\neg x_i \in C_{j+1}$ &  				 $s_{j+1}$  & $M(\Seg{s_{j+1}c_{j+1}})$\\

 & $z_{j+1}$ & $c_{j+1}$ \\
 &  & $w_{j+1}$ \\

 &  $\Seg{g_{j+1}c_{j+1}}	\dashv \Seg{s_{j+1}g_{j+2}}$ & $\Seg{g_{j+1}c_{j+1}}	\dashv \Seg{w_{j+1}\alpha_{j+1}}$ \\

if $x_i \notin C_{j+1} \& \neg x_i \notin C_{j+1}$ &  				 $s_{j+1}$  & $\Seg{\beta_jw_{j+1}}	\dashv \Seg{c_{j+1}s_{j+1}}$\\
& $z_{j+1}$ & $w_{j+1}$\\
 & & $c_{j+1}$ \\
 & & $w_{j+1}$ \\
&  $\Seg{g_{j+1}c_{j+1}}	\dashv \Seg{s_{j+1}g_{j+2}}$ & $\Seg{g_{j+1}c_{j+1}}	\dashv \Seg{w_{j+1}\alpha_{j+1}}$ \\

\hline

	&  $h_5$ s.t.  $\| h_5\alpha_{j+1} \| \le \eps$ & $\alpha_{j+1}$\\
	&  	$h_6$ s.t.  $\| h_6 \beta_{j+1} \| \le \eps$ & $\beta_{j+1}$\\
\hline

if $ x_i \in C_{j+2}$ &  		$\Seg{s_{j+2}c_{j+2}}	\dashv \Seg{s_{j+1}g_{j+2}}$ 		   & $M(\Seg{s_{j+2}c_{j+2}})$\\

 & $z_{j+2}$ & $c_{j+2}$ \\
 &  & $w_{j+2}$ \\

 &  $g_{j+2}$ & $\Seg{g_{j+2}c_{j+2}}	\dashv \Seg{w_{j+2}\alpha_{j+2}}$ \\

if $\neg x_i \in C_{j+2}$ &  				$\Seg{s_{j+2}c_{j+2}}	\dashv \Seg{s_{j+1}g_{j+2}}$ 		   &   $\Seg{s_{j+2}c_{j+2}}	\dashv \Seg{\beta_{j+1}w_{j+2}}$   \\
 & $z_{j+2}$ & $w_{j+2}$ \\
 &  & $c_{j+2}$ \\

 &   $g_{j+2}$ & $M(\Seg{c_{j+2}g_{j+2}})$ \\

if $x_i \notin C_{j+2} \& \neg x_i \notin C_{j+2}$ &  				
$\Seg{s_{j+2}c_{j+2}}	\dashv \Seg{s_{j+1}g_{j+2}}$ 		   &   $\Seg{s_{j+2}c_{j+2}}	\dashv \Seg{\beta_{j+1}w_{j+2}}$  \\
& $z_{j+2}$ & $w_{j+2}$\\
 & & $c_{j+2}$ \\
 & & $w_{j+2}$ \\
&  $g_{j+2}$ & $\Seg{g_{j+2}c_{j+2}}	\dashv \Seg{w_{j+2}\alpha_{j+2}}$ \\

\end{tabular}
\caption{Distance between pair of points is less or equal to one}
\label{tab:PathB}
\end{table}

Finally, if $k$ is an odd number, then  
$\Dir{g_kv}$ is the last segment along $B$, otherwise, 
$\Dir{s_kv}$ is the last one. In any case, 
that edge crosses  circle $\CB(\eta,1)$, where $\eta$ is the last vertex of 
$\cfev_i$ before $v$ (point $\eta$ is computed after the condition checking 
in line \ref{l:ComputeV} of 
Algorithm \ref{alg:reduction}). Therefore, 
 $\CO_B$ can walk to $v$, while keeping distance $1$ to $\CO_L$.

\qed

\end{proof}

\begin{figure}
	\centering
	\includegraphics[width=0.9\columnwidth]{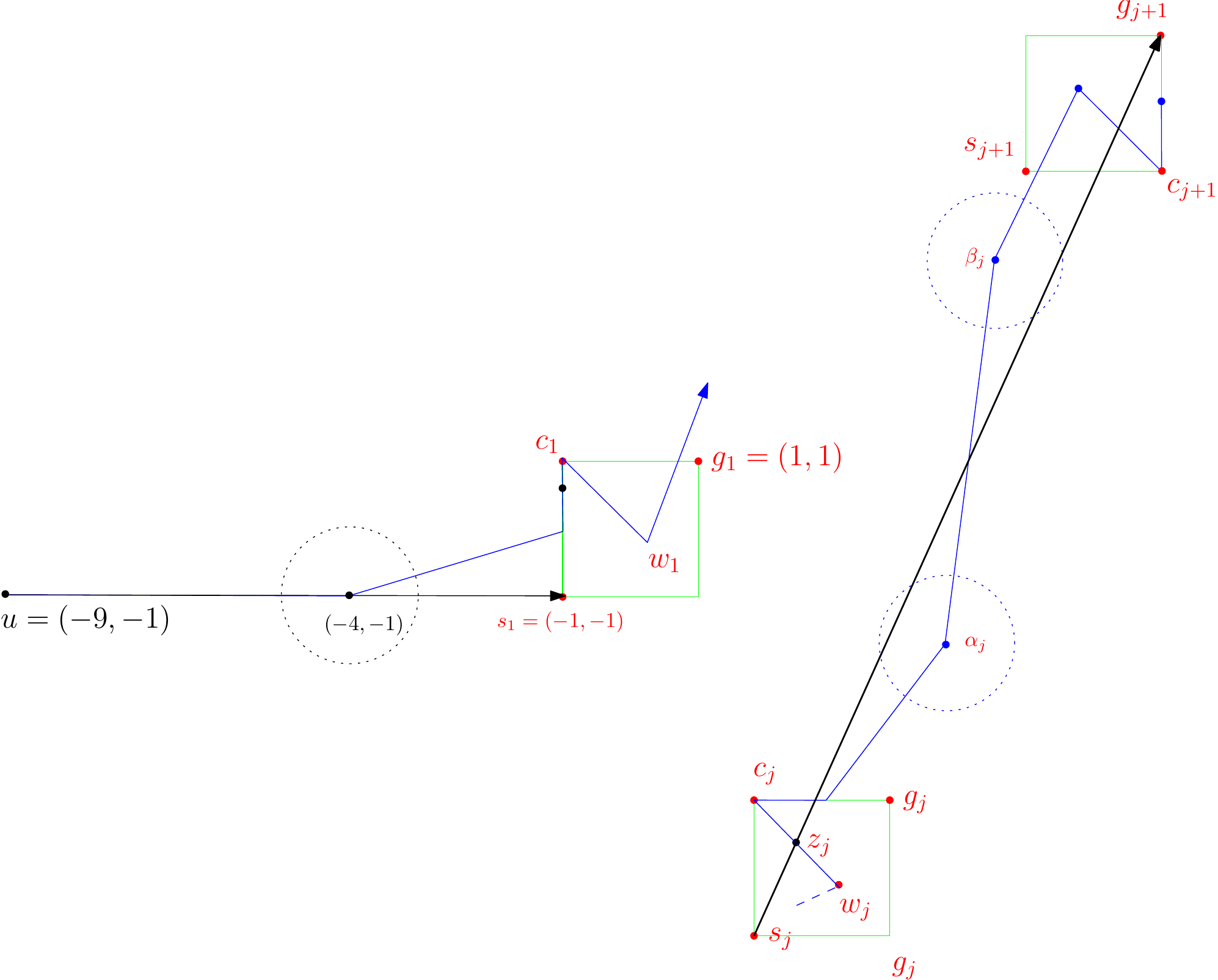}
	\caption{Proof of Lemma \ref{lemma:NoSwitchFromAtoB}}
	\label{fig:noswitch}
\end{figure}

\begin{lemma}\label{lemma:NoSwitchFromAtoB}
Consider any curve $\cfev_i \subset P $, $1\le i \le n+2$. Imagine that  
a point object $\CO_L$ is walking from $u$ to $v$ on $\cfev_i$. Furthermore, imagine
two point objects $\CO_A$ and $\CO_B$ which are walking on curves $A$ and $B$ 
(from Lemmas \ref{lemma:PathA} and \ref{lemma:PathB}), respectively,
while keeping distance $1$ to $\CO_L$. 
If $\CO_A$ goes to any vertex of $B$ or $\CO_B$ goes to any vertex of $A$, 
then they loose distance $\le 1$ to $\CO_L$.

\end{lemma}

\begin{proof}
Let 
$cl_i$ refer to 
points $\{s_i,g_i,c_i\}$.
Notice that we have placed the $cl_{i+1}$ points far enough from 
the $cl_{i}$ points so that 
no curve can go to $cl_{i+1}$
and come back to $cl_i$ and stay 
in \Frechet distance 1 to $\cfev_i$.
Therefore, to prove the lemma, 
we only focus on two consecutive c-squares.
We show that no subcurve $l' \subseteq \cfev_i$ exists such 
that (for an illustration, see Figure \ref{fig:noswitch}) :

\begin{itemize}

\item $\distF(l',\Dir{\sma_j\gre_j}) \le 1$ because:

for all $j$, $1 \le j \le k$, point $c_j$ is always a vertex of $\cfev_i$. 
A point on $\cfev_i$ at distance 1 
to $\sma_j$ lies before $c_j$ 
in direction $\Dir{\cfev_i}$, 
while a point on $\cfev_i$ at distance 1 
to point $\gre_j$ lies after $c_j$ in direction $\Dir{\cfev_i}$.
Since $dist(c_j, \Seg{\sma_j\gre_j}) >1$, 
no subcurve $l' \subseteq \cfev_i$ exists such that 
$\distF(l',\Dir{\sma_j\gre_j}) \le 1$.

\item $\distF(l',\langle \sma_jc_j\gre_{j}\rangle) \le 1$ or $\distF(l',\langle\gre_jc_j\sma_{j}\rangle) \le 1$, because:

For all $j$, $1 \le j \le k$, $w_j$
is a vertex of $\cfev_i$. 
A point on $\cfev_i$ at distance 1 
to $\sma_j$ lies before $w_j$ 
in direction $\Dir{\cfev_i}$, 
while a point on $\cfev_i$ at distance 1 
to point $\gre_j$ lies after $w_j$ in direction $\Dir{\cfev_i}$.
Since $dist(w_j, \Seg{\sma_jc_j}) >1$ and 
$dist(w_j, \Seg{\gre_j\gre_j}) >1$,
no subcurve $l' \subseteq \cfev_i$ exists such that 
$\distF(l', \langle \sma_jc_j\gre_{j}\rangle ) \le 1$.
Similarly,  no subcurve $l' \subseteq \cfev_i$ exists such that 
$\distF(l', \langle\gre_jc_j\sma_j\rangle ) \le 1$.

\item $\distF(l',\langle \sma_j\sma_{j+1} \rangle) \le 1$ or  $\distF(l',\langle\gre_j\gre_{j+1} \rangle) \le 1$ because:

Vertex $\alpha_{i}$ of $\cfev_{i}$
guarantees the first part as $dist( \alpha_{i},\Seg{\sma_j\sma_{j+1}}) > 1 $, 
and vertex $\beta_{i}$ of $\cfev_{i}$
guarantees the second part, 
as $dist( \beta_{i},\Seg{\gre_j\gre_{j+1}}) > 1$.

\item $\distF(l',\langle c_jc_{j+1} \rangle ) \le 1$,  because $dist( \alpha_{i},\Seg{c_jc_{j+1}}) > 1$

\item $\distF(l',\langle uc_1 \rangle) \le 1$, because $dist( (-4,-1),\Seg{uc_1}) >1$

\item  $\distF(l',\langle c_j\gre_{j+1} \rangle) \le 1$, because $dist( \alpha_i,\Seg{c_j\gre_{j+1}}) >1$

\item  $\distF(l',\langle c_j\sma_{j+1} \rangle ) \le 1$, because $dist( \alpha_i,\Seg{c_j\sma_{j+1}}) >1$

\item  $\distF(l',\langle c_kv \rangle) \le 1$, because $dist( \eta,\Seg{c_kv}) >1$ 
\end{itemize}

\end{proof}

\noindent To establish the correctness of
our reduction algorithm, 
from now on, we define:
$(a_i = \sma_i, b_i = \gre_i)$,  
 when $i$ is an odd number, 
and
$(a_i = \gre_i, b_i = \sma_i)$, 
when $i$ is an even number, for $1 \le i \le k$.

\begin{lemma}\label{lemma:ABCanSeeC}
Consider the curve $A = \langle ua_1a_2a_3 \dots a_k v \rangle$ from Lemma \ref{lemma:PathA}. Let 
$A_1$ be a subcurve of $A$ which starts at $u$ and ends at $a_j$, $1 \le j \le k$.
Furthermore, let $A_2$ be a subcurve of $A$ which starts at $a_j$ and ends at $v$.
For any curve $\cfev_i$ , $1\le i \le n+2$,
if $x_i \in C_j$, 
$\distF(A_1 \ap c_j \ap A_2, \cfev_i) \le \eps$. 
Similarly, consider the curve $B = \langle ub_1b_2b_3 \dots b_k v \rangle$ from Lemma \ref{lemma:PathB}. Let 
$B_1$ be a subcurve of $B$ which starts at $u$ and ends at $b_j$, $1 \le j \le k$.
Furthermore, let $B_2$ be a subcurve of $B$ which starts at $b_j$ and ends at $v$.
For any curve $\cfev_i$ , $1\le i \le n+2$,
if $\neg x_i \in C_j$,   
$\distF(B_1 \ap c_j \ap B_2, \cfev_i) \le \eps$. 
\end{lemma}

\begin{proof}
When $x_i$ appears in clause $C_j$, point $z = M(\Seg{c_ja_j})$ is 
a vertex of $\cfev_i$. 
Since $\| c_ja_j \| = 2$ and $z$ is the  midpoint of 
 $\Seg{c_ja_j}$,
$\CO_L$ can wait at $z$ while $\CO_A$ visits $c_j$.
Therefore, as the lemma states, 
we can cut curve $A$ at vertex $a_j$,
add two edges $\Dir{a_jc_j}$
and then $\Dir{c_ja_j}$ to $A$,
and continue with the same 
curve $A$
from $a_j$ to $A$'s endpoint. 
For the 
modified $A$, still $\distF(A, \cfev_i) \le \epsilon$ holds. 

When $\neg x_i$ appears in clause $C_j$, point $z = M(\Seg{c_jb_j})$ is 
a vertex of $\cfev_i$. 
Since $\| c_jb_j \| = 2$ and $z$ is the  midpoint of 
 $\Seg{c_jb_j}$,
$\CO_L$ can wait at $z$ while $\CO_A$ visits $c_j$
and comes back to $b_j$.
Therefore, as the lemma says, 
we can cut curve $B$ at vertex $b_j$,
add two edges $\Dir{b_jc_j}$
and then $\Dir{c_jb_j}$ to $B$,
and continue with the same 
curve $B$
from $b_j$ to $B$'s endpoint. 
For the 
modified $B$, still $\distF(B, \cfev_i) \le \epsilon$ holds.

\end{proof}

\begin{lemma}\label{lemma:NOTABCanSeeC}
Consider curve $A$ (respectively, $B$) from previous lemma.
For any curve $\cfev_i$, $1\le i \le n+2$,
when $ x_i \notin C_j$ and $\neg x_i \notin C_j$,   
curve $A$ (resp., $B$) can not be modified to visit $c_j$.
\end{lemma}
\begin{proof}
This holds because $dist(w_j,\Seg{a_jc_j}) >1$
and $dist(w_j,\Seg{b_jc_j}) >1$.

\end{proof}


\vspace{0.1 in}

\begin{theorem}
Given a formula $\phi$ with $k$ clauses $C_1, C_2, \dots, C_k$ and $n$ variables $x_1, x_2 \dots, x_n$,
as input, let curve $P$ and pointset $\pset$ be the output of Algorithm \ref{alg:reduction}. 
Then, $\phi$ is satisfiable iff a 
curve $Q \in Curves(S)$ exists such that 
$\distF(P,Q) \le 1$.
\end{theorem}

\begin{proof}


For $(\Rightarrow)$: 
Assume that  formula $\phi$ is satisfiable. 
In Algorithm \ref{alg:buildQ}, we show that 
knowing the truth value of the literals in $\phi$, 
we can build a curve $Q$ which 
visits every point in $\pset$ and $\distF(P,Q) \le 1$.

\begin{algorithm} [h]
\caption {{\sc Build a feasible curve $Q$ }} 
\label{alg:buildQ}
\algsetup{indent=1.5em}
\begin{algorithmic}[1]	
		\baselineskip=0.9\baselineskip
	\REQUIRE  Truth table of variables $x_1, x_2, \dots, x_n$ in $\phi$

	\STATE $Q \eq \emptyset$
	\STATE $Q \eq Q \ap t$ \label{l:startPoint}
	 
	\FOR {$i=1$ to $n$}   
	\IF {$(x_i = 1)$}
	\STATE $\pi \eq \langle ua_1a_2a_3 \dots a_kv \rangle$
	\FORALL {$C_j$ clauses, if $x_i \in C_j$ }
	\STATE let $\pi_1 $ be  subcurve of $\pi$ from $u$	 to $a_j$
	\STATE let $\pi_2 $ be  subcurve of $\pi$ from $a_j$	 to $v$
	\STATE $\pi \eq \pi_1 \ap c_j \ap \pi_2$  \label{l:visitCone}
	\ENDFOR
	\STATE $Q \eq Q \ap \pi$ \label{l:x1}
	\ELSE 	
	\STATE $\pi \eq \langle ub_1b_2b_3 \dots b_kv \rangle$
	\FORALL {$C_j$ clauses, if $\neg x_i \in C_j$ }
	\STATE let $\pi_1 $ be  subcurve of $\pi$ from $u$	 to $b_j$
	\STATE let $\pi_2 $ be  subcurve of $\pi$ from $b_j$	 to $v$
	\STATE $\pi \eq \pi_1 \ap c_j \ap \pi_2$ \label{l:visitCzero}
	\ENDFOR
	\STATE $Q \eq Q \ap \pi$ \label{l:x0}

	\ENDIF
	\STATE $Q \eq Q \ap t$
	\ENDFOR
	\STATE $Q \eq Q \ap \langle ua_1a_2a_3 \dots a_kv \rangle$\label{l:nplusone}
	\STATE $Q \eq Q \ap t$

	\STATE $Q \eq Q \ap \langle ub_1b_2b_3 \dots b_kv \rangle$\label{l:nplustwo}
	\STATE $Q \eq Q \ap t$  \label{l:endPoint}

	\STATE {\bf return} {\sc Q}  
\end{algorithmic}
\end{algorithm}


First, we show $\distF(P,Q) \le 1$, where $Q$
is the output curve of Algorithm \ref{alg:buildQ}.
Recall that by Algorithm \ref{alg:reduction}, 
curve $P$ includes $n$ subcurves $\cfev_{i}$ each corresponds 
to a variable $x_i$. 
Both curves $P$ and $Q$ start and end at a same point $t$.
For each curve $\pi$ which is appended to $Q$ 
in the $i$-th iteration of Algorithm \ref{alg:buildQ} 
(Line \ref{l:x1} or Line \ref{l:x0}), 
$\distF(\pi,\cfev_i) \le 1$  by Lemma \ref{lemma:ABCanSeeC}. 
Notice that $P$ also includes two additional subcurves $\cfev_{n+1}$ and $\cfev_{n+2}$ 
whereas there is no variable $x_{n+1}$ and $x_{n+2}$ in formula $\phi$. 
These two curves are to resolve two special cases: 
when all variables $x_i$ are 1,  no $\neg x_i$ appears in $\phi$,
and when all variables $x_i$ are 0,  
no $x_i$ appears in $\phi$.
Because of these two curves, 
we added two additional curves in line \ref{l:nplusone}
and \ref{l:nplustwo} to $Q$. Finally, by  Observation 
\ref{obs:concat}, $\distF(P,Q) \le 1$.

Next, we show that curve $Q$ visits every point in $S$. First of 
all, by the curves added to $Q$ 
in Line \ref{l:nplusone} and \ref{l:nplustwo}, 
all $a_j$ and $b_j$, $1\le j \le k$, in $S$ will be visited. 
It is sufficient to show that $Q$ will visit all $c_j$ points in $S$  as well.
Since  formula $\phi$ is satisfied, every clause $C_i$ in $\phi$ must be satisfied 
too. Fix clause $C_j$. At least one of the literals in $C_j$
must have a truth value $1$. If $x_i \in C_j$ and $x_i = 1$, 
then by line \ref{l:visitCone}, curve $Q$ visits $c_j$.
On the other hand, if $\neg x_i \in C_j$ and $x_i = 0$, 
by Line \ref{l:visitCzero}, curve $Q$ visits $c_j$. We conclude that 
curve $Q$ is feasible.

Now $(\Leftarrow)$ part:

Let $Q$ be a feasible curve with respect to $P$ and pointset $\pset$.
Notice that curve $P$ consists of $n$ subcurves $\cfev_i$, 
$1 \le i \le n $, where each corresponds to one variable $x_i$. 
From the configuration of each $\cfev_i$ in c-squares, 
one can easily construct formula $\phi$ with 
all of its clauses and literals. 


Imagine two point objects $\CO_Q$ 
and $\CO_P$ walk on $P$ and $Q$, respectively. 
We find the truth value of variable $x_i$ in the formula
by looking at the path that $\CO_Q$ takes to stay in \Frechet distance 1 to $\CO_P$, 
when $\CO_P$ walks on curve $\cfev_i$ corresponding to $x_i$.
If $\CO_Q$ takes path $A$ from Lemma \ref{lemma:PathA} 
while $\CO_P$ is walking on  $\cfev_i$, then $x_i = 1$. But if $\CO_Q$ takes path $B$ from Lemma \ref{lemma:PathB} 
while $\CO_P$ is walking on  $\cfev_i$, then $x_i = 0$. 
Object $\CO_Q$ decides between path $A$ or $B$,  when both $\CO_Q$ and $\CO_P$ are at point $u$. 
Lemma \ref{lemma:NoSwitchFromAtoB} ensures that  
once they start walking, 
$\CO_Q$ can not change its path from $A$ to $B$ 
or from $B$ to $A$. 
Therefore, the truth value of a variable $x_i$ is consistent.

The only thing left to show is the reason that formula $\phi$ is satisfiable. 
It is sufficient to show every clause of $\phi$ is satisfiable. 
Consider any clause $C_j$.
Since curve $Q$ is feasible, 
it uses every point in $\pset$.  
Assume w.l.o.g. that $\CO_Q$ visits $c_j$ 
when $\CO_P$ is walking along curve $\cfev_i$.  
By Lemmas 
\ref{lemma:NoSwitchFromAtoB} and \ref{lemma:ABCanSeeC},
this only happens when either ($x_i$ appears in $C_i$ and $x_i = 1$)
or ($\neg x_i$ appears in $C_i$ and $x_i =0$). 
Therefore, $C_j$ is satisfiable.


%
%

%
%
%
%

The last ingredient of the NP-completeness proof is
to show that the reduction takes polynomial time.  
One can easily see that Algorithm \ref{alg:reduction}
has running time $O(nk)$, 
where $n$ is the number of variables in 
the input formula with $k$ clauses.

\end{proof}

\subsection{Implementation Results}
To show the simplicity of our reduction algorithm, we have implemented it
in Java.  The figures in this chapter are all generated by our program. 
Our test case, as 
an input to the program, is a formula $\phi$ with four clauses. The output is 
three sets $S$, $L$ and $C$ as follows.

Set $S$ is a pointset computed by 
Algorithm \ref{alg:reduction}.
Since $\phi$ has four clauses, 
$S$
contains the following points:

$S = \{ s_1,g_1,c_1, s_2,g_2,c_2, s_3,g_3,c_3, s_4,g_4,c_4 , u, v, t \}$.

Set $L$ is a set of curves, 
where each of it is a configuration of 
$\cfev_i$ in the reduction algorithm.
Choosing  a formula with four clauses as an input, 
enables us to check all possible configurations of curve $\cfev_i$ built by Algorithm \ref{alg:reduction}. 
Let  $x_i$ be a variable in  formula
$\phi$.
Since $x_i$ or $\neg x_i$ or none could appear in a clause, and the formula has four clauses,  set $L$ contains 81 curves $\cfev_i$. 

Set $C$ contains all possible curves $\alpha$, 
each built in this way: 
$\alpha$ starts from point $u$, 
goes through arbitrary points from $\{g_1,c_1,s_1\}$, 
then to arbitrary points in $\{g_2,c_2,s_2\}$, next
to arbitrary points from $\{g_3,c_3,s_3\}$, and lastly 
from $\{g_4,c_4,s_4\}$
and at the end, 
$\alpha$ ends at $v$.
Therefore, set $C$ has almost 1,000,000,000 polygonal curves.

Let $\alpha$ be any curve in $C$ 
and $\ell$ be any curve in $L$. 
We compute in our program, the \Frechet distance between every curve $\alpha$ and $\ell$.
Notice that $C$ has huge amounts of curve data. We implemented our 
program in an efficient way so 
that we could do this computation in 
a fair amount of time. First, 
all 81 curves $\ell$ are computed and then, 
by computing each $\alpha$ in $C$, 
we compute 81 \Frechet distances 
$\distF(\alpha,\ell)$. Therefore, in total, 
almost $81 \times 1,000,000,000$ 
\Frechet distances have been computed by our program. 
The experiment is 
performed on four machines in parallel, each 
has an Intel(R) Core(TM) i7 CPU 2.67GHz and 12GB RAM.

The results show that in all cases,
$\distF(\alpha,\ell) > 1$ 
except for the following cases:
\vspace{0.1 in}

Case I:  $\alpha = \langle u,s_1,g_2,s_3,g_4,v \rangle$, then 
$\distF(\alpha,\ell) \le  1$, for any curve $\ell$ in $L$.

Case II:  $\alpha = \langle u,g_1,s_2,g_3,s_4,v \rangle$, then 
$\distF(\alpha,\ell) \le  1$, for any curve $\ell$ in $L$.
 
Case III:  $\alpha = \langle u,g_1,c_1,g_1,s_2,g_3,s_4,v \rangle$, then 
$\distF(\alpha,\ell) \le  1$, for $\ell$ corresponding to the case 
where $\neg x_i$ appeared in the 
first clause.

Case IV: $\alpha = \langle u,s_1,c_1,s_1,g_2,s_3,g_4,v \rangle$, 
then 
$\distF(\alpha,\ell) \le  1$, for $\ell$ corresponding to the case 
where $x_i$ appeared in the 
first clause 
and so on for other occurrence of 
variable $x_i$ in other clauses.

\vspace{0.1 in}
Case I confirms Lemma \ref{lemma:PathA},
case II confirms Lemma \ref{lemma:PathB},
cases III and IV confirm Lemma \ref{lemma:ABCanSeeC},
and all together confirm Lemma \ref{lemma:NoSwitchFromAtoB}.

\section{Conclusions}
\label{sec:conc}
In this chapter, we investigated the problem of deciding whether a polygonal curve through a given pointset $\pset$ exists, which 
visits every point in $\pset$ and
 is in $\eps$-\Frechet distance 
to a curve $P$. We showed that this problem is NP-complete. 


\chapter{Conclusions and Open Problems}
\label{ch:OpenProblems}

In the first part of 
this thesis, we introduced 
a new generalization of the well-known \Frechet distance between two polygonal
curves, and provided an efficient algorithm for computing it. In our  variant of the \Frechet problem, the speed of 
traversal along each segment of the curves is restricted to be within a specified range. This setting is more realistic than the classical \Frechet distance setting, specially in GIS applications.
We presented an efficient algorithm to solve the decision problem
in $O(n^2 \log n)$ time, which
led to an $O(n^3 \log n)$ time algorithm for finding 
the exact value of the \Frechet distance with speed limits.
Getting a better running time than  $O(n^3 \log n)$ remains open. 

We also studied speed-constrained \Frechet distance in the case where the curves
are located inside a simple polygon.
Several open problems arise from our 
work.
It is interesting to consider speed limits in other variants of the \Frechet distance studied in the literature,
such as the \Frechet distance between two curves lying 
on a convex polyhedron~\cite{AnilFrechet}, or on a polyhedral surface~\cite{WenkC08a}. 

Another open problem, in the context of the first 
part of the thesis, is whether 
our results can be applied in matching planar maps, where the objective is 
to find a path in a given road network that is as close as possible to a vehicle trajectory. 
In~\cite{AltERW03a}, the traditional \Frechet metric is used to match 
a trajectory to a road network. 
If the road network is very congested,  
the \Frechet distance with speed limits introduced in this thesis seems to find a more realistic path in the road network, 
close to the trajectory of the vehicle.
It is also interesting to extend our variant of \Frechet distance 
to the setting where the speed limits on the segments of the curves change as functions over time.

In the second part of this thesis, 
we introduced a data structure, 
called the free space map, which can be used 
as an alternative to the free-space diagram 
in applications related to 
\Frechet distance. 
Our data structure has the same size and construction time as the standard
free-space diagram, and 
encapsulates all the information available in that diagram, 
yet it is capable of answering more general types of queries efficiently. 
One open problem in the 
context of the second part of the thesis
is to quickly answer dynamic queries. Such queries include:
removing  vertices from the input curves 
or adding new vertices to the curves. How fast 
 could we solve the corresponding decision problems or
recompute the \Frechet distance.
\REM{
In the following, we explain the difficulty of constructing such data structure:

There are several interesting dynamic problems related to \Frechet metric
that to our knowledge, have not yet been addressed  in the literature:

Let  $P$ and $Q$ be 
two polygonal curves of length $n$ and $m$, respectively.
By spending $O(nm \log nm)$  , one can compute the \Frechet distance between   $P$ and $Q$ (i.e. $\distF(P,Q)$). 
Now, we want to know how fast we can recompute $\distF(P,Q)$, after:
  splitting some edges of $P$ (resp., $Q$),
  or deleting some vertices  of polyline $P$ (resp., $Q$),
  or merging some consecutive vertices of $P$ (resp., $Q$) into one vertex.

%


Consider the free-space diagram $\BNM$ corresponding to curves $P$ and $Q$, and input parameter $\epsilon \ge 0$.
Let $\row{i} = \Feps \cap y = i, i \in \{1,\dots,m+1\}$. 
Notice that given an arbitrary point $p \in \row{i}$, the number of intervals in $\row{j}$, $j > i$, which are monotonically reachable from $p$, 
is linear in the number of cells in one row. 
Let $\RSet{i}{j}$ denote the set of intervals in $\row{j}$ which are reachable from 
every point $p \in \row{i}$. We observed that for any $i$ and $j$, $j>i$, 
$\RSet{i}{j}$ has linear complexity (proportional to the number of cells in a row of $\BNM$). 
We store the reachability information in $\BNM$ in a data structure 
called {\em free space tree} in order to quickly update that information 
after performing the above operations. 




The free space tree is a balanced binary search tree, denoted by $T_{\epsilon}$, 
which we construct on the top 
of the free-space diagram (See Figure \ref{fig:freeSpaceTree}.a). 
Each leaf $v_i$ in the tree corresponds to $\row{i}$ and
stores  reachability set $\RSet{i}{i+1}$.
The tree is built bottom-up from the leaves to the root:
at an internal node $w$, with leaf $v_i$ as the left child
and leaf $v_j$ as the right child, the reachability information $\RSet{i}{i+1}$
is {\em merged} with $\RSet{j}{j+1}$ and consequently, set $\RSet{i}{j+1}$
is stored at node $w$.
The merge process is recursively continued at each internal node of the tree
by merging the reachability information stored at the left and right child of 
that node, until the root of the tree is reached.
Likewise, at the root of the tree, 
$\RSet{1}{\floor{\frac{m}{2}}}$ (stored at the left child of the root) 
is  merged with $\RSet{\floor{\frac{m}{2}}}{m+1}$ (stored at the right child of it),
and the reachability set $\RSet{1}{m+1}$ is stored  at the root. 
Having built the tree, one can report all reachable intervals in $\row{i+1}$ from a point $p \in \row{1}$ spending $O(1) + O(k)$ time ($k$ is the number of those reachable intervals).
In fact, finding the leftmost and the rightmost point in $\row{m+1}$ reachable from $p \in \row{1}$ can be done in constant time.


In the following, we describe the difficulty of constructing $T_{\epsilon}$ by showing the merge at 
the root of the tree as an example of the challenge in merge operation at internal nodes 
(see Figure \ref{fig:freeSpaceTree}.b for the illustration).
Let  a {\em feasible }interval be the maximal set of white points on the bottom side of a cell. 
Assume the free-space diagram is split in half 
at $\row{\floor{\frac{m}{2}}}$. Consider a feasible interval $I_1 \in \row{\floor{\frac{m}{2}}}$ and 
two pointers $\ell_{m+1}(I_1)$ and $r_{m+1}(I_1)$, where the pointers store the leftmost and 
the rightmost point in $\row{m+1}$, reachable from some points in $I_1$, respectively. 
Imagine another feasible interval $I_2 \in \row{1}$
where  $p = \ell_{m+1}(I_2)$ and $p \in I_1$.
Now, if we could determine the leftmost point in $\row{m+1}$ reachable from $p$, then 
we would be able to determine $\ell_{m+1}(I_2)$ in constant time
and consequently, perform the merge at the root of the tree in linear time. 

Our target is to obtain a linear running time
(or near linear) for the merge operation at internal nodes. 
If we could achieve a linear merge, then we can exploit the tree for computing 
the \Frechet distance after different dynamic operations on curves. 
As an example, suppose that we delete a vertex 
from curve $Q$. This means that two rows are deleted from the free-space diagram. 
Using the free space tree, we update the reachability information 
in $O(n \log n )$ time rather than $O(n^2)$ time. Therefore, 
we could decide whether the target point in $\BNM$ is monotonically reachable or not. 
}


%
 
\REM{
\begin{figure}[h]
	\centering
	\includegraphics[width=0.8\columnwidth]{Pics/Freespacetree2}
	\caption{ The free space tree is built on the top of the free-space diagram}
	\label{fig:freeSpaceTree}
\end{figure}
}

\REM{
We would also want to investigate about 
computing \Frechet distance in the following setting.
Consider a polygonal curve $P$ corresponding to the trajectory 
of a moving vehicle. Assume that to each edge of $P$, a parameter {\em weight}
is associated, which shows the accuracy of the sampling around that edge.
Our goal is to compute the \Frechet distance between two curves 
taking into consideration the weights on the edges. 
In this setting, the free space inside each cell in the free-space diagram 
corresponding to the curves, shrinks. To solve the decision problem, 
one needs to first discover the properties of the free space inside each cell, 
whether it is convex or $xy$-monotone or if it is connected or not. 
Next, the boundary of the free space inside the cells must be computed.
Finally, the free-space diagram is searched to find a monotone path from the bottom-left corner to the top-right corner. 
}


%

In the third part of this thesis, 
we introduced a new variant of
\Frechet distance problem, 
called Curve-Pointset (CPM) Problem, and showed that
given a polygonal curve $P$ of $n$ segments and a set $S$ of $k$ points in $\IR^d$, 
a polygonal curve $Q$ through $S$ of size $O(\min\set{n,k})$ minimizing $\distF(P, Q)$ 
can be computed in $O(nk^2 \log (nk))$ time. 
Then, we proved that if curve $Q$	is required to visit every point in the pointset, 
the problem becomes NP-complete. 
It is interesting to study special cases of 
this problem when the input is a specific 
type of curve, e.g., an $xy$-monotone convex curve, 
or a monotone curve or a non-intersecting curve. 
Also, it will be interesting to study special cases 
of the All-Points CPM problem  under the
condition that each point in the pointset 
can be visited only once.

\REM{	
	At the end, we studied a special case of the last problem  when the input is a convex polygon. 
Many open problems arise from this
work.	An open problem which we would like to mention is whether a polynomial time algorithm could be designed for the case where the input	
to the All-points CPM problem
is some other special type of curves or polygons, for instance, a non-intersecting curve, or a monotone polygon.  
It would be also interesting to study special cases of All-Points CPM problem  under the condition that each point in the pointset can be visited only once or a weight is 
assigned to each point and 
a curve with minimum or maximum weight should be constructed.
}	

}

\bibliographystyle{abbrv}
\bibliography{abbrv,Proposal}

\end{document}